\documentclass[a4paper]{amsbook}

\usepackage{cite}
\def\BibTeX{{\rm B\kern-.05em{\sc i\kern-.025em b}\kern-.08em
    T\kern-.1667em\lower.7ex\hbox{E}\kern-.125emX}}
\makeindex
\usepackage[arrow, matrix, curve]{xy}
\usepackage{amsfonts}
\usepackage{paralist}
\usepackage[english]{babel}
\usepackage{amssymb}
\usepackage{amsthm}

\newtheorem{dummytheorem}{Dummy-Theorem}[section]
\newtheorem{lemma}[dummytheorem]{Lemma}
\newtheorem{theorem}[dummytheorem]{Theorem}
\newtheorem{proposition}[dummytheorem]{Proposition}
\newtheorem{corollary}[dummytheorem]{Corollary}
\newtheorem{example}[dummytheorem]{Example}

\theoremstyle{definition}
\newtheorem{definition}[dummytheorem]{Definition}
\newtheorem{remark}[dummytheorem]{Remark}
\numberwithin{section}{chapter}
\numberwithin{equation}{chapter}

\begin{document}

\thispagestyle{empty}
\voffset-20pt

\begin{center}






{\bf{\LARGE Universit\"at Paderborn \\[3mm]
Fakult\"at f\"ur Elektrotechnik, Informatik und Mathematik }}\\[3mm]
{\bf{\LARGE Universit\'e Lorraine-Metz \\[3mm]
\'Ecole Doctorale, IAEM Lorraine }}\\[5cm]
{\bf{\sf{\Huge A Dynamical Interpretation of Patterson-Sullivan Distributions
}}}\\[6cm]





{\Large
{\huge Dissertation}											\\[1cm]
{\huge 
Th\`ese de doctorat
}


vorgelegt von															\\[0.5cm]
{\huge Jan Emonds}									\\[0.5cm]
Advisors:  {Prof.~Dr.~J.~Hilgert} 	\\[1mm]
\hspace*{4em}  {Prof.~Dr.~S.~Mehdi}	\\[5mm]
Version: 26.05.2014
}

\vspace{10mm}

\vspace{10mm}

\vspace{10mm}

\vspace{10mm} 
\vspace{4mm}

\end{center}
\newpage

\cleardoublepage
\thispagestyle{empty}

\noindent{\bf{\LARGE Zusammenfassung}}
\\[0.2cm]

Dieser Arbeit untersucht den Zusammenhang zwischen Patterson-Sullivan Distributionen und dynamischen Zetafunktionen auf reellhyperbolisch kompakten R\"aumen. In \cite{AZ} wurde im Fl\"achenfall gezeigt, dass die Residuen von gewissen Zetafunktionen, die mithilfe von Daten des geod\"atischen Flusses definiert werden, durch Patterson-Sullivan Distributionen beschrieben werden. In dieser Arbeit wird der h\"oher dimensionale Fall behandelt. 
\cleardoublepage
\thispagestyle{empty}

\noindent{\bf{\LARGE Summary}}
\\[0.2cm]
Given a compact real hyperbolic space we study the connection between certain
phase space distributions, so called Patterson-Sullivan distributions, and dynamical
zeta functions. These zeta functions generalize logarithmic derivatives of classical
Selberg zeta functions which are defined by closed geodesics which is data from the
geodesic 
flow on phase space. Patterson-Sullivan distributions are asymptotically
equivalent to Wigner distributions which play a key role in quantum ergodicity but
they are also invariant under the geodesic 
flow.
The surface case was studied before in \cite{AZ} and thus the emphasis in this
work lies on the higher dimensional case.
\cleardoublepage
\thispagestyle{empty}
\voffset-20pt

\noindent{\bf{\LARGE Acknowledgements}}
\\[0.2cm]

First, I want to thank my advisors Joachim Hilgert and Salah Mehdi. They gave me the chance to work in this fascinating area of mathematics and without their constant support and patience I would not be able to finish this work. 

I would also like to thank the staff of the mathematical institutes in Paderborn and Metz, in particular my colleagues from the working group. Finally, I want to to thank my family for always believing in me.
\tableofcontents
\chapter{Statement of results}
This work is mainly based on the articles \cite{AZ} by N. Anantharaman and S. Zelditch and \cite{Z} by the latter author. We will shortly summarize the results from \cite{AZ} and \cite{Z}, which we want to examine. 

Let $X_\Gamma$ be a compact hyperbolic surface which can be written as \index{$X_\Gamma$, $X_\Gamma=\Gamma\backslash X$} \[X_\Gamma:=\Gamma\backslash G/K.\]
Here $G=\mathrm{PSL}_2(\mathbb R)$, $K=PSO(2)$ and $\Gamma\subset PSL_2(\mathbb R)$ is a cocompact, discrete and torsionfree subgroup of $G$. We also consider the right-regular representation $\pi_R$ of $G$ on $L^2(\Gamma\backslash G)$ defined by \[g\cdot f(x)=f(xg).\]

Since $\Gamma$ is cocompact, this unitary representation decomposes into a discrete sum 
\begin{equation}\label{eq: l2decdisc}
L^2(\Gamma\backslash G)=\bigoplus_{\pi\in\widehat G} m_\pi V_\pi,
\end{equation}
where $m_\pi\in\mathbb N$. We call the irreducible component $V_\pi$ spherical, if it possesses a $K$-fixed vector. These $K$-fixed vectors are unique up to constants, smooth and give an orthogonal basis of $L^2(X_\Gamma)$ of Laplace eigenfunctions. We fix such a normalized basis $\{\varphi_n\}_n$ with eigenvalues $-\left(\rho_0^2+\lambda_n^2\right)$. It can be shown that one can either assume $\lambda_n\in i\mathbb R$ or $\lambda_n\in\mathbb R^+$. In the former case, we say that $\lambda_n$ is in the complementary series, else we say that $\lambda_n$ is in the principal series. There are only finitely many $\lambda_n$ in the complementary series. Associated to $\{\varphi_n\}_n$ there are some specific distributions on phase space. They are called Wigner distributions and are given by \index{$W_n(\sigma)$} \[W_n(\sigma):=\langle \mathrm{Op}(\sigma)\varphi_n,\varphi_n\rangle_{L^2(X_\Gamma)}\]
for $\sigma\in C^\infty(S^*X_\Gamma)$. Here one needs a calculus for pseudodifferential operators, i.e. an assignment \index{$\mathrm{Op}(a)$} \[C^\infty(S^*X_\Gamma)\to B(L^2(X_\Gamma))\mbox{ ,  } a\mapsto \mathrm{Op}(a).\]
 This is the data from the quantum side we need.

On the side of classical dynamics, we consider the geodesic flow on $T^*X_\Gamma$. Since the geodesic flow preserves the metric on $T^*X_\Gamma$, it can be considered as a mapping on phase space $S^*X_\Gamma\cong SX_\Gamma$, which can also be identified with $\Gamma\backslash G$. Under this identification the geodesic flow is given by right translation with $\exp tH_0$, where $\mathfrak a=\mathbb RH_0 =\mathrm{Lie}(A)$ comes from an Iwasawa decomposition of $G=ANK$. Periodic orbits of the geodesic flow are called closed geodesics and it can be shown that there is a bijection between closed geodesics and nontrivial conjugacy classes \index{$[\gamma]$} $[\gamma]$ in $\Gamma$. The smallest possible period is called the length of $[\gamma]$ and denoted by $L_\gamma$. Finally, a closed geodesic $[\gamma]$ is called prime, if there is no $\gamma_0\in \Gamma$ and $n\in\mathbb N$, $n>1$, such that
\begin{equation}\label{eq: primegeodst}
\gamma=\gamma_0^n.
\end{equation}

With this data from classical mechanics, in \cite{AZ} a dynamical zeta function 
\begin{equation}\label{eq: zetasurf}
\mathcal{Z}(k;\sigma):= \sum_{[\gamma]\neq 1} \frac{e^{-kL_\gamma}}{1-e^{-L_\gamma}}\cdot \int_{\gamma_0} \sigma
\end{equation}
for $k\in \mathbb{C}$, $\mathrm{Re}(k)>1$, is defined. Here $[\gamma_0]$ is the unique prime closed geodesic belonging to $[\gamma]$ and $\sigma$ is a test function on phase space, that is $\sigma\in C^\infty(\Gamma\backslash G)$. 

For $\sigma\equiv 1$ constant,  $\mathcal{Z}(\cdot;\sigma)$ just equals the logarithmic derivative of the dynamical Selberg zeta function $Z_S$\index{$Z_S$} \[\frac{d}{dk}\mathrm{ln } Z_S(k)= \sum_{[\gamma]\neq 1} \frac{e^{-kL_\gamma}}{1-e^{-L_\gamma}}\cdot L_{\gamma_0}  .\] 

In this context Anantharaman and Zelditch show the following theorem, see \cite[Th.1.3]{AZ}:
\begin{theorem}\label{th: mainth}
If $\sigma$ is a real analytic function on $SX_\Gamma$, then $\mathcal Z(\cdot;\sigma)$ admits a meromorphic continuation to $\mathbb C$. The poles in the critical strip $0<\mathrm{Re}(k)<1$ appear at $s=\frac{1}{2}+i\lambda$, where $-\left(\frac{1}{4}+\lambda^2\right)$ is an eigenvalue of the Laplacian. The residue at $k=\frac{1}{2}+i\lambda$ is given by \[\sum \langle \sigma,\widehat{\mathrm{PS}}_{\varphi}\rangle,\] where this finite sum runs over all (normalized) Patterson-Sullivan distributions $\widehat{PS}_{\varphi}$ associated with an eigenfunction $\varphi$ for the eigenvalue $-\left(\frac{1}{4}+\lambda^2\right)$. 
\end{theorem}

Here, (normalized) Patterson-Sullivan distributions $\widehat{\mathrm{PS}}_{\varphi}$  are certain other kind of phase space distributions which one can associate to the eigenfunctions $\{\varphi_n\}_n$, but in contrast to Wigner distributions they are invariant under the geodesic flow. They were first defined in \cite{AZ} using the same calculus which was used to construct Wigner distributions. See \cite{HHS} for an extension of this calculus and a definition of Patterson-Sullivan distributions by this calculus for arbitrary compact locally symmetric spaces. Patterson-Sullivan distributions also play a role in quantum unique ergodicity, since they are asymptotically equivalent to Wigner distributions, see \cite{AZ}, \cite{HHS}. 

Note that if the spectrum is simple, i.e. if $m_\pi=1$ for all spherical components $V_\pi$, this theorem yields a definition of Patterson-Sullivan distributions, which only uses knowledge of the length spectrum $\{L_\gamma:[\gamma]\text{ conjugacy class in } \Gamma\}$. 
See for example the survey article \cite{Sar} for more information on simple spectra of the Laplacian.  

\cite{AZ} presents two kinds of proofs for Theorem \ref{th: mainth}. The first one uses thermodynamic formalism, the second one is of representation theoretical nature and uses a generalized Selberg trace formula, which can be found in \cite{Z}. We will pursue the second proof as this seems to be more feasible in generalizing \ref{th: mainth} to higher dimensional spaces. It has the disadvantage as it so far only works for test functions $\sigma$, which have only finitely many nontrivial components in the decomposition $(\ref{eq: l2decdisc})$. On the other hand this approach makes it possible also to determine poles/residues of $\mathcal Z(\sigma)$ outside the strip $0<\mathrm{Re}(k)<1$. Furthermore, a generalized Selberg trace formula could be of independent interest, as it connects periodic orbit measures to Wigner distributions.

We will now come to the content of our work. \index{$X_\Gamma$} $X_\Gamma:=\Gamma\backslash G/K$ will be a compact locally symmetric space of (real) rank one. Here $G$ is a real semisimple, noncompact Lie group with finite center, $K\subset G$ a maximal compact subgroup and $\Gamma\subset G$ a discrete, cocompact and torsionfree subgroup. We also fix an Iwasawa decomposition $G=ANK$ and set $M=Z_K(A)$ for the centralizer of $A$ in $K$. In the course of this work, we will specialize $X_\Gamma$ to be a compact real hyperbolic locally symmetric space, this means $G=SO(1,n)$. We restrict to real rank one symmetric spaces, as this ensures the absence of non trivial elliptic elements in the uniform lattice $\Gamma$. The major reason for specializing further to real hyperbolic spaces is that in this case $M$ acts on $N$ with a one dimensional slice and the differential equation which will occur during this work can be solved on this simple slice. 

The representation theoretic proof of Theorem \ref{th: mainth} relies mainly on three results, which we will now explain. The first one is an observation using the fact that Patterson-Sullivan distributions are invariant under the geodesic flow. Then one uses the result that the representation of $A$ on irreducible components $V_\pi$ of $L^2(\Gamma\backslash G)$, which we obtain by restricting the right-regular representation of $G$, is particularly simple. More precisely, if $G=PSL_2(\mathbb R)$ and $V_\pi $ is a spherical component of $(\ref{eq: l2decdisc})$, the representation of $A$ on $V_\pi$ has exactly two invariant subspaces, one of which is generated by the $K$-fixed vector in $V_\pi$. If $V_\pi$ is not spherical, then it has already a cyclic vector, see \cite[Prop. 2.2.]{AZ}. This yields that in the proof instead of considering all possible $\sigma\in C^{\infty}(SX_\Gamma)$ one can restrict to $\sigma$ coming from three basic series. For $X_\Gamma$ a real hyperbolic space we can generalize the result on the action of $A$ on the spherical spectrum. It will turn out that if $X_\Gamma$ is of dimension at least 3, the action of $A$ on any spherical component $V_\pi$ is already irreducible, the $K$-fixed vector being cyclic, see Theorem \ref{th: geodirrsphere}. Here we use the fact that the set of $M$-invariants in the universal enveloping algebra $U(\mathfrak n)$ is generated by the (euclidean) Laplacian on $\mathfrak n$. If $V_\pi$ is not spherical, we state a result, but we will not use it, see Proposition \ref{th: geodirrnonsphere}. 

The procedure for defining $\mathcal Z(\sigma)$ is computing the (geometric) trace of a suitable trace class operator $\sigma\cdot \pi_R(f)$, which depends on $\sigma$ and a suitable function $f\in C^\infty(G)$. Here $\pi_R(f)$ is just the Fourier transform of $f$ with respect to the representation $\pi_R$ and $\sigma$ is viewed as a multiplication operator. The trace can be computed in two ways, the first one is the so-called geometric, the latter one is called the spectral trace. We start by computing the geometric trace and thus, the next ingredient is a generalization of Selberg's trace formula depending on $\sigma$ coming from the three series. We first mention that the computation of the trace formula heavily depends on the rank one assumption. Namely, if $X_\Gamma$ is a locally symmetric compact space of (real) rank one, then all nontrivial elements in $\Gamma$ are hyperbolic which means they are conjugate to some element $ma\in MA$.  We state a trace formula for general compact locally symmetric spaces of rank one in Theorem \ref{theo: generaltr}. 

We further specialize $\sigma$ to be a function in $C(SX_\Gamma)$, which is only allowed to have finitely many nontrivial components in the \textit{spherical} spectrum and no components in the nonspherical spectrum. Hence in the case of a real hyperbolic space, Lemma  \ref{lem: sphcy} allows us now to reduce to the case where $\sigma$ equals some eigenfunction $\varphi$ on $X_\Gamma$. The obstacle that occurs in computing a satisfactory trace formula now, is a factor we call the weight function $I_\gamma(\sigma)$. It is a real valued function on $N$ and depends on the element $\gamma$ in $\Gamma$ and the function $\sigma$.\footnote{In \cite{AZ} this weight function is called orbital integral but this seems to not quite compatible with the terminology of the classical Selberg trace formula.} 
Basically, for $n\in N$, $I_\gamma(\sigma)(n)$ is the integral of the $n$-translation of $\sigma$ over the prime closed geodesic  belonging to $[\gamma]$. For $\varphi=\sigma\equiv 1$ this weight function is constant and just equals $L_{\gamma_0}$ but for non constant eigenfunctions the evaluation of the weight function is more complicated. Since $\varphi$ is a Laplace eigenfunction by assumption, we obtain a differential equation using an expression for the Casimir $\Omega$ from Chapter \ref{chap: casimir}. But this differential equation is a priori an equation on $N$ and thus it is not clear how  $I_\gamma(\varphi)$, which is as mentioned above a real valued function on $N$, depends in higher dimensions on its value at the neutral element, which is just $\int_{\gamma_0} \varphi$ from the definition of $\mathcal Z$. We circumvent this problem by decomposing $I_\gamma(\varphi)$ into a sum \[I_\gamma(\varphi)=\sum_{\pi\in\widehat M}d_\pi \chi_\pi*I_\gamma(\varphi),\]
see Theorem \ref{th: gstf}. Here $\widehat M$ consists of all irreducible representations $(\pi,V_\pi)$ of $M$, $d_\pi$ is the dimension of $V_\pi$, $\chi_\pi=\mathrm{Tr}(\pi)$ the character of $\pi$ and $*$ denotes convolution. Then $d_\pi\chi_\pi*I_\gamma(\varphi)$ is the projection of $I_\gamma(\varphi)$ in the space of $M$-finite functions of type $\check\pi$, $\check\pi$ the contragradient representation to $\pi$. Furthermore, each $\chi_\pi*I_\gamma(\varphi)$ is also a Casimir eigenfunction with the same eigenvalue as $\varphi$. The observation which helps us now is the fact that in real hyperbolic spaces, the subgroup $M$ acts transitively on spheres in $N$ if the dimension of the real hyperbolic space is at least 3, i.e. slices for this action of $M$ on $N$ are one dimensional. For any slice $S$ we can now restrict the equation for each $\chi_\pi*I_\gamma(\varphi)$ to $S$ and the results of Chapter \ref{chap: radial} and \ref{chap: ode} allow us to determine $\chi_\pi*I_\gamma(\varphi)|_S$ as a product of a hypergeometric function with a monomial and a scalar which is connected to $I_\gamma(\varphi)$ at the neutral element, see equations (\ref{eq: lghyp}) and (\ref{eq: l2hyp}). 

The resulting trace formula is now sufficient to define a zeta function 
\begin{equation}\label{eq: ourzeta}
\mathcal Z(k;\sigma):=\sum_{1\neq [\gamma]\in C\Gamma}\sum_{\pi\in\widehat{M}}c(\gamma,\sigma,\pi,k)e^{(-k+\rho_0)L_\gamma},
\end{equation}
 which converges at least on the half plane $\{k\in\mathbb C:\mathrm{Re}(k)>2\rho_0\}$ and generalizes the one from \cite{AZ}. Here $\rho_0$ is a number depending only on the dimension and $\sigma$ is as above and the coefficients $c(\gamma,\sigma,\pi,k)$ depends on $k$, the test function $\sigma$, $\pi$ in $\widehat M$ and the period length $L_{\gamma_0}$ of the prime geodesic $[\gamma_0]$.

The meromorphic continuation now follows from the computation of the spectral trace by using the basis of eigenfunctions $\{\varphi_n\}_n$. The preliminary result of the trace, for an eigenfunction $\varphi$, is given in Proposition \ref{thm: strac}. It uses a calculus for pseudodifferential operators from \cite{Sch}, see also \cite{HS}, which is adapted to the rank one setting. Proposition \ref{thm: strac} is valid for any locally symmetric compact space of rank one, but it only involves Wigner distributions. To connect it with Patterson-Sullivan distributions we have to perform computations similar to the ones on the geometric side. In particular, we again encounter a differential equation, which we can only solve when the underlying space is real hyperbolic. The final result for the spectral trace can be found in Theorem \ref{thm: stracspec}. 

The meromorphic continuation now follows by standard arguments, we only mention that we have to make some explicit computations for functions on hyperbolic space. These calculations also seem to be possible in other rank one spaces by the classification results but we have not tried to do so.

The main result is as follows, see Proposition \ref{th: geodirrnonsphere}.

\begin{theorem}\label{th: ourmain}
Let $X_\Gamma$ be a compact, locally symmetric real hyperbolic space and $\sigma$ a function in $C^\infty(SX_\Gamma)$ with only finitely many nontrivial components in the spherical spectrum and no components in the non-spherical spectrum. The associated zeta function $\mathcal Z(\sigma)$ defines a meromorphic function on $\mathbb C$. In the strip $\rho_0-\frac{1}{2} <\mathrm{Re}(k)<\rho_0+\frac{1}{2}$ the poles of $\mathcal{Z}(\sigma)$ are at $k=\rho_0+i\lambda$, where $-(\rho_0^2+\lambda^2)$ is an eigenvalue of the Laplacian, and $k=\rho_0$. 
Their  residues are determined by Patterson-Sullivan distributions. If $\lambda$ comes from the principal series, then the residue at $k=\rho_0+i\lambda$ is given (up to a non-zero constant) by normalized Patterson-Sullivan distributions $\widehat{\mathrm{PS}}_\varphi$ \begin{equation}\sum \langle \sigma ,\widehat{\mathrm{PS}}_{\varphi}\rangle,
\end{equation} 
where the sum runs over all eigenfunctions $\varphi$ with eigenvalue $-\left(\lambda^2+\rho_0^2\right)$. 
\end{theorem}
  
\begin{remark}
Let $\sigma=\varphi$ be a Laplace eigenfunction. There are two statements in Theorem 1.3 from \cite{AZ} which we cannot verify in Theorem \ref{th: ourmain} in the case of a compact surface. 

The first one is about the form of the zeta function $\mathcal Z(\sigma)$ from (\ref{eq: zetasurf}). It seems to come from an incomplete integral substitution in \cite{AZ} which is used to get from equation (9.8) to (9.9). Our definition (\ref{eq: ourzeta}) differs in the surface case by a constant which depends on $k\in\mathbb C$, the geodesic $[\gamma]$ and the eigenvalues from the $K$-fixed vectors in the spherical components of $\sigma$, see Section \ref{sec: out} for the definition. This constant is furthermore holomorphic in $k$ on $\{k\in \mathbb{C}:\mathrm{Re}(k)>0\}$ and approaches $1$, as $L_\gamma$ goes to infinity. After normalizing $\mathcal Z(\sigma)$ this constant is \[ \left(\frac{\cosh L_\gamma}{\cosh L_\gamma-1}\right)^{k-1/2}\cdot {}_2F_1\left(k-\frac{1}{4}-\frac{ir}{2},k-\frac{1}{4}+\frac{ir}{2},k;1-\frac{\cosh L_\gamma}{\cosh L_\gamma-1}\right),\] 
where $-\frac{1}{2}(r^2+\frac{1}{4} )$ is the eigenvalue of $\varphi$, see Section \ref{sec: out}. We can only deduce from (\ref{eq: ourzeta}) that (\ref{eq: zetasurf}) has a meromorphic continuation to the half plane  $\{k\in\mathbb C:\mathrm{Re}(k)>0\}$ with the same poles and residues as (\ref{eq: ourzeta}). In particular, in the strip $0<\mathrm{Re}(k)<1$ the residues of (\ref{eq: ourzeta}) are given by (normalized) Patterson-Sullivan distributions.

The second difference is about the location and the residues of the poles of the continuation of $\mathcal Z(\sigma)$ in the strip $0<\mathrm{Re}(k)<1$. Here the problem seems to be that the constant $\mu_0(s)$ in \cite{AZ} which relates Wigner- to Patterson-Sullivan distributions could have poles at values $s=1/2+ir$, if $-(1/4+r^2)$ is an eigenvalue from the complementary series, which are not considered in \cite{AZ}. We can only recover the result on the poles/residues from Theorem 1.3 in \cite{AZ} in the case where the pole $k=\frac{1}{2}+i\lambda$ corresponds to an eigenvalue $-(\frac{1}{4}+\lambda^2 )$ from the principal series. 
 
In the complementary series case, the poles/residue are more complicated. Especially, it makes a difference if $-\rho_0^2$ is an eigenvalue or not.  
\end{remark}  


\chapter{Preliminaries}\label{chap: prelim}

In this chapter we collect some facts about the geometry of semisimple Lie groups $G$ and locally symmetric spaces. In Section \ref{chap: casimir} we compute a formula for the Casimir operator $\Omega$ which becomes useful if we apply $\Omega$ to functions in $C(A\backslash G/K)$ for a fixed Iwasawa decomposition $G=ANK$. In Section \ref{chap: geometry} we discuss the geometry and dynamics of locally symmetric spaces $\Gamma\backslash G/K$. In Section \ref{chap: hyper} we discuss a model for real hyperbolic spaces $G/K$ which makes it possible to do some concrete computations in Section \ref{sec: spher traf} where we compute the spherical transform $\mathcal S(f)$ of a certain bi-$K$-invariant function $f$.  

\section{Some computations on the Casimir element}\label{chap: casimir}
In \cite[p.40]{Z} a decomposition of the Casimir operator in $SL_2(\mathbb{R})$ is given, which we want to generalize to arbitrary semisimple Lie algebras. The result is formula (\ref{eq: casfin}), see also the example at the end of the section.

In what follows let $G$ be a connected, noncompact and semisimple Lie group with finite center and $\mathfrak{g}$ be its Lie algebra with fixed \index{Cartan decomposition} \index{$\mathfrak{k}$, Lie algebra of $K$} \index{$\mathfrak p$, part of Cartan dec. of $\mathfrak g$} Cartan decomposition \[\mathfrak{g}=\mathfrak{k}\oplus\mathfrak p\] and \index{$t$@$\theta$, Cartan involution} \index{Cartan involution} Cartan involution $\theta$. Let \index{$K$, max. compact in $G$} $K$ be the analytic subgroup of $G$ corresponding to $\mathfrak{k}$. Then $K$ is compact. Let \index{$\mathfrak{a}$, Lie algebra of $A$} $\mathfrak{a}\subset\mathfrak p$ be a maximal abelian subspace and \index{$A$, abelian subgroup of $G$} $A=\exp\mathfrak{a}$. 

Further, let \index{$\mathfrak{m}$} $ \mathfrak{m}$ and \index{$M$, centralizer of $A$ in $K$} $M$ be the centralizers of $\mathfrak{a}$ in $\mathfrak{k}$ resp. $K$. Then we have the \index{Iwasawa decomposition} \textit{Iwasawa decomposition} \[\mathfrak{g}= \mathfrak{n}\oplus \mathfrak{a}\oplus\mathfrak{k}\] which gives the decompositions \[G=NAK=ANK=KAN,\] where \index{$\mathfrak{n}$, Lie algebra of $N$} $\mathfrak{n}$ is a nilpotent Lie subalgebra and \index{$N$, nilpotent subgroup of $G$} $N=\exp\mathfrak{n}$. By slight abuse of notation, we will call all these decompositions \textit{Iwasawa decomposition of $G$}. If we fix $K$, the Iwasawa decomposition is unique up to conjugation in $K$, that is, if \[G=KAN=KA_1N_1,\] then there is an element $k\in K$ such that \begin{equation}\label{eq: iwauniq}A_1=kAk^{-1} \mbox{ and  } N_1=kNk^{-1},\end{equation}
also \index{$ \mathfrak{a}^k $, $\mathfrak a$ conjugated by $k\in K$} $\mathfrak a^k:=\mathrm{Ad}(k)\mathfrak{a}=\mathfrak{a}_1$ and \index{$ \mathfrak{n}^k $, $\mathfrak n$ conjugated by $k\in K$} $\mathfrak n^k:=\mathrm{Ad}(k)\mathfrak{n}=\mathfrak{n}_1$, 
\cite[(2.2.12)]{GV}.

Let \index{Killing form} \index{$B(\cdot,\cdot)$, Killing form} $B(\cdot,\cdot)$ be the Killing form of $\mathfrak g$, then \index{$B_\theta(\cdot,\cdot)$, inner product on $\mathfrak g$ induced by $B$ and $\theta$} \[B_\theta(\cdot,\cdot):=-B(\cdot,\theta\cdot)\] defines an inner product on $\mathfrak g$. We also have the \index{root space decomposition} root space decomposition

\begin{equation*} \mathfrak{g}=\mathfrak{g}_0\oplus\underset{\alpha\in \Delta(\mathfrak{g},\mathfrak{a})}\bigoplus \mathfrak{g}_\alpha,\end{equation*} 
where \[ \mathfrak{g}_\alpha=\{X\in \mathfrak{g}:\mbox{ad}(H)X=\alpha(H)X \mbox{ for all } H\in\mathfrak{a}\}\] for $\alpha\neq 0$ is called the root space of $\mathfrak{g}$ with respect to $\alpha$ and \index{$Delta$@$\Delta(\mathfrak{g},\mathfrak{a})$, restricted root system of $\mathfrak g$} \[\Delta(\mathfrak{g},\mathfrak{a})=\{\alpha\in\mathfrak{a}^*-\{0\}:\mathfrak{g}_\alpha\neq\{0\}\}.\]

We temporarily assume that the rank $n=\mathrm{dim}_\mathbb{R}\mathfrak{a}$ of $\mathfrak{g}$ is arbitrary. We need some lemmata:
\begin{lemma}\label{lem: theta}
For each $\alpha\in\Delta(\mathfrak{g},\mathfrak{a})\cup\{0\}$ we have $\theta \mathfrak{g}_\alpha=\mathfrak{g}_{-\alpha}$. 
\end{lemma}
\begin{proof}
\cite[12.3.2]{HN} 
\end{proof}
\begin{lemma}\label{lem: orth}
If $\alpha,\beta\in\mathfrak{a}^*$ with $\alpha+\beta\neq 0$, then $B(\mathfrak{g}_\alpha,\mathfrak{g}_{\beta})=0$. 
\end{lemma}
\begin{proof}
\cite[12.3.4]{HN} 
\end{proof}

Let \index{$X_i$, part of ONB in $\mathfrak n$}$X_1,\ldots, X_m$ be any basis of orthonormal elements of $\mathfrak{n}$ with respect to $B_\theta(\cdot,\cdot)$ and set \index{$Z_i$, $Z_i=-\theta X_i$} \[Z_i=-\theta X_i\] for $i=1,\ldots,m$. Then \[B(X_i,Z_j)=-B(X_i,\theta X_j)=\delta_{ij}.\] Further, let \index{$H_i$, part of ONB in $\mathfrak a$} $H_1,\ldots,H_n$ and \index{$M_i$, part of ONB in $\mathfrak m$} $M_1,\ldots,M_k$ be any orthonormal bases of $\mathfrak{a}$ resp. $\mathfrak m$ with respect to $B_\theta(\cdot,\cdot)$. \footnote{Note that $B(\cdot,\theta\cdot)=-B(\cdot,\cdot)$ on $\mathfrak{a}$ while $B(\cdot,\theta\cdot)=B(\cdot,\cdot)$ on $\mathfrak m$.}

Then \[H_1,\ldots,H_n, M_1,\ldots,M_k, X_1,\ldots,X_m,Z_1,\ldots,Z_m\] is a basis of $\mathfrak{g}$. We denote the dual basis with respect to $B_\theta(\cdot,\cdot)$ with \index{$H^1,\ldots,H^n$, dual basis to $H_1,\ldots,H_n$} \index{$M^1,\ldots,M^k$, dual basis to $M_1,\ldots,M_k$} \index{$X^1,\ldots,X^m$, dual basis to $X_1,\ldots,X_m$} \index{$Z^1,\ldots,Z^m$, dual basis to $Z_1,\ldots,Z_m$} \[H^1,\ldots,H^n,M^1,\ldots,M^k, X^1,\ldots,X^m,Z^1,\ldots,Z^m.\] 

Since 
\[\mathfrak{g}_0=\mathfrak{a}+\mathfrak m,\] while $\mathfrak m\subset\mathfrak{k}$ and $\mathfrak{a}\subset\mathfrak p$, we see by Lemma \ref{lem: orth} that \[H^i=H_i\] and \[M^i=-M_i.\]

Furthermore, by Lemma \ref{lem: orth}, \[X^i=Z_i\] and \[Z^i=X_i.\] 

There is an object of special interest, the Casimir operator \index{$Omega$@$\Omega$, Casimir operator} \index{Casimir operator} $\Omega$, which is an element of\index{$U(\mathfrak{g})$, universal enveloping of $\mathfrak g$} $U(\mathfrak g)$, \index{universal enveloping algebra} the universal enveloping algebra of $\mathfrak g$. More precisely, $\Omega$ lies in \index{$Z(\mathfrak g)$, center of $U(\mathfrak g)$} $Z(\mathfrak g)$, that means it commutes with every element in $U(\mathfrak g)$. If $Q_j$ is any basis of $\mathfrak g$ and $Q^j$ the dual basis with $B_\theta(Q_i,Q^j)=\delta_{ij}$, then $\Omega$ is defined by \begin{equation}\label{def: casimir}\Omega:=\sum_{j}Q_jQ^j.\end{equation} 

This is independent of the choice of $Q_j$, \cite[(2.6.58)]{GV}. Consequently, we can write the Casimir operator as 
\begin{equation*}\Omega=\sum_{i=1}^{n}H_i^2-\sum_{i=1}^{k}M_i^2+\sum_{i=1}^m(X_iZ_i+Z_iX_i).\end{equation*}
We work on the last sum:
\begin{eqnarray*}
\sum_{i=1}^m(X_iZ_i+Z_iX_i) &=& -\sum_{i=1}^m(X_i\theta X_i+ \theta X_iX_i+X_i^2-X_i^2)\\ &=&\sum_{i=1}^mX_i(X_i-\theta X_i)-\sum_{i=1}^m(X_i+\theta X_i)X_i\\ &=& \sum_{i=1}^mX_i(2X_i-(X_i+\theta X_i))-\sum_{i=1}^m(X_i+\theta X_i)X_i =(*). 
\end{eqnarray*}

Set \index{$W_i$, $W_i=X_i+\theta X_i$} \[W_i:=X_i+\theta X_i,\] then \[W_i\in\mathfrak{k}=\mathfrak m\oplus \mathfrak{m}^{\bot_\mathfrak{k}}=\mathfrak m\oplus (1+\theta)\mathfrak{n}.\] 

Thus,
\begin{equation*}(*)=2\sum_{i=1}^mX_i^2-\sum_{i=1}^m(X_iW_i+W_iX_i).\end{equation*}
Now,
\begin{equation*} [W_i,X_i]=[X_i+\theta X_i,X_i]=[\theta X_i,X_i]\in \mathfrak{a}.\end{equation*} 

More precisely, a direct computation shows that \[[\theta X_i,X_i]=H_\alpha,\] if $X_i\in\mathfrak{g}_\alpha$, see \cite[(4.2.1)]{GV}, where \index{$H_\alpha$, defining vector for root $\alpha$ in $\mathfrak a^+$} $H_\alpha\in\mathfrak{a}$ is defined by \[\alpha(H)=B(H_\alpha,H)\] for all $H\in\mathfrak{a}$. Therefore,
\begin{equation*} (*)=2\sum_{i=1}^mX_i^2-2\sum_{i=1}^mX_iW_i-2H_\rho,\end{equation*}  and 
\begin{equation}\label{eq: casfin} \Omega=\sum_{i=1}^{n}H_i^2-\sum_{i=1}^{k}M_i^2+2\sum_{i=1}^mX_i^2-2
\sum_{i=1}^mX_iW_i-2H_\rho,
\end{equation} if \index{$H_\rho$, defining vector for $\rho$ in $\mathfrak a^+$} \index{$rho$@$\rho$, half-sum of positive roots times dimension} \[2\rho=\sum_{\alpha\in\Delta(\mathfrak{g},\mathfrak{a})^+}\dim(\mathfrak{g}_\alpha)\alpha.\] 

This decomposition (\ref{eq: casfin}) is the central result of this section, as it will suffices for our purpose, see Chapter \ref{chap: ode}. 
\begin{remark}
The homogeneous space $X=G/K$ is a symmetric space, in particular, it is a Riemannian space $(X,g)$ with metric $g$. Using this metric one can define a special differential operator \index{$Deltaa$@$\Delta$, Laplacian on $G/K$} $\Delta$ called \index{Laplace operator} \textit{Laplace}- or \textit{Laplace-Beltrami}-operator, see \cite[Ch. II \S 2.4]{GGA}. The Laplace operator equals $-\Omega$ on  $C^\infty(X)$, see for example \cite[p.97]{Sch}. This means, if we let elements $X$ in $\mathfrak g$ act on functions $f\in C^\infty(G/K)$ by \[ X\cdot f(g):=\frac{d}{dt}|_{t=0}f(g\exp tX)\] as left invariant differential operators and extend this definition to $U(\mathfrak g)$, then \[\Delta f=-\Omega f.\]

That is, we have shown in this section, that
 \[\Delta=-\sum_{i=1}^{n}H_i^2+2H_\rho-2\sum_{i=1}^mX_i^2,\]
see also \cite[(4.4.2)]{AJ},
\end{remark}

Now we apply the computation to the case of $SL_2(\mathbb R)$. 
\begin{example}\label{ex: casimir}
Let $\mathfrak{g}=\mathfrak{sl}_2(\mathbb{R})$ with the standard basis $H=\left(\begin{array}{cc} 1&  0\\ 0&-1 \end{array}\right)$, $V=\left(\begin{array}{cc} 0&  1\\ 1& 0 \end{array}\right)$, $W=\left(\begin{array}{cc} 0&  1\\ -1&0 \end{array}\right)$. Then the Cartan-Killing form $B$ is defined by the matrix \[8\left(\begin{array}{ccc}1 & 0 & 0\\0 & 1 &0\\ 0& 0 & -1 \end{array}\right),\] in particular \[B(H,H)=8=-B(W,W).\] Let $X=\frac{1}{2}V+\frac{1}{2}W =\left(\begin{array}{cc} 0&  1\\ 0&0 \end{array}\right)$. We want to express $\Omega$ in terms of $H$, $X$ and $W$. We see that $Z=-2\theta X$ satisfies \[B(X,Z)=8\] since $\theta V=-V$ and $\theta W=W$. Thus,\begin{equation*} 8\Omega=H^2+4X^2-4XW-2H
\end{equation*} is (eight times) the operator from above, see also \cite[p.40]{Z}.
\end{example}
\begin{remark}
This example shows a difference between \cite{AZ} and our work. For the formula (\ref{eq: casfin}) we have worked with an orthonormal basis of $\mathfrak n$ with respect to $B_\theta$, while in \cite{AZ} the vector $X$ from Example \ref{ex: casimir} is used to identify $N$ with $\mathbb R$, $\exp tX\to t$. But $B_\theta(X,X)=4$ not 1. The difference comes from the fact that in \cite{AZ} the invariant form $\tilde B(X,Y)=2\cdot\mathrm{tr}(XY)$ is implicitly used. This is off by the factor $2$ from our definition $B(X,Y)=4\cdot\mathrm{tr}(XY)$.    
\end{remark}

\section{Geometry and dynamics of (locally) symmetric spaces}\label{chap: geometry}
We summarize some notions and facts on the geometry of the (locally) symmetric spaces $X$ and $X_\Gamma$. We will mainly cite from the articles \cite{Ga} and \cite{Ga1}.

Let $G$ be a semisimple, noncompact connected Lie group of real rank one with finite center, maximal compact subgroup $K$ and Lie algebra $\mathfrak g$. We fix a Cartan decomposition $\mathfrak g=\mathfrak{k}\oplus \mathfrak{p}$. Let $X=G/K$ be the associated (Riemannian) symmetric space and $\Gamma\subset G$ a discrete co-compact and torsionfree subgroup. For short we call $\Gamma$ a \index{uniform lattice} \textit{uniform lattice} in $G$. Then $\Gamma$ acts freely (by isometries) on $X$ and the locally symmetric space \index{$X_\Gamma$} $X_\Gamma:=\Gamma\backslash X$ is a compact (Riemannian) manifold with simply connected covering $X$, in particular the fundamental group $\pi_1(X)$ is isomorphic to $\Gamma$. Further for any $Z\in\mathfrak g$ we put \index{$|Z|$, norm on $\mathfrak g$ induced by $B_\theta$} \[|Z|:=-B( Z,\theta Z),\] 
where $B$ denotes the Killing form on $\mathfrak g$ and $\theta$ the Cartan involution. We fix an Iwasawa decomposition $G=ANK$ and set $M$ resp. $M'$ to be the centralizer resp. normalizer of $A$ in $K$. Then $M$ is normal in $M'$ and the quotient group $M'/M$ is called the \index{Weyl group} \textit{Weyl group} \index{$W$, Weyl group} $W$. Since the rank of $G$ is one, $W$ is isomorphic to the group with two elements and we denote the nontrivial element in $W$ by \index{$w$, nontrivial weyl element} $w$. On $G$ we fix a Haar measure $dg$ such that $$\int_G f(g)dg=\int_{ANK}f(ank)dkdnda$$ for integrable function $f$ on $G$, where $da$ and $dn$ are defined by the euclidean structure on $A$ and $N$ coming from the inner product $B_\theta(.,.)$. The Haar measure $dk$ is assumed to give $K$ unit mass. Then we fix also a $G$-invariant measure $dx$ on $\Gamma\backslash G$ such that \[\int_Gf(g)dg=\int_{\Gamma\backslash G}\sum_{\gamma\in\Gamma}f(\gamma x)dx\] for all $f\in C_c(G)$. 

We call $x\in G$ \index{elliptic} \textit{elliptic}, if it is conjugated to some element in $K$, which implies in particular that every elliptic element is semisimple, i.e. $\mbox{Ad}(x)\in \mbox{End}(\mathfrak{g})$ is semisimple. If $x\in G$ is semisimple but not elliptic, we call it \index{hyperbolic} \textit{hyperbolic}.\footnote{In all other cases, we call $x$ parabolic. Since $\Gamma\backslash G$ is compact, $\Gamma$ does not contain any parabolic elements.} It is known that $\gamma\in \Gamma$ is elliptic iff it is of finite order, which is again equivalent to  $\gamma$ having a fixedpoint in $X$. Since $\Gamma$ is torsionfree, every nontrivial element $\gamma$ in $\Gamma$ is hyperbolic. Finally, we note:
\begin{proposition}\label{chap: geom gcon}
Every $\gamma\in \Gamma-\{e\}$ is conjugated to some element $a_{\mathfrak k} a_{\mathfrak{p}}$ in the Cartan subgroup \index{$a_{\mathfrak k}$, element of $A_{\mathfrak k}$} \index{$a_{\mathfrak p}$, element in $A$} $A_{\mathfrak k} A$. Here $A_{\mathfrak k}$ is a subgroup contained in $K$.  Even more, by possibly conjugating $a_\mathfrak k a_{\mathfrak p}$ with the nontrivial Weyl group element $w$ one can assume that $\gamma\in \Gamma-\{e\}$ is conjugated to some \index{$deltagamma$@$\delta_\gamma$, $\delta_\gamma=m_\gamma a_\gamma\in MA^+$} $\delta_\gamma=m_\gamma a_\gamma\in MA^+$. 
\end{proposition}
\begin{proof}
See \cite[Cor.11.5]{Wi}.
\end{proof}
Note that $a_\gamma$ is uniquely determined in $A^+$ by $\gamma$ while $m_\gamma$ is determined up to conjugacy in $M$, see \cite[Lem.6.6]{Wal}. In particular, $a_\gamma$ and $a_\mathfrak p$ generate the same (cyclic) subgroup of $A$  and $|\log a_\gamma|=|\log a_\mathfrak p|$, see the proof of \cite[Prop.11.4]{Wi}. 
By \index{$\Gamma_\gamma$, centralizer of $\gamma$ in $G$} $G_\gamma$ and $\Gamma_\gamma=\Gamma_\gamma\cap \Gamma$ we denote the centralizer of $\gamma$ in $G$ and $\Gamma$. 
\begin{proposition}\label{chap: geom icg}
$\Gamma_\gamma$ is co-compact in $G_\gamma$ which is reductive and unimodular. Furthermore, the centralizer $\Gamma_\gamma$ of any $\gamma$ in $\Gamma-\{e\}$ is always an infinite cyclic group. 
\end{proposition}
\begin{proof}
See \cite[p. 407]{Ga1} for the claim on $G_\gamma$ and \cite[Lem. 4.1.]{Ga1} for the claim on $\Gamma_\gamma$.
\end{proof}

Then we call $\gamma$ \index{primitive} \textit{primitive} if $\Gamma_\gamma$ is generated by $\gamma$. Since $\Gamma$ is torsionfree, every $\gamma\in \Gamma-\{e\}$ is the unique positive power of a primitive element \index{$gamma0$@$\gamma_0$, primitive element to $\gamma$} $\gamma_0$.

A \index{geodesic loop} \textit{geodesic loop} in $X_\Gamma$ is the image of closed path in $X$ under the orbit map $X\to X_\Gamma$. We call it a \index{closed geodesic} \textit{closed geodesic} if it is a geodesic when viewed as a subset of $X$. It is well-known that there is a one-to-one correspondence between closed geodesics and nontrivial conjugacy classes in $\Gamma$, \cite[p. 404]{Ga1}. We call a closed geodesic a \index{prime geodesic} \textit{prime geodesic}, if the corresponding conjugacy class contains a primitive element. By \index{$C\Gamma$,  conjugacy classes in $\Gamma$} $C\Gamma$ we denote the set of conjugacy classes in $\Gamma$. For $[\gamma]\in C\Gamma$ let $l_\gamma:=|\log a_\gamma|$, then \index{$l_\gamma$, $l_\gamma=|\log a_\gamma|$} 
\begin{equation}\label{eq: legamma} l_\gamma =\inf_{x\in G}|x^{-1}\gamma x|,
\end{equation} 
see \cite[p. 413]{Ga1}. Here for $g\in G$ \index{$|g|$, $|g|=|X|$ induced by$|.|$ on $\mathfrak g$}, $|g|:=|X|$, if $g=k\exp X$ with $k\in K$ and $X\in \mathfrak{p}$. We fix a Haar measure \index{$dx_\gamma$, Haar measure on $G_\gamma$} $dx_\gamma$ on $G_\gamma$ analogous to the Haar measure on $G$, following the Iwasawa decomposition of $G_\gamma=A_\gamma K_\gamma N_\gamma$ such that $K_\gamma$ has unit measure. We also have a Haar measure \index{$d\overset{\cdot}x_\gamma$, Haar measure on $\Gamma_\gamma \backslash G_\gamma$} $d\overset{\cdot}x_\gamma$ on the quotient $\Gamma_\gamma \backslash G_\gamma$. and we set \index{$|\Gamma_\gamma\backslash G_\gamma|$, volue induced by $d\overset{\cdot}x_\gamma$} $|\Gamma_\gamma\backslash G_\gamma|:=\int_{\Gamma_\gamma\backslash G_\gamma}d\overset{\cdot}x_\gamma$. We can make this more precise, if we consider the centralizer of $\alpha_\gamma \gamma \alpha_\gamma^{-1}=\delta_\gamma$. Then $G_{\delta_\gamma}=A(G_{\delta_\gamma}\cap K)=A(G_{\delta_\gamma}\cap M)$, see \cite[p. 414]{Ga1}. As before we fix Haar measures $dx_{\delta_\gamma}=dadk_{\delta_\gamma}$ on $G_{\delta}$ following the Iwasawa decomposition of $G_{\delta_\gamma}$ such that $K_{\delta_\gamma}:=(K\cap G_{\delta_\gamma})=(M\cap G_{\delta_\gamma})=:M_{\delta_\gamma}$ has unit measure with respect to $dk_{\delta_\gamma}$. We set $|\Gamma_\gamma\backslash G_\gamma|:=\int_{\Gamma_\gamma\backslash G_\gamma}d\overset{\cdot}x_\gamma$ etc.. The number $l_\gamma$ can now be related to $\Gamma_\gamma\backslash G_\gamma$.
\begin{proposition}\label{chap: geoprop l}
For $\gamma\in\Gamma$ let \index{$alphagamma$@$\alpha_\gamma$, element in $G$ conjugating $\gamma$ into $\delta_\gamma$} $\delta_\gamma=\alpha_\gamma \gamma \alpha_\gamma^{-1}\in MA^+$ and $H_\gamma:=\alpha_\gamma \Gamma_\gamma \alpha^{-1}_\gamma$\index{$H_\gamma$, centralizer of $\delta_\gamma$ in $G$}. Furthermore, let $K_{\delta_\gamma}$ and $G_{\delta_\gamma}$ be the centralizer of $\delta_\gamma$ in $K$ resp. $G$. If we identify $A$ with $G_{\delta_\gamma}/M_{\delta_\gamma}$ then \index{$l_{\gamma_0}$, $l_{\gamma_0}=|\log a_{\gamma_0}|$}
\begin{equation*} 
|\Gamma_\gamma\backslash G_\gamma|=|H_\gamma\backslash G_{\delta_\gamma}|=|H_\gamma\backslash G_{\delta_\gamma}/K_{\delta_\gamma}|=| A/\langle a_{\gamma_0} \rangle|=|\log a_{\gamma_0}|=l_{\gamma_0}=l_\gamma j^{-1}.
\end{equation*} 

Here $|.|$ always denotes the volume with respect to the invariant measure on quotient space, i.e. $|\Gamma_\gamma\backslash G_\gamma|=\int_{\Gamma_\gamma\backslash G_\gamma}d\overset{\cdot}x_\gamma$, $|H_\gamma\backslash G_{\delta_\gamma}|=\int_{H_\gamma\backslash G_{\delta_\gamma}}d\overset{\cdot}x_{\delta_\gamma}$ etc.

\end{proposition}

\begin{proof}
See also \cite[(4.5) pp. 414]{Ga1}. Assume that $\gamma$ is conjugated to $\delta_\gamma=m_\gamma a_\gamma\in MA^+$, i.e. $\alpha_\gamma\gamma\alpha_{\gamma^{-1}}=\delta_\gamma$ for some $\alpha_\gamma\in G$. The centralizer $G_{\delta_\gamma}$ of $\delta_\gamma$ in $G$ equals now $M_{m_\gamma}A$, where $M_{m_\gamma}$ is the centralizer of $m_\gamma$ in $M$, see \cite[p.185]{Wi}. Let $\gamma=\gamma_0^j$ with $\gamma_0$ primitive and $j\in\mathbb N$. If $\gamma_0$ is conjugated to $m_{\gamma_0}a_{\gamma_0}\in MA^+$, then it follows by the uniqueness of $a_\gamma$ that $a_\gamma=a_{\gamma_0}^j$. Then  
\begin{equation*} H_\gamma=\alpha_\gamma \Gamma_\gamma \alpha^{-1}_\gamma=\alpha_\gamma\langle\gamma_0\rangle\alpha_{\gamma^{-1}}\subset G_{\delta_\gamma}=M_{m_\gamma}A\subset MA.
\end{equation*}

If $\delta'=m'a'\in H_\gamma$, where $a'\in A$ and $m'\in M$, it follows that 
\begin{equation*}
m'\in M_{m_\gamma}=G_{\delta_\gamma}\cap M=M_{\delta_\gamma}=K_{\delta_\gamma}.
\end{equation*}
 Therefore, the action of $H_\gamma$ on $G_{\delta_\gamma}/K_{\delta_\gamma}$ equals the action of $\{a':\delta'=m'a'\in H_\gamma\}$ by left translation on $A$, where we identified 
\begin{equation*} 
 G_{\delta_\gamma}/K_{\delta_\gamma}=M_{\delta_\gamma}A/M_{\delta_\gamma}
\end{equation*} 
with $A$. Now $\{a':\delta'=m'a'\in H_\gamma\}=\langle \tilde a\rangle$, where $\tilde a\in A$ such that $\alpha_\gamma \gamma_0\alpha_{\gamma^{-1}}=\tilde a\tilde m\in AM$. Since $\gamma_0$ is also conjugated to $a_{\gamma_0}m_{\gamma_0}\in A^+M$, we deduce by possibly replacing $\alpha_\gamma$ with $\alpha_\gamma w$, that $\tilde a\in\{a_{\gamma_0},a_{\gamma_0}^{-1}\}$, i.e. $\langle \tilde a\rangle=\langle a_{\gamma_0}\rangle$. If we fix a Haar measure on $G_{\delta_\gamma}=K_{\delta_\gamma}A$ as we did for $G_\gamma$ with normalized measure on $K_{\delta_\gamma}$, then we get by unimodularity 
\begin{equation*} 
|\Gamma_\gamma\backslash G_\gamma|=|H_\gamma\backslash G_{\delta_\gamma}|=|H_\gamma\backslash G_{\delta_\gamma}/K_{\delta_\gamma}|=| A/\langle a_{\gamma_0} \rangle|=|\log a_{\gamma_0}|=l_{\gamma_0}=l_\gamma j^{-1}.
\end{equation*} 
\end{proof}

One can show the following. 
\begin{proposition}\label{prop: inflspec}
The set \[\{l_\gamma: 1\neq [\gamma]\in C\Gamma\}\] is bounded away from zero and has no upper bound. 
\end{proposition}
\begin{proof}
See \cite[Th. 4.4.]{Ga1}. 
\end{proof}

We denote its infimum by \index{$l_{\mathrm{\inf}}$, infimum of $\{l_\gamma: 1\neq [\gamma]\in C\Gamma\}$} $l_{\mathrm{inf}}$. The \index{geodesic flow}\index{$G_t$, geodesic flow} geodesic flow $G_t$ on the tangent bundle $TX$ is given by \[G_t(V):=\gamma_V'(t),\] see \cite[Ex. 1.1]{Hi} where $V\in T_xX$ and $\gamma_V:\mathbb R\to X$ is the (unique) geodesic with $\gamma_V'(0)=V$  and $\gamma_V(0)=x$.\footnote{More precisely, $\gamma_V(t)=g\exp (tX)K$ for $V=d(xK\mapsto gxK)_oX\in T_g(G/K)$, see \cite[p.14]{Hi}.} $G_t$ preserves the Riemannian metric on $TX$, in particular $G_t$ maps the spherical bundle $SX$ into itself. 
\begin{lemma}\label{eq: sphbund}
The bundle $SX$ can be identified with $G/M$. Also in the same vein 
\begin{equation*}
SX_\Gamma\cong \Gamma\backslash G/M.
\end{equation*}
\end{lemma}
\begin{proof}
See \cite[p.15]{Hi}. We have the identification
\begin{equation*}
TX=T(G/K)\cong G\times_K (\mathfrak{g}/\mathfrak k)=G\times_K\mathfrak p.
\end{equation*} 

The bundle $SX$ gets identified with $G/M$, since $K$ acts transitively on the spheres in $\mathfrak p$, i.e. the unit sphere in $\mathfrak{p}$ can be written as $K/M$. Thus $SX\cong G/K\times K/M\cong G/M$. For $SX_\Gamma$ we project from $X$ to $\Gamma\backslash X$. 
\end{proof}

The geodesic flow on the spherical bundle $SX\cong G/M$ is then given by right-translation with $A$, i.e. by 
\begin{equation*} (gM,\exp(tH_1))\mapsto G_t(gM)=g\exp(-tH_1)M,
\end{equation*} 
where \index{$H_1$, unit vector in $\mathfrak a^+$} $H_1$ is the unique unit vector in $\mathfrak a^+$, see \cite[p.89]{BO}. It also projects to a flow on $\Gamma\backslash G/M$ via 
\begin{equation*} (\Gamma gM,\exp(tH_1))\mapsto G_t(\Gamma g M)=\Gamma g\exp(-tH_1)M.
\end{equation*}

We already remarked that a closed geodesic corresponds to a closed orbit of the geodesic flow. We want to make this more precise. Let $\gamma\neq e$ be conjugated to $m_\gamma a_\gamma$ via $\alpha_\gamma$. Then the corresponding orbit is given by 
\begin{equation*} c_\gamma:=\{\Gamma\alpha_{\gamma^{-1}}\exp (-tH_1)M:t\in\mathbb R\}\subset \Gamma\backslash G/M,
\end{equation*} see \cite[p.90]{BO}. Note that indeed $G_t(\Gamma\alpha_{\gamma^{-1}}M)=\Gamma\alpha_{\gamma^{-1}}M$ for $t=l_\gamma$ and that $\gamma$ is primitive iff $t=l_\gamma$ is the smallest $t>0$ with $G_t(\Gamma\alpha_{\gamma^{-1}}M)=\Gamma\alpha_{\gamma^{-1}}M$.

Assume now that a continuous, left-$\Gamma$- and right-$K$-invariant function $\sigma$ is given. By the the same arguments as in the proof of \ref{chap: geoprop l} we get \begin{eqnarray*} \int_{H_\gamma\backslash G_{\delta_\gamma}}\sigma(\alpha_{\gamma^{-1}}x)dx&=& \int_{H_\gamma\backslash G_{\delta_\gamma}/K_{\delta_\gamma}}\sigma(\alpha_{\gamma^{-1}}x)dx\\&=& \int_{A/ \langle a_{\gamma_0}\rangle}\sigma(\alpha_{\gamma^{-1}}x)dx\\&=:& \int_{c_{\gamma_0}}\sigma,
\end{eqnarray*} where \index{$c_{\gamma_0}$, prime, closed geodesic to $\gamma$} \begin{equation}\label{eq: primclos}c_{\gamma_0}=\{\Gamma\alpha_{\gamma^{-1}}\exp (-tH_1) M:0\leq t\leq l_{\gamma_0}\}\end{equation} is the prime, closed orbit in $SX_\Gamma$ belonging to $\gamma$. Furthermore, we note an easy observation:
\begin{proposition}\label{chap: geom prop minv} 
The mapping $G\to\mathbb C$, $g\mapsto \int_{H_\gamma\backslash G_{\delta_\gamma}}\sigma(\alpha_{\gamma^{-1}}xg)dx$ is invariant under left translation by elements of $\in G_{\delta_\gamma}$. That is, for all $z\in G_{\delta_\gamma}$ $$\int_{H_\gamma\backslash G_{\delta_\gamma}}\sigma(\alpha_{\gamma^{-1}}xzg)dx=\int_{H_\gamma\backslash G_{\delta_\gamma}}\sigma(\alpha_{\gamma^{-1}}xg)dx.$$
\end{proposition}
\begin{proof} 
This follows, because $H_\gamma$, which is discrete, and $G_{\delta_\gamma}=M_{m_\gamma}A$ are unimodular. By \cite[p. 5]{Wi} the measure $dx$ on $H_\gamma\backslash G_{\delta_\gamma}$ is invariant under left translation by elements in $G_{\delta_\gamma}$, that is    $\int_{H_\gamma\backslash G_{\delta_\gamma}}\sigma(\alpha_{\gamma^{-1}}xzg)dx=\int_{H_\gamma\backslash G_{\delta_\gamma}}\sigma(\alpha_{\gamma^{-1}}xg)dx$ for all $z\in G_{\delta_\gamma}$ and continuous functions $\sigma$.
\end{proof}

\section{Real hyperbolic spaces}\label{chap: hyper}
We collect some facts concerning real hyperbolic symmetric spaces. We follow \cite{Ko} and \cite{Q}. 

The real hyperbolic symmetric spaces \index{$X$, (real hyperbolic) symmetric space} $X$ form one of the three main series for real symmetric spaces of noncompact type of rank one. If $X$ is a real hyperbolic space of dimension $l$, it can be described as a quotient $X=\tilde G/\tilde K$, where $\tilde G=O(1,l)$ and $\tilde K\subset G$ is a maximal compact subgroup in $\tilde G$ isomorphic to $O(l)$. Here $O(1,l)$ is the group of real $(l+1)\times (l+1)$ matrices which preserve the quadratic form \index{$[x,y]$, quadratic form on $\mathbb R^{l+1}$}

\begin{equation*}
[x,y]=x_0y_0-x_1y_1-\ldots x_ly_l
\end{equation*}
on $\mathbb{R}^{l+1}$. 

Furthermore $\tilde K=\{\pm 1\}\times SO(l)$. Real hyperbolic spaces can also be written as $X=G/K$, where \index{$SO_o(1,l)$, identity component of $SO(1,l)$} $G=SO_o(1,l)$ is the connected component of identity in $SO(1,l)$ and $SO(1,l)\subset O(1,l)$ is the subgroup of elements with determinant 1. Further, \index{$K$, max. compact in $SO_o(1,l)$} $K$ equals $\{1\}\times SO(l)$. This group $G$ has the advantage of being semisimple while $\tilde G$ is only reductive. 

We denote the Lie algebras of $G$ and $K$ by $ \mathfrak{g}$ resp. $ \mathfrak{k}$. Then we have the Cartan decomposition $ \mathfrak{g}=\mathfrak{k}\oplus \mathfrak{p}$ with Cartan involution $\theta$ defined by \[\theta(X):=JXJ,\] where \index{$J$, matrix in $\mathbb R^{l+1}$ defining $\theta$} 
\[
J:=\left(\begin{array}{cc} -1 & 0  \\ 0 & I_l   \end{array}\right),
\] $I_l$ the identity matrix in $SO(l)$. The Iwasawa decomposition reads $G=NAK$, where \index{$A$, abelian part of Iwasawa dec. of $SO_o(1,l)$} $A=\exp \mathfrak{a}\cong \mathbb{R}$, $ \mathfrak{a}\subset \mathfrak{p}$ is maximal abelian and \index{$N$, nilpotent part of Iwasawa dec. of $SO_o(1,l)$} $N=\exp \mathfrak{n}\cong \mathbb{R}^{l-1}$, $ \mathfrak{n}\subset \mathfrak{g}$ is an abelian Lie subalgebra. We will make the identification of $A$ resp. $\mathfrak a$ with $\mathbb R$ and of $ N $ resp. $\mathfrak n$ with $\mathbb R^{l-1}$ more precise below, see (\ref{eq: aidr}) and (\ref{eq: nidr}). The root space decomposition is simply

\begin{equation*}
\mathfrak{g}=\mathfrak{g_{-\alpha}}\oplus \mathfrak{g}_0\oplus \mathfrak{g}_\alpha.
\end{equation*}

Here \index{$alpha$@$\alpha$, positive root} $\alpha$ is the unique positive root of the pair $(\mathfrak g, \mathfrak a)$ with $ \mathfrak{n}=\mathfrak{g}_\alpha$,

\begin{equation*}
\mathfrak{g}_\alpha=\{X\in \mathfrak{g}:[H,X]=\alpha(H)X \mbox{ for all } H\in \mathfrak{a}\}.
\end{equation*}

It follows that \index{$rho$@$\rho$, $\rho=\frac{l-1}{2}\alpha$} $\rho=\frac{l-1}{2}\alpha$. We can also decompose $ \mathfrak{g}_0$ into $\mathfrak g_0= \mathfrak{a}\oplus \mathfrak{m}$, where $ \mathfrak{m}\subset \mathfrak{k}$ is the Lie algebra of the centraliser $M$ of $A$ in $K$. In our setting $M$ is just $\{1\}\times SO(l-1)\times \{1\}$. The subgroup $A$ is given by the set of all matrices of the form \index{$a_t$, element of $A$}
\[
a_t:=\left(\begin{array}{ccc}\cosh t & 0 & \sinh t\\ 0 & I_{l-1} & 0\\ \sinh t & 0 & \cosh t\end{array}\right),
\]
$t\in \mathbb{R}$ and \index{$I_{l-1}$, identity matrix} $I_{l-1}\in SO(l-1)$ the identity matrix. The Lie algebra $\mathfrak a$ of $A$ is spanned by \[H_0:=\left(\begin{array}{ccc}0 & 0 & 1\\ 0 & 0_{l-1} & 0\\ 1 & 0 & 0\end{array}\right),\]
where \index{$0_{l-1}$, null matrix} $0_{l-1}$ is the null matrix in $\mathfrak{so}(l-1)$. One computes $\alpha(H_0)=1$, $B(H_0,H_0)=2(l-1)\alpha(H_0)^2=2(l-1)$ and \index{$rho0$@$\rho_0$, $\rho_0=\rho(H_0)$} $\rho_0:=\rho(H_0)=\frac{l-1}{2}$, see \cite[p.135]{GV}. Here \[B(X,Y)=(l-1)\mathrm{Tr}(XY)=\mathrm{Tr}\left(\mathrm{ad}X\mathrm{ad} Y\right)\] for $X,Y\in \mathfrak{g}$ denotes the Cartan-Killing form. We identify $A$ resp. $\mathfrak a$ with $\mathbb R$ via 
\begin{equation}\label{eq: aidr}
\mathbb{R}\to \mathfrak{a} \to A \mbox{ , }
t\mapsto tH_0\mapsto \exp tH_0. 
\end{equation}  

The subgroup $N$ is given by \index{$n_u$, element of $N$}
\[
n_u:=\left(\begin{array}{ccc}1+\frac{1}{2}|u|^2 & u^T & -\frac{1}{2}|u|^2\\ u & I_{l-1} & -u\\ \frac{1}{2}|u|^2 & u^T & 1-\frac{1}{2}|u|^2\end{array}\right),
\]
$u\in\mathbb R^{l-1}$. Its Lie algebra $\mathfrak n$ consists of all matrices $X_u$ \index{$X_u$, element of $\mathfrak n$}, $u\in \mathbb{R}^{l-1}$, of the form \[
X_u:=\left(\begin{array}{ccc}0 & u^T & 0\\ u & 0_{l-1} & -u\\ 0 & u^T & 0\end{array}\right)
\] with $0_{l-1}$ the zero matrix in $\mathbb R^{(l-1)\times (l-1)}$. It follows that \index{$X_{e_1}$, element of $\mathfrak n$ to unit vector $e_1$} \begin{equation}\label{eq: normx1}B_\theta(X_{e_1},X_{e_1})=-B(X_{e_1},\theta X_{e_1})=4(l-1), \end{equation} where $e_1=(1,0,\ldots,0)^T\in\mathbb R^{l-1}$. Similar to (\ref{eq: aidr}) the identification here is
\begin{equation}\label{eq: nidr}
\mathbb{R}^{l-1}\to \mathfrak{n} \to N \mbox{ , }
u\mapsto X_u \mapsto \exp X_u 
\end{equation}
 
Conjugation by $A$ resp. $M$ on $N$ can now be described as

\begin{equation*}
a_tn_ua_{-t}=n_{e^tu} \mbox{ resp. } mn_um^{-1}=n_{m\cdot u}.
\end{equation*}

Furthermore,
 \[n_un_{u'}=n_{u+u'},\] 

while \[\mathrm{Ad}(a_t)X_u=e^tX_u=X_{e^tu}\] and \begin{equation}\label{eq: mactonn}\mathrm{Ad}(m)X_u=X_{m\cdot u}\end{equation} for $a_t\in A$, $m\in M\cong SO(l-1)$ and $X_u\in\mathfrak n\cong\mathbb R^{l-1}$. We further recall that the  Weyl group $W=N_K(A)/M$, where $N_K(A)$ is the normalizer of $A$ in $K$, is isomorphic to the symmetric group of two elements.

The space $X=G/K$ can be described as the image of the open set \[\{x\in\mathbb R^{l+1}:[x,x]>0\}\] under the mapping from $ \mathbb{R}^{l+1}$ to the unit ball \index{$B(\mathbb R^l)$, unit ball in $\mathbb R^l$} $B(\mathbb R^l)$ in $\mathbb R^l$ given by \[x\mapsto y \mbox{ , } y_i=x_ix_0^{-1}.\] 

The group $G$ acts then on $B(\mathbb R^l)$ by fractional linear transformation, that is if $g\in G$ is given by
\begin{equation}\label{eq: gmatrix}g=\left(\begin{array}{cc} a & b^T \\ c & d\end{array}\right)\end{equation} with $a\in \mathbb{R}$, $b,c\in \mathbb{R}^{l}$ and $d\in \mathbb{R}^{l\times l}$, then 

\begin{equation}\label{def: gaction}
g\cdot y=(dy+c)(\langle b,y\rangle+a)^{-1}.
\end{equation} 

Here \index{$\langle.,.\rangle$, inner product} $\langle.,.\rangle$ denotes the standard inner product on $\mathbb R^l$.
One can show that, see \cite[p.111]{DH},

\begin{equation}\label{eq: bik}
1-\left|g\cdot y\right|^2=\frac{1-|y|^2}{\left|\langle b,y\rangle +a\right|^{2}},
\end{equation}
where $|.|$ means the euclidean norm derived from $\langle.,.\rangle$ on $\mathbb R^{l}$. Hence,
\begin{equation}\label{eq: bikspec}
1-\left|g\cdot 0\right|^2=|a|^{-2}.
\end{equation}

\begin{lemma}\label{lem: bikinv}
The mapping \[g\mapsto 1-\left|g\cdot 0\right|^2\] from $SO_o(1,l)\to\mathbb R$ is bi-$K$-invariant.
\end{lemma}
\begin{proof}
Since $K$ equals the stabilizer of $0$ in $G$, the right-$K$-invariance is clear. For the left-$K$-invariance we note that the product of $g=\left(\begin{array}{cc} a & b^t \\ c & d\end{array}\right)$ as in (\ref{eq: gmatrix}) and $k=\left(\begin{array}{cc} 1 & 0 \\ 0 & f\end{array}\right)$, $f\in SO(l)$, is of the form
\[k\cdot g= \left(\begin{array}{cc} a & b^t \\ f\cdot c & f\cdot d\end{array}\right)\] it follows from (\ref{eq: bikspec}) that for $k\in K$
\[1-\left|kg\cdot 0\right|^{2}=1-\left|g\cdot 0\right|^{2}.\]

\end{proof}

For later purpose let us make two computations. For $n_u\in N$ and $a_t\in A$ resp. $m\in M$ we have by (\ref{eq: bikspec})
\begin{eqnarray}\label{comp: atn}
1-\left|a_tn_u\cdot 0\right|^2&=&\left(\cosh t\left(1+\frac{1}{2}|u|^2 \right)+\frac{1}{2}|u|^2\sinh t\right)^{-2}\\&=& \left(\cosh t +\frac{|u|^2}{2}e^t\right)^{-2}. 
\end{eqnarray}

For the next computation we identify $M$ with $SO(l-1)$ and $m$ with the $(l-1)\times (l-1)$-matrix $(m)_{i,j=1}^{l-1}$ in 
$SO(l-1)$. Let $(m)_{1,1}$ be the first entry on the diagonal. Then by (\ref{eq: bikspec}) 
\begin{eqnarray*}
1-\left|n_{-u}a_tmn_u\cdot 0\right|^2&=&1-\left|n_{-u}a_tn_{m\cdot u}\cdot 0\right|^2\\&=& 
\left(\left(1+\frac{|u|^2}{2} \right)\left(\cosh t+ \frac{|u|^2}{2}e^{-t}\right)-\langle u,m\cdot u\rangle+\frac{|u|^2}{2}\left(\sinh t-\frac{|u|^2}{2}e^{-t}\right)\right)^{-2}\\&=& 
\left(-\langle u,m\cdot u\rangle+(1+|u|^2)\cosh t\right)^{-2}.
\end{eqnarray*}

In particular, for $u=r\cdot e_1=(r,0,\ldots,0)^T\in\mathbb R^{l-1}$, $r\in\mathbb R$.
\begin{eqnarray}\label{comp: natmn}
1-|n_{-re_1}a_tmn_{re_1}\cdot 0|^2&=&
\left(-(m)_{1,1}r^2 +(1+r^2)\cosh t\right)^{-2}.
\end{eqnarray}

Finally, let \index{$H:SO_o(1,l)\to \mathfrak{a}$, Iwasawa projection} $H:SO_o(1,l)\to \mathfrak{a}$ be the \index{Iwasawa projection} Iwasawa projection to $SO_o(1,l)=KAN$ defined by
\[H(k a_tn)=H(k\exp tH_0 n):=tH_0,\]
where $H_0\in\mathfrak a$ with $\alpha(H_0)=1$. From \cite[(5.9)]{Ko} we cite that for $g=(g)_{i,j=1}^{l+1}$ in $SO_o(1,l)$ 
\begin{equation}\label{eq: hproj}
H(g)=\ln|g_{1,1}+g_{0,l+1}|.
\end{equation}

Then let \index{$w$}
\[
w:=\left(\begin{array}{ccc} I_{l-1} & 0 & 0 \\ 0 & -1 & 0 \\ 0&0& -1  \end{array}\right),
\]
be a representative for the nontrivial Weyl group element in $K=\{1\}\times SO(l)$. From (\ref{eq: hproj}) we conclude that
\begin{equation*}
H(n_uw)=tH_0\end{equation*}
with $t=\ln\left(1+|u|^2 \right)
$. Also 
\begin{equation}\label{eq: Hnuw}
H(n_{-u}w)=H(n_uw)=\ln\left(1+|u|^2 \right)H_0.
\end{equation}

\section{Spherical Transforms}\label{sec: spheretra}
In this section we recall some facts about spherical transforms on semisimple Lie groups. We start by the general definition of a spherical function and spherical transform. In the next section we explicitly calculate the spherical transform of a function $f_k$ depending on some complex parameter $k$ when $G/K$ is a real hyperbolic space.

Let $G=NAK$ be a semisimple, noncompact Lie group, \index{$A$ $:G\to\mathfrak a$, projection} $A:G\to \mathfrak{a}$ the corresponding projection such that \[g=n(g)\exp A(g)k(g)\] and \begin{equation}\label{def: horobrac}\langle gK, kM \rangle:=A(k^{-1}g)\end{equation} is the so-called horocycle bracket \index{horocycle bracket} where $X=G/K$ and boundary \index{$B=K/M$} $B=K/M$.
\begin{definition} \cite{GGA}\\
For $f\in C_c^\infty(G/K)$ we 
define its \index{Fourier transform} \index{$\mathcal{F}$, Fourier transform} Fourier transform by 
\begin{equation}\label{def: fouriertransf} 
\mathcal{F}(f,\lambda,b):=\int_{G/K}f(gK)e^{(-i\lambda+\rho)\langle gK,b\rangle}d(gK),
\end{equation} where \index{horocycle bracket} $\langle\cdot,\cdot\rangle$ is the horocycle bracket.
\end{definition}
\begin{definition}\index{spherical function}
We call a function in $C^\infty(G//K)$ \index{(elementary) spherical} \textit{(elementary) spherical}, if it is an eigenfunction of $U(\mathfrak{g}_{\mathbb{C}})^K$, the subset of $K$-invariants in $U(\mathfrak{g}_\mathbb{C})$.
\end{definition}

The following theorem gives the integral formula for spherical functions. It describes spherical functions as integral transforms of the wave $e^{(i\lambda+\rho)A(\cdot)}$. 
\begin{theorem}\label{Thm: HC}
Every spherical function is given by \index{$fi$@$\varphi_\lambda$, spherical function to $\lambda\in \mathfrak{a}^+_{\mathbb C}$}
\begin{equation*}
\varphi_\lambda(g)=\int_Ke^{(i\lambda+\rho)A(kg)}dk,
\end{equation*} where $\lambda\in \mathfrak{a}^*$ and $\varphi_\lambda=\varphi_{s\lambda}$ for $s\in W$.
\end{theorem}
\begin{proof}
\cite[p.418]{GGA}
\end{proof}
Note that $A(.)$ is bi-$M$-invariant, hence the domain of integration can be changed to $K/M$ if $dk_M$ corresponds to the (normalized) measures $dk$ and $dm$.

We state some useful formulae for spherical functions.
\begin{remark}\label{rem: sphere}
If the real rank of $G$ is one and \index{$H$ $:G\to \mathfrak{a}$, projection} \index{$\log a$, element of $\mathfrak a$} $H:G\to \mathfrak a, H(kna):=\log a$, where $\exp \log a=a$, corresponds to the Iwasawa decomposition $G=KAN$, then $H(g^{-1}k)=-A(k^{-1}g)$ and 
\begin{eqnarray*}
\varphi_\lambda(g)&=&\int_Ke^{(i\lambda+\rho)A(kg)}dk\\ 
&=&  \int_Ke^{(i\lambda+\rho)A(k^{-1}g)}dk \mbox{ by unimodularity}\\
&=& \int_Ke^{(-i\lambda-\rho)H(g^{-1}k)}dk\\
&=& \int_Ke^{(i\lambda-\rho)H(g^{-1}k)}dk \mbox{ by Weyl group invariance}\\
&=& \int_K e^{(i\lambda-\rho)H(gk)}dk \mbox{ , see \cite[p.89]{Sc}}.
\end{eqnarray*}
\end{remark}


With this preparation we can define the spherical transform on bi-$K$-invariant functions.
\begin{definition} See \cite[p.457]{GGA}\\ \index{spherical transform} \index{$\mathcal S(f,\lambda)$, spherical transform} Let $f\in C^\infty(G//K)$ and $\lambda\in\mathfrak a_{\mathbb C^*}$. We define define the spherical transform of $f$ by  
\begin{equation}\label{def: sphericaltransform} 
\mathcal{S}(f,\lambda):=\int_Gf(g)\varphi_{-\lambda}(g)dg,
\end{equation} 
whenever the integral is finite and where $\varphi_\lambda(g)=\int_Ke^{(i\lambda+\rho)A(kg)}dk$ is an elementary spherical function.
\end{definition}

From Remark \ref{rem: sphere} one can deduce that the spherical transform can be factorized into a product of the euclidean Fourier transform on $A$ \index{$\hat f$, euclidean Fourier transform} \[\hat{f}(\lambda):=\int_A f(a)e^{-i\lambda\log a}da\] and the Abel transform, see Remark \ref{rem: gang} for the definition of the Abel transform $F_f$.
\begin{proposition}\label{prop: spherefact}
The spherical transform can be factorized as
\begin{eqnarray*}
\mathcal{S}(f,\lambda)&=&\widehat{F_f}(\lambda)\\&=& \int_{A}\int_N f(an)e^{(-i\lambda+\rho)\log a}dnda.
\end{eqnarray*}
\end{proposition}
\begin{proof}
\cite[Ch.II \S 5 (37)]{GGA} or \cite[Prop.3.3.1.]{GV}
\end{proof}
\begin{theorem}\label{th: hccf}\index{Harish-Chandra's $c$-function} \index{$c(\lambda)$, Harish-Chandra's $c$-function}
(\textit{Harish-Chandra's} $c$-function)\\ Let $G$ be semisimple, noncompact of real rank one. For $\lambda\in \mathfrak{a}^*$ with $\mathrm{Re}(i\lambda)\in \mathfrak{a}_{+}^*$ the integral 
\begin{equation*} 
c(\lambda)=\int_{\overline{N} }e^{-(i\lambda+\rho)H(\bar n)}d\bar{n}
\end{equation*} converges absolutely and admits a meromorphic extension in the parameter $\lambda$ to $\mathfrak a_{\mathbb{C}}^*\cong \mathbb{C}$. Here the measure $d\bar n$ on \index{$\overline N$, $\overline{N}=\theta N$} $\overline{N}:=\theta N$ is normalized such that $\int_{\bar N}e^{-2\rho\left(H(\bar n)\right)}dn=1$.
\end{theorem}
\begin{proof}
\cite[p.436]{GGA} 
\end{proof}

\section{The spherical transform of $f_k$}\label{sec: spher traf}
In this section we specialize to $G=SO_o(1,l)$. We want to compute the spherical transform of the bi-$K$-invariant function, see Lemma \ref{lem: bikinv}, \index{$f_k$, bi-$K$-invariant function} \[f_k:G\to \mathbb{R}\mbox{ , } g\mapsto \left(1-|g\cdot 0|^2\right)^{k/2}\] depending on $k\in \mathbb{C}$. We will use $f_k$ to obtain the zeta function $ \mathcal{R}(k;\varphi_n) $, see Section \ref{sec: zeta}. The spherical transform of $f_k$ will in turn produce the meromorphic continuation of $ \mathcal{R}(k;\varphi_n)$ to $ \mathbb{C}$ via the spectral trace of the operator $ \varphi_n\cdot \pi_R(f_k)$ in Chapter \ref{chap: meromorph}.

The spherical transform of $f_k$ is given by Proposition \ref{prop: spherefact} by

\begin{equation*}
\mathcal{S}(f_k,\mu)= \int_A\int_N f_k(a n) e^{(-i\mu+\rho)\log a}dnda.
\end{equation*}

Now we identify $A$ with $\mathbb{R}$ by $a_t=\exp tH_0\mapsto t$, $\alpha(H_0)=1$, and $N$ with $\mathbb R^{l-1}$, $\exp X_u\mapsto u$, see (\ref{eq: nidr}). From (\ref{comp: atn}ff) it follows that \begin{equation}\label{eq: fkatsp}f_k(a_t \exp X_u)=\left(\cosh t +\frac{|u|^2}{2}e^t\right)^{-k} .\end{equation}

The spherical transform of $f_k$ is hence given by
\begin{eqnarray*}
\mathcal{S}(f_k,\mu)&=& \int_{-\infty}^{\infty}\int_N f_k(\exp tH_0  n) e^{(-i\mu+\rho_0)t}dndt
\\&\overset{(\ref{eq: fkatsp})}=&\int_{-\infty}^\infty\int_N \left( \cosh t+ \frac{1}{2}|u|^2e^{t}\right)^{-k} e^{(-i\mu+\rho_0)t}dudt\\
&\overset{t\mapsto -t}=&
\int_{-\infty}^{\infty}\int_N \left( \cosh t+ \frac{1}{2}|u|^2e^{-t}\right)^{-k} e^{(i\mu-\rho_0)t}dudt
\\&=& \omega_{l-1}\cdot \int_{-\infty}^{\infty}\int_{0}^\infty s^{l-2}\left(\cosh t +\frac{1}{2}s^2e^{-t}\right)^{-k}e^{(i\mu-\rho_0)t}dsdt\\
&=& \omega_{l-1}\cdot \int_{-\infty}^{\infty} (\cosh t)^{-k}e^{(i\mu-\rho_0)t} \int_{0}^\infty s^{l-2}\left(1+s^2(1+e^{2t})^{-1}\right)^{-k} ds dt,
\end{eqnarray*}
where we used \index{polar coordinates} polar coordinates, namely we identify $N\cong \mathbb{R}^{l-1}$, $u\mapsto X_u$, see (\ref{eq: nidr}),
\begin{eqnarray*}
\int_Nf(n)dn&=&\int_{\mathbb{R}^{l-1}}f(\exp X_u)du\\
&=& \int_{0}^\infty \int_{\partial B_s(0)}fdSds\\
&=& \omega_{l-1}\int_0^\infty s^{l-2}  f(\exp sX_{e_1})ds
\end{eqnarray*}
for radial functions $f$, i.e. $f(n)=\int_Mf(m\cdot n)dm$ for all $n\in N$. Here \index{$omega$@$\omega_{l-1}$, volume of unit ball}\index{$B_1(0)$, unit ball} $B_1(0)=\{x\in\mathbb R^{l-1}:|x|\leq 1\}$ and $\omega_{l-1}:=|\partial B_1(0)|$ with respect to the Lebesgue measure in $\mathbb R^{l-2}$ for $l\geq 3$. For $l=2$ we set $\omega_1:=2$. Next we substitute $s\mapsto (1+e^{2t})^{1/2}s$ to get
\begin{eqnarray*}
&&\omega_{l-1}\cdot\int_{-\infty}^{\infty} (\cosh t)^{-k}(1+e^{2t})^{\frac{l-1}{2}}e^{(i\mu-\rho_0)t}dt\int_{0}^{\infty}s^{l-2}(1+s^2)^{-k}ds
\\&=& \omega_{l-1}\cdot 2^{ \frac{l-1}{2}}\int_A (\cosh t)^{-k+ \frac{l-1}{2}}e^{ \frac{l-1}{2}t }\cdot e^{(i\mu-\rho_0)t}dt\int_{0}^{\infty}s^{l-2}(1+s^2)^{-k}ds.
\end{eqnarray*}

Now we remember that $\rho_0= \frac{l-1}{2}$ to find that we have to compute the integral 
\begin{eqnarray*}
2^{ \frac{l-1}{2}}\int_{-\infty}^\infty (\cosh t)^{-k+ \frac{l-1}{2} }e^{i\mu t}dt&=&2^{ \frac{l-1}{2}} \int_0^\infty 2^{k- \frac{l-1}{2} }\left(v+v^{-1}\right)^{-k+ \frac{l-1}{2} } v^{i\mu-1} dv\\ 
&=& 2^{k}\int_0^\infty \frac{v^{i\mu-1}}{\left( \frac{v^2+1}{v}\right)^{k-\frac{l-1}{2}}}dv,
\end{eqnarray*}
by the substitution $e^t=v$. Then we substitute $v\mapsto v^{1/2}$ to obtain
\begin{equation}\label{eq: intbeta}
2^{k-1} \int_0^\infty v^{ \frac{k+i\mu-\rho_0}{2}-1 }\left(v+1 \right)^{-k+\rho_0}dv.
\end{equation}

Next we use the integral formula for the Beta function \index{Beta function} \index{$B(x,y)$, beta function}
\begin{equation}\label{eq: betafunc}B(x,y)=\int_0^\infty \frac{u^{x-1}}{(1+u)^{x+y}}du=\frac{\Gamma(x)\Gamma(y)}{\Gamma(x+y)}
\end{equation} to obtain that (\ref{eq: intbeta}) equals 
\begin{equation*}
2^{k-1}B\left( \frac{k+i\mu-\rho_0}{2}, \frac{k-i\mu-\rho_0}{2}\right).
\end{equation*}

We recall the following facts on the Gamma function:
\begin{remark}\label{rem: gamma}
The Gamma function $\Gamma(z)=\int_0^\infty t^{z-1}e^{-t}dt$ converges absolutely for $\mathrm{Re}(z)>0$ and defines there an entire function. It can be continued analytically to $\mathbb{C}-\{0,-1,-2,\ldots\}$. The poles in $-n$, $n\in\mathbb{N}_0$, are of first order with residues $\mathrm{Res}(\Gamma;-n)=\frac{(-1)^n}{n}$. The Gamma function has no zeros, \cite[Ch. IV 1]{FB}.

Furthermore, for $x$ and $y$ real we have \begin{equation}\label{eq: gammaasympt}\lim_{|y|\to\infty}\left|\Gamma(x+iy)\right|e^{\frac{\pi}{2}|y| }|y|^{\frac{1}{2}-x }=\sqrt{2\pi},\end{equation}
see \cite[8.328]{GR}.
\end{remark}

That is, $B(x,y)$ is defined for \[\mathrm{Re}(x),\mathrm{Re}(y)\neq 0,-1,-2,\ldots,\] in our case \[k\neq \rho_0\pm i\mu,\rho_0-2\pm i\mu,\ldots.\]

To sum up:
\begin{proposition}
The spherical transform of $f_k$ is given by
\begin{eqnarray}\label{eq: stfk}
\mathcal{S}(f_k,\mu)&=&\omega_{l-1}\int_0^\infty s^{l-2}(s^2+1)^{-k}ds\cdot 2^{k-1}B\left( \frac{k+i\mu-\rho_0}{2}, \frac{k-i\mu-\rho_0}{2}\right)\nonumber\\
&=& \omega_{l-1}\cdot 2^{k-1} \frac{\Gamma(k-\rho_0)\Gamma(\rho_0)}{2\Gamma(k)}B\left( \frac{k+i\mu-\rho_0}{2}, \frac{k-i\mu-\rho_0}{2}\right)\nonumber\\
&=& \omega_{l-1}\cdot 2^{k-2}B(k-\rho_0,\rho_0)B\left( \frac{k+i\mu-\rho_0}{2}, \frac{k-i\mu-\rho_0}{2}\right) .
\end{eqnarray}
\end{proposition}

As we will later see in Section \ref{sec: zeta} this function $f_k$ leads to an operator of trace class and its trace in turn will produce the (auxiliary) zeta function $\mathcal R(\varphi)$, see (\ref{def: auxzet}). 
\chapter{Polar decomposition and radial parts}\label{chap: radial}

In this chapter $G$ will always come from a real hyperbolic symmetric space, which is at least 3 dimensional. That is, $G=SO_o(1,l)$ with $l\geq 3$, see Chapter \ref{chap: hyper}. We fix an Iwasawa decomposition $G=ANK$ resp. $\mathfrak{g}=\mathfrak{a}\oplus\mathfrak{n}\oplus\mathfrak{k}$. Then $M=Z_K(A)\cong SO(l-1)$. For $l=2$, $M$ is hence trivial and we exclude this case.

In Section \ref{sec: mac} we discuss the action of $M$ on $N$ which is essentially described by the natural action of $SO(l-1)$ on $\mathbb R^{l-1}$. In the next Section \ref{sec: rapo} we develop a theory of polar decomposition for differential operators $D$ on $G$ for a certain decomposition of $G$ which is similar to the polar decomposition $G=KAK$. This also leads to the definition of radial parts of differential operators for $M$-invariant functions. For the succeeding chapters, in particular Chapter \ref{chap: ode}, Section \ref{sec: radial.1} is crucial where we apply the theory of polar decomposition from Section \ref{sec: mac} to the Casimir operator $\Omega$. While Section \ref{sec: rapo} suggests formula (\ref{eq: omegapolar}) for $\Omega$, see Theorem \ref{th:radial part1}, the calculations of Section \ref{sec: radial.1} which lead to (\ref{eq: omegapolar}) only depend on Section \ref{chap: casimir} and \ref{sec: mac}.
\section{The action of $M$ on $\mathfrak n$}\label{sec: mac}

We assume that the dimension of $\mathfrak n=\mathrm{Lie}(N)$ is $l-1$ and we consider the adjoint action of $M$ on $\mathfrak n $. It is content of Kostant's double transitivity theorem, see \cite[Th. 8.11.3]{Wal2}, that the generic orbits are spheres, i.e. if $0\neq X\in \mathfrak{n}$, then \[M\cdot X=\{X'\in \mathfrak{n}:|X'|=|X|\}.\] 

It follows that the tangent space $\mathfrak m\cdot X$ to the $M$-orbit through $X$ is $l-2$ dimensional and thus \[\left(\mathfrak{m}\cdot X\right)^{\bot_{\mathfrak n}}:=\{X'\in\mathfrak n: B_\theta(X',\mathfrak{m}\cdot X)=0\}\] is 1-dimensional. Now it is a general fact that for any compact group $H$ and real, finite dimensional representation of $H$ on some vector space $V$, the space $\left(H\cdot v\right)^{\bot}$ meets every $H$-orbit in $V$, see \cite[Lem. 1]{Dad}. We call a linear subspace of $\mathfrak n$ a \index{section} \textit{section}, if it intersects every $M$-orbit. 
 A \textit{slice} \index{slice} is then a subset of a section which intersects every regular orbit exactly once.

\begin{lemma}\label{lem: mpol}
Let $0\neq X\in\mathfrak n$. Then $\left(\mathfrak{m}\cdot X\right)^{\bot_{\mathfrak n}}$ meets every orbit orthogonally, i.e $$B_\theta(Z\cdot X,X')=0$$ for all $Z\in\mathfrak m$ and $X'\in \left(\mathfrak{m}\cdot X\right)^{\bot_{\mathfrak n}}$.\footnote{The Lemma would still be true if the dimension of $\left(\mathfrak{m}\cdot X\right)^{\bot_{\mathfrak n}}$ is two. In particular for any semisimple $\mathfrak g$ of real rank one, $M$ acts polarly on $\mathfrak n$.}
\end{lemma} 
\begin{proof}
The claim will follow as soon as we show $\left(\mathfrak{m}\cdot X\right)^{\bot_{\mathfrak n}}=\mathbb{R}X$. For this we compute for arbitrary $Z\in\mathfrak m$
\begin{eqnarray*}
B_\theta(Z\cdot X,X)&=& -B([Z,X,],\theta X)\\&=& -B(Z,[X,\theta X])\\ &=&0,
\end{eqnarray*}
since $[X,\theta X]\in\mathfrak a$ which is orthogonal to $\mathfrak k$ with respect to $B_\theta$. 
\end{proof}

In other words $M$ acts \textit{polarly} on $\mathfrak n$, that is, there is a section which intersects every $M$-orbit orthogonally. See also \cite[Chap.VI 30]{M} for more information on polar actions.

We note the following:
\begin{lemma}
Let $\mathfrak{s}_1$ and $\mathfrak{s}_2$ be 1-dimensional sections for $M$ acting on $N$. If
\begin{eqnarray*}
\mathfrak{g}&=&\mathfrak{k}\oplus \mathfrak{a}_1\oplus\mathfrak{s}_1\oplus\mathfrak{s}_1^{\bot_{\mathfrak{n}_1}}\\
&=& \mathfrak{k}\oplus\mathfrak{a}_2\oplus\mathfrak{s}_2\oplus\mathfrak{s}_2^{\bot_{\mathfrak{n}_2}},
\end{eqnarray*}
with $\mathfrak{s}_1\subset\mathfrak{n}_1$, $\mathfrak{s}_2\subset\mathfrak{n}_2$, then there is some $k\in K$ such that \[\mathfrak{a}_1^k=\mathfrak{a}_2 \mbox{ and } \mathfrak{s}_1^k=\mathfrak{s}_2.\]

In particular, if $\mathfrak k$ and $\mathfrak a_1=\mathfrak{a}_2$ are fixed, then there is some $m\in M$ such that $\mathfrak s_1^m=\mathfrak{s}_2$.
\end{lemma}
\begin{proof}
This follows since by (\ref{eq: iwauniq}) in Chapter \ref{chap: casimir} there is some $k\in K$ such that \[\mathfrak{a}_1^k=\mathfrak{a}_2 \mbox{ and } \mathfrak{s}_1^k+\left(\mathfrak{s}_1^{\bot_{\mathfrak{n}_1}}\right)^ k=\mathfrak{s}_2+\mathfrak{s}_2^{\bot_{\mathfrak{n}_2}}.\]

But $\mathfrak{s}_1^k$ is a one dimensional subspace of $\mathfrak n_2$, hence a section for $M_2:=Z_K(A_2)$ acting on $\mathfrak{n}_2$ and consequently there is some $m\in M_2$ with \[(\mathfrak{s}_1^k)^m=\mathfrak{s}_2.\]

Still of course, \[(\mathfrak{a}_1^k)^m=\mathfrak{a}_2^m=\mathfrak{a}_2.\]
\end{proof}


We fix \index{$X_1$, unit vector in $\mathfrak n$} $X_1\neq 0$ of length 1 with respect to $B_\theta(.,.)$ and \index{$\mathfrak{s}$, section in $\mathfrak n$} \[\mathfrak{s}:= (\mathfrak{m}\cdot X_1)^{\bot_\mathfrak{n}}= \mathbb{R}X_1.\] 

We denote by \index{$\mathfrak{s}'$, regular elements in $\mathfrak s$} \[\mathfrak{s}':= \mathbb{R}X_1\setminus\{0\}\] the set of \index{regular element} \textit{regular elements}, i.e. the set of elements with maximal orbit dimension. The first aim of this section is to find a decomposition of $\mathfrak{g}$ resp. $G$ according to the section $\mathfrak{s}$. Later we will study differential equations of invariant functions on $\mathfrak{s}$. 

Let \index{$\mathfrak{z}_{\mathfrak{m}}(\mathfrak{s})$, Lie algebra of $Z_M(\mathfrak s)$}
\begin{equation*}
\mathfrak{z}_{\mathfrak{m}}(\mathfrak{s}):=\{Z\in\mathfrak{m}:[Z,X']=0 \mbox{ for all } X'\in \mathfrak{s}\}.
\end{equation*}
be the centralizer of $\mathfrak{s}$ in $\mathfrak m$ and let \index{$\mathfrak{z}_{\mathfrak{m}}(\mathfrak{s})^{\bot_{\mathfrak{m}}}$, orthognal complement to $\mathfrak{z}_{\mathfrak{m}}(\mathfrak{s})$}
\begin{equation*}
\mathfrak{z}_{\mathfrak{m}}(\mathfrak{s})^{\bot_{\mathfrak{m}}}:=\{W\in\mathfrak{m}:B(Z,W)=0 \mbox{ for all } Z\in \mathfrak{z}_{\mathfrak{m}}(\mathfrak{s})\}.
\end{equation*}
be its orthogonal complement with respect to the Killing form. Finally,
\begin{equation*}
\mathfrak{s}^{\bot_\mathfrak{n}}:=\{Y\in\mathfrak n:B_\theta(X', Y)=0 \mbox{ for all } X'\in \mathfrak{s}\}.
\end{equation*}

We start by a lemma showing that $\mathfrak{z}_{\mathfrak{m}}(\mathfrak{s})^{\bot_{\mathfrak{m}}}$ is isomorphic to  $\mathfrak{s}^{\bot_{\mathfrak n}}$.
\begin{lemma}\label{l2}
The map $\mathrm{ad } X':\mathfrak{z}_{\mathfrak m}(\mathfrak{s})^{\bot_{\mathfrak{m}}}\to \mathfrak{s}^{\bot_\mathfrak{n}}$, $Z\mapsto [Z,X']$ is an isomorphism for all $X'\in\mathfrak{s}'$.
\end{lemma}
\begin{proof}
Let $X'\in\mathfrak{s}'$. Since the action is polar, this implies \[Z_{M}(X')=Z_{M}(\mathfrak{s}),\] \cite[30.23]{M}. Also by polarity of the action we get the orthogonal decomposition 
\begin{equation*} \mathfrak{n}=[\mathfrak{m},X']\oplus \mathfrak{s},
\end{equation*} i.e. $[\mathfrak{m},X']=\mathfrak{s}^{\bot_\mathfrak{n}}$. Furthermore, \begin{equation*} [\mathfrak{m},X']=[\mathfrak{z}_{\mathfrak{m}}(X')\oplus\mathfrak{z}_{\mathfrak{m}}(X')^{\bot_{\mathfrak{m}}},X']=[\mathfrak{z}_{\mathfrak{m}}(X')^{\bot_{\mathfrak{m}}},X'],\end{equation*} 
where $\mathfrak{z}_{\mathfrak{m}}(X')$ is the Lie algebra of $Z_M(X')=Z_M(\mathfrak{s})$. That is, $\mathrm{ad }X'$ maps $\mathfrak{z}_{\mathfrak{m}}(\mathfrak{s})^{\bot_{\mathfrak{m}}}$ onto $\mathfrak{s}^{\bot_\mathfrak{n}}$ and both spaces have the same dimension, since \[M\cdot X'\cong M/Z_{M}(X').\] 
\end{proof}

The proof shows that lemma is also true for any subgroup $M'\subset M$ acting polarly on $\mathfrak n$, if we replace $\mathfrak s$ with a section for $M'$ acting on $\mathfrak n$.

Lemma \ref{l2} allows us now to write $\mathfrak g$ as a direct sum involving $\mathfrak{s}$ and $\mathfrak z(\mathfrak{s})$. This decomposition is analogous to the one derived from $K$ acting on $\mathfrak p$ for a semisimple Lie algebra $\mathfrak g$, see for example \cite[Chap.4]{GV}. 

\begin{proposition}\label{t1}
Let $\mathfrak{g}=\mathfrak{so}(1,l)=\mathfrak{k}\oplus\mathfrak{a}\oplus\mathfrak{n}$, $\mathfrak{m}=\mathfrak{z}_\mathfrak{k}(\mathfrak{a})$ and $\mathfrak{s}$ any section for $M=\exp\mathfrak m$ acting on $\mathfrak n$. We have the following direct decompositions for any $X'\in\mathfrak{s}'$:
\begin{eqnarray*}
\mathfrak{g}&=& (\mathfrak{z}_{\mathfrak{m}}(\mathfrak{s})^{\bot_{\mathfrak{m}}})^{\exp -X'}\oplus \mathfrak{a}\oplus\mathfrak{s}\oplus \mathfrak{k}\\&=&\left(\mathfrak{z}_{\mathfrak{m}}(\mathfrak{s})^{\bot_{\mathfrak{m}}}\oplus (1+\theta)\mathfrak{s}\right)^{\exp -X'}\oplus\mathfrak{s}\oplus\mathfrak k. 
\end{eqnarray*}
 
\end{proposition}
\begin{proof}
Let $\mathfrak s=\mathbb{R}X_1$ for some $0\neq X_1\in\mathfrak{n}$. Since the sum $\mathfrak{a}\oplus\mathfrak{n}\oplus\mathfrak{k}$ is direct, it suffices to show that $(\mathfrak{z}_\mathfrak{m}(\mathfrak{s})^\bot_\mathfrak{n})^{\exp -X'}$ is not contained in $\mathfrak{a}\oplus\mathfrak{s}\oplus\mathfrak{k}$. The first claim then follows by dimension counting, see Lemma \ref{l2}. 

So let $Z^{\exp-X'}\in \left(\mathfrak{z}_\mathfrak{m}(\mathfrak{s})^{\bot_{\mathfrak{m}}}\right)^{\exp -X'}\cap \mathfrak{a}\oplus\mathfrak{s}\oplus\mathfrak{k}$. Then \index{$Z^{\exp -X'}$, $Z^{\exp -X'}=\mathrm{Ad}(-X')Z$}
\begin{eqnarray*}
Z^{\exp -X'}&=& \mathrm{Ad}(\exp -X')Z\\&=&\sum_{l= 0}^\infty\frac{(\mathrm{ad}-X')^l}{l!}Z\\
&=& Z+[Z,X']-\frac{[X',[Z,X']]}{2}+\ldots\\
&=&Z+[Z,X'],
\end{eqnarray*} as $\mathfrak{n}$ is abelian and $[Z,X']\in\mathfrak n$. That is, $Z^{\exp -X'} \in \mathfrak{m}\oplus\mathfrak{s}^{\bot_\mathfrak{n}}$ for all $Z\in\mathfrak{m}$. So if we assume \[ Z^{\exp-X'}=Z+[Z,X']=H+rX_1+W\in\mathfrak{a}\oplus\mathfrak{s}\oplus\mathfrak{k},\] we see that \[[Z,X']=rX_1\in \mathfrak{s}^{\bot_\mathfrak{n}}\cap \mathfrak{s}=0\]

But by Lemma \ref{l2} this is true for $Z\in\mathfrak z_{\mathfrak{m}}(\mathfrak{s})^{\bot_{\mathfrak{m}}}$ precisely iff $Z=0$. For the second claim, we note that $\dim\mathfrak{s}=1$, i.e. $\mathfrak{s}=\mathbb{R}\cdot X'$ for all $X'\in\mathfrak{s}'$. Hence, \begin{equation*} \left((1+\theta)\mathfrak{s}\right)^{\exp -X'}\oplus\mathfrak{k}\supset\mathfrak a\end{equation*} for all $X'\in\mathfrak{s}'$.
\end{proof}

The next lemma is an easy consequence of the fact that spheres are the generic orbits of $M$ in $\mathfrak n$. 
\begin{lemma}\label{lem: radial main}
Let $Y\in\mathfrak{s}^{\bot_\mathfrak{n}}$, then there is an analytic function $f$ on $\mathfrak{s}'$ and a (unique) \index{$Z_Y$, element of $\mathfrak{z}_{\mathfrak{m}}(\mathfrak{s})^{\bot_{\mathfrak{m}}}$} $Z=Z_Y\in \mathfrak{z}_{\mathfrak{m}}(\mathfrak{s})^{\bot_{\mathfrak{m}}}$ such that \begin{equation*} Y= f(\exp X')Z^{\exp -X'}-f(\exp X')Z\end{equation*}
for all $X'\in\mathfrak{s}'$ 
\end{lemma}
\begin{proof}
If the dimension of $\mathfrak{n}$ is one, then $\mathfrak{s}=\mathfrak{n}$, i.e. $\mathfrak{s}^{\bot_{\mathfrak{n}}}$ is trivial and we set \[f\equiv 0.\] If $\dim\mathfrak n>1$ we define \[f(\exp X')=\frac{1}{r}\] for $X'=rX_1\in \mathfrak{s}'$. By Lemma \ref{l2} there is a unique $Z=Z_{Y}\in\mathfrak{z}_{\mathfrak{m}}(\mathfrak{s})^{\bot_{\mathfrak{m}}}$ such that $[Z_{Y},X_1]=Y$. Now \[Z^{\exp -X'}=Z+[Z,X']=Z+r[Z,X_1]\] and thus \begin{equation*} f(\exp X') Z_{Y}^{\exp -X'}-f(\exp X')Z_{Y}=[Z,X_1]=Y.\end{equation*} 
\end{proof}


\section{Polar decompositions and radial parts for $G=M (A\exp\mathfrak{s}) K$}\label{sec: rapo}

We denote by  $U(\mathfrak{g})$ the universal enveloping algebra of $\mathfrak{g}$ and by \index{$S(\mathfrak{g})$, symmetric algebra of $\mathfrak g$} $S(\mathfrak{g})$ the symmetric algebra symmetric algebra. There is a linear bijection \index{$lambda$@$\lambda:S(\mathfrak{g})\to U(\mathfrak{g})$, symmetrizer}\[\lambda:S(\mathfrak{g})\to U(\mathfrak{g})\] called \textit{symmetrization}, see \cite[p.161]{Wr1}. 

From the decomposition $G=M (A\exp\mathfrak{s}) K$ we want to obtain a theory of radial parts for differential operators on functions, which are left-$M$- and right $K$-invariant. That is, we construct for given $D\in U(\mathfrak{g})$ some differential operator $\delta(D)$ on $A\exp\mathfrak{s}'$ such that \index{$delta(D)$@$\delta(D)$, radial part}
\begin{equation*} 
D\phi=\delta(D)\phi
\end{equation*} on $A\exp\mathfrak{s}'$ for all $\phi\in C^\infty(M\backslash G/K)$. This theory is developed in analogy with the theory for $G=KAK$, see for example \cite[Chap.9]{Wr2} or \cite[Chap.4]{GV}.

\begin{proposition}
The mapping $\psi :M\times A\exp \mathfrak{s}\times K\to G,(m,ax,k)\mapsto maxk$ is surjective and a submersion on $M\times A\exp \mathfrak{s}'\times K$. 
\end{proposition}
\begin{proof}
For $Z\in\mathfrak m,R\in\mathfrak a\oplus\mathfrak{s}, W\in\mathfrak k$, the differential $(d\psi)_{m,ax,k}$ computes to
\begin{eqnarray*}
(d\psi)_{m,ax,k}(Z,R,W)&=& (d\psi)_{m,ax,k}(Z,0,0)+ (d\psi)_{m,ax,k}(0,R,0)+(d\psi)_{m,ax,k}(0,0,W)\\&=& \frac{d}{dt}|_{t=0}\psi(m\exp tZ,ax,k)+\frac{d}{dt}|_{t=0}\psi(m,ax\exp tR,k)\\&&+\frac{d}{dt}|_{t=0}\psi(m,ax,k\exp tW)\\&=& Z^{(axk)^{-1}}+R^{k^{-1}}+W\\ &=& Z^{(xk)^{-1}}+R^{k^{-1}}+W,
\end{eqnarray*}
since $\frac{d}{dt}|_{t=0}\psi(g,\exp tL,h)=L^{h^{-1}}$ for all $L\in\mathfrak g$, $g,h\in G$ and since $M$ centralizes $A$. Thus,
\begin{equation*} (\mbox{Ad}(k)\circ (d\psi)_{m,ax,k})(Z,R,W)=Z^{x^{-1}}+R+W^{k}\end{equation*} and because $\mbox{Ad}(k)$ is an isomorphism of $\mathfrak{g}$, the surjectivity of $d\psi_{m,ax,k}$ follows from the decomposition \begin{equation*}\mathfrak{g}=\mathfrak{m}^{\exp -X'}+\mathfrak{a}+\mathfrak{s}+ \mathfrak{k},\end{equation*} which is valid for all $X'\in\mathfrak{s}'$.
\end{proof}

The decomposition of $\mathfrak g$ and the \textit{Poincar\'{e}-Birkhoff-Witt}-Theorem together yield a decomposition of the universal enveloping algebra which we view as the algebra of differential operators on $C^\infty(G)$ with constant coefficients.

\begin{lemma}\label{lem: unideco}
We have the following decompositions 
\begin{eqnarray*}
U(\mathfrak{g})&=& U(\mathfrak{a}\oplus\mathfrak{s})U(\mathfrak{k})\oplus \left((\mathfrak{s}^{\bot_\mathfrak{n}})U(\mathfrak{g})\right)\\
&=& U(\mathfrak{a}\oplus\mathfrak{s})\oplus \left((\mathfrak{s}^{\bot_\mathfrak{n}})U(\mathfrak{g})+U(\mathfrak{g})\mathfrak{k}\right).
\end{eqnarray*}
with projections \index{$pi1$@$\pi_1$, projection from $U(\mathfrak g)$ to $U(\mathfrak{a}\oplus\mathfrak{s})U(\mathfrak{k})$} \[\pi_1:U(\mathfrak{g})\to U(\mathfrak{a}\oplus\mathfrak{s})U(\mathfrak{k})\] and \index{$pi2$@$\pi_2$, projection from $U(\mathfrak g)$ to $U(\mathfrak{a}\oplus\mathfrak{s})$} 
\[\pi_2:U(\mathfrak{g})\to U(\mathfrak{a}\oplus\mathfrak{s}).\] Furthermore, \[\pi_1\equiv \pi_2 \mbox{ }\mathrm{ mod }\mbox{ } (U(\mathfrak{a}\oplus\mathfrak{s})U(\mathfrak{k})\mathfrak{k}).\]
\end{lemma}
\begin{proof}

We modify the proof of \cite[Lemma 2.6.6.]{GV}.
For the second decomposition we use the direct decomposition of $\mathfrak g$

\begin{equation*}
\mathfrak g=\mathfrak{a}\oplus \mathfrak{s} \oplus\mathfrak{s}^{\bot_\mathfrak{n} }\oplus\mathfrak{k}.
\end{equation*}

It follows by the Poincar\'{e}-Birkhoff-Witt-Theorem, see \cite[Th.7.1.9]{HN}, that the mapping \[U(\mathfrak{a}\oplus \mathfrak{s})\otimes U(\mathfrak{s}^{\bot_\mathfrak{n} })\otimes U(\mathfrak{k})\to U(\mathfrak g)\mbox{ , } a\otimes b\otimes c\mapsto abc\] defines a linear isomorphism. Because $\mathfrak n$ is abelian, $\mathfrak{s}$ and $\mathfrak{s}^{\bot_\mathfrak{n}}$ form indeed subalgebras of $\mathfrak g$. Furthermore,
\begin{eqnarray*}
U(\mathfrak g)&=& U(\mathfrak{a}\oplus \mathfrak{s}) U(\mathfrak{s}^{\bot_\mathfrak{n} })U(\mathfrak{k})
\\&=& U(\mathfrak a\oplus\mathfrak{s})U(\mathfrak{s}^{\bot_\mathfrak{n}})\oplus U(\mathfrak a\oplus \mathfrak{s})U(\mathfrak{s}^{\bot_\mathfrak{n}})U(\mathfrak{k})\mathfrak{k}\\&=&
U(\mathfrak{a}\oplus\mathfrak{s})\oplus\mathfrak{s}^{\bot_\mathfrak{n}}U(\mathfrak{s}^{\bot_\mathfrak{n}})U(\mathfrak a\oplus \mathfrak{s}) \oplus U(\mathfrak g)\mathfrak{k},
\end{eqnarray*}
since if $\mathfrak{g}$ is the direct sum of two subalgebras $\mathfrak{g}_1$ and $\mathfrak{g}_2$, then \[U(\mathfrak{g})=U(\mathfrak{g}_1)\oplus \mathfrak{g}_2U(\mathfrak{g})=U(\mathfrak{g}_1)\oplus U(\mathfrak{g})\mathfrak{g}_2,\] see \cite[Cor. 3.2.7.]{V}. If we apply this to \[\mathfrak{s}^{\bot_\mathfrak{n}}U(\mathfrak g)\subset U(\mathfrak g),\] we get

 
\begin{equation*}
\mathfrak{s}^{\bot_\mathfrak{n}}U(\mathfrak g)\subset \mathfrak{s}^{\bot_\mathfrak{n}}U(\mathfrak{s}^{\bot_\mathfrak{n}})U(\mathfrak a\oplus \mathfrak{s}) \oplus U(\mathfrak g)\mathfrak{k},
\end{equation*}
i.e. 
\begin{equation*}
\mathfrak{s}^{\bot_\mathfrak{n}}U(\mathfrak g)+U(\mathfrak g)\mathfrak{k}=\mathfrak{s}^{\bot_\mathfrak{n}}U(\mathfrak{s}^{\bot_\mathfrak{n}})U(\mathfrak a\oplus \mathfrak{s}) \oplus U(\mathfrak g)\mathfrak{k}.
\end{equation*}

For the first decomposition we proceed as before by using 
\begin{equation*}
U(\mathfrak g)\cong  U(\mathfrak{a}\oplus \mathfrak{s})\otimes U(\mathfrak{k})\otimes U(\mathfrak{s}^{\bot_\mathfrak{n} })
\end{equation*}
from which we derive this time
\begin{eqnarray*}
U(\mathfrak g)&=& U(\mathfrak{a}\oplus \mathfrak{s})U(\mathfrak{k})\oplus \mathfrak{s}^{\bot_\mathfrak{n}}U(\mathfrak{a}\oplus \mathfrak{s})U(\mathfrak{k})U(\mathfrak{s}^{\bot_\mathfrak{n} })
\\&=& U(\mathfrak{a}\oplus \mathfrak{s})U(\mathfrak{k})\oplus \mathfrak{s}^{\bot_\mathfrak{n}}U(\mathfrak g).
\end{eqnarray*}


The last claim \[\pi_1(D)\equiv \pi_2(D) \mbox{ }\mathrm{mod}\mbox{ } U(\mathfrak a\oplus\mathfrak{s})U(\mathfrak k)\mathfrak{k}\] follows since $$U(\mathfrak a\oplus \mathfrak s)U(\mathfrak k)=U(\mathfrak a\oplus\mathfrak s)\oplus U(\mathfrak a\oplus\mathfrak s)U(\mathfrak k)\mathfrak{k}$$ and since $U(\mathfrak a\oplus\mathfrak s)U(\mathfrak k)\subset U(\mathfrak g)$.
\end{proof}

For the definition of a radial part we need to express any 
$D\in U( \mathfrak{g} )$ in polar coordinates adapted to the decomposition of $G=M(A\exp \mathfrak s)K$.
The following theorem is the major ingredient for this.

\begin{proposition}\label{th: bijec}
Let $\mathfrak{g}=\mathfrak{so}(1,l)=\mathfrak{a}\oplus\mathfrak{n}\oplus\mathfrak{k}$, $\mathfrak{m}=\mathfrak{z}_\mathfrak{k}(\mathfrak{a})$ and $\mathfrak s$ a section for $M=\exp\mathfrak m$ acting on $\mathfrak n$. The map \index{$Gammaa$@$\Gamma_{a\exp X'}$, map from $\lambda(S(\mathfrak{z}_{\mathfrak{m}}(\mathfrak{s})^{\bot_{\mathfrak{m}}}))\otimes U(\mathfrak{a}\oplus\mathfrak{s}) \otimes U(\mathfrak{k})$ to $U(\mathfrak g)$}
\begin{eqnarray*} \Gamma_{a\exp X'}&:&\lambda(S(\mathfrak{z}_{\mathfrak{m}}(\mathfrak{s})^{\bot_{\mathfrak{m}}}))\otimes U(\mathfrak{a}\oplus\mathfrak{s}) \otimes U(\mathfrak{k})\to U(\mathfrak{g}),\\ &&\xi\otimes u\otimes \xi'\mapsto \xi^{\exp (-X')}u\xi'
\end{eqnarray*}
defines a linear isomorphism for all $X'$ in $\mathfrak{s}'$.
\end{proposition}
\begin{proof}
The bijectivity of $\Gamma_{a\exp X'}$ follows from the decomposition \begin{equation*}
\mathfrak{g}=(\mathfrak{z}_{\mathfrak{m}}(\mathfrak{s})^{\bot_{\mathfrak{m}}})^{\exp -X'}\oplus\left(\mathfrak{s}\oplus\mathfrak a\right)\oplus \mathfrak{k},\end{equation*} see Proposition \ref{t1}, which is valid for all $X'\in\mathfrak{s}'$, and the \textit{Poincar\'{e}-Birkhoff-Witt}-Theorem. 
\end{proof}

We denote the inverse of $\Gamma_{a\exp X'}$ by \index{$\stackrel{\circ}{\bot}_{\mathfrak{a}\oplus\mathfrak{s},a\exp X'}$, inverse of $\Gamma_{a\exp X'}$} $\stackrel{\circ}{\bot}_{\mathfrak{a}\oplus\mathfrak{s},a\exp X'}$ and interpret it as the local expression of $D$ in polar coordinates  $\stackrel{\circ}\bot_{\mathfrak{a}\oplus\mathfrak{s}}$ of $D$ around the point $a\cdot \exp X'$ for the decomposition of the dense, open subset in $G$ 
\begin{equation*}
G'=M(A\exp\mathfrak{s}')K.
\end{equation*}
 Then we define \index{$\mathcal F$, algebra of functions} $\mathcal F$ to be the algebra with unit of functions generated by the function $f$ from Lemma \ref{lem: radial main} and let \index{$\mathcal F^+$, subset of $\mathcal F$} $\mathcal{F}^+$ be the linear span of monomials of positive degree in this generator $f$. 

\begin{theorem}\label{th:radial part1}
Let $\mathfrak{g}=\mathfrak{so}(1,l)=\mathfrak{a}\oplus\mathfrak{n}\oplus\mathfrak{k}$ with universal enveloping $U(\mathfrak{g})$, $\mathfrak{m}=\mathfrak{z}_\mathfrak{k}(\mathfrak{a})$ and $\mathfrak s$ a section for $M=\exp\mathfrak m$ acting on $\mathfrak n$. Let $D$ be in $U(\mathfrak{g})$. There are $\Delta_i\in U(\mathfrak{m})\otimes U(\mathfrak{s}\oplus\mathfrak a) \otimes U(\mathfrak{k})$ and \index{$Delta0$@$\Delta_0$, element of $U(\mathfrak a\oplus\mathfrak{s})\otimes U(\mathfrak{k})$} $\Delta_0\in U(\mathfrak a\oplus\mathfrak{s})\otimes U(\mathfrak{k})$, $\varphi_j\in\mathcal F^+$ analytic on $\mathfrak{s}'$ such that 
\begin{equation*} \stackrel{\circ}{\bot}_{\mathfrak{a}\oplus\mathfrak{s},a\exp X'}(D)=\Delta_0+\sum_j\varphi_j(\exp X')\Delta_j
\end{equation*}
for all $X'\in\mathfrak{s}'$.
\end{theorem}
\begin{proof}
By Proposition \ref{th: bijec} it is equivalent to prove that for any $D\in U(\mathfrak g)$ there are $D_0\in U(\mathfrak a \oplus\mathfrak s)U(\mathfrak k)$, $\xi_j\in \lambda(S(\mathfrak{z}_{\mathfrak{m}}(\mathfrak{s})^{\bot_{\mathfrak{m}}}))$, $u_j\in U(\mathfrak{a}\oplus\mathfrak{s})$ and $\xi'_j\in U(\mathfrak k)$ such that \[D=D_0+\sum_j \varphi_j \xi_j^{\exp-X'}u_j\xi_j', \]
where $\varphi_j\in\mathcal F^+$. 

For $D\in U(\mathfrak{g})$ we set $\Delta_0:=\pi_1(D)$. Next if $D\in U(\mathfrak{s}\oplus\mathfrak a)U(\mathfrak{k})$, then the theorem is clear, since \begin{equation*} D=\pi_1(D)+\Gamma_{a\exp X'}(0\otimes 0\otimes 0)=\Gamma_{a\exp X'}(\pi_1(D))\end{equation*} by the definition of $\Gamma_{a\exp X'}$.

Now we continue with an induction on the degree of $D\in U(\mathfrak{g})$ and note that the claim is true for constants. If the degree is one and $D$ is an element of $\mathfrak{s}^{\bot_\mathfrak{n}}$, then there is by Lemma \ref{lem: radial main} a unique  $Z\in\mathfrak{z}_{\mathfrak{m}}(\mathfrak{s})^{\bot_{\mathfrak{m}}}$ such that 
\begin{eqnarray*}
D&=&f(\exp X')\Gamma_{a\exp X'}(Z\otimes 1\otimes 1)-f(\exp X')\Gamma_{a\exp X'}(1\otimes 1\otimes Z)\\ &=& f(\exp X')\Gamma_{a\exp X'}(Z\otimes 1\otimes 1-1\otimes 1\otimes Z)
\end{eqnarray*}
proving the claim for $D\in\mathfrak{s}^{\bot_n}$.

Since $U(\mathfrak g)=U(\mathfrak s\oplus \mathfrak a)\oplus(\mathfrak s^{\bot_{\mathfrak n}})U(\mathfrak g)$ by Lemma \ref{lem: unideco}, we only have to consider $D\in(\mathfrak{s}^{\bot_\mathfrak{n}}) U(\mathfrak{g})$ of degree $m+1$ assuming the theorem is true for degree $m$. Without loss we can assume that $D=YD_1$, where $Y$ is an element of $\mathfrak{s}^{\bot_n}$ and $D_1$ of $ U(\mathfrak{g})$ of degree $m$. Again by Lemma \ref{lem: radial main} \begin{equation*} Y=f(\exp X')Z^{\exp -X'}-f(\exp X')Z\end{equation*} for some $Z\in \mathfrak{z}_{\mathfrak{m}}(\mathfrak{s})^{\bot_{\mathfrak{m}}}$ and we find for $X\in\mathfrak{s}'$ that
\begin{eqnarray*}
 YD_1&=&f(\exp X')Z^{\exp -X'}D_1-f(\exp X')ZD_1\\&=& f(\exp X')Z^{\exp-X'}D_1-f(\exp X')D_1Z-f(\exp X')[Z,D_1].
\end{eqnarray*}
But $D_1$ and $[Z,D_1]$ are of degree $\leq m$, so we can apply the induction hypothesis \[\stackrel{\circ}{\bot}_{\mathfrak{a}\oplus\mathfrak{s},a\exp X'}(D_1)=\Delta_0+\sum_j\varphi_j(\exp X')\Delta_j \mbox{ , } \stackrel{\circ}{\bot}_{\mathfrak{a}\oplus\mathfrak{s},a\exp X'}([Z,D_1])=\tilde{\Delta}_0+\sum_j\tilde{\varphi}_j(\exp X')\tilde{\Delta}_j\] resp. \[D_1=\Gamma_{a\exp X'}\left(\Delta_0+\sum_j\varphi_j(\exp X')\Delta_j(D_1)\right)=\Gamma_{a\exp X'}(\Delta_0)+\sum_j\varphi_j(\exp X')\Gamma_{a\exp X'}(\Delta_j)\] and 
\begin{eqnarray*}[Z,D_1]&=&\Gamma_{a\exp X'}\left(\Delta_0+\sum_j\varphi_j(\exp X')\Delta_j([Z,D_1])\right)=\Gamma_{a\exp X'}(\tilde{\Delta}_0)+\sum_l\tilde{\varphi}_l(\exp X')\Gamma_{a\exp X'}(\tilde{\Delta}_l)\\
&=&\tilde{\Delta}_0+\sum_l \varphi_j(\exp X')\Gamma_{a\exp X'}(\tilde{\Delta}_l)
\end{eqnarray*}
with $D_j$, $\varphi_j$, $\tilde{D}_l$, $\tilde{\varphi}_j$ as claimed in the theorem. That is,
\begin{eqnarray*}
Z^{\exp -X'}D_1&=& Z^{\exp -X}\Gamma_{a\exp X'}(\Delta_0)+\sum_j\varphi_j(\exp X')Z^{\exp -X}\Gamma_{a\exp X'}(\Delta_j)\\
&=& Z^{\exp -X'}\Delta_0+\sum_j\varphi_j(\exp X') Z^{\exp -X'}\Gamma_{a\exp -X'}(\Delta_j),\\
D_1Z&=&\Gamma_{a\exp X'}(\Delta_0)Z+\sum_j\varphi_j(\exp X')\Gamma_{a\exp -X'}(\Delta_j)Z\\&=& 
\Delta_0 Z+\sum_j\varphi_j(\exp X')\Gamma_{a\exp -X'}(\Delta_j)Z,
\end{eqnarray*}

It follows that
\begin{eqnarray*}
YD_1&=& f(\exp X')Z^{\exp -X'}\Delta_0+\sum_j f(\exp X')\varphi_j(\exp X')Z^{\exp -X'}\Gamma_{a\exp X'}(\Delta_j)-f(\exp X')\Delta_0Z\\&&-\sum_j f(\exp X')\varphi_j(\exp X')\Gamma_{a\exp X'}(\Delta_j)Z- f(\exp X')\tilde{\Delta}_0-\sum_l f(\exp X')\varphi_l(\exp X')\Gamma(\tilde\Delta_lj)\\
&=& -f(\exp X')\Delta_0Z-f(\exp X')\tilde{\Delta}_j+f(\exp X')Z^{\exp -X'}\Delta_0+\\&&\sum_j f(\exp X')\varphi_j(\exp X')\left(Z^{\exp -X'}\Gamma_{a\exp X'}(\Delta_j)+\Gamma_{a\exp X'}(\Delta_j)Z\right)\\&&+\sum_l f(\exp X')\varphi_l(\exp X')\Gamma(\tilde\Delta_l). 
\end{eqnarray*}
 
 Then we are done, since $\pi_1(YD_1)=0$ which implies \[0=-f(\exp X')\Delta_0Z-f(\exp X')\tilde{\Delta}_0.\]
\end{proof}

It follows that for functions $\phi\in C^\infty(G')$  and $D\in U(\mathfrak g)$
\begin{equation*}
(D\phi)(a\exp X')=\left(\Gamma_{a\exp X'}\left(\stackrel{\circ}{\bot}_{\mathfrak{a}\oplus\mathfrak{s},a\exp X'}(D)\right)\phi\right)(a\exp X')
\end{equation*}
valid for all $a\in A$, $X'\in \mathfrak{s}'$. 

For certain $D\in U(\mathfrak g)$, the expression of $\stackrel{\circ}{\bot}_{\mathfrak{a}\oplus\mathfrak{s},a\exp X'}(D)$ is simpler. 
\begin{corollary}
With the notation of Theorem \ref{th:radial part1} we have
\begin{equation*} \stackrel{\circ}\bot_{\mathfrak{a}\oplus\mathfrak{s},a\exp X'}(\pi_1(D))=\Delta_0. \end{equation*} Further, if $$\stackrel{\circ}\bot_{\mathfrak{a}\oplus\mathfrak{s},a\exp X'}(D)=\tilde{\Delta}_0+\sum_{j=1}^n\varphi_j(\exp X')\Delta_j=\Delta_0+\sum_{l=1}^m\tilde\varphi_l(X')\tilde{\Delta}_l$$ for all $X'\in\mathfrak s'$, where $\Delta_j,\tilde{\Delta}_l,\varphi_j,\tilde{\varphi}_l$ as in Theorem \ref{th:radial part1} and $n\leq m$, then $\Delta_0=\tilde{\Delta}_0$ and after a possible permutation of $j$, $$ \Delta_1=\tilde{\Delta}_1,\ldots, \Delta_k=\tilde{\Delta}_k, \mbox{ } \varphi_1=\tilde{\varphi}_1,\ldots,\varphi_k=\tilde{\varphi}_k .$$ and $\tilde \Delta_{k+1}=\ldots=\tilde{\Delta}_{m}=0$.
\end{corollary}
\begin{proof}
The first claim, $\stackrel{\circ}\bot_{\mathfrak{a}\oplus\mathfrak{s},a\exp X'}(\pi_1(D))=\Delta_0$, follows by the proof of 
Theorem \ref{th:radial part1}. For the second claim we note that $\varphi_j(\exp tX')\to 0$, $t\to\infty$, for all $X'\in\mathfrak{s}'$ 
and $\varphi_j\in \mathcal{F}^+$. So if we assume 
\begin{eqnarray*}  
\Delta_0+\sum_j\varphi_j(\exp X')\Delta_j&=&\Delta_0'+\sum_k\varphi_k'(\exp X')\Delta_k' \mbox{  } \mbox{ resp.}
\\ \Delta_0-\Delta_0'&=& \sum_l \varphi_l''(\exp X')D_l'' 
,
\end{eqnarray*}
for some $\Delta_0,\Delta_0'\in U(\mathfrak{s}\oplus\mathfrak a)U(\mathfrak k)$, $\Delta_j,\Delta_k', D_l''\in U(\mathfrak m)U(\Sigma\oplus\mathfrak a)U(\mathfrak k)$, $\varphi_j,\varphi_k',\varphi_l''\in\mathcal F^+$ and all $X'\in\mathfrak{s}'$, then replacing $\exp X'$ by $\exp tX'$ and letting $t\to\infty$ gives the claim.

The last claim, follows since $\Delta_j,\tilde{\Delta}_l$ are independent of $X'\in \mathfrak{s}'$ and since $\varphi_j.\tilde{\varphi}_k\in\mathcal{F}^+$.
\end{proof}

We state now a lemma which tells us more about the nature of the projection $\pi_2$ from $U(\mathfrak g)$ onto $U(\mathfrak{s}\oplus \mathfrak a)$.

\begin{lemma}
The mapping $\pi_2$ is a homomorphism from $ U(\mathfrak{g})^K$ into $ U(\mathfrak{s}\oplus \mathfrak{a})$. 
\end{lemma}
\begin{proof}
Let $D,D'\in U(\mathfrak{g})^K$. Then \begin{equation*} DD'-\pi_2(D)\pi_2(D')=\pi_2(D)(D'-\pi_2(D'))+(D-\pi_2(D))D'.\end{equation*} 
By definition $D-\pi_2(D),D'-\pi_2(D')$ are elements of $(\mathfrak{s}^{\bot_\mathfrak{n}}) U(\mathfrak{g})+ U(\mathfrak{g})\mathfrak{k}$ and since 
$D'$ is in $ U(\mathfrak{g})^K$, it follows that \[WD'=D'W\] for $W\in\mathfrak{k}$. Hence,
\begin{equation*} (D-\pi_2(D))D'\in(\mathfrak{s}^{\bot_\mathfrak{n}}) U(\mathfrak{g})
+ U(\mathfrak{g})\mathfrak{k}.
\end{equation*}
 Since $[\mathfrak{s}\oplus\mathfrak{a},\mathfrak{s}^{\bot_n}]\subset \mathfrak{s}^{\bot_\mathfrak{n}}$ it also follows that 
\begin{equation*}
\pi_2(D)(D'-\pi_2(D'))\in(\mathfrak{s}^{\bot_\mathfrak{n}}) U(\mathfrak{g})+ U(\mathfrak{g})\mathfrak{k}.
\end{equation*} Thus, 
\begin{equation*}
\pi_2(DD')=\pi_2(D)\pi_2(D').
\end{equation*}

\end{proof}

Now we explain how 
$\Gamma_{a\exp X'}(D)$ transforms if we apply it to invariant functions $\phi\in C^\infty(MA\backslash G/K)$. This will lead us to the definition of the radial part of $D$. We therefore compute for $Z\in\mathfrak{z}_{\mathfrak{m}}(\mathfrak{s})^{\bot_{\mathfrak{m}}}, u\in U(\mathfrak{a}\oplus\mathfrak{s}),\xi\in U(\mathfrak{k})$, $a\exp X'\in A\exp \mathfrak{s}'$ and $\phi\in C^\infty(G)$  
\begin{eqnarray}\label{eq: compminv} \Gamma_{a\exp X'}(Z\otimes u \otimes \xi)\phi(a\exp X')&=&\left(Z^{\exp -X'}u\xi \right)\phi(a\exp X')\nonumber\\
&=& \frac{d}{dt}|_{t=0}u\xi \phi(a\exp X'\exp \mathrm{Ad}(\exp -X')tZ)\nonumber\\&=& \frac{d}{dt}|_{t=0}u\xi \phi(a\exp X' \exp -X'\exp tZ\exp X')\\
&=& \frac{d}{dt}|_{t=0}u\xi \phi(a\exp tZ\exp X')\nonumber\\ &=& \frac{d}{dt}|_{t=0}u\xi \phi(\exp tZ a\exp X'),\nonumber
\end{eqnarray}
as $\exp tZ\in M$. In particular it is clear that \begin{equation}\label{eq: compminv2}\Gamma_{a\exp X'}(\xi\otimes u\otimes \xi') \phi(a\exp X')=0\end{equation} if $\phi\in C^\infty(M\backslash G/K)$, $\xi\in  U(\mathfrak{m})$,
$\xi'\in U(\mathfrak{k})$ and $\xi$ or $\xi'$ is not a constant. If $u=H$ is in $\mathfrak{a}$ and $\phi$ is left-$A$-invariant, then a similar computation gives 
\begin{eqnarray*}
Z^{\exp -X'}H\xi \phi(a\exp X')&=& \frac{d}{dt_1}|_{t_1=0}\frac{d}{dt_2}|_{t_2=0}\xi \phi(\exp t_1Z\exp X'\exp t_2H)\\&=&\frac{d}{dt_1}|_{t_1=0}\frac{d}{dt_2}|_{t_2=0}\xi \phi(\exp t_1Z\exp e^{-t_2\alpha(H)}X').
\end{eqnarray*}

Let $c:U(\mathfrak k)\to \mathbb{C}$ be the trivial one dimensional representation of $U(\mathfrak k)$. The map $Z\mapsto c(Z)1$ can be regarded as the projection $U(\mathfrak k)\to \mathbb{C}\cdot 1$ which belongs to the decomposition
\begin{equation*} 
 U(\mathfrak{k})=\mathfrak{k}\oplus \mathrm{ker}(c),
\end{equation*}
see also \cite[p.129]{GV}. Then for all $D\in U(\mathfrak{g})$ and $\phi\in C^\infty(M\backslash G/K)$, $x=a\exp X'\in A\exp\mathfrak{s}'$,
\begin{eqnarray*} D\phi(x)&=& \Gamma_{a\exp X'}\left(\stackrel{\circ}\bot_{\mathfrak{a}\oplus\mathfrak{s},a\exp X'}(D)\right)\phi(x)\\
&=& \Gamma_{a\exp X'}\left(\Delta_0+\sum_j \varphi_j(x)\Delta_j\right)\varphi(x)\\ 
&=& \Gamma_{a\exp X'}(\Delta_0)\varphi(x)+\sum_j \varphi_j(\exp X') \Gamma_{a\exp X'}(\Delta_j)\varphi(x)\\ 
&=& \Gamma_{a\exp X'}\left(\pi_1(D)\right)\varphi(x)+\sum_j \varphi_j(\exp X') \Gamma_{a\exp X'}(\Delta_j)\varphi(x) \mbox{ by the definition of } \Delta_0\\
&=& \pi_1(D)\varphi(x)+\sum_j \varphi_j(\exp X') \Gamma_{a\exp X'}(\Delta_j)\varphi(x) \mbox{ by the definition of } \Gamma_{s\exp X'}\\
&=&\pi_2(D)\phi(x)+\sum_j \varphi_j(\exp X') \Gamma_{a\exp X'}(\Delta_j)\varphi(x), \\&&\mbox{ as $\phi$ is right-$K$-invariant and } \pi_1\equiv \pi_2 \mbox{ }\mathrm{ mod }\mbox{ } (U(\mathfrak{a}\oplus\mathfrak{s})U(\mathfrak{k})\mathfrak{k}\\
&=&\pi_2(D)\phi(x)+\sum_j\varphi_j(x)c(\xi)c(\xi')u_j\phi(x)\end{eqnarray*}
by (\ref{eq: compminv}) and (\ref{eq: compminv2}), where $\Delta_0,\Delta_j,\varphi_j$ are as in Theorem \ref{th:radial part1}, i.e. $\varphi_j\in\mathcal F^+$ and $u_j\in U(\mathfrak a \oplus \mathfrak{s})$. Thus we may define a differential operator \index{$delta$@$\delta(D)$} $\delta(D)$ on $
A\exp {\mathfrak{s}'}$ by 
\begin{equation}\label{eq: radpart}\delta(D):=\pi_2(D)+\sum_j\psi_ju_j,
\end{equation} where $\psi_j=c(\xi_j)c(\xi_j')\varphi_j$ and \begin{equation*} D\phi=\delta(D)\phi
\end{equation*} on $A\exp\mathfrak{s}'$ for all  $\phi\in C^\infty(M\backslash G/K)$. We call $\delta(D)$ the \index{radial part} \textit{radial part} of $D$.
\section{Example: Polar decomposition of $\Omega$}\label{sec: radial.1}

In this section we apply the theory developed above to the Casimir operator $\Omega$.
We start with the computation of the polar decomposition $\stackrel{\circ}{\bot}_{\mathfrak{a}\oplus\mathfrak{s},a\exp X'}(\Omega)$ of $\Omega$. 

We recall that we fixed some $X_1\in \mathfrak{n}$ of unit length. We want to complete $X_1$ to an orthonormal basis of $\mathfrak{n}$. Therefore, we remark:
\begin{lemma}\label{lem: jrelations}
We can complete $X_1$ to an orthonormal basis $\{X_1,Y_2,\ldots, Y_{l-1}\}$ of $\mathfrak n$ such that for all $j$ \[ [Z_{Y_j},X_1]=Y_j \]
 always implies \[[X_1,\theta Y_j]=2Z_{Y_j}\] and \[[Y_j,Z_{Y_j}]=X_1.\]

Here $Z_{Y_j}\in \mathfrak{z}_\mathfrak{m}(\mathfrak{s} )^\bot_\mathfrak{m}$ is uniquely determined by the relation $[Z_{Y_j},X_1]=Y_j$, see Lemma \ref{l2}.\end{lemma}
\begin{proof}
In \cite[Chap.2 (2.9)]{Ju} we find a special orthonormal basis $\{L_j\}_j$ of $\mathfrak n$ which satisfies the requirements of this lemma. Because $M$ acts transitively on the sphere in $\mathfrak n$, we can find some $m\in M$ such that $X_1=m\cdot L_1$. We set \index{$Y_j$, $Y_j=m\cdot L_j$} $Y_j:=m\cdot L_j$ which implies $Z_{Y_j}=m\cdot Z_{L_j}$. But then $$[X_1,\theta Y_j]=m\cdot [L_1,\theta L_j]=m\cdot 2Z_{L_j}=2Z_{Y_j}$$ and $$[Y_j,Z_{Y_j}]=[m\cdot L_j,m\cdot Z_{L_j}]=m\cdot L_1=X_1.$$
\end{proof}

We take the orthonormal basis \index{$Y_j$} $\{X_1,Y_j\}_j$ of $\mathfrak{s}^{\bot_\mathfrak{n}}$ from Lemma \ref{lem: jrelations}. Then we know from Chapter \ref{chap: casimir}, equation (\ref{eq: casfin}), that 
\begin{eqnarray*} \Omega
&=&H_1^2-\sum_{i=1}^{k}M_i^2+2X_1^2+2\sum_{j=2}^{l-1}Y_j^2
-2X_1W_1-2\sum_{j=2}^{l-1}Y_jW_j-2H_\rho\\ 
&=& H_1^2-2H_\rho-\sum_{i=1}^kM_i^2+2\left(X_1^2-X_1W \right)+2\sum_{j=2}^{l-1} \left(Y_j^2-Y_jW_j\right),\end{eqnarray*} 
where $M_i \in\mathfrak m$, $H_1\in\mathfrak a$, $H_\rho\in\mathfrak a$ and $W=W_1,W_j\in\mathfrak m^{\bot_\mathfrak{k}}$, see Chapter \ref{chap: casimir} for the details. 

Let $X'=rX_1\in\mathfrak{s}'$, then there exists by Lemma \ref{lem: radial main} for any $Y_j\in\mathfrak n$ some $Z_{Y_j}\in\mathfrak z_{\mathfrak{m}}(\mathfrak{s})^{\bot_{\mathfrak{m}}}$ such that 
\begin{equation*} Y_j=f(\exp X')Z_{Y_j}^{\exp -X'}-f(\exp X')Z_{Y_j},\end{equation*} where $Z_{Y_j}\cdot X_1=Y_j$. Hence,
\begin{eqnarray*}
\sum_j Y_j^2&=& f(\exp X')^2\sum_{j}(Z_{Y_j}^{\exp -X'}-Z_{Y_j})(Z_{Y_j}^{\exp -X'}-Z_{Y_j})
\\ &=&  f(\exp X')^2\sum_{j}(Z_{Y_j}Z_{Y_j})^{\exp-X}-Z_{Y_j}^{\exp -X'}Z_{Y_j}-Z_{Y_j}Z_{Y_j}^{\exp -X}+Z_{Y_j}Z_{Y_j}
\\&=& f(\exp X')^2\sum_{j}(Z_{Y_j}Z_{Y_j})^{\exp -X'}-2Z_{Y_j}^{\exp -X'}Z_{Y_j}-[Z_{Y_j},Z_{Y_j}^{\exp -X'}]+Z_{Y_j}Z_{Y_j}
\\&=& f(\exp X')^2\sum_{j}(Z_{Y_j}Z_{Y_j})^{\exp -X'}-2Z_{Y_j}^{\exp -X'}Z_{Y_j}-[Z_{Y_j},Z_{Y_j}+[Z_{Y_j},X']]+Z_{Y_j}Z_{Y_j}
\\&=& f(\exp X')^2\sum_{j}(Z_{Y_j}Z_{Y_j})^{\exp -X'}-2Z_{Y_j}^{\exp -X'}Z_{Y_j}-[Z_{Y_j},f(\exp X')^{-1}[Z_{Y_j},X_1]]+Z_{Y_j}Z_{Y_j}
\\&=&f(\exp X')^2\sum_{j}(Z_{Y_j}Z_{Y_j})^{\exp -X'}-2Z_{Y_j}^{\exp -X'}Z_{Y_j}-f(\exp X')^{-1}[Z_{Y_j},Y_j]+Z_{Y_j}Z_{Y_j}
\\&\overset{\text{Lem. \ref{lem: jrelations}}}=&f(\exp X')^2\sum_{j}(Z_{Y_j}Z_{Y_j})^{\exp -X'}-2Z_{Y_j}^{\exp -X'}Z_{Y_j}+f(\exp X')^{-1}X_1+Z_{Y_j}Z_{Y_j}.
\end{eqnarray*}

For $\Omega$ we obtain
\begin{eqnarray}\label{eq: omegapolar}
\Omega=\Gamma_{a\exp X'}\left(\stackrel{\circ}{\bot}_{\mathfrak{a}\oplus\mathfrak{s},a\exp X'}(\Omega)\right)&=& H_1^2-\sum_i M_i^2-2H_\rho+
2X^2_1+2(l-2)f(\exp X')X_1-2X_1W\nonumber\\&&+2f(\exp X')^2\sum_{j}(Z_{Y_j}Z_{Y_j})^{\exp -X'}-2Z_{Y_j}^{\exp -X'}Z_{Y_j}
+Z_{Y_j}Z_{Y_j}\nonumber \\&&
-f(\exp X')^{-1}(Z_{Y_j}^{\exp -X'}-Z_{Y_j})W_j.
\end{eqnarray}

From this we can easily derive the radial part of $\Omega$ by dropping all terms involving elements from $\mathfrak m$. Note that $(Z_{Y_i}Z_{Y_i})^{\exp -X'}$ operates from the left, i.e. it vanishes if the function is left-$M$-invariant. Thus, \begin{equation}\label{eq: radpartomega}  \delta(\Omega)=H_1^2-2H_\rho+2 X_1^2+2(n-2)f\cdot X_1,\end{equation} 
where $f(\exp X')=\frac{1}{r}$ for $X'=rX_1\in\mathfrak{s}'$.
\section{Restrictions of bi-$M$-invariant functions}
Here we want to settle the question what happens if we restrict functions which are bi-$M$-invariant to the section $\mathfrak{s}$.
\begin{lemma}
Let $\mathfrak{s}$ be a section for $M$ acting on $\mathfrak{n}$, $F$ a fundamental domain for $W=N_{M}(\mathfrak{s})/Z_{M}(\mathfrak{s})$ acting on $\mathfrak{s}$. Then one can choose a slice $F'\subset F$ consisting of regular elements which intersects every regular orbit. The mapping $\phi:M/Z_{M}(\mathfrak{s})\times F'\to \mathfrak{n}'$, where $\mathfrak{n}'$ is the subset of regular points in $\mathfrak{n}$ (always relative to the action of $M$), is a diffeomorphism.
\end{lemma}
\begin{proof}
By definition the map is bijective. For $m_0\in M$ and $X_0\in F'$ the differential computes to \begin{equation*} d\phi_{(m_0,X_0)}(d\tau(m_0)Z,T)=\mathrm{Ad}(m_0)([Z,X_0]+T),\end{equation*} where $Z$ is an element of the orthogonal complement of $\mathfrak{z}_{\mathfrak m}(\mathfrak{s})$ in $\mathfrak m$, $T\in\mathfrak{s}$ and $\tau(x)$ is the mapping $mZ_{M}(\mathfrak{s})\mapsto xmZ_{M}(\mathfrak{s})$ from $M/Z_{M}(\mathfrak{s})$ onto itself. 

Now \[\langle Z\cdot X_0,T\rangle=0,\] since the action is polar. Hence, the differential vanishes iff $[Z,X_0]=T=0$. But $[Z,X_0]=0$ implies that \[Z\in \mathfrak{z}_{\mathfrak m}(H_0))=\mathfrak{z}_{\mathfrak m}(\mathfrak{s}),\] since $X_0$ is regular. Thus, $Z=0$.
\end{proof}
\begin{corollary}\label{cor1}
The restriction map from $C^\infty(\mathfrak{n}')^{M}\to C^\infty(\mathfrak{s}')^{W}$ is an isomorphism. 
\end{corollary}

\begin{remark}
Due to the structure of (non exceptional) rank one symmetric spaces most of the theory of this chapter is also applicable to any semisimple, rank one group $G$ coming from the 3 non exceptional series and a subgroup $M'$ of $M$ acting polarly on $\mathfrak n$. Examples for $M'$ are centralizer $M_{m}$, $m\in M$. The only case when such centralizers do not act polarly on $\mathfrak n$, is when  $G=SU(1,n)$ and $M_m$ is a maximal torus in $SU(1,n)$.
\end{remark}
\chapter{Some differential equations of hypergeometric type}\label{chap: ode}

In this chapter we consider a special differential equation coming from the Casimir operator. We make use of the theory of radial parts for the decomposition $SO_o(1,l)=G=M(A\exp \mathfrak{s}) K$ from Chapter \ref{chap: radial} and consider only functions satisfying a certain equivariance property. We show then that the differential equation resembles a hypergeometric equation and we determine the space of solutions thereof.

\section{Differential equations for $M$-equivariant functions on a slice}\label{sec: odeimp}

Let \index{$V_\pi$, representation space of $\pi$} $(\pi,V_\pi)$ be an irreducible representation of $M$ on $V_\pi$ and \index{$C^\infty\left(X\times_M\mathrm{End}(V_\pi)\right)$, space of sections} $F\in C^\infty\left(X\times_M\mathrm{End}(V_\pi)\right)$, where 
\begin{equation*}
C^\infty\left(X\times_M \mathrm{End}(V_\pi) \right):=\{F\in C^\infty(X,\mathrm{End}(V_\pi)):F(m\cdot x)=\pi(m)F(x) \mbox{ f.a. $m\in M$ and $x\in X$}\}.
\end{equation*}

Here we define the action of $M$ on $X=G/K$ by $m\cdot x=mx$. We can also identify $X$ with $AN$, $xK=anK\mapsto an\cdot o \mapsto an$, \index{$o$, $o=K$}$o=K\in G/K$, via the Iwasawa decomposition. The action of $M$ on $X$ is then given by \[m\cdot x=mx=manK=amnK=a(mnm^{-1})K\] for $x=an=anK$. We can also define an action of $M$ on $G$ via \label{$*$} left translation \[m*g:=mg.\]  

Let $\mathrm{pr}:G\to X=G/K=AN$ be the canonical projection associated to the Iwasawa decomposition $G=ANK$ which maps $g=ank$ to $anK=an$. Then we have the following commutative diagram
 
\begin{center}
    \makebox[0pt]{
\begin{xy}
  \xymatrix{
      G \ar[r]^{m*} \ar[d]_{\mathrm{pr}}    &   G \ar[d]^{\mathrm{pr}}  \\
      X=G/K=AN \ar[r]_{m\cdot }             &   X=G/K=AN   
  }
\end{xy}
}
\end{center}
because \[\mathrm{pr}(m*g)=\mathrm{pr}(m*ank)=\mathrm{pr}(mank)=\mathrm{pr}(amnm^{-1}mk)=amnm^{-1}K=m\cdot anK=m\cdot \mathrm{pr}(g).\]

Furthermore, we fix a slice $S\subset N$ for the action of $M$ on $N$. We can assume that $S=\exp \mathbb{R}^+X_1$, where $X_1\in\mathfrak n\cong \mathbb{R}^{l-1}$ is of unit length with respect to $B_\theta(.,.)$. We set $\mathfrak s:=\mathbb{R}X_1$. We fix an orthonormal basis $\{X_1,Y_2,\ldots,Y_{l-1} \}$ of $\mathfrak n$ according to Lemma \ref{lem: jrelations}. Then we compute for a differential operator $Z\in \mathfrak m$, $X'=sX_1$ \index{$pi(Z)$@$\pi(Z)$, derived operator from $\pi(\exp Z)$} and for any $F:G\to \mathrm{End}(V_\pi)$ smooth with $F(mg)=\pi(m)F(g)$ for all $m\in M$, $g\in G$  
\begin{eqnarray}\label{eq: zmequi}
Z^{\exp -X'}F(\exp sX_1)&=&\frac{d}{dt}|_{t=0}F(\exp sX_1\exp tZ^{\exp -sX_1})\nonumber\\ 
&=& \frac{d}{dt}|_{t=0}F(\exp tZ\exp sX_1)\nonumber
\\&=& \frac{d}{dt}|_{t=0}\pi(\exp tZ) F(\exp sX_1)\nonumber
\\&=:&\pi(Z)F(\exp sX_1). 
\end{eqnarray}

We assume in addition that $F$ solves the differential equation 
\begin{equation}\label{eq: ode}
\Omega F+\mu F=0
\end{equation}
for some $\mu\in \mathbb{C}$. Then we use the expression for $\Omega$ from equation (\ref{eq: omegapolar})

\begin{eqnarray*}
\Omega&=& H_1^2-2H_\rho+2X_1^2+ 2(l-2)f(\exp X')X_1+2f(\exp X')^2\sum_j \left((Z_{Y_j}Z_{Y_j})^{\exp -X'}\right)\\
&&-2X_1W-\sum_iM_i^2-\sum_j 2Z_{Y_j}^{\exp -X'}+Z_{Y_j}Z_{Y_j}-f(\exp X')^{-1}(Z_{Y_j}^{\exp -X'}-Z_{Y_j})W_j, 
\end{eqnarray*} 
which is valid for any $X'=sX_1$, $s\neq 0$. Here $f(\exp X')=f(\exp sX_1)=\frac{1}{s}$, $W,W_j\in \mathfrak{m}^{\bot_\mathfrak{k} }$ and $Z_{Y_j}\in \mathfrak{z}_\mathfrak{m}(\mathfrak s)^{\bot_\mathfrak m}$ as in Lemma \ref{lem: jrelations}.  
Hence, $\Omega$ can be written modulo $U(\mathfrak{g})\mathfrak{k}$ for any $X'=sX_1$, $s\neq 0$, in the form

\begin{equation}\label{eq: radom}
H_1^2-2H_\rho+ 2X^2_1 +2\frac{l-2}{s}X_1+\frac{2}{s^2}\sum_{j}(Z_{Y_j}Z_{Y_j})^{\exp -X'}.
\end{equation}

We call this the \index{polar coordinate form of $\Omega$} polar coordinate form of $\Omega$. Let $\overline{F}$ be the restriction of $F$ to $S\cdot o\subset X$, where $o=K$ and $\tilde F$ the lift of $F$ to a function $\tilde F:G\to \mathrm{End}(V_\pi)$. Then it follows from (\ref{eq: radom}) that for $s>0$

\begin{eqnarray*}
(\Omega F)(\exp sX_1\cdot o)&=& \left[\left(H_1-2H_\rho+2X_1^2+\frac{2(l-2)}{s}X_1 \right)\overline{F}\right](\exp sX_1\cdot o)\\
&&+\left[\left(\frac{2}{s^2}\mathrm{Ad}(\exp -sX_1)\sum_j Z_{Y_j}^2 \right)\tilde{F}\right](\exp sX_1).
\end{eqnarray*}

It follows from (\ref{eq: zmequi}) that \[\left(\mathrm{Ad}(\exp -sX_1)\sum_j Z_{Y_j}^2\right) \tilde{F}(\exp sX_1)=\sum_j \pi(Z_{Y_j}^2)\tilde{F}(\exp sX_1).\]

Hence, the restriction $\overline F$ of $F$ to $S\cdot o$ satisfies

\begin{equation}\label{eq: vecfield}
\left(H_1^2-2H_\rho+2X^2+2\frac{l-2}{s}X_1+\frac{2}{s^2}\sum_{j}\pi(Z_{Y_j}^2)+\mu\right) \overline{F}=0, 
\end{equation}
$s\in\mathbb R^+$. Let \index{$Z_M(S)$, centralizer of $S$ in $M$} \[Z_M(S):=\{m\in M:m\cdot \exp sX_1=\exp sX_1 \mbox{ for all } s\in\mathbb R\}.\]  

Since $\mathrm{pr}:G\to X=AN$ is the identity when restricted to $S\subset N$ and by the commutative diagram from above, it follows that \begin{eqnarray*}Z_M(S)&=&\{m\in M:m\cdot \exp sX_1=\exp sX_1 \mbox{ for all } s\in \mathbb{R}\}\\&=&\{m\in M:\exp sX_1m=m\exp sX_1 \mbox{ for all } s\in \mathbb{R}\}
.\end{eqnarray*}

That is, $Z_M(S)$ also equals the centralizer of $\exp X_1$ in $M$.  
\begin{lemma}\label{chap: ode aux}
For any $(\pi,V_\pi)\in \widehat M$ we have that $\dim(V_\pi^{Z_M(S)})$ is either 0 or 1. 
\end{lemma}

\begin{proof}
This follows since $(M,Z_M(S))$, $M\cong SO(l-1)$, $Z_M(S)\cong SO(l-2)$, is a Gelfand pair. Hence for any irreducible $\pi$, the space of
vectors fixed under $Z_M(S)$ is at most one dimensional.  
\end{proof}

\begin{lemma}\label{chap: ode lem: inv}
Assume $V_\pi^{Z_M(S)}$ is spanned by $v$. The restriction $\overline{F}$ of $F$ to $S\cdot o$ maps $V_\pi$ to $V_\pi^{Z_M(S)}$  
\end{lemma}

\begin{proof}
This is a consequence of $F(m\cdot x)=\pi(m)F(x)$ for any $m\in M$, $x\in X$. Then for any $v\in V_\pi$
\begin{eqnarray*} \pi(m)F(\exp sX_1)v&=&F(m\cdot \exp sX_1)v\\
&=&F(\exp sX_1 )v
\end{eqnarray*} for all $s\in \mathbb{R}$ and $m\in Z_M(S)$. This means that $F(\exp sX_1)v$ is $Z_M(S)$-invariant. 
\end{proof}

\begin{lemma}\label{chap: ode d}
$\sum_jZ_{Y_j}^2$ is $Z_M(S)$-invariant and $\sum_j \pi(Z_{Y_j}^2)|_{V_\pi^{Z_M(S)}}$ is negative semidefinite for any $(\pi,V_\pi)\in\widehat M$. 
\end{lemma}

By Lemma \ref{chap: ode aux} we can hence identify the operator $\sum_j \pi(Z_{Y_j}^2)|_{V_\pi^{Z_M(S)}}\in \mathrm{End} (V_\pi^{Z_M(S)})$ with a nonnegative number.

\begin{proof}
 For the first claim we use the orthogonal decomposition 
\begin{equation*} 
\mathfrak{m}=\mathfrak{z}_{ \mathfrak m}(S)\oplus \mathfrak{z}_{ \mathfrak{m} }(S)^{\bot_{ \mathfrak{m} }}.
\end{equation*}

Recall the relations from Lemma \ref{lem: jrelations} on $X_1,Y_j,Z_{Y_j}$. We compute for $Y_k\neq Y_l$ from $\{Y_j\}_j$ 
\begin{eqnarray*} 
0=-B_\theta(Y_k,Y_l)&=&B([Z_{Y_k},X_1],\theta Y_l)
\\&=& B(Z_{Y_k},[X_1,\theta Y_l])
\\&\overset{\text{Lem. } \ref{lem: jrelations}}=& 2B(Z_{Y_k},Z_{Y_l}).
\end{eqnarray*}

By the same computation $B(Z_{Y_j},Z_{Y_j})=1/2$. That is, $\{Z_{Y_j}\}_j$ is an orthogonal set and $|Z_{Y_j}|^2=1/2$. Since the dimension of $\mathfrak z_{\mathfrak m}(S)^{\bot_\mathfrak{m}}$ equals the dimension of $\mathfrak s^{\bot_\mathfrak{n}}$, see Lemma \ref{l2}, it follows that the set $\{\sqrt{2} Z_{Y_j}\}_j$ forms an orthonormal basis for $\mathfrak z_{\mathfrak m}(S)^{\bot_\mathfrak{m}}$. If $\{K_i\}_i$ is an orthonormal basis of $ \mathfrak{z}_{ \mathfrak{m} }(S)$, then $ 2 \sum_j Z_{Y_j}^2$ can be written as the difference 
\begin{equation*}
2\sum_j Z_{Y_j}^2=\Omega_{ \mathfrak{m} }-\sum_iK_i^2
\end{equation*}
of two $Z_M(S)$-invariant operators, where $\Omega_{ \mathfrak{m} }$ is the Casimir of $U( \mathfrak{m} )$.

The second claim is clear, if $V^{Z_M(S)}_\pi=0$. Otherwise we fix some representation $(\pi,V_\pi)$ with $V_\pi^{Z_M(S)}\neq 0$. Since $\sum_j Z_{Y_j}^2$ is $Z_M(S)$-invariant, this implies that $\sum_j\pi(Z_{Y_j}^2)$ maps $V_\pi^{Z_M(S)}$ into itself. Further, $\pi$ defines also a representation of $ \mathfrak{m}$ by differentiating. Since $\pi$ is unitary, i.e. there is some preserved inner product on $V_\pi$, the mapping $\pi(Z)$ is skew-symmetric for any $Z$ in $\mathfrak{m}$. Thus, 
\begin{equation*}
 \mathrm{spec}(\pi(Z))\subset i \mathbb{R} \mbox{ resp. } \mbox{spec}(\pi(Z^2))\subset \mathbb{R}^-
\end{equation*} 
for any $Z$ in $\mathfrak{m}$.
\end{proof}

Thus, we see that $\sum_j\pi(Z_{Y_j}^2)$, just as $F(\exp sX_1)$ for $s\in\mathbb R$, maps $V_\pi^{Z_M(S)}$ into itself. Next, we define the action of $X_1$ resp. for $H\in \mathfrak{a}$ on $C^\infty(X\times_M V_\pi)$. While the action of $X_1$ on $\exp \mathbb{R}X_1$ is standard, i.e. by translation, we define an action of $A$ on $N$ by $a\cdot n=a^{-1}na$. This really  defines an action, since $A$ is commutative. Indeed, \begin{equation*} a_1a_2\cdot n=a_2a_1\cdot n=a_1^{-1}a_2^{-1}na_2a_1=a_1\cdot(a_2\cdot n).\end{equation*}

Furthermore, \[a^{-1}\exp( X_1)a=\exp(-tH)\exp( X_1)\exp(tH)=\exp( e^{-t\alpha(H)}X_1)\in \exp \mathbb{R}X_1.\] 

Thus, we see that the action of $A$ on $N$ restricts to an action on $\exp \mathbb{R}X_1$. The reason for defining the action of $A$ in this non-standard way is that we consider functions on $G$ which are right-$K$- and left-$A$-invariant, i.e. we will consider functions on $N$ using the Iwasawa decomposition $G=ANK$. In this way the action of $A$ on $N$ is compatible with the action of vectorfields $V\in\mathfrak g$ on $C^\infty(G)$, i.e. for $n\in N$, $H\in\mathfrak a$ and $f\in C^\infty(A\backslash G/K)$
\begin{equation*}
Hf(n)=\frac{d}{dt}|_{t=0}f(n\exp tH)=\frac{d}{dt}|_{t=0}f(\exp -tHn\exp tH)=\frac{d}{dt}|_{t=0}f(\exp tH\cdot n) .
\end{equation*}Next we determine how the vector fields occurring in (\ref{eq: vecfield}) act on functions of $C^\infty(X,\mathrm{End}(V_\pi))$.

\begin{lemma}
Identifying $\exp \mathbb{R}X_1$ with $\mathbb R$, the vectorfield $X_1$ corresponds to $\frac{d}{ds}$ while $H_\rho$ is 
$-\alpha(H_\rho)s\frac{d}{ds}$ and $H_1^2$ acts as $\alpha(H_1)^2\left(s^2\frac{d^2}{ds^2}+s\frac{d}{ds}\right)$.
\end{lemma}

\begin{proof}
Let $F\in C^\infty(X\times_M V_\pi) $. We have for any $H\in \mathfrak{a}$ and $s\in \mathbb{R}$
\begin{eqnarray*}
(HF)(\exp sX_1) &=& \frac{d}{dt}|_{t=0}F( \exp -tH \exp sX_1\exp tH)\\ 
&=& \frac{d}{dt}|_{t=0}F( \exp sX^{\exp -tH})\\ &=&  \frac{d}{dt}|_{t=0}F(\exp (\mathrm{Ad}(\exp -tH)sX_1))\\ &=& \frac{d}{dt}|_{t=0} F(\exp (\exp -t\mathrm{ad}(H)sX_1))\\ 
&=& \frac{d}{dt}|_{t=0}F(\exp e^{-\alpha(H)t}sX_1)\\ &\overset{\text{Identification}}=&\frac{d}{dt}|_{t=0} F( e^{-\alpha(H)t}s)\\
&=& -\alpha(H)sF'(s).
\end{eqnarray*}
Similarly,
\begin{eqnarray*}
(H^2F)(\exp sX_1) &=& \frac{d^2}{dt^2}|_{t=0}  F(\exp -tH\exp sX_1\exp tH)\\ 
&=& \frac{d^2}{dt^2}|_{t=0} F(\exp e^{-t\alpha(H)}sX_1)\\ 
&\overset{\text{Identification}}=& \frac{d^2}{dt^2}|_{t=0} F(e^{-t\alpha(H)}s)\\ 
&=& \frac{d}{dt}|_{t=0}  -\alpha(H)sF'(e^{-t\alpha(H)}s)\\ 
&=& \alpha(H)^2 \left( s F'(s)+s^2 F''(s) \right),
\end{eqnarray*}
where $F(s):=F(\exp sX_1)$.
\end{proof}
 
We apply equation (\ref{eq: vecfield}) to $v$ and get 
\begin{equation*}
\left(H_1^2-2H_\rho+ 2X^2_1+2\frac{l-2}{s}X_1+\frac{2}{s^2}\sum_{j}\pi(Z_{Y_j}^2)|_{V_\pi^{Z_M(S)}}+\mu\right) F(\exp sX_1)v=0.
\end{equation*}
i.e.
\begin{eqnarray}\label{e11}
\left((\alpha(H_1)^2s^2+2)\frac{d^2}{ds^2}+\left(\left(\alpha(H_1)^2+2\alpha(H_\rho)\right)s+2\frac{l-2}{s}\right)\frac{d}{ds}\right.\nonumber\\
\left.+
\frac{2}{s^2}\sum_{j}\pi(Z_{Y_j}^2)|_{V_\pi^{Z_M(S)}}+\mu\right)F(\exp sX_1)v=0.
\end{eqnarray}

Let $V_\pi^{Z_M(S)}\neq\{0\}$. As $F(\exp sX_1)|_{V_\pi^{Z_M(S)}}$ and $\sum_{j}\pi(Z_{Y_j}^2)|_{V_\pi^{Z_M(S)}}$ are elements of $\mathrm{End}(V_\pi^{Z_M(S)})$ which we identify with $\mathbb C$ we deduce that 

\begin{eqnarray}\label{e1}
\left((\alpha(H_1)^2s^2+2)\frac{d^2}{ds^2}+\left(\left(\alpha(H_1)^2+2\alpha(H_\rho)\right)s+2\frac{l-2}{s}\right)\frac{d}{ds}\right.\nonumber\\
\left.+
\frac{2}{s^2}\sum_{j}\pi(Z_{Y_j}^2)|_{V_\pi^{Z_M(S)}}+\mu\right)F(s)=0
\end{eqnarray}
with the convention $F(\exp sX_1)v=:F(s)v$. We view (\ref{e1}) as an ordinary differential equation for a scalar valued function. 
\begin{lemma}\label{lem: ffinitesl}
Let $F\in C^\infty(G\times_M V_\pi)$. Then \[\mathrm{Tr}(F(s))=\langle F(s)v,v\rangle_{V_\pi},\] where $v\in V_\pi$ spans $V_\pi^{Z_M(S)}$.
\end{lemma}
\begin{proof}
Let $v,v_2,\ldots, v_{d_\pi}$ be an orthonormal basis of $V_\pi$ w.r.t. the inner product $\langle.,,\rangle_{V_\pi}$. Then
\[\mathrm{Tr}(F(\exp sX_1))=\langle F(\exp sX_1)v,v\rangle_{V_\pi} +\sum_{i=2}^{d_\pi}\langle F(\exp sX_1)v_i,v_i\rangle_{V_\pi}\]

By Lemma \ref{chap: ode lem: inv}, $F(\exp sX_1)v_i\in V_\pi^{Z_M(S)}=\mathbb{C}v$, hence $\langle F(\exp sX_1)v_i,v_i\rangle_{V_\pi}=0$ for all $i$.

\end{proof}

Let $M$ act on $C^\infty(X)$ by the left regular representation, i.e. $m\cdot f(x)=f(m^{-1}x)$ for $m\in M$ and $x\in X$. We call $f\in C^\infty(X)$ \index{$M$-finite of type $\pi$} $M$-finite of type $\pi\in\widehat M$, if the left regular representation of $M$ restricted to  $\mathrm{span}\{M\cdot f\}$ decomposes into finitely many copies of $\pi$. We set for $\pi\in\widehat M$ \index{$C^\infty(X)_\pi$, space of $M$-finite fucntions of type $\pi$} \[C^\infty(X)_\pi:=\{f\in C^\infty(X): f \mbox{ $M$-finite of type } \pi\}.\]

For $\pi\in \widehat M$ we denote its contragradient representation by \index{$picheck$@$\check\pi$, contragradient representation of $\pi$} $\check \pi$.
\begin{lemma}\label{lem: finfacts}
For any $M$-finite function $f\in C^\infty(X)$ of type $\check\pi\in \widehat{M}$ we define \[f^\pi(x):=d_\pi \int_M f(m\cdot x)\pi(m^{-1})dm,\]
where \index{$d_\pi$, dimension of $V_\pi$} $d_\pi=\mathrm{dim}(V_\pi)$. Then \begin{itemize}\item[a)] $f^\pi \in C^\infty(X\times_M V_\pi)$,
\item[b)] $f(x)=\mathrm{Tr}(f^\pi(x))$ for all $x\in X$,
\item[c)] $f(\exp sX_1)=\langle F(\exp sX_1)v,v\rangle_{V_\pi}$, in particular $f|_S=0$, if $V_\pi^{Z_M(S)}=\{0\}$.
\item[d)] There exist finitely many slices $S_i$ for $M$ acting on $N$ such that the restriction to $N$ of $f|_N$ vanishes iff all its restrictions $f|_{S_i}$ vanish. \footnote{ Later in Section \ref{sec: specialization} we will see that this is equivalent to the existence of certain $p_i\in \frac{1}{2}\mathbb{N}_0$ such that $\left(X_i^{2p_i}f\right)(e)=0 $ for all $i$, where $S_i=\exp \mathbb{R}^+X_i$, $X_i$ of unit length.}
\end{itemize}
\end{lemma}
\begin{proof}
Claim $a)$ follows from the computation \[f^\pi(m'\cdot x)=d_\pi\int_M f(mm'\cdot x)\pi(m^{-1})dm=d_\pi\int_M f(m\cdot x)\pi(m'm^{-1})dm=\pi(m)f^\pi(x).\]

From \cite[Ch. IV Lem. 1.7]{GGA} it follows that the mapping from $C(X)$ to $C(X)_{\check{\pi}}$ given by convolution with \index{$\chi_\pi$, character of $\pi$} $d_\pi\chi_\pi:=d_\pi\mathrm{Tr}(\pi)$ \[C(X)\to C(X)_{\check{\pi}} \mbox{ , } f\mapsto d_\pi\chi_\pi*f, \]
where $\chi_\pi*f(x):=\int_M f(m\cdot x)\chi_\pi(m^{-1})dm$, is a continuous projection. If $f$ is already $M$-finite of type $\check{\pi}$, then \begin{eqnarray*} f(x)&=&d_\pi\chi_\pi*f(x)\\&=& d_\pi\int_M f(m\cdot x)\chi_\pi(m^{-1})dm\\
&=& d_\pi\int_M f(m\cdot x)\mathrm{Tr}(\pi(m^{-1}))dm\\
&=&  \mathrm{Tr}(f^\pi(x)). \end{eqnarray*}

Claim $c)$ is a direct consequence of $b)$ and Lemma \ref{lem: ffinitesl}. For $d)$ let $n\in N$. Then there is some $m\in M$ and $s\geq 0$ such that $n=m\cdot s$. This implies \[f(n)= f(m\cdot s)=m^{-1}\cdot f(s). \] 

We fix a basis $f_1,f_2,\ldots,f_d$ of $\mathrm{span}\{M\cdot f\}$. We can assume that $f_1=f$ and $f_j=m_j\cdot f$ for some $m_j\in M$, $j=2,\ldots,d$. We set $S_j:=m_{j}^{-1}\cdot S$, then $f_j(s)=m_j\cdot f(s)=f(m_j^{-1}s)$, i.e. $f_j|_S=f_{S_j}$. Since $\{M\cdot f\}\subset\mathrm{span}\{f_1,\ldots,f_d\}$ the claim follows. 
\end{proof} 

We go back to equation (\ref{e1}). The application we have in mind is the following. Let $f\in C^\infty(X)_{\check{\pi}}$, then $f^\pi \in C^\infty(X\times_M V_\pi)$ and with the convention $f^\pi(\exp sX_1)v=:f^\pi(s)v$, $s\geq 0$, it follows that \[f^\pi(s)=\langle f^\pi(\exp sX_1)v,v\rangle_{V_\pi}=\mathrm{Tr}\left(f^\pi(\exp sX_1)\right)=f(\exp sX_1).\] 

Thus, if $f$ satisfies $\Omega f=\mu f$, then by dominated convergence $\Omega f^\pi=\mu f^\pi$ and the restriction $f|_S$ satisfies equation (\ref{e1}). 
 
We note that if $\sum_{j}\pi(Z_{Y_j}^2)|_{V_\pi^{Z_M(S)}}=0$, in particular for $\pi=\mathbf{1}$ the trivial representation, this can be transformed 
to a well-known hypergeometric differential operator, i.e. $F:\mathbb R\to \mathbb{R}$ is given by some hypergeometric function as we will see later. More precisely, we will solve equation (\ref{e1}) in the next section.

Let us summarize the results of this section.
\begin{theorem}\label{th: odesum}
Let $G=SO_o(1,l)$ and $(\pi,V_\pi)\in\widehat M$. Furthermore, let $F\in C^\infty(X\times_MV_\pi)$ with $\Omega F=\mu F$ and $\exp \mathbb{R}X_1$ a section for $M$ acting on $N$. We can assume that $V_\pi^{Z_M(S)}=\mathbb{C}v$, where $v\in V_\pi$ is of unit length or 0. 
\begin{itemize}
\item[a)]
If we define $F:\mathbb{R}\to \mathbb{R}$ by $F(s)v:=F(\exp sX_1)v$ for $v\neq 0$ or $F(s):=0$ for all $s\in\mathbb R$ otherwise, then $F:\mathbb{R}\to \mathbb{R}$ satisfies equation (\ref{e1}) on $\mathbb R^+$.
\item[b)]
For $s\in\mathbb R$: 
\begin{eqnarray*}
F(s)&=& \langle F(s)v,v\rangle_{V_\pi}\\
&=& \langle F(\exp sX_1)v,v\rangle_{V_\pi}\\
&=& \mathrm{Tr}\left(F(\exp sX_1)\right).
\end{eqnarray*}

If $F(x)=f^\pi(x)=d_\pi \int_Mf(m\cdot x)\pi(m^{-1})dm$ for some $f\in C^\infty(X)$, then also 
\begin{eqnarray*}
F(s)&=& d_\pi \int_M f(m\cdot s)\chi_\pi(m^{-1})dm\\
&=& d_\pi(\chi_\pi*f)(\exp sX_1)\\
&=& d_\pi^2\left(\chi_\pi*(\chi_\pi*f)\right)(\exp sX_1).
\end{eqnarray*}
\end{itemize}
\end{theorem}

\section{Differential equations of hypergeometric type}\label{sec: hyppr}

We continue with working on equation (\ref{e1}) and make the substitution $s^2=u$. The differentials transform according to
\begin{equation*}
s\frac{d}{ds}=2u \frac{d}{du} \mbox{ resp. } \frac{1}{s}\frac{d}{ds}=2\frac{d}{du} \mbox{ and } \frac{d^2}{ds^2}=4u \frac{d^2}{du^2}+2 \frac{d}{du}.
\end{equation*}

Thus, we obtain the new equation

\begin{eqnarray*}
\left((4\alpha(H_1)^2u^2+8u)\frac{d^2}{du^2}+\left(\left(4\alpha(H_1)^2+4\alpha(H_\rho)\right)u+4(l-1)\right)\frac{d}{du}\right.\nonumber\\
\left.+
\frac{2}{u}\sum_{j}\pi(Z_{Y_j}^2)|_{V_\pi^{Z_M(S)}}+\mu\right)F(w)=0.
\end{eqnarray*}

Then we divide by the positive scalar $4\alpha(H_1)^2$ to obtain
\begin{eqnarray*}
 \left(u(u+k_1)\frac{d^2}{du^2}+\left((1+k_2)u+k_3\right)\frac{d}{du}+
\frac{k_4}{u}\sum_{j}\pi(Z_{Y_j}^2)|_{V_\pi^{Z_M(S)}}+k_5\right)F(w)=0,
\end{eqnarray*} 
where we set $k_1=\frac{2}{\alpha(H_1)^2}$, $k_2=\frac{\alpha(H_\rho)}{\alpha(H_1)^2}$, $k_3=\frac{l-1}{\alpha(H_1)^2}$, $k_4=\frac{1}{2\alpha(H_1)^2}$ and $k_5=\frac{\mu}{4\alpha(H_1)^2}$. 

Then obviously $k_1,k_4$ and $k_3\in\mathbb R^+$ but also $k_2$ is positive, since $H_\rho$ lies in the \index{positive Weyl chamber} \index{$\mathfrak a^+$, positive Weyl chamber} positive Weyl chamber $\mathfrak{a}^+$ of $\mathfrak{a}$. The (regular) singularities of this ordinary differential equation are $0$, $-k_1$ and $\infty$.

Now we substitute $u\mapsto w=-\frac{u}{k_1}$, i.e. $u=-k_1w$, 
\begin{equation*}
\frac{d}{du}=-\frac{1}{k_1}\frac{d}{dw} \mbox{ and } \frac{d^2}{du^2}=\frac{1}{k_1^2}\frac{d^2}{dw^2}.
\end{equation*}

We get the new equation
\begin{eqnarray*}
\left( (-k_1w)\left((-k_1w)+k_1\right)\frac{1}{k_1^2}\frac{d^2}{dw^2}+\left((1+k_2)(-k_1w)+k_3\right)\frac{1}{(-k_1)}\frac{d}{dw}\right.\nonumber\\-\left. 
\frac{k_4}{k_1w}\sum_{j}\pi(Z_{Y_j}^2)|_{V_\pi^{Z_M(S)}}+k_5\right)F(w)=0.
\end{eqnarray*}

After multiplying with $(-1)$ we arrive at 
\begin{equation}\label{ode3}\left( w(1-w)\frac{d^2}{dw^2}+\left(\frac{k_3}{k_1}-(1+k_2)w\right)\frac{d}{dw}+\frac{k_4}{k_1w}\sum_{i}\pi(Z_{Y_i}^2)|_{V_\pi^{Z_M(S)}}-k_5\right)F(w)=0
\end{equation} 

We call this is an ordinary differential equation of hypergeometric type 
\begin{equation}\label{eh} \left(x(1-x)\frac{d^2y}{dx^2}+[c-(a+b+1)x]\frac{dy}{dx}-aby+\frac{d}{x}\right)F(x)=0\end{equation}
with (regular) singularities at $0$, $1$ and $\infty$. Here \index{$c=\frac{l-1}{2}=\rho_0$} $c=\frac{k_3}{k_1}=\frac{l-1}{2}$, $a+b=k_2=-\frac{\alpha(H_\rho)}{\alpha(H_1)^2}$, $ab=k_5$ and \index{$d=\frac{k_4}{k_1}\sum_{j}\pi(Z_{Y_j}^2)|_{V_\pi^{Z_M(S)}}$} $d=\frac{k_4}{k_1}\sum_{j}\pi(Z_{Y_i}^2)|_{V_\pi^{Z_M(S)}}\leq 0$.

We recall that the dimension of $N$ is $l-1$. Then $\rho=\frac{l-1}{2}\alpha$ and the Cartan-Killing on $\mathfrak a$ form is given by, see \cite[(4.2.10)]{GV} 
\begin{equation*}
B(H,H)=2(l-1)\alpha(H)^2.
\end{equation*}

We also defined $H_0\in\mathfrak{a}^+$ such that $H_0$ satisfies \begin{equation}\label{def: h0}\alpha(H_0)=1.\end{equation} It follows that 
\begin{equation*}
B(H_0,H_0)=2(l-1)
\end{equation*} and 
\begin{equation*}
H_\alpha=\frac{1}{2(l-1)}H_0.
\end{equation*}
 
Furthermore, we can assume that $H_1$ with $B(H_1,H_1)=1$ lies in the positive Weyl chamber $\mathfrak{a}^+$. That is, 
\begin{equation*}
H_1=\frac{1}{\sqrt{2(l-1)}}H_0
\end{equation*}
and 
\begin{equation*}
H_\rho=\frac{1}{4}H_0.
\end{equation*}
For all these facts see \cite[p.135]{GV}. 

Now \begin{equation*}
a+b=k_2=\frac{\alpha(H_\rho)}{\alpha(H_1)^2}=\frac{\frac{1}{4}\alpha(H_0)}{\frac{1}{2(l-1)}\alpha(H_0)^2}=\frac{(l-1)}{2}=\rho_0 
\end{equation*} and 
\begin{equation*}
a\cdot b=k_4=\frac{\mu}{4\alpha(H_1)^2}=\frac{\mu (l-1)}{2\alpha(H_0)^2}=-\frac{\mu (l-1)}{2}=\mu\rho_0.
\end{equation*}  

Let us assume that the eigenvalue is given by  
\begin{equation*}
\mu=\frac{1}{4}\left(\rho_0+\frac{r^2}{\rho_0}\right)
\end{equation*} for some $r\in\mathbb C$. Then \index{$a=\frac{1}{2}(\rho_0+ir)$} $a=\frac{1}{2}(\rho_0+ir)$ and \index{$b=\frac{1}{2}(\rho_0-ir)$} $b=\frac{1}{2}(\rho_0-ir)$ solves the equations for $a+b$ and $a\cdot b$, where $\rho_0=\rho(H_0)=\frac{l-1}{2} $.

\section{Solutions to $\left( x(1-x)\frac{d^2y}{dx^2}+[c-(a+b+1)x]\frac{dy}{dx}-aby+
\frac{d}{x}\right)f(x)=0$}\label{sec: ode.1}

Motivated by the chapter in \cite[III,2,3]{Y} on solutions to the hypergeometric equation we define $D=x\frac{d}{dx}$. Then it can be shown that the differential operator 
\begin{equation*}
E(a,b,c):=\left( x(1-x)\frac{d^2y}{dx^2}+[c-(a+b+1)x]\frac{dy}{dx}-aby\right)
\end{equation*}
can be factorized as
\begin{equation*}
E(a,b,c)=\left((c+D)(1+D)\frac{1}{x}-(a+D)(b+D)\right),
\end{equation*}
see \cite[p.61]{Y}. Thus we define 

\begin{eqnarray*}
F(a,b,c,d)&:=&  x(1-x)\frac{d^2y}{dx^2}+[c-(a+b+1)x]\frac{dy}{dx}-aby +\frac{d}{x}
\nonumber\\&=&E(a,b,c)+\frac{d}{x}\nonumber\\
&=& \left((c+D)(1+D)+d\right)\frac{1}{x}-(a+D)(b+D). 
\end{eqnarray*}

Now it is easy to show that
\begin{equation}\label{eq: hytra}
D(x^pu(x))=x^p(p+D)u(x) \mbox{ , i.e. } Dx^p=x^p(p+D), 
\end{equation} if we view $x^p$ as an operator $u\mapsto x^pu$ for functions $u=u(x)$. Hence also \[D^2x^p=x^p(p+D)^2.\]

We now claim that 
\begin{equation*}
x^{p_1}{}_2F_1(a+p_1,b+p_1,1+p_1-p_2;x)
\end{equation*} and 
\begin{equation*}
x^{p_2}{}_2F_1(a+p_2,b+p_2,1+p_2-p_1;x)
\end{equation*} both solve 
\begin{equation*}
F(a,b,c,d)f(x)=0,
\end{equation*}
where in turn \index{$p_{1,2}$, roots of indicial equation} $p_{1,2}$ solve the indicial equation 
\begin{equation}\label{eq: indic1}
p^2+(c-1)p+d=0.
\end{equation}

Therefore we compute what happens if we apply $F(a,b,c,d)$ to functions of the form $x^p g(x)$. Again we omit the function $g(x)$ viewing $x^p$ as the operator $g\mapsto x^pg$ to keep formulas simple. We compute

\begin{eqnarray*}
F(a,b,c,d)x^p&=&\left[\left[(c+D)(1+D)+d\right]\frac{1}{x}-(a+D)(b+D)\right]x^p\\ 
&=& \left[c+d+(c+1)D+D^2\right]x^{p-1}-\left[(ab+(a+b)D+D^2)\right]x^p\\
&=& x^{p-1}\left[c+d+(c+1)(p-1+D)+(p-1+D)^2\right]-x^p \left[ab+(a+b)(p+D)+(p+D^2)\right]\\
&=& x^p\left[x^{-1}\left(c+d+(c+1)(p-1+D)+(p-1+D)^2\right)\right]-x^p\left[(a+p+D)(b+p+D)\right]\\
&=&  x^p\left[x^{-1}\left(p^2+(c-1)p +d+(c-1+2p)D+D^2\right)\right]-x^p\left[(a+p+D)(b+p+D)\right]\\
&=:&(*).
\end{eqnarray*}

To get this into the desired form we must therefore have 
\begin{equation*}
p^2+(c-1)p+d=0,
\end{equation*}
i.e.
\begin{equation*}
p_{1,2}=\frac{1-c}{2}\pm \sqrt{\left(\frac{1-c}{2}\right)^2-d} \mbox{ and } p_1-p_2=2\sqrt{\left(\frac{1-c}{2}\right)^2-d}.
\end{equation*}

It follows that 
\begin{equation*}
c-1+2p_{1,2}=\pm 2\sqrt{\left(\frac{1-c}{2}\right)^2-d}=\pm(p_1-p_2)
\end{equation*}
and
\begin{eqnarray*}
&&x^{-1}\left(p_{1,2}^2+(c-1)p_{1,2} +d+(c-1+2p_{1,2})D+D^2\right)\\&=& x^{-1}\left(\pm (p_1-p_2)D+D^2\right)
\\ 
&=& x^{-1}(-1+1\pm(p_1-p_2)D+D^2)
\\ &=& x^{-1}\left(1\pm(p_1-p_2)-(2\pm(p_1-p_2))\right.
+\left.(2\pm(p_1-p_2))D+1-2D+
D^2\right)\\
&=& x^{-1}\left(1\pm(p_1-p_2)+(2\pm(p_1-p_2))(-1+D)\right.
+\left.(-1+D)^2\right)\\
&\overset{(\ref{eq: hytra})}=& \left(1\pm(p_1-p_2)+(2\pm(p_1-p_2))D+D^2\right)x^{-1}
\\ &=& (1+D)(1\pm(p_1-p_2)+D)\frac{1}{x},
\end{eqnarray*}
that is,

\begin{eqnarray*}
(*)&=& x^{p_{1,2}}\left[(1+D)(1\pm(p_1-p_2)+D)\frac{1}{x}-(a+p+D)(b+p+D)\right]\\
&=& x^{p_{1,2}}E\left(a+p_{1,2},b+p_{1,2},1\pm (p_1-p_2)\right).
\end{eqnarray*}

Thus we have shown 
\begin{equation*}
F(a,b,c,d)[x^{p_{1,2}} g(x)]=x^{p_{1,2}} E\left(a+p_{1,2},b+p_{1,2},1\pm (p_1-p_2)\right)g(x),
\end{equation*}
hence 
\begin{eqnarray*}
&&F(a,b,c,d)\left[x^{p_{1,2}} {}_2F_1\left(a+p_{1/2},b+p_{1/2},1\pm (p_1-p_2)\right)\right]
\\&=& x^{p_{1,2}} E\left(a+p_{1,2},b+p_{1,2},1\pm (p_1-p_2)\right){}_2F_1\left(a+p_{1/2},b+p_{1/2},1\pm (p_1-p_2)\right)
\\&=&0
\end{eqnarray*}
as we claimed. 

Another point is which solutions are smooth at the origin. We are only interested in solutions which are smooth everywhere, in particular at the origin. While 
$x^{p_1}{}_2F_1(a+p_1,b+p_2,1+p_1-p_2)$ is always defined as $p_1\geq 0$, the second solution $x^{p_2}{}_2F_1(a+p_1,b+p_2,1-p_1+p_2)$ is only defined even for $x\neq 0$ if $p_2-p_1\neq -1,-2,-3,\ldots$.
Finally, $x^{p_{1,2}} {}_2F_1\left(a+p_{1/2},b+p_{1/2},1\pm (p_1-p_2)\right)$ is smooth at the origin iff $p_{1,2}\in \mathbb{N}_0$.

It remains to consider when the two solutions are linearly independent, i.e. span the space of solution to $F(a,b,c,d)f=0$. We will come to this in the next section.
\subsection{Application of Frobenius method}
Since the second solution \[x^{p_2}{}_2F_1(a+p_1,b+p_2,1-p_1+p_2)\] is not defined for special $p_2$ we shortly explain another method for finding two linearly independent solutions of (\ref{e1}) which has the advantage of covering all cases but is less explicit. It is called the \index{Frobenius method}\textit{Frobenius method} and we follow Chapter 6 in \cite{Mi}. Let us consider the second order differential equation

\begin{equation}\label{eq: so}
\frac{d^2y}{dz^2}+P(z)\frac{dy}{dz}+Q(z)y=0
\end{equation}
on the punctured disc $D^*=\{z\in \mathbb{C}:0<|z|<1\}$, where $P$ has a pole of at most order 1 at 0 and $Q$ of order 2. Let 

\begin{equation*}
zP(z)=\sum_{r=0}^\infty a_rz^r
\end{equation*} 
and
\begin{equation*}
z^2Q(z)=\sum_{s=0}^\infty b_sz^s
\end{equation*}
be the Taylor series in $D$. Then we make the ansatz for a formal solution \[y(p,z)=z^p\sum_{t=0}^\infty c_tz^t.\] 

The coefficients can now be determined by putting $y(p,z)$ into (\ref{eq: so}). If 
\begin{equation}\label{eq: indic2}
f(p):=p(p-1)+pa_0+b_0
\end{equation}
denotes the indicial equation, one can show that
\begin{equation}\label{eq: indic}
c_t=\frac{\sum_{k=1}^t\left((t-k+p)a_k+b_k\right)c_{t-k}}{f(p+t)}.
\end{equation}

Let now $r$ and $s$, $s<r$, be the roots of (\ref{eq: indic2}). Since $r+\mathbb N$ does not contain another root the succeeding coefficients can be calculated recursively starting with $c_0=1$. Thus we obtain a first  solution to (\ref{eq: so}) $y(r,z)=z^r\sum_{t=0}^\infty c_t z^t$. If $r-s$ is not an integer we also find a second solution $y(s,z)$ which is linearly independent from the first. If on the other hand $r-s\in \mathbb{N}$, then one can construct another solution independent from $y(r,z)$ but having a logarithmic singularity at $z=0$. 

Going back to equation (\ref{e1}) we divide by $\alpha(H_1)^2s^2+2$ to find 
\begin{equation}\label{e2}
\left(\frac{d^2}{ds^2}+P(s)\frac{d}{ds}+Q(s)\right)F(s)=0,
\end{equation} 
where 
\begin{equation*}
P(s):=\frac{\left(\alpha(H_1)^2+2\alpha(H_\rho)\right)s+2 \frac{l-2}{s} }{\alpha(H_1)^2s^2+2}
\end{equation*}
has a pole of order 1 at $s=0$ and 
\begin{equation*}
Q(s):=\frac{ \frac{2}{s^2}\sum_j\pi(Z_{Y_j}^2)|_{V_\pi^{Z_M(S)}}+\langle \lambda , \lambda \rangle + \langle \rho,\rho,\rangle }{\alpha(H_1)^2s^2+2}
\end{equation*}
has a pole of order 2, i.e. the method of Frobenius is applicable. We find that $a_0=l-2$, $b_0=4d$ in accordance with (\ref{eq: indic1}). We already know by Lemma \ref{chap: ode d} that $d\leq 0$ and we distinguish the two cases $d=0$ and $d<0$.

If $d<0$, then $p_1>0$, while $p_2<0$. The first solution \[x^{p_1}{}_2F_1(a+p_1,b+p_1,1+p_1-p_2;x)\] is always defined and smooth at the origin for $p_1\in\mathbb N_0$, the second solution is only defined even for $x\neq 0$ if $p_2-p_1\neq -1,-2,-3,\ldots$, as we have seen above. By the Frobenius method one can construct another solution which  has a singularity at the origin, i.e. there is only one solution which is smooth at $x=0$.

If $d=0$, then $p_1=0$ and $p_2=1-c$, i.e. $F(a,b,c,0)$ reduces to $E(a,b,c)$. The two solutions simplify to 
\begin{equation*}
x^{p_1}{}_2F_1(a+p_1,b+p_1,1+p_1-p_2;x)={}_2F_1(a,b,c;x)
\end{equation*}
and 
\begin{equation*} 
x^{p_2}{}_2F_1(a+p_2,b+p_2,1+p_2-p_1;x)=x^{1-c}{}_2F_1(a+1-c,b+1-c,2-c;x).
\end{equation*}

In our case where $G=SO_o(1,l)$ this implies $c=\rho_0=\frac{l-1}{2}$, i.e. $c=\frac{1}{2},1,\frac{3}{2},\ldots$. Thus, except for $c=\frac{1}{2}$ or $c=1$ there is at most one solution which is smooth at the origin, namely ${}_2F_1(a,b,c;x)$. 

Let us discuss the two cases $c=\frac{1}{2}$ and $c=1$, i.e. the cases $G=SO_o(1,2)$ and $G=SO_o(1,3)$. We start with $c=\frac{1}{2}$. For $c=\frac{1}{2}$, $x^{1-c}=x^{1/2}$. Since we made the substitution $s^2=-u$ to get from equation (\ref{e1}) to the hypergeometric type equation (\ref{eh}) there is a second smooth solution to (\ref{e1}) for $d=0$ and $c=\frac{1}{2}$, i.e. $G=SO_o(1,2)$. This solution is - up to constants - 
\begin{equation*}
is \cdot {}_2F_1\left(a+\frac{1}{2},b+\frac{1}{2}, \frac{3}{2};-\frac{s^2}{4} \right),  
\end{equation*}
which is an odd function.

For $c=1$, i.e. $G=SO_o(1,3)$, the solutions ${}_2F_1(a,b,c;x)$ and $x^{1-c}{}_2F_1(a+1-c,b+1-c,2-c;x)$ coincide. Then 
\begin{equation*}
\frac{d}{dc}|_{c=1}x^{1-c}{}_2F_1(a+1-c,b+1-c,2-c;x) 
\end{equation*} 
is another solution of (\ref{eh}) having a singularity at $x=0$, see \cite[p.64]{Y}. 

In any case there is for $d=0$ - up to constants - only one even, everywhere smooth solution to (\ref{e1}). 

We make the resubstitution $s^2=u=-k_1w=-k_1x$, where $k_1=\frac{2}{\alpha(H_1)^2}=4(l-1)$. It follows that $x=-\frac{s^2}{4(l-1)}$ and $x^{p_{1,2}}=\left(\frac{-s^2}{4(l-1)}\right)^{p_{1,2}}$.
Finally let us state the results of this chapter in summarized form:

\begin{theorem}\label{chap: ode main}
For $l\geq 3$ the equation (\ref{e1}) has up to constants at most one solution which is smooth at the origin $s=0$. This solution is up to a constant given by \index{$p$, $p=\frac{1-c}{2}+ \sqrt{\left(\frac{1-\rho_0}{2}\right)^2-d}$}
\begin{equation*}
\left(\frac{-s^2}{4(l-1)}\right)^{p}{}_2F_1\left(a+p,b+p,1+2\sqrt{\left(\frac{1-\rho_0}{2}\right)^2-d};\frac{-s^2}{4(l-1)}\right).
\end{equation*}

Here $a=\frac{1}{2}(\rho_0+ir)$, $b=\frac{1}{2}(\rho_0-ir)$ depend only on the eigenvalue $\mu$ given by $\mu=\frac{1}{4}(\rho_0+\frac{r^2}{\rho_0})$, while $d\leq 0$ depends only on the $K$-type $\pi\in \widehat{K}$ and the section $S$. Here
\begin{equation*}
p=\frac{1-c}{2}+ \sqrt{\left(\frac{1-\rho_0}{2}\right)^2-d}\geq 0
\end{equation*}
solves the indicial equation (\ref{eq: indic1}). For the solution to be 
smooth at the origin it is necessary and sufficient that $2p\in \mathbb{N}_0$. This implies also $1+2\sqrt{\left(\frac{1-\rho_0}{2}\right)^2-d}\in \mathbb{Q}$. Further, if $d=0$, then $p=0$.

For $l=2$ the space of solution to (\ref{e1}) is 2-dimensional and spanned by
${}_2F_1\left(a,b,\frac{1}{2};\frac{-s^2}{4}\right)$  and $is\cdot {}_2F_1\left(a+\frac{1}{2},b+\frac{1}{2}, \frac{3}{2};-\frac{s^2}{4} \right)$.
\end{theorem}

\begin{remark}\label{rem: hypsymxp}
Let $S=\exp \mathbb{R}^+X_{1}$ be a slice. If\begin{eqnarray*}F(\exp sX_1)&=& (-1)^p\left(\frac{s^2}{4(l-1)}\right)^{p}{}_2F_1\left(a+p,b+p,1+2\sqrt{\left(\frac{1-\rho_0}{2}\right)^2-d};\frac{-s^2}{4(l-1)}\right),
\end{eqnarray*}
then it follows that
\[F(\exp 2 \sqrt{l-1}sX_1)= (-1)^ps^{2p}\cdot{}_2F_1\left(a+p,b+p,1+2\sqrt{\left(\frac{1-\rho_0}{2}\right)^2-d};-s^2\right).\]

In particular, if $X_1$ is proportional to $X_{e_1}$, see Section \ref{chap: hyper} for the definition of $X_{e_1}$, this implies $2 \sqrt{l-1}X_1 = X_{e_1}$, see (\ref{eq: normx1}), and
\[F(\exp sX_{e_1})= (-1)^ps^{2p}\cdot{}_2F_1\left(a+p,b+p,1+2\sqrt{\left(\frac{1-\rho_0}{2}\right)^2-d};-s^2\right).\]
\end{remark}

\chapter{The zeta function on $\{\mathrm{Re}(k)>2\rho_0\}$}\label{chap: gtrace}
In this chapter we want to compute the trace of a certain convolution operator $\sigma\cdot \pi_R(f)$ in order to obtain  the (logarithmic derivative of an) auxiliary zeta function $\mathcal{R}(\sigma)=\mathcal{R}(\cdot;\sigma)$ in a half plane of $\mathbb C$. We show that this operator is of trace class and has a kernel. Then we compute its trace by integrating the kernel over the diagonal. For $\sigma\equiv 1$ this procedure yields Selberg's trace formula which is used to derive the classical dynamical zeta function. We generalize this approach to nontrivial eigenfunctions $\sigma=\varphi$ of the Laplacian. From the auxiliary zeta function $\mathcal{R}(\varphi)=\mathcal{R}(\cdot;\varphi)$ we derive the zeta function $\mathcal{Z}(\varphi)=\mathcal Z(\cdot;\varphi)$ as a superposition of shifted $\mathcal{R}(\varphi)$. This idea was first used in \cite{AZ} for $G=SL_2(\mathbb R)$. 

The theory of this chapter is developed for real hyperbolic spaces $X$ as discussed in Section \ref{chap: hyper}. In particular, $X=G/K$ with $G=SO_o(1,l)$ and $K\cong SO(l)$, where $l\in \mathbb{N}$, $l\geq 2$. We identify $X$ with $AN$ by using the Iwasawa decomposition $G=ANK$ , $A\cong\mathbb R$, $N\cong \mathbb{R}^{l-1}$ and $M=Z_K(A)\cong SO(l-1)$. Finally, $\Gamma$ denotes a 
uniform lattice, that is, $\Gamma\subset G$ is a discrete, torsion-free and cocompact subgroup. Every nontrivial $\gamma\in\Gamma$ can be conjugate to some $a_\gamma m_\gamma\in A^+M$, see Proposition \ref{chap: geom gcon}.

\section{A generalized trace formula}\label{sec: traceform}

We let $G=SO_o(1,l)$. Recall the identification of $\mathbb R$ with $A$, $t\mapsto \exp tH_0$ and of $\mathbb R^{l-1}$ with $N$ via $u\mapsto \exp X_u$ from Section \ref{chap: hyper}. Let \index{$da$, Haar measure on $A$} $da$ and \index{$dn$, Haar measure on $N$} $dn$ be the left-invariant Haar measures on $A$ and $N$ obtained from this identification.\footnote{ For the integral definition of Harish-Chandra's $c$-function one normally requires $\int_{\theta N} \exp\left(-2\rho\left(H(\theta n)\right)\right)d\theta n=1$, see \cite[Ch. IV Th. 6.14]{GGA}. We discuss this in Remark \ref{rem: hccfc}.} Further we fix a Haar measure \index{$dk$, Haar measure on $K$} $dk$ on $K$ by requiring $K$ to have unit measure, then it follows from \cite[Ch.I Prop. 5.1]{GGA} and \cite[Ch. I Cor. 5.3]{GGA} that there is a  Haar measure \index{$dg$, Haar measure on $G$} $dg$ on $G$ such that \[\int_G f(g)dg=\int_{ANK}f(ank)dadndk\] for all integrable functions $f$. We fix a Haar measures on $X$, $\Gamma\backslash G$ and on $X_\Gamma=\Gamma\backslash G/K$ such that for all integrable functions on $G$ $$\int_G f(g)dg=\int_{\Gamma\backslash G}\sum_{\gamma\in \Gamma}f(\gamma g)d(\Gamma g)$$ and $$\int_G f(g)dg=\int_X\int_Kf(gk)dkd(gK)$$ resp. for all integrable functions on $\Gamma\backslash G$ $$\int_{\Gamma \backslash G}f(\Gamma g)d(\Gamma g)=\int_{X_\Gamma}\int_K f(\Gamma g k)dkd(\Gamma gK).$$  

Let \index{$C_c^\infty(G//K)$, space of bi-$K$-invariant functions of compact support} $C_c^\infty(G//K)$ be the set of smooth functions with compact support that are bi-$K$-invariant. Denote the right-regular representation of $G$ on $L^2(\Gamma\backslash G)$ by \index{$piR$@$\pi_R$, right-regular representation of $G$} $\pi_R$ and let $\varphi,\sigma\in C^\infty(\Gamma\backslash G)$ and $f\in C_c^\infty(G)$. The Fourier transform of $f$ is defined as an operator on $L^2(\Gamma\backslash G)$ by \index{$piRf$@$\pi_R(f)$, Fourier transform w.r.t. $\pi_R$}
\begin{equation}\label{def: fourierr}
\pi_R(f)\varphi(x):=\int_{G}f(g)\left(\pi_R(g)\varphi\right)(x)dg=\int_Gf(g)\varphi(xg)dg=\left(*\tilde{\varphi}\right)(x^{-1}),
\end{equation}
where \index{$fit$@$\tilde \phi$, $\tilde{\phi}(x)=\phi(x^{-1})$} $\tilde{\phi}(x):=\phi(x^{-1})$ and convolution is defined by \index{$phif$@$\varphi*f$, convolution} $(\varphi*f)(x):=\int_G\varphi(h)f(h^{-1}x)dh$. Note that $\pi_R(f)\varphi$ is trivially square integrable as $\Gamma\backslash G$ is compact. 

We now combine $\pi_R(f)$ with the multiplication operator on $L^2(X_\Gamma)$ which sends $f\mapsto \sigma\cdot f$.
By the compact support of $f$ and an application of Fubini's theorem
\begin{eqnarray}\label{eq: convolution}
\left[\sigma\cdot \pi_R(f)\right] \varphi(\Gamma x)&=& \sigma(\Gamma x)\int_Gf(g)\varphi(\Gamma xg)dg\nonumber\\
&=& \sigma(\Gamma x)\int_G\varphi(\Gamma g)f(x^{-1}g)dg\\ &=& \sigma(\Gamma x) \int_{\Gamma\backslash G}\underset{ \gamma\in \Gamma}\sum \varphi(\Gamma\gamma g)f(x^{-1}\gamma g)d(\Gamma g)\nonumber\\ &=& \sigma(\Gamma x)\int_{\Gamma \backslash G}\varphi(\Gamma g)\underset{\gamma\in\Gamma}\sum f(x^{-1}\gamma g)d(\Gamma g)\nonumber.
\end{eqnarray}

So, if we define \index{$K_\sigma(x,g)$, intgeral kernel of $\sigma\cdot\pi_R(f)$} $K_\sigma(x,g):=\sigma(x)\underset{\gamma\in \Gamma}\sum f(x^{-1}\gamma g)$, then $K_\sigma(\gamma x,g)=K_\sigma(x,\gamma g)=K_\sigma(x,g)$ defines a smooth function $K_\sigma:\Gamma\backslash G\times \Gamma \backslash G\to\mathbb C$ satisfying 
\begin{equation*} \left[\sigma\cdot \pi_R(f)\right] \varphi(x)=\int_{\Gamma\backslash G}\varphi(\Gamma g)K_\sigma(x,\Gamma g)d (\Gamma g).
\end{equation*} 

From \cite{GW} we adopt the following definition. We call $f\in C^\infty(G//K)$ \index{admissible} \textit{admissible}, if \begin{itemize}
\item the operator $\pi_R(f)$ is of trace class on $L^2(\Gamma\backslash G)$ and
\item the series $ \sum_{\gamma\in\Gamma} f(x\gamma y^{-1})$ converges to a continuous function of $(x,y)$.

\end{itemize}
\begin{proposition}\label{prop: traceclass}
For $\sigma\in C^\infty(\Gamma\backslash G/K)$ and $f\in C^\infty(G//K)$ admissible the operator $\sigma\cdot\pi_R(f)$ is an integral operator with kernel $K_\sigma(\cdot,\cdot)$ as from above. It maps 
$L^2(X_\Gamma)$ into itself and is a trace class operator on 
$L^2(X_\Gamma)$.
\end{proposition}
\begin{proof}
Standard theory implies that $\pi_R(f)$ is of trace class on $L^2(\Gamma\backslash G)$, see \cite[p.172]{Wal}, and hence also on $L^2(X_\Gamma)$ as soon as we check that $\pi_R(f)$ resp. $\sigma\cdot\pi_R(f)$ leave $L^2(X_\Gamma)$ invariant. Note that $\pi_R(f)$ being of trace class implies that $\sigma\cdot\pi_R(f)$ is of trace class as well, since $\sigma$ is bounded. So we only have to check invariance. Since by assumption $\sigma$ is right-$K$-invariant and $f$ is also bi-$K$-invariant, we get for $k\in K$
\\
\begin{eqnarray*}
\left[\sigma\cdot \pi_R(f)\right]\varphi(xk)&=& \sigma(xk)[\pi_R(f)\varphi](xk)\\
&=& \sigma(x)\int_{G}f(g)\varphi(xkg)dg\\ &=&
\sigma(x)\int_G f(k^{-1}g)\varphi(xg)dg\\&=&
\sigma(x)\int_Gf(g)\varphi(xg)dg\mbox{ }=\mbox{ }\left[\sigma\cdot \pi_R(f)\right]\varphi(x),
\end{eqnarray*}
where we applied the transformation $g\mapsto k^{-1}g$ and used the unimodularity of $G$. 
\end{proof}

From now on we assume that $\sigma\in C^\infty(\Gamma\backslash G/K)$ and $f\in C^\infty(G//K)$ is admissible. We can now compute the trace by

\begin{eqnarray*}
\mathrm{Tr}(\sigma\cdot \pi_R(f)) 
&=& \int_{X_\Gamma}\sigma(\Gamma y)\sum_{\gamma\in\Gamma}f(y^{-1}\gamma y)d(\Gamma y) \\
&=& f(e)\int_{F_\Gamma}\sigma(x)dx +\sum_{1\neq\gamma\in\Gamma}\int_{F_\Gamma} \sigma(x)f(x^{-1}\gamma x)dx,
\end{eqnarray*}
where \index{$F_\Gamma$, fundamental domain for $\Gamma$} $F_\Gamma\subset X$ is some \index{fundamental domain} fundamental domain for $\Gamma$ in $X$, i.e. $F_\Gamma$ is measurable and up to a possible set of measure zero it contains exactly one element of every orbit $\Gamma x$, $x\in X$. In particular, any fundamental domain $F_\Gamma$ satisfies $\int_{\Gamma\backslash X} f(x)dx=\int_{F_\Gamma}f\circ\mathrm{pr}(g)dg$, $\mathrm{pr}:X\to \Gamma\backslash X$ canonical projection, for any integrable function $f:\Gamma\backslash X\to\mathbb C$. 

We recall notations from Chapter \ref{chap: geometry}, in particular that $C\Gamma$ denotes the set of conjugacy classes $[\gamma]$ in $\Gamma$ and $\Gamma_\gamma$ is the centralizer of $\gamma\in \Gamma$ . Then we continue our computation with

\begin{eqnarray*}
\sum_{1\neq\gamma\in\Gamma}\int_{F_\Gamma} \sigma(x)f(x^{-1}\gamma x)dx&=& \sum_{1\neq[\gamma]\in C\Gamma}\sum_{\gamma'\in[\gamma]} \int_{F_\Gamma}\sigma(x)f(x^{-1}\gamma' x)dx\\
&=& \sum_{1\neq[\gamma]\in C\Gamma}\sum_{\Gamma_\gamma\gamma'\in \Gamma_\gamma\backslash \Gamma}\int_{F_\Gamma}\sigma(\gamma' x)f(x^{-1}\gamma'^{-1}\gamma\gamma'x)dx
\end{eqnarray*}

In Proposition \ref{chap: geom icg} it was mentioned that for any nontrivial $\gamma\in \Gamma$ its centralizer $\Gamma_\gamma$ is cyclic and infinite as $\Gamma$ is torsion-free. Furthermore, there is a unique primitive element $\gamma_0\in \Gamma$ such that $\gamma=\gamma_0^{n_\gamma}$ for some \index{$n_\gamma$, $\gamma=\gamma_0^{n_\gamma}$} $n_\gamma\in \mathbb{N}$. It follows $\Gamma_{\gamma}=<\gamma_0>$ and 
\begin{equation*}
\sum_{1\neq[\gamma]\in C\Gamma}\sum_{\Gamma_\gamma\gamma'\in \Gamma_\gamma\backslash \Gamma}\int_{F_\Gamma}\sigma(\gamma' x)f(x^{-1}\gamma'^{-1}\gamma\gamma'x)dx= \sum_{1\neq[\gamma]\in C\Gamma}\int_{F_{\gamma_0}}\sigma(x)f(x^{-1}\gamma x)dx, 
\end{equation*}
where \index{$F_{\gamma_0}$, fundamental domain for $\Gamma_\gamma$} \[F_{\gamma_0}:=\biguplus_{\Gamma_\gamma\gamma'\in \Gamma_\gamma\backslash \Gamma} \gamma'F_\Gamma\subset X\] is a fundamental domain for $\Gamma_{\gamma}=<\gamma_0>$. Here $\biguplus$ means the disjoint union. Then $<\gamma>\backslash \Gamma_\gamma=<\gamma>\backslash<\gamma_0>\cong \mathbb{Z}/n_\gamma\mathbb Z$, i.e. $|<\gamma>\backslash \Gamma_\gamma|=n_\gamma$  and

\begin{equation*}
\int_{F_{\gamma_0}}\sigma(x)f(x^{-1}\gamma x)dx=  \frac{1}{n_\gamma}\sum_{<\gamma> h\in <\gamma>\backslash \Gamma_\gamma} \int_{F_{\gamma_0}}\sigma(hx)f(x^{-1}h^{-1}\gamma hx)dx
\end{equation*}
as $f(x^{-1}h^{-1}\gamma hx)=f(x^{-1}\gamma x)$ for $h\in \Gamma_\gamma$. Applying now the transformation $x\mapsto hx$ and using the unimodularity of $X$, we get
\begin{eqnarray*}
&&\sum_{1\neq[\gamma]\in C\Gamma}\frac{1}{n_\gamma}\sum_{<\gamma> h\in <\gamma>\backslash \Gamma_\gamma} \int_{F_{\gamma_0}}\sigma(hx)f(x^{-1}h^{-1}\gamma hx)dx\\&=& \sum_{1\neq[\gamma]\in C\Gamma} \frac{1}{n_\gamma}\sum_{<\gamma >h\in <\gamma>\backslash \Gamma_\gamma}\int_{hF_{\gamma_0}}\sigma(x)f(x^{-1}\gamma x)dx
\end{eqnarray*}

Now we change to a fundamental domain \index{$F_\gamma$, fundamental domain for $<\gamma>$} $F_\gamma\subset X$ for $<\gamma>$ to obtain

\begin{equation}\label{equation}
\sum_{1\neq[\gamma]\in C\Gamma} \frac{1}{n_\gamma}\sum_{<\gamma >h\in <\gamma>\backslash \Gamma_\gamma}\int_{hF_{\gamma_0}}\sigma(x)f(x^{-1}\gamma x)dx= \sum_{1\neq[\gamma]\in C\Gamma} \frac{1}{n_\gamma}
\int_{F_\gamma}\sigma(x)f(x^{-1}\gamma x)dx 
\end{equation}

But $\gamma$ is conjugated to $m_\gamma a_\gamma\in MA^+$, i.e. \index{$alphagamma$@$\alpha_\gamma$, $\alpha_\gamma \gamma \alpha_\gamma^{-1}=a_\gamma m_\gamma$} 
\begin{equation}\label{eq: alga}
\alpha_\gamma \gamma \alpha_\gamma^{-1}=a_\gamma m_\gamma
\end{equation} for some $\alpha_\gamma\in G$, and a fundamental domain for $a_\gamma m_\gamma$ can be stated explicitly.

\begin{lemma}(See \cite[Lemma 3.1]{BO2})

Let $ma_t\in MA^+$, where $A=\exp \mathbb{R}H_0$, $a_t=\exp tH_0$, and identify $X$ with $AN$. A fundamental domain for the cyclic subgroup generated by $a_tm$ is given by \index{$F_{a_tm}$, fundamental domain for $<a_tm>$} 
\begin{equation*}
F_{a_tm}:=\{a_sn:n\in N,0\leq s\leq t\}.
\end{equation*}
\end{lemma}
\begin{proof} Note that $a_tm$ acts on $x=a_sn$ by $a_tm\cdot x=a_{t+s}mnm^ {-1}$. Further, \[<a_tm>=\{a_{zt}m^z:z\in\mathbb Z\}.\] We fix an arbitrary $a_hn\in X=AN$ with $h=c\cdot t + s$, where $c\in\mathbb Z$ and $0\leq s\leq t$. 
That is, \[<a_tm>\cdot a_hn=\{a_{(z+c)t+s}m^znm^{-z}:z\in \mathbb{Z}\},\] in particular $a_sm^{-c}nm^c$ is contained in the orbit of $a_hn$ and $F_{a_tm}$. 

Finally, $h\equiv s \mbox{ }\mathrm{ mod }\mbox{ } \mathbb{Z}$, if $a_hn$ is contained in the orbit of $a_sn'$. This shows that $F_{a_tm}$ contains at most one point from any orbit (up to a set of measure zero). 
\end{proof}

\begin{lemma}\label{lco1}
Let $G$ be a group acting on a manifold $M$, $H\subset G$ a subgroup and $F_H$ a fundamental domain for $H$. Then $gF_H$ is a fundamental domain for $gHg^{-1}$ for all $g\in G$. 
\end{lemma}
\begin{proof}
The proof is simple. Fix $g\in G$ and $m\in M$. If $y\in Hm\cap F_H$, then for some $h\in H$ \[gy=ghm=ghg^{-1}gx\in gHg^{-1}(gx).\] 

Since left translation by $g$ is a bijection on $G$, the claim follows.
\end{proof}


Thus, a fundamental domain for the cyclic subgroup generated by $\gamma=\alpha_\gamma^{-1}a_\gamma m_\gamma \alpha_\gamma$, see (\ref{eq: alga}), is given by
\begin{equation*}
F_\gamma:=\{\alpha_\gamma^{-1}x:x\in F_{a_\gamma m_\gamma}\}=:\alpha_\gamma^{-1}F_{a_\gamma m_\gamma}.
\end{equation*}

Next we rewrite the integral from the right side of (\ref{equation}) as an integral over a subset of $AN$. Note that the identification $X=AN$ implies $dx=dadn$, see \cite[Ch.I §5 Cor. 5.3]{GGA}
\begin{eqnarray*}\sum_{1\neq[\gamma]\in C\Gamma} \frac{1}{n_\gamma}
\int_{F_\gamma}\sigma(x)f(x^{-1}\gamma x)dx  &=&\sum_{1\neq[\gamma]\in C\Gamma} \frac{1}{n_\gamma}
\int_{A/<a_\gamma>}\int_N\sigma(\alpha_\gamma^{-1}an)f(n^{-1}a^{-1}a_\gamma m_\gamma an)dnda\\
&=&  \sum_{1\neq[\gamma]\in C\Gamma}  \int_N\int_{A/<a_{\gamma_0}>}\sigma(\alpha_\gamma^{-1}an) f(n^{-1}a^{-1}a_\gamma m_\gamma a n)dadn\\&=&
\sum_{1\neq[\gamma]\in C\Gamma}  \int_N\int_{A/<a_{\gamma_0}>}\sigma(\alpha_\gamma^{-1}an)daf(n^{-1}a_\gamma m_\gamma n)dn
\end{eqnarray*}

In the last but one equation we used again Lemma \ref{lco1} for a fundamental domain of $<a_\gamma m_\gamma>$ and the fact $a_\gamma=a_{\gamma_0}^{n_\gamma}$, see the proof of Proposition \ref{chap: geoprop l}. We call \index{$I_\gamma(\sigma)$, weight function} $I_\gamma(\sigma)$ defined by  \begin{equation}\label{def: weight}
I_\gamma(\sigma)(x):=\int_{A/<a_{\gamma_0}>}\sigma(\alpha_\gamma^{-1}ax)da 
\end{equation} for $x\in X=AN$ the \index{weight function} weight function of the \index{orbital integral} \textit{(weighted) orbital integral} \index{$\mathcal{O}_\gamma(f)$, (weighted) orbital integral} \begin{equation}\label{def: orbitalint}
\mathcal{O}_\gamma(f):=\int_N f(n^{-1}a_\gamma m_\gamma n)I_\gamma(\sigma)(n)dn.
\end{equation}

Thus, 

\begin{eqnarray*}
\sum_{1\neq[\gamma]\in C\Gamma}  \int_N\int_{A/<a_{\gamma_0}>}\sigma(\alpha_\gamma^{-1}an)daf(n^{-1}a_\gamma m_\gamma n)dn&=&  \sum_{1\neq[\gamma]\in C\Gamma}  \int_N f(n^{-1}a_\gamma m_\gamma n)I_\gamma(\sigma)(n)dn\\
&=& \sum_{1\neq[\gamma]\in C\Gamma}  \mathcal{O}_\gamma(f).
\end{eqnarray*}

\section{The weight $I_\gamma(\sigma)$}

In this section we want to examine the weight $I_\gamma(\sigma)$ more closely. The next theorem will us show how we can decompose $I_\gamma(\sigma)$. We make use of the following fact which can be found in Helgason's book \cite{GGA}.

\begin{theorem}(\cite[Ch. V Cor. 3.4]{GGA})\label{th helgason}

Let $X$ be a manifold with countable base and $H$ a compact, connected Lie transformation group, then for any $f\in C^\infty(X)$
\begin{equation*}
f=\sum_{\delta\in\widehat H}d_\delta\overline{\chi}_\delta*f = \sum_{\delta\in\widehat H}d_\delta \chi_\delta*f
\end{equation*} 
with absolute convergence, where $(\phi*f)(x):=\int_H\phi(h)f(h^{-1}x)dh$. Here \index{$d_\delta$, dimension of $\delta$} $d_\delta$ is the dimension of $\delta$ and \index{$chidelta$@$\chi_\delta$, character of $\delta$} $\chi_\delta:=\mathrm{Tr}(\delta)$ its character.  
\end{theorem}

For the second equality we use that $\overline{\chi_\delta}=\chi_{\delta^*}$, if $\delta^*$ is the dual of $\delta$. Absolute convergence of a series $\sum_{a\in A}v_a$, $\{v_a\}_{a\in A}\subset C^\infty(X)$ means that $\sum_{a\in A}\nu(v_a)$ is absolute convergent for every continuous seminorm $\nu$ on $C^\infty(X)$.

We apply this to $M$ acting on $X=AN$ by conjugation to obtain a decomposition of the weight $I_\gamma(\sigma)$. A useful guiding for the theory presented in the remainder of this section is the theory of generalized spherical functions (Eisenstein integrals) for a rank one symmetric space which can be found in \cite[Ch.III \S 2 and \S 11]{GASS}. Roughly speaking we replace $K$ acting on $\mathfrak p$ with section $\mathfrak a$ by $M$ acting on $\mathfrak n$ with section $\mathfrak s$.

For a unitary representation $(\pi,V_\pi)\in \widehat{M}$ of \index{$d_\pi$} dimension $d_\pi$ and $f\in C^\infty(X)$ we define its $\pi$-projection by \index{$pi$@$\pi$-projection} \index{$f^\pi$, $\pi$-projection of $f$} \index{$()^{\pi}$, $\pi$-projection}
\begin{equation}\label{def: projection}
f^\pi(x):=(f)^\pi(x):=d_\pi \int_M f(m\cdot x)\pi(m^{-1})dm
\end{equation}
for $x\in X=AN$, where $m\cdot x:=mx$. We recall the function space \index{$C^\infty\left(X\times_M \mathrm{End}(V_\pi)\right)$} for $(\pi,V_\pi)\in \widehat{M}$ 
\begin{equation*}
C^\infty\left(X\times_M (V_\pi) \right):=\{F\in C^\infty(X,\mathrm{End}(V_\pi)): F(m\cdot x)=\pi(m)F(x) \mbox{ for all } x\in X,m\in M\}.
\end{equation*}

Then for $f\in C^\infty(X)$, $x\in X$ and $m'\in M$
\begin{eqnarray}\label{eq: equivar}f^\pi(m'\cdot x)&=&d_\pi \int_M f(m\cdot (m'\cdot x))\pi(m^{-1})dm\nonumber\\&=& 
d_\pi\int_M f(mm'\cdot x)\pi(m^{-1})dm\nonumber\\ &=& d_\pi \int_M f(m\cdot x)\pi(m'm^{-1})dm=\pi(m')f^\pi(x) 
\end{eqnarray}
by unimodularity of $M$. Hence, $f^\pi\in C^\infty(X\times_M V_\pi)$.

\begin{remark}\label{rem: comdiag}
We recall the definition of the action of $M$ on $G$ via \label{$*$} left translation \[m*g:=mg.\]  

In Section \ref{sec: odeimp} we showed that the following diagram is commutative:

 
\begin{center}
    \makebox[0pt]{
\begin{xy}
  \xymatrix{
      G \ar[r]^{m*} \ar[d]_{\mathrm{pr}}    &   G \ar[d]^{\mathrm{pr}}  \\
      X=G/K=AN \ar[r]_{m\cdot }             &   X=G/K=AN   
  }
\end{xy}
}
\end{center}

We can also extend the definition of the \index{$()^\pi$, $\pi$-projection} \index{$pi$@$\pi$-projection} $\pi$-projection to functions $f\in C^\infty(G)$ for $(\pi,V_\pi)\in\widehat M$, $\mathrm{dim}(V_\pi)=d_\pi$,    
\begin{equation}\label{def: projex}
f^\pi(g):=d_\pi\int_M f(m*g)\pi(m^{-1})dm.
\end{equation}

It follows again that $f^\pi$ is in \index{$C^\infty(G,\times_M V_\pi)$, space of sections} \[ C^\infty(G,\mathrm{End}_M(V_\pi):=\{F\in C^\infty(G,\mathrm{End}(V_\pi)): F(m*g)=\pi(m)F(g) \mbox{ for all } m\in M,g\in G\}.\]
\end{remark}

\begin{proposition}\label{prop: deco} For any $f\in C^\infty(X)$
\begin{equation*}
f=\sum_{\pi\in \widehat M}d_\pi \left(\chi_\pi* f\right)(x)=\sum_{\pi\in\widehat{M}}\mathrm{Tr}(f^\pi).
\end{equation*}
\end{proposition}
\begin{proof}
By the preceding Theorem \ref{th helgason} we have for $x\in X$
\begin{eqnarray*}
f(x)&=&\sum_{\pi\in \widehat M}d_\pi \left(\chi_\pi* f\right)(x)\\
&=& \sum_{\pi\in\widehat M} d_\pi \left(\mathrm{Tr}(\pi)* f\right)(x)\\
&=& \sum_{\pi\in\widehat M} d_\pi\int_M \chi_\pi(m) f(m^{-1}x)dm\\
&=& \sum_{\pi\in\widehat M} d_\pi\int_M \mathrm{Tr}(\pi(m^{-1})) f(m\cdot x)dm\\
&=& \sum_{\pi\in\widehat M} \mathrm{Tr}\left(d_\pi\int_M \pi(m^{-1}) f(m\cdot x)dm\right)=\sum_{\pi\in\widehat{M}}\mathrm{Tr}(f^\pi)
\end{eqnarray*}
by unimodularity of $M$ and the linearity of the trace.
\end{proof}

Let us summarize the preliminary result of this chapter.
\begin{theorem}\label{theo: generaltr}
Let $G=SO_o(1,l)=ANK$, $M=Z_K(A)$ and $\Gamma\subset G$ a uniform lattice. Denote by $\pi_R$ the right regular representation of $G$ on $L^2(\Gamma\backslash G)$ and fix $\sigma\in C^\infty(X_\Gamma)$ and $f\in C^\infty(G//K)$ admissible. Then the operator $\sigma\cdot \pi_R(f)$, given by (\ref{eq: convolution}), maps $L^2(X_\Gamma)$ into itself and is of trace class. Its trace can be computed to
\begin{eqnarray*}\mathrm{Tr}(\sigma\cdot \pi_R(f))&=&f(e)\int_{F_\Gamma}\sigma(x)dx+\sum_{1\neq [\gamma]\in C\Gamma}\int_Nf(n^{-1}a_\gamma m_\gamma n)\sum_{\pi\in\widehat M}\mathrm{Tr}(I_\gamma(\sigma)^\pi(n))dn\\
&=& f(e)\int_{F_\Gamma}\sigma(x)dx+\sum_{1\neq [\gamma]\in C\Gamma}\int_Nf(n^{-1}a_\gamma m_\gamma n)\sum_{\pi\in\widehat M}d_\pi \left(\chi_\pi*I_\gamma(\sigma)\right)(n)dn,
\end{eqnarray*}
where $1\neq \gamma$ is conjugate to $a_\gamma m_\gamma$, $F_\Gamma\subset X$ a fundamental domain for $\Gamma$  and $I_\gamma(\sigma)^\pi$ is defined by (\ref{def: projection}).
\end{theorem}

\begin{remark}\label{chap: gtrace rem}
Let us examine $I_\gamma(\sigma)^\pi$ more closely. We slightly generalize Lemma \ref{chap: ode lem: inv} in Chapter \ref{chap: ode}.

\begin{lemma}\label{chap: gtrace scalar}
If $a\exp sX_1\in A\exp \mathbb{R}X$ and $F\in C(X,\times_MV_\pi)$, then $F(a\exp sX_1)$ maps $V_\pi$ into $V_\pi^{Z_M(S)}$.  
\end{lemma}

\begin{proof}
The same proof as before applies. We use $F(m\cdot x)=\pi(m)F(x)$ for any $m\in M$, $x\in X$. But then for $v\in V_\pi$ \begin{eqnarray*}\pi(m)F(a\exp sX_1)v&=&F(m\cdot a\exp sX_1)v=F(am\cdot\exp sX_1)v\\&=&F(a\exp sX_1)v\end{eqnarray*} for all $a\exp sX_1\in A\exp \mathbb{R}X_1$, $m\in Z_M(S)$. 
\end{proof}

We recall that $r_g$, $l_g$ denote right- resp. left-translation on $G$. Let $\gamma\neq 1$ be in $\Gamma$ with primitive $\gamma_0$. We assume that $\gamma_0$ is conjugated to $a_{\gamma_0}m_{\gamma_0}\in A^+M$, where $a_{\gamma_0}=\exp L_{\gamma_0}H_0$. This implies $L_{\gamma_0}=\left(\sqrt{2(l-1)}\right)^{-1}l_{\gamma_0}$, where $l_{\gamma_0}$ is defined by $a_{\gamma_0}=\exp l_{\gamma_0}H_1$. By definition (\ref{def: projection})

\begin{eqnarray}\label{eq: ini}
\left(\chi_\pi*I_\gamma(\sigma)\right)(\exp sX_1)&=& \int_M I_\gamma(\sigma)(m\cdot \exp sX_1)\chi_\pi(m^{-1})dm \nonumber \\
&=&   \int_M\int_{A/<a_{\gamma_0}>}\sigma (\alpha_\gamma^{-1}ma\exp sX_1)da\mbox{ }\chi_\pi(m^{-1})dm
\nonumber\\
&=& \int_M\int_0^{L_{\gamma_0}}\sigma(\alpha_\gamma^{-1}\exp tH_0m\exp sX_1)da\chi_\pi(m^{-1})dm \nonumber\\
&=& \sqrt{2(l-1)}\int_M\int_0^{l_{\gamma_0}}\sigma(\alpha_\gamma^{-1}\exp tH_1m\exp sX_1)da\chi_\pi(m^{-1})dm\nonumber\\
&=& \sqrt{2(l-1)}\int_M \int_0^{l_{\gamma_0}}\sigma\circ l_{\alpha_\gamma^{-1}}\circ r_{m\cdot \exp sX_1}da\chi_\pi(m^{-1})dm\nonumber\\
&=& \sqrt{2(l-1)}\int_M\left(\int_{c_{\gamma_0}}\sigma\circ r_{m\cdot \exp sX_1}\right)\chi_\pi(m^{-1})dm\nonumber\\
&=&\sqrt{2(l-1)} \left(\chi_\pi*\left(x\mapsto \int_{c_{\gamma_0}}\sigma\circ r_{x}\right)\right)(\exp sX_1)
\end{eqnarray}
where we used the definition of the $\pi$-projection from (\ref{def: projection}) and for functions $f$ on $\Gamma\backslash X$ the integral over the (prime) closed geodesic belonging to primitive $\gamma_0$ \[c_{\gamma_0}=\{\Gamma\alpha_{\gamma^{-1}}\exp (-tH_1) M:0\leq t\leq l_{\gamma_0}\},\] see (\ref{eq: primclos}), is given by
\[\int_{c_{\gamma_0}}f:=\int_{0}^{l_{\gamma_0}}f\circ l_{\alpha_\gamma^{-1}}(\exp tH_1)dt.\]

Finally, we remark that  
\begin{eqnarray}\label{eq: inivalsimp}
\sqrt{2(l-1)}\int_{c_{\gamma_0}}\sigma&=&\int_{A/<a_{\gamma_0}>}\sigma(\alpha_\gamma^{-}x)dx\nonumber
\\&=&I_\gamma(\sigma)(e)\\&=&\int_M I_\gamma(\sigma)(m)dm,\nonumber \end{eqnarray} since for any nontrivial $\pi$
\[\left(\chi_\pi*I_\gamma(\sigma)\right)(e)=0,\]
where $e\in G$ is the neutral element, as $I_\gamma(\sigma)$ is right-$K$-invariant and $\int_M\chi_\pi(m)dm=0$.

\end{remark}

\section{The case of $\sigma=\varphi$ an automorphic eigenfunction}\label{sec: specialization}

Now we specify $\sigma$ to be an \index{automorphic eigenfunction} automorphic eigenfunction of the Laplacian $\Delta_\Gamma$ on the compact manifold $X_\Gamma$. So let $\varphi$ be some fixed Laplace eigenfunction, i.e. \[\Delta_\Gamma \varphi=\Omega \varphi=\mu \varphi\] for some eigenvalue $\mu$. More specifically, there are \index{$r$, $\mu=-\frac{1}{4}\left(\rho_0+\frac{r^2}{\rho_0}\right)$} $r\in\mathfrak{a}_{\mathbb C}^*\cong \mathbb{C}$ such that \[\mu=-\frac{1}{4}\left(\rho_0+\frac{r^2}{\rho_0}\right)=-\frac{1}{4\rho_0}\left(\rho_0^2+r^2\right) .\] 

By the ellipticity of $\Delta_\Gamma$ we know that $\varphi$ is smooth and we can apply the theory of the last section to $\sigma=\varphi$. Then we consider $I_\gamma(\varphi)$ defined in (\ref{def: weight}), so that 
\begin{equation*}I_\gamma(\varphi)(x)=\int_{A/<a_{\gamma_0}>}\varphi(\alpha_\gamma^{-1}ax)da .\end{equation*}

Since $\Omega$ lies in the center $Z(\mathfrak g)$ of $U(\mathfrak g)$ it commutes with the right-translation \index{$r_g$, right-translation with $g$} $r_g$ and the left-translation \index{$l_g$, left-translation with $g$} $l_g$ \[\Omega(\varphi\circ l_{g})=(\Omega \varphi)\circ l_{g}=\mu\varphi\circ l_{g} \]  for all $g$ in $G$, in particular for $g=\alpha_\gamma^{-1}a$, $a\in A$.

If we apply now $\Omega$ to $I_\gamma(\varphi)$ we can use, as $A/<a_{\gamma_0}>$ is compact, the theorem of dominated convergence and exchange differentiation with integration to see that $I_\gamma(\varphi)$ is also a $\Omega$-eigenfunction, i.e. \[\Omega I_\gamma(\varphi)=\mu I_\gamma(\varphi).\] 

Since also the action of $M$ and $\Omega$ commute, even for any $(\pi,V_\pi)\in\widehat M$, by the same reasoning
\begin{equation}\label{equal}
\Omega I_\gamma(\varphi)^\pi=\mu I_\gamma(\varphi)^\pi,
\end{equation} 
where, according to (\ref{def: projection}) 
\begin{equation}\label{def: piphi}
I_\gamma(\varphi)^\pi(x)=d_\pi\int_M I_\gamma(\varphi)(m\cdot x)\pi(m^{-1})dm.
\end{equation}

Equations of the form (\ref{equal}) have been discussed in Chapter \ref{chap: ode}. We fix a slice $S\subset N$ for $M$ acting on $N$ and assume that $S=\exp \mathbb{R}^+X_1$, where $X_1\in\mathfrak n $ is of unit length. We recall that if $F$ is an element of $C^\infty(X\times_M V_\pi)$, then $F|_S$ maps $V_\pi$ onto $V_\pi^{Z_M(S)}$, see Lemma \ref{chap: ode lem: inv}, where $V^{Z_M(S)}_\pi=\mathbb{C}\cdot v$ for some $v\in V_\pi$. 

In Section \ref{sec: odeimp} we have shown that the equation \[\Omega F-\mu F=0\] leads to an equation on $S$ 

\begin{eqnarray}\label{chap: gtrace ef}
\left(\left(\alpha(H_1)^2t^2+2\right)\frac{d^2}{dt^2}+\left(\left(\alpha(H_1)^2+2\alpha(H_\rho)\right)t+2\frac{l-2}{t}\right)\frac{d}{dt}\right.\\
\left.+
\frac{2}{t^2}\sum_{j}\pi(Z_{Y_j}^2)|_{V_\pi^{Z_M(S)}}-\mu\right)F(t)=0,\nonumber
\end{eqnarray}
for the scalar valued function $F(t)$ defined by $F(t)v:=F(\exp tX_1)v$, see Theorem \ref{th: odesum}. Here $\alpha$ is the unique positive root, $H_1,H_\rho\in\mathfrak a^+$ the unique unit vector resp. the defining vector for $\rho$, $\{Z_{Y_j}\}_j$ the orthogonal set of $\mathfrak z_{\mathfrak{m} }(\mathfrak s)^{\bot_\mathfrak{m} }$ which corresponds to the orthonormal basis $\{Y_j\}_j$ of $\{X_1\}^{\bot_\mathfrak n}$ via $[Z_{Y_j},X_1]=Y_j$. By Lemma \ref{chap: ode d} we know that $\sum_{j}\pi(Z_{Y_j}^2)|_{V_\pi^{Z_M(S)}}\in\mathrm{End}(V_\pi^{Z_M(S)})$ can be identified with a scalar $\leq 0$. We note that $F(t)=d_\pi
\left(\chi_\pi*f\right)(\exp tX_1)$, if $F=f^\pi$ for some $f\in C^\infty(X)$, see Theorem \ref{th: odesum}.
  
The results of Theorem \ref{chap: ode main} now imply that there are constants $a,b,d,p$, where $a=\rho_0+ir$ and $b=\rho_0-ir$ only depend on the eigenvalue, $2p\in \mathbb{N}_0$ and $d\leq 0$ depend only on $\pi$ and $S$, such that for some other constant $C$ and for $l\geq 3$ any solution $F$ to (\ref{chap: gtrace ef}), which is smooth at $t=0$, can be expressed as
\begin{equation}\label{eq: general solution}\tilde F(t):=F(2\sqrt{l-1}t)=(-1)^pC\cdot t^{2p}\cdot{}_2F_1\left(a+l,b+l,1+2\sqrt{ \frac{(1-\rho_0)^2}{4}-d };-t^2\right)\end{equation}
for $t>0$, see Remark \ref{rem: hypsymxp}. By the Leibniz rule \begin{eqnarray*}
 \tilde{F}^{(2p)}(t)&=&(-1)^pC\left( t^{2p}\cdot{}_2F_1\left(a+p,b+p,1+2\sqrt{ \frac{(1-\rho_0)^2}{4}-d };-t^2\right)\right)^{(2p)}
\\&=& (-1)^pC\sum_{i=0}^{2p}\dbinom{2p}{i}\frac{1}{i!}t^{2p-i}{}_2F_1\left(a+p,b+p,1+2\sqrt{ \frac{(1-\rho_0)^2}{4}-d };-t^2\right)^{(2p-i)}\\&=&(-1)^pC\cdot {}_2F_1\left(a+p,b+p,1+2\sqrt{ \frac{(1-\rho_0)^2}{4}-d };-t^2\right)\\&&+(-1)^pC\sum_{i=0}^{2p-1}\dbinom{2p}{i}\frac{1}{i!}t^{2p-i}{}_2F_1\left(a+p,b+p,1+2\sqrt{ \frac{(1-\rho_0)^2}{4}-d };-t^2\right)^{(2p-i)}\\&=& (-1)^pC\cdot {}_2F_1\left(a+p,b+p,1+2\sqrt{ \frac{(1-\rho_0)^2}{4}-d };-t^2\right)\\&&+(-1)^pC\cdot t\sum_{i=0}^{2p-1}\dbinom{2p}{i}\frac{1}{i!}t^{2p-i-1}{}_2F_1\left(a+p,b+p,1+2\sqrt{ \frac{(1-\rho_0)^2}{4}-d };-t^2\right)^{(2p-i)}
\end{eqnarray*}

Hence, \[\tilde{F}^{(2p)}(0)=\left(4(l-1)\right)^pF(0)^{(2p)}=(-1)^pC\cdot {}_2F_1\left(a+p,b+p,1+2\sqrt{ \frac{(1-\rho_0)^2}{4}-d };0\right)=(-1)^pC\] because ${}_2F_1\left(a+p,b+p,1+2\sqrt{ \frac{(1-\rho_0)^2}{4}-d };0\right)=1$ by the definition of hypergeometric functions.

For $l=2$ the space of solutions to (\ref{chap: gtrace ef}) is spanned by ${}_2F_1\left(a,b,\frac{1}{2};\frac{-t^2}{4}\right)$  and $it\cdot {}_2F_1\left(a+\frac{1}{2},b+\frac{1}{2}, \frac{3}{2};-\frac{t^2}{4} \right)$, where $a=\rho_0+ir$ and $b=\rho_0-ir$ as before. The general solution in this case has thus the form
\begin{equation}\label{eq: n2}F(2t):=F(\exp 2tX_{1})=F(0)\cdot {}_2F_1\left(a,b,\frac{1}{2};-t^2\right)+ 2F'(0)it\cdot {}_2F_1\left(a+\frac{1}{2},b+\frac{1}{2}, \frac{3}{2};-t^2 \right).\end{equation}

Now we consider $F:=I_\gamma(\varphi)^\pi:\mathbb R^+\to \mathbb{C}$ defined by $I_\gamma(\varphi)^\pi(s)v=I_\gamma(\varphi)^\pi(\exp sX_1)v$, then also \[I_\gamma(\varphi)^\pi(s)=d_\pi \left(\chi_\pi*I_\gamma(\varphi)\right)(\exp sX_1)
,\] 
see Theorem \ref{th: odesum}. Under the map $\exp tX_{1}\mapsto t$, $X_{1}$ is identified with $ \frac{d}{dt} $, thus 
\begin{equation}\label{eq: idiff}
\left(I_\gamma(\varphi)^\pi\right)^{(2p)}(0)= d_\pi\frac{d^{2p}}{dt^{2p}}|_{t=0}\left(\chi_\pi*I_\gamma(\varphi)\right)(\exp tX_{1})= d_\pi \left(X_{1}^{2p}\chi_\pi*I_\gamma(\varphi)\right)(e),
\end{equation}
where \index{$e$, neutral element} $e\in N$ is the neutral element. For $I_\gamma(\varphi)^\pi:\mathbb{R}^+\to\mathbb{C}$, we get from (\ref{eq: general solution}) for $l\geq 3$

\begin{eqnarray}\label{eq: lghyp}
d_\pi\left(\chi_\pi*I_\gamma(\varphi)\right)(\exp 2\sqrt{l-1}sX_1)&=&I_\gamma(\varphi)^\pi( 2\sqrt{l-1}s)\nonumber\\
&\overset{(\ref{eq: general solution})}=& (-1)^p\widetilde{I_\gamma(\varphi)}^{(2p)}(0)\cdot\nonumber
\\ &&\cdot
s^{2p}\cdot{}_2F_1\left(a+p,b+p,1+2\sqrt{ \frac{(1-\rho_0)^2}{4}-d };-s^2\right)
\nonumber\\
&=& (-1)^p\left(4(l-1)\right)^p\left(I_\gamma(\varphi)^\pi\right)^{(2p)}(0)\cdot\nonumber
\\ &&\cdot
s^{2p}\cdot{}_2F_1\left(a+p,b+p,1+2\sqrt{ \frac{(1-\rho_0)^2}{4}-d };-s^2\right)
\nonumber\\&\overset{(\ref{eq: idiff})}=&(-1)^p\left(4(l-1)\right)^{p}d_\pi\frac{d^{2p}}{dt^{2p}}|_{t=0}\left(\chi_\pi*I_\gamma(\varphi)\right)(\exp tX_{1})s^{2p}\nonumber\\&&\cdot{}_2F_1\left(a+p,b+p,1+2\sqrt{ \frac{(1-\rho_0)^2}{4}-d };-s^2\right)\nonumber
\\&=&(-1)^p\left(4(l-1)\right)^{p}d_\pi \left(X_{1}^{2p}\chi_\pi*I_\gamma(\varphi)\right)(e)\cdot s^{2p}\\&&\cdot{}_2F_1\left(a+p,b+p,1+2\sqrt{ \frac{(1-\rho_0)^2}{4}-d };-s^2\right) \nonumber
.
\end{eqnarray}

For $l=2$ we see from (\ref{eq: n2}) that the weight of the orbital integral is given by 

\begin{eqnarray}\label{eq: l2hyp}
I_\gamma(\varphi)(2s)&=& I_\gamma(\varphi)(\exp 2sX_1)\nonumber\\&=& \left(\int_{c_{\gamma_0}}\varphi\right) \cdot {}_2F_1\left(a,b,\frac{1}{2} ;-s^2\right)\\&&+ 2i\left(\int_{c_{\gamma_0}}X_{1}\varphi \right) \cdot s \cdot {}_2F_1\left(a+\frac{1}{2},b+\frac{1}{2}, \frac{3}{2};-s^2 \right).\nonumber
\end{eqnarray}

From Lemma \ref{lem: finfacts}, Theorem \ref{chap: ode main} and Proposition \ref{prop: deco} we get:

\begin{theorem}\label{th: gstf}
Let $X=G/K$ be a real hyperbolic space of dimension $l$, $\Gamma\subset G$ a uniform lattice, $1\neq\gamma\in \Gamma$ and $\varphi$ be an automorphic eigenfunction (on $X_\Gamma$) with eigenvalue $\mu$. The weight $I_\gamma(\varphi)$ can be decomposed as
\[I_\gamma(\varphi)=\sum_{\pi\in\widehat M}d_\pi\left(\chi_\pi*I_\gamma(\varphi)\right).\]

Here $d_\pi\left(\chi_\pi*I_\gamma(\varphi)\right)=\mathrm{Tr}\left(I_\gamma(\varphi)^\pi\right)\in C^\infty(X)_{\check{\pi}}$ is a Casimir eigenfunction. On each slice $S=\exp \mathbb{R}^+X_1$, $|X_1|=1$, for $M$ acting on $N$, $d_\pi\left(\chi_\pi*I_\gamma(\varphi)\right)$ satisfies equation (\ref{chap: gtrace ef}) and is given by (\ref{eq: lghyp}) for $l\geq 3$ resp. (\ref{eq: l2hyp}) for $l=2$. 

\end{theorem}

\begin{remark}
We digress to examine $X_{1}^{2p}I_\gamma^\pi(\varphi)$ more closely. Using the definition of $I_\gamma(\varphi)^\pi$, see (\ref{def: piphi}), and the computations made at the end of the last section, see Remark \ref{chap: gtrace rem}, as $\exp tX_{1}\in S$ for all $t>0$,
\begin{eqnarray*}\label{chap: gtrace initial1}
X_{1}^{2p}\left(\chi_\pi*I_\gamma(\varphi)\right)(e)&=& \frac{d^{2p}}{dt^{2p}}|_{t=0}\int_M I_\gamma(\varphi)(m\cdot \exp tX_{1})\chi_\pi(m^{-1})dm\\&=&
\frac{d^{2p}}{dt^{2p}}|_{t=0}\int_M \int_{A/<a_{\gamma_0}>}\varphi(\alpha_\gamma^{-1}am\exp tX_{1})da \chi_\pi(m^{-1})dm\\
&\overset{(\ref{eq: ini})}=& \sqrt{2(l-1)} \frac{d^{2p}}{dt^{2p}}|_{t=0} \left(\chi_\pi*\left(x\mapsto \int_{c_{\gamma_0}}\varphi\circ r_x\right)\right)(\exp tX_{1}).
\end{eqnarray*}

On the other hand,  
\[\left(\chi_\pi*\left(\varphi\circ l_{\alpha_\gamma^{-1}}\right)\right)(x)=d_\pi\int_M \varphi\circ l_{\alpha_\gamma^{-1}}(m\cdot x)\chi_\pi(m^{-1})dm\]

and
\begin{eqnarray}\label{chap: gtrace initial2}
X_{1}^{2p}\left(\chi_\pi*I_\gamma(\varphi)\right)(e)&=& \frac{d^{2p}}{dt^{2p}}|_{t=0}\int_M I_\gamma(\varphi)(m\cdot \exp tX_{1})\chi_\pi(m^{-1})dm\nonumber\\
&=&
\frac{d^{2p}}{dt^{2p}}|_{t=0}\int_M \int_{A/<a_{\gamma_0}>}\varphi(\alpha_\gamma^{-1}am\exp tX_{1})da \chi_\pi(m^{-1})dm\nonumber\\
&=&\frac{d^{2p}}{dt^{2p}}|_{t=0}\int_M \int_{A/<a_{\gamma_0}>}\varphi\circ l_{\alpha_\gamma^{-1}}(am\exp tX_{1})da \chi_\pi(m^{-1})dm\nonumber\\
&=&\frac{d^{2p}}{dt^{2p}}|_{t=0} \int_{A/<a_{\gamma_0}>}\int_M\varphi\circ l_{\alpha_\gamma^{-1}}(m\cdot a\exp tX_{1}) \chi_\pi(m^{-1})dmda\nonumber\\
&=&\frac{d^{2p}}{dt^{2p}}|_{t=0} \int_{A/<a_{\gamma_0}>}\left(\chi_\pi*\left(\varphi\circ l_{\alpha_\gamma^{-1}}\right)\right)( a\exp tX_{1})da\nonumber\\
&=&\int_{A/<a_{\gamma_0}>}\frac{d^{2p}}{dt^{2p}}|_{t=0}\left(\chi_\pi*\left(\varphi\circ l_{\alpha_\gamma^{-1}}\right)\right)(a\exp tX_{1})da\nonumber\\
&=&\int_{A/<a_{\gamma_0}>}\left(X_{1}^{2p}(\chi_\pi*\left(\varphi\circ l_{\alpha_\gamma^{-1}})\right)\right)(a)da
\end{eqnarray}

For $\pi=\mathbf{1}$ trivial, (\ref{chap: gtrace initial2}) simplifies to 
\begin{eqnarray*}&&\int_{A/<a_{\gamma_0}>}\left(X_{e_1}^{2p}\left(x\mapsto \int_M \varphi\circ l_{\alpha_\gamma^{-1}}(m\cdot x)dm\right)\right)(a)da\\
&=&\int_{A/<a_{\gamma_0}>}\left(X_{e_1}^{2p}\left(\left[x\mapsto \int_M\varphi(m\cdot x)dm\right]\circ l_{\alpha_\gamma^{-1}}\right)\right)(a)da\\
&=&\int_{A/<a_{\gamma_0}>}\left(\left(X_{e_1}^{2p}\left[x\mapsto \int_M\varphi(m\cdot x)dm\right]\right)\circ l_{\alpha_\gamma^{-1}}\right)(a)da\\
&=&\sqrt{2(l-1)}\int_{c_{\gamma_0}}\left(X_{e_1}^{2p}\left[x\mapsto \int_M\varphi(m\cdot x)dm\right]\right).\end{eqnarray*}
\end{remark}

\section{An auxiliary zeta function $ \mathcal{R}(\varphi)$}\label{sec: zeta}

We choose now $f$ from  a special family of bi-$K$-invariant functions $f_k$. We consider on $G$ the function \index{$f_k$, bi-$K$-invariant function} \begin{equation}\label{def: fk}f_k:g\mapsto \left(1-|g\cdot 0|^2\right)^{k/2}\end{equation} for $k\in \mathbb{C}$, where we recall the action of $G$ on $B_1(\mathbb R^{n})$ given by (\ref{def: gaction}) in Chapter \ref{chap: hyper}. While this function is bi-$K$-invariant on $G$, see Lemma \ref{lem: bikinv}, in particular it is a function on $X=G/K$, it does not have compact support and we have to show that $\pi_R(f_k)$ is indeed of trace class for suitable $k\in\mathbb C $. 

Recall that we fixed $H_0,H_1\in\mathfrak a^+$ where $\alpha (H_0)=1$ and $|H_1|=1$. It follows that $|H_0|^2=2(l-1)$, see \cite[(4.2.10)]{GV} and thus $H_0=\sqrt{2(l-1)}H_1$. Let now $\gamma \in \Gamma$ be conjugated to $a_\gamma m_\gamma\in MA^+$. By definition of $l_\gamma$, $a_\gamma=\exp l_\gamma H_1=\exp \frac{l_\gamma}{\sqrt{2(l-1)}}H_0$. Let us set \index{$L_\gamma$, lenght of $[\gamma]$} \begin{equation}\label{eq: lengthgamma}L_\gamma:=\frac{l_\gamma}{\sqrt{2(l-1)}}.\end{equation} 

Then by definition $a_\gamma=\exp L_\gamma H_0$ and we call \index{length of a closed geodesic} $L_\gamma$ the length of the closed geodesic $[\gamma]$. Further, we call the set \index{length spectrum} \[\{L_\gamma:1\neq [\gamma]\in C\Gamma\}\] the length spectrum of $\Gamma$. It follows from Proposition \ref{prop: inflspec} that $\{L_\gamma:1\neq [\gamma]\in C\Gamma\}$ has an infimum \index{$L_{\mathrm{inf}}$, infimum of length spectrum} \index{infimum of the length spectrum} $L_{\mathrm{inf}}>0$.

\begin{remark}
Let $[\gamma]\in C\Gamma$. The definition of the length of $[\gamma]$ from (\ref{eq: lengthgamma}) differs from $l_\gamma$ in \cite{Ga1}, see also Section \ref{chap: geometry}, by the factor $\left(\sqrt{2(l-1)}\right)^{-1}$. That is $l_\gamma=\sqrt{2(l-1)}L_\gamma$. If $G=SO_o(1,2)$, then our definition of $L_\gamma$ agrees with the one which is used in \cite{AZ}. In this case $\sqrt{2}L_\gamma=l_\gamma$.
\end{remark}

The computations made at the end of Chapter \ref{chap: hyper}, especially equation (\ref{comp: natmn}ff), give

\begin{equation}\label{eq: fk1}
f_k\left((\exp -sX_{e_1}) a_\gamma m_\gamma^{m^{-1}}(\exp sX_{e_1})\right)=\left(-(m_\gamma^{m^{-1}})_{1,1}s^2+(1+s^2) \cosh L_\gamma\right)^{-k},
\end{equation} 
where \index{$a_\gamma$, element of $A^+$} $a_\gamma=\exp l_\gamma H_1=\exp L_\gamma H_0=a_{L_\gamma}$ and \index{$(m_\gamma^{m^{-1}})_{1,1}$, first entry in first row of $m_\gamma^{m^{-1}}$} 
\begin{equation}\label{def: cm}
(m_\gamma^{m^{-1}})_{1,1}:=(m^{-1}m_\gamma m)_{1,1}
\end{equation}
is the first entry in the first row of $m_\gamma^{m^{-1}}$ as we identified $m^{-1}m_\gamma m$ with the $(l-1)\times (l-1)$-matrix $(m^{-1}m_\gamma m)_{i,j=1}^{l-1}$ in $SO(l-1)$. In other words, $(m_\gamma^{m^{-1}})_{1,1}$ is just the matrix coefficient belonging to the defining representation of $M$ and the first basis vector evaluated at $m_\gamma^{m^{-1}}$. Also, see  equation (\ref{eq: fkatsp}),
\begin{equation}\label{eq: fkatn}
f_k(\exp tH_0\exp X_u)=f_k(a_t\exp X_u)=\left(\cosh t +\frac{|u|^2}{2}e^t\right)^{-k}.
\end{equation}
for $X_u$ in $\mathfrak{n}$, i.e. $u\in\mathbb{R}^{l-1}$.

In \cite{Ga} on page 8 we find the following useful remark for determining whether a function $f\in C^\infty(G//K)$ is admissible.
\begin{remark}\label{rem: gang}
For a function $f\in C^\infty(G//K)$ we consider its Abel transform. Here, we use the formula for the Abel transform given in  \cite[(9.37)]{Wi}. We define \index{$F_f$, Abel transform of $f$}
\[F_{f}(a_t):=F_f(t):= e^{\rho_0 t}\int_N f(a_tn)dn,\] 
whenever this integral is finite.

The transform $F_f$, if defined, is a smooth function on $A$ resp. $\mathbb R$. As a function on $\mathbb R$ it is even by the bi-$K$-invariance which translates to a Weyl group invariance on $A$. If $F_f$ also satisfies for some $\varepsilon>0$
\begin{equation}\label{crit}
\sup_{t\in \mathbb R}(\exp(\rho_0+\varepsilon)|t|)|F_f(t)|<\infty,
\end{equation}
then  $\pi_R(f)$ is admissible. 
\end{remark}

\begin{proposition}\label{th: traceclass}
Let $G=SO_o(1,l)$, $\pi_R$ the right-regular representation of $G$ on $\Gamma\backslash G$, where $\Gamma$ is a uniform lattice, and let $f_k$ as defined in (\ref{def: fk}). The function $f_k$ is admissible for $\mathrm{Re}(k)>2\rho_0$. In particular, the operator $\pi_R(f_k)$ is of trace class for $\mathrm{Re}(k)>2\rho_0$.
\end{proposition}
\begin{proof}
We make use of the criterion from \cite{Ga} recalled in Remark \ref{rem: gang}. Thus we have to compute the Abel transform $F_{f_k}$ and check (\ref{crit}).
\begin{eqnarray*}
F_{f_k}(t)&=& e^{\rho_0 t}\int_N f_k(a_tn)dn  \\&\overset{(\ref{eq: fkatn})}=&
e^{\rho_0 t}\int_{\mathbb R^{l-1}} \left(\cosh t +\frac{|u|^2}{2}e^t\right)^{-k}du\\&=&
e^{\rho_0 t}\cosh^{-k}t\int_{\mathbb R^{l-1}} \left(1+\frac{|u|^2e^t}{2\cosh t}\right)^{-k}du\\&=&
2^{\rho_0}\cosh^{-k+\rho_0}t\int_{\mathbb R^{l-1}}\left(1+|u|^2 \right)^{-k}du,
\end{eqnarray*}
where we applied for the last equation the coordinate transform $u\mapsto e^{-t/2}\sqrt{2}\cosh^{1/2} t\cdot u$. Thus $e^{\rho_0|t|}|F_{f_k}(t)|$ behaves as $e^{(2\rho_0-\mathrm{Re}(k))|t|}$ for large $|t|$.
\end{proof}


Let us then assume that the automorphic eigenfunction $\varphi$ on $X_\Gamma$ with eigenvalue $\mu=-\frac{1}{4}(\rho_0+\frac{r^2}{\rho_0})$, $r\in \mathfrak{a}^*_\mathbb{C}\cong \mathbb{C}$, is not trivial, so $\int_{\Gamma\backslash G}\varphi(x)dx=0$. Then the first summand in the formula for $\varphi\cdot\pi_R(f_k)$ in Theorem \ref{theo: generaltr} vanishes and the trace computes for $l\geq 3$ by (\ref{eq: fk1}) to 

\begin{eqnarray*}
\mathrm{Tr}(\varphi\cdot\pi_R(f_k))
&=& \sum_{1\neq [\gamma]\in C\Gamma}\int_Nf_k(n^{-1}a_\gamma m_\gamma n)\sum_{\pi\in\widehat M}d_\pi\left(\chi_\pi*I_\gamma(\varphi)\right)(n)dn
\\&=&\omega_{l-1}\sum_{1\neq[\gamma]\in\Gamma} \sum_{\pi\in\widehat M}d_\pi\int_M\int_{0}^\infty s^{l-2} f_k\left(\exp -sX_{e_1}m^{-1}m_\gamma m \exp sX_1\right) \\&&\cdot\left(\chi_\pi*I_\gamma(\varphi)\right)(m\cdot \exp sX_{e_1})dsdm,
\end{eqnarray*}

where we used \index{polar coordinates} polar coordinates to obtain the second equality. We recall that if we identify $N\cong \mathbb{R}^{l-1}$, $u\mapsto X_u$, see (\ref{eq: nidr}), then
\begin{eqnarray*}
\int_Nf(n)dn&=&\int_{\mathbb{R}^{l-1}}f(\exp X_u)du\\
&=& \int_{0}^\infty \int_{\partial B_s(0)}fdSds\\
&=& \omega_{l-1}\int_0^\infty s^{l-2}  f(\exp sX_{e_1})ds
\end{eqnarray*}
for radial functions $f$, i.e. $f(n)=\int_Mf(m\cdot n)dm$ for all $n\in N$. Here $B_1(0)=\{x\in\mathbb R^{l-1}:|x|\leq 1\}$ and $\omega_{l-1}:=|\partial B_1(0)|$ with respect to the Lebesgue measure in $\mathbb R^{l-2}$ for $l\geq 3$. For $l=2$ we set $\omega_1:=2$.

Now we plug in equation (\ref{eq: fk1}) to get 

\begin{eqnarray*}
&&\mathrm{Tr}\left(\varphi\cdot\pi_R(f_k)\right)\\
&=&\omega_{l-1}\sum_{1\neq[\gamma]\in\Gamma} \sum_{\pi\in\widehat M}d_\pi\int_M\int_{0}^\infty s^{l-2} f_k\left(\exp -sX_{e_1}m^{-1}m_\gamma m \exp sX_1\right) 
\\&&\cdot \left(\chi_\pi*I_\gamma(\varphi)\right)(m\cdot \exp sX_{e_1})dsdm
\\&\overset{(\ref{eq: fk1})}=&  \omega_{l-1}\sum_{1\ne[\gamma]\in\Gamma} \sum_{\pi\in\widehat M} \int_M\int_{0}^\infty s^{l-2}\left(-(m_\gamma^ {m^ {-1}})_{1,1}s^2+(1+s^2)\cosh L_\gamma\right)^{-k}\\&&\cdot \left(\chi_\pi*I_\gamma(\varphi)\right)(m\cdot \exp sX_{e_1})dsdm
\\&=& \omega_{l-1}\sum_{1\neq[\gamma]\in\Gamma} \sum_{\pi\in\widehat M} d_\pi \cosh^{-k} L_\gamma\int_M\int_{0}^\infty s^{l-2}  \left(s^2\left(\frac{\cosh L_\gamma-(m_\gamma^{m^{-1}})_{1,1}}{\cosh L_\gamma}\right)+1\right)^{-k} \\&&\cdot\left(\chi_\pi*I_\gamma(\varphi)\right)(m\cdot \exp sX_{e_1})dsdm
\\
&=&\omega_{l-1}\sum_{1\neq[\gamma]\in\Gamma} \cosh^{-k+ \frac{l-1}{2} }L_\gamma\sum_{\pi\in\widehat M}d_\pi 
\int_M \frac{1}{(\cosh L_\gamma-(m_\gamma^{m^{-1}})_{1,1})^{\frac{l-1}{2}}}
\\ 
&&\cdot \int_{0}^\infty s^{l-2}\left(s^2+1 \right)^{-k} \left(\chi_\pi*I_\gamma(\varphi)\right)\left(m\cdot \exp \sqrt{\frac{\cosh L_\gamma}{\cosh L_\gamma-(m_\gamma^{m^{-1}})_{1,1}}}sX_{e_1}\right)dsdm,
\end{eqnarray*}
where we applied the transformation $s\mapsto \sqrt{\frac{\cosh L_\gamma}{\cosh L_\gamma-(m_\gamma^{m^{-1}})_{1,1}}}s$ to obtain the last equality.

We define for the automorphic eigenfunction $1\neq \varphi$, $\gamma\neq 1$ and $\pi\in\widehat M$ \index{$c(\varphi,\gamma,\pi,k)$, coefficient in zeta function $\mathcal Z(\varphi)$}

\begin{eqnarray}\label{def: coeff}
c(\varphi,\gamma,\pi,k)&:=&\omega_{l-1}d_\pi\int_M \frac{1}{(\cosh L_\gamma-(m_\gamma^{m^{-1}})_{1,1})^{\frac{l-1}{2}}}
\int_{0}^\infty s^{l-2}\left(s^2+1 \right)^{-k}\nonumber\\&&\cdot \left(\chi_\pi*I_\gamma(\varphi)\right)\left(m\cdot \exp \sqrt{\frac{\cosh L_\gamma}{\cosh L_\gamma-(m_\gamma^{m^{-1}})_{1,1}}}sX_{e_1}\right)dsdm,
\end{eqnarray} $L_\gamma= \sqrt{2(l-1)}^{-1}l_\gamma>0$. Thus we can write for $\mathrm{Re}(k)>2\rho_0$ the trace of $\varphi\cdot \pi_R(f_k)$ for any nontrivial automorphic eigenfunction $\varphi$ as 
\begin{equation}\label{eq: tracezeta}
\mathrm{Tr}\left(\varphi\cdot \pi_R(f_k)\right)= \sum_{1\neq[\gamma]\in C\Gamma}\sum_{\pi\in\widehat M}c(\varphi,\gamma,\pi,k) (\cosh L_\gamma)^{-k+\rho_0}.
\end{equation}

For this special choice of $f_k$ we call this trace (\ref{eq: tracezeta}) the \index{auxiliary zeta function} \textit{auxiliary zeta function} for $k\in\mathbb C$, $\mathrm{Re}(k)>2\rho_0$,
\index{$\mathcal{R}(k;\varphi)$, $\mathcal{R}(\varphi)$, auxiliary zeta function} 
\begin{equation}\label{def: auxzet}
\mathcal{R}(k;\varphi):=\mathrm{Tr}(\varphi\cdot\pi_R(f_k))=\sum_{1\neq[\gamma]\in C\Gamma}\sum_{\pi\in\widehat M}c(\varphi,\gamma,\pi,k) (\cosh L_\gamma)^{-k+\rho_0}, 
\end{equation}
where $\rho_0=\rho(H_0)=\frac{l-1}{2}$.

We will often omit the argument of the (auxiliary) zeta function and write $\mathcal R(\varphi)$ (resp. $\mathcal Z(\varphi)$, see the next section) instead of $\mathcal R(k;\varphi)$ (resp. $\mathcal Z(k;\varphi)$).

\begin{remark}\label{rem: inival}
In (\ref{eq: normx1}) we have shown that $|X_{e_1}|^2=4(l-1)$, i.e. $ \frac{1}{2\sqrt{l-1}}X_{e_1}$ has unit length. On the slice \index{$S_m$, slice in $N$} $$S_m:= m\cdot \exp\mathbb{R}^+X_{e_1}=\exp \mathbb{R}^+m\cdot X_{e_1}=\exp \mathbb{R}^+m\cdot \frac{1}{2\sqrt{l-1}}X_{e_1}$$ the function $d_\pi\left(\chi_\pi*I_\gamma(\varphi)\right)$ satisfies for $l\geq 3$, see (\ref{eq: lghyp}),
\begin{eqnarray*} d_\pi \left(\chi_\pi*I_\gamma(\varphi)\right)\left(\exp s \left(m\cdot X_{e_1}\right)\right)&=& d_\pi \left(\chi_\pi*I_\gamma(\varphi)\right)\left(\exp  2\sqrt{l-1}s\left(\frac{m\cdot X_{e_1}}{2\sqrt{l-1}}\right)\right) 
\\&=&(-1)^p\left(4(l-1)\right)^{p}d_\pi \left(\left(\frac{m\cdot X_{e_1}}{2\sqrt{l-1}}\right)^{2p}\chi_\pi*I_\gamma(\varphi)\right)(e)\cdot s^{2p}\nonumber\\&&\cdot{}_2F_1\left(a+p,b+p,1+2\sqrt{ \frac{(1-\rho_0)^2}{4}-d };-s^2\right)
\\&=& (-1)^p d_\pi \left(\left(m\cdot X_{e_1}\right)^{2p}\chi_\pi*I_\gamma(\varphi)\right)(e)\cdot s^{2p}\nonumber\\&&\cdot{}_2F_1\left(a+p,b+p,1+2\sqrt{ \frac{(1-\rho_0)^2}{4}-d };-s^2\right)
\end{eqnarray*} for some constants $a,b,d,p$ depending on the eigenvalue $\mu$, $\pi$ and the slice $S_m$. For $l=2$, see (\ref{eq: l2hyp}) 
\begin{eqnarray*}
I_\gamma(\varphi)(\exp sX_{e_1})&=& I_\gamma(\varphi)\left(\exp 2s \frac{X_{e_1}}{2}\right)
\\&=& \left(\int_{c_{\gamma_0}}\varphi\right) \cdot {}_2F_1\left(a,b,\frac{1}{2} ;-s^2\right)\\&&+ 2i\left(\int_{c_{\gamma_0}}X_{1}\varphi \right) \cdot s \cdot {}_2F_1\left(a+\frac{1}{2},b+\frac{1}{2}, \frac{3}{2};-s^2 \right).
\end{eqnarray*}

\end{remark}
\begin{remark}
For the case of $G=SO_o(1,2)$ we note that $N=\{\exp uX_{e_1}:u\in\mathbb R\}$. It follows thus from (\ref{def: orbitalint}), as $M$ is trivial,
\begin{eqnarray*}
\mathcal{O}_\gamma(f_k)
&=& \int_N f_k(n^{-1} a_\gamma  n)I_\gamma(\varphi_j)(n)dn\\
&=& \int_{-\infty}^\infty f_k(\exp (-uX_{e_1}) a_\gamma \exp (uX_{e_1}))I_\gamma(\varphi_j)(\exp uX)du\\
&=& \left(\int_{c_{\gamma_0}}\varphi_j\right) \int_{-\infty}^\infty f_k(\exp (-uX_{e_1})a_\gamma \exp (uX_{e_1})) {}_2F_1\left(a,b,\frac{1}{2} ;-u^2\right)du
\\&&+ 2i\left(\int_{c_{\gamma_0}}X_{e_1}\varphi_j \right) \int_{-\infty}^\infty f_k(\exp -uX a_\gamma \exp uX) \cdot u \cdot {}_2F_1\left(a+\frac{1}{2},b+\frac{1}{2}, \frac{3}{2};-u^2 \right) du.
\end{eqnarray*}

The restriction $f|_N$ for any $f\in C(G//K)$ is even, see Corollary \ref{cor1}. Hence
\begin{equation*}
\int_{-\infty}^\infty uf(\exp (-uX_{e_1}) a_\gamma \exp (uX_{e_1}))  \cdot {}_2F_1\left(a+\frac{1}{2},b+\frac{1}{2}, \frac{3}{2};-\frac{u^2}{4} \right)du=0
\end{equation*}
and 
\begin{equation*}
\mathcal{O}_\gamma(f_k)=\left(\int_{c_{\gamma_0}}\varphi_j\right) \int_{-\infty}^\infty (u^2+1)^{-k} {}_2F_1\left(a,b,\frac{1}{2} ;-u^2 \sqrt{\frac{\cosh L_\gamma}{\cosh L_\gamma-1}} \right)du.
\end{equation*}
\end{remark}

We will now come to the theorem which shows that the auxiliary zeta function is defined at least for all $k\in\mathbb C$ with $\mathrm{Re}(k)>2\rho_0$.
\begin{theorem}\label{th: conauxzeta} 
Let $X=G/K$ be a real hyperbolic space and $\Gamma\subset G$ a uniform lattice. Let $\varphi$ be a non trivial automorphic eigenfunction and $k\in\mathbb C$. The auxiliary zeta function $\mathcal{R}(k;\varphi)$ converges for $ \mathrm{Re}(k)>2\rho_0 $.
\end{theorem}
\begin{proof}
From Proposition \ref{th: traceclass} it follows that $\varphi\cdot \pi_R(f_k)$ is of trace class for $\mathrm{Re}(k)>2\rho_0$, because if $\pi_R(f_k)$ of trace class then also $\varphi\cdot \pi_R(f)$ as $\varphi$ is bounded. But then also $\mathrm{Tr}\left(\varphi\cdot\pi_R(f_k)\right)$ converges for $\mathrm{Re}(k)>2\rho_0$.
\section{From $\mathcal{R}(\varphi)$ to $\mathcal{Z}(\varphi)$}\label{chap: gtrace sec: zeta}
From the auxiliary zeta function $ \mathcal{R}(\varphi)$ we can also derive the \index{zeta function} zeta function $ \mathcal{Z}(\varphi)$ via the following lemma, which can be found in \cite[Lemma 9.3]{AZ}.

\begin{lemma}\label{binom}
Let $k\in \mathbb{C}$ and $y\in (1,\infty)$ then there exist coefficients \index{$beta$@$\beta(k;m)$, coefficient} $\beta(k;m)$, which tend to zero for $m\to \infty$,  such that 
\begin{equation*}
\left(1-\sqrt{1- \frac{1}{y} }\right)^k=y^{-k}\sum_{m=0}^\infty \beta(k;m)y^{-m}.
\end{equation*}
Further, $\beta(k;0)=2^{-k}$ and $k\mapsto \beta(k;m)$ is holomorphic for any $m\in\mathbb N_0$.
\end{lemma}

\begin{proof}
Let $y\in(1,\infty)$. For $k\in \mathbb{C}$ we find by the binomial theorem

\begin{eqnarray*}
\left(1-\sqrt{1-y^{-1}}\right)^k&=& \left(1-\sum_{m=0}^\infty \binom{\frac{1}{2} }{m}(-1)^my^{-m}\right)^k\\
&=& \left(-\sum_{m=1}^\infty \binom{ \frac{1}{2} }{m}(-1)^my^{-m}\right)^k\\
&=& \left(\sum_{m=1}^\infty \binom{ \frac{1}{2} }{m}(-1)^{m-1}y^{-m}\right)^k\\
&=& \left(y^{-1}\sum_{m=1}^\infty \binom{ \frac{1}{2}}{m}(-1)^{m-1}y^{-m+1}\right)^k\\
&=& y^{-k}\left(\sum_{m=0}^\infty \binom{\frac{1}{2} }{m+1}(-1)^my^{-m}\right)^k.
\end{eqnarray*}

Now we assume temporarily $k\in\mathbb N_0$. Then by induction and Cauchy's product formula we find coefficients $\beta(k;m)$ such that
\begin{equation}\label{eq: bkntemp}
\left(\sum_{m=0}^\infty \binom{ \frac{1}{2} }{m+1}(-1)^my^{-m}\right)^k=\sum_{m=0}^\infty \beta(k;m)y^{-m},
\end{equation}
where \[\beta(k;m)=(-1)^m\sum_{r_1+\ldots+r_k=m,r_i\in\mathbb N_0}\binom{ \frac{1}{2} }{r_1+1}\cdots \binom{ \frac{1}{2} }{r_k+1}.\]

In particular, $\beta(k;0)=2^{-k}$. Let now $z\in (0,\infty)$ and $k\in \mathbb{C}$ arbitrary. Once again the binomial theorem gives 
\begin{eqnarray}\label{eq: binotemp}z^k&=&\left(1+(z-1)\right)^k\nonumber\\
&=& \sum_{l=0}^\infty \binom{ k }{l}(z-1)^l\nonumber\\
&=& \sum_{l=0}^\infty \sum_{m=0}^l\binom{ k }{l} \binom{l}{m}(-1)^{l-m}z^m.
\end{eqnarray}

So let $k\in\mathbb C$ again be arbitrary. We take for $z$ in (\ref{eq: binotemp}) the series $\sum_{m=0}^\infty \binom{\frac{1}{2} }{nm+1}(-1)^my^{-m}(=y-\sqrt{y^2-y})$. Then by (\ref{eq: bkntemp})  
\begin{eqnarray*}
\left(\sum_{m=0}^\infty \binom{\frac{1}{2} }{m+1}(-1)^my^{-m}\right)^k&=& \sum_{l=0}^\infty \sum_{s=0}^l \binom{ k}{l}\binom{l}{s}(-1)^{l-s}\left(\sum_{m=0}^\infty \binom{\frac{1}{2} }{m+1}(-1)^my^{-m}\right)^s\\ 
&\overset{(\ref{eq: bkntemp})}=& \sum_{m=0}^\infty\left(\sum_{l=0}^\infty \sum_{s=0}^l \binom{k}{l}\binom{l}{s}\beta(s;m)(-1)^{l-s}\right)y^{-m}\\
&=:& \sum_{m=0}^\infty \beta(k;m)y^{-m}.
\end{eqnarray*} 

In particular, 
\begin{eqnarray*}
\beta(k;0)&=& \sum_{l=0}^\infty \sum_{s=0}^\infty \binom{k}{l}\binom{l}{s}\beta(s;0)(-1)^{l-s}\\
&=& \sum_{l=0}^\infty (-1)^l\sum_{s=0}^l \binom{l}{s}\left(\frac{1}{2}\right)^s(-1)^{-s}\\
&=& \sum_{l=0}^\infty \binom{k}{l}(-1)^l\sum_{s=0}^l\binom{l}{s}\left(\frac{1}{2}\right)^s(-1)^s\\
&=& \sum_{l=0}^\infty \binom{k}{l}(-1)^l\left(1-\frac{1}{2}\right)^l\\
&=& \sum_{l=0}^\infty \binom{k}{l}(-1)^l\left(\frac{1}{2}\right)^l=\left(\frac{1}{2}\right)^k.
\end{eqnarray*}

Finally, $k\mapsto \beta(k;m)$ is holomorphic as $k\mapsto \binom{k}{l}=\frac{k(k-1)\ldots(k-l+1)}{l!}$ is holomorphic for $l\in \mathbb{N}_0$.
\end{proof}

Now we come to the definition of the \textit{zeta function} 
\begin{equation}\label{def: zetfunc}\index{$\mathcal{Z}(k;\varphi)$, $\mathcal{Z}(\varphi)$, zeta function}
 \mathcal{Z}(k;\varphi):=\sum_{1\neq[\gamma]\in C\Gamma}\sum_{\pi\in\widehat M}c(\varphi,\gamma,\pi,k) e^{ -(k-\rho_0)L_\gamma}.
\end{equation}
\begin{proposition}\label{prop: zetf}
The zeta function $\mathcal{Z}(k;\varphi)$ converges for $\mathrm{Re}(k)>2\rho_0$.
\end{proposition}
\begin{proof}
This will follow by the next lemma.

\begin{lemma}\label{auxi}
Let $k\in\mathbb{C}$ with $\mathrm{Re}(k)>2\rho_0$. For any $\gamma\neq 1$ there exists a constant \index{$C(\varphi)$, constant depending on $\varphi$} $C(\varphi)>0$ only depending on $\varphi$ such that 
\begin{equation*}
\left|\sum_{\pi\in\widehat{M}}c(\varphi,\gamma,\pi,k)\right|\leq C(\varphi)\cdot L_\gamma e^{-\rho_0L_\gamma}.
\end{equation*}  
Here $L_\gamma> 0$ is the length of the closed geodesic $[\gamma]\neq 1$.
\end{lemma}
We start with the orbital integral $\mathcal{O}_\gamma(f_k)$ for $\gamma\neq 1$ which equals by (\ref{def: orbitalint}) 
\begin{eqnarray*}\label{eq: og}
\mathcal{O}_\gamma(f_k)&=& \int_Nf_k(n^{-1}a_\gamma m_\gamma n)I_\gamma(\varphi)(n)dn\nonumber\\&=&
\int_{\mathbb R^{l-1}}\left(-(m_\gamma^{m^{-1}})_{1,1}|u|^2+(1+|u|^2)\cosh L_\gamma \right)^{-k} I_\gamma(\varphi)(\exp X_u)du\\&& \mbox{ , where  } (m_\gamma)_{1,1} \mbox{ as in } (\ref{def: cm})\nonumber\\&=&
\cosh^{-k} L_\gamma \int_{ \mathbb{R}^{l-1} }\left( -(m_\gamma^{m^{-1}})_{1,1} \frac{|u|^2}{\cosh L_\gamma}(1+|u|^2)\right)^{-k}I_\gamma(\varphi)(\exp X_u)du\nonumber\\&=&
\cosh^{-k}L_\gamma\int_{ \mathbb{R}^{l-1} }\left( |u|^2\left(1- \frac{(m_\gamma^{m^{-1}})_{1,1}}{\cosh L_\gamma}\right)+1 \right)^{-k}I_\gamma(\varphi)(\exp X_u)du\nonumber\\&=&
\cosh^{-k}L_\gamma\left(1-\frac{(m_\gamma^{m^{-1}})_{1,1}}{\cosh L_\gamma}\right)^{-\rho_0}\int_{\mathbb{R}^{l-1}}(|u|^2+1)^{-k}I_\gamma(\varphi)(\exp X_{T(u)})du\nonumber \\&=&
\cosh^{\rho_0-k}L_\gamma (\cosh L_\gamma-(m_\gamma^{m^{-1}})_{1,1})^{-\rho_0} \int_{\mathbb{R}^{l-1}}(|u|^2+1)^{-k}I_\gamma(\varphi)(\exp X_{T(u)})du\\&=:& (*),
\end{eqnarray*}
where $T(u)=((1-(m_\gamma^{m^{-1}})_{1,1}\cosh^{-1}L_\gamma)^{-1/2}u_1,\ldots,(1-(m_\gamma^{m^{-1}})_{1,1}\cosh^{-1}L_\gamma)^{-1/2}u_{l-1})^T$. By comparing the right side of the definition of $\mathcal{R}(k;\varphi)$, (\ref{def: auxzet}), with the last equation $(*)$ for $\mathcal{O}_\gamma(f_k)=\cosh^{-k+\rho_0}L_\gamma \sum_{\pi\in\widehat M}c(\varphi,\gamma,\pi,k)$ 
we see that
\begin{equation*}
\sum_{\pi\in\widehat M}c(\varphi,\gamma,\pi,k)=\cosh^{k-\rho_0} L_\gamma\cdot \mathcal{O}_\gamma(f_k)\overset{(*)}=(\cosh L_\gamma-(m_\gamma^{m^{-1}})_{1,1})^{-\rho_0} \int_{\mathbb{R}^{l-1}}(|u|^2+1)^{-k}I_\gamma(\varphi)(\exp X_{T(u)})du.
\end{equation*}

Now $|(m_\gamma^{m^{-1}})_{1,1}|\leq 1$ and $\varphi$ is bounded as a continuous function on compact $\Gamma \backslash X$, so 
\begin{eqnarray*}\left| I_\gamma(\varphi)(\exp X_{T(u)})\right|&\leq& \int_{A/<a_{\gamma_0}>}|\varphi(\alpha_\gamma^{-1}a\exp X_{T(u)})|da\\&\leq& C \int_{A/<a_{\gamma_0}>}= CL_\gamma  
\end{eqnarray*}
for $C:=\max_{x\in X_\Gamma}\{|\varphi(x)|\}$. For $k\in \mathbb{C}$ with $\mathrm{Re}(k)>2\rho_0$ let
\begin{equation}\label{eq: conpk}C(\varphi;k):=C\int_{\mathbb{R}^{l-1}}(|u|^2+1)^{-\mathrm{Re}(k)}du\geq 0.
\end{equation} 

The integral \[\int_{\mathbb{R}^{l-1}}(|u|^2+1)^{-\mathrm{Re}(k)}du\] surely converges for all $k\in\mathbb C$ with $\mathrm{Re}(k)>2\rho_0$ and \[\int_{\mathbb{R}^{l-1}}(|u|^2+1)^{-\mathrm{Re}(k)}\leq \int_{\mathbb{R}^{l-1}}(|u|^2+1)^{-2\rho_0}du=\frac{\omega_{l-1}}{2}B(\rho_0,\rho_0)\] for all $k\in\mathbb C$ with $\mathrm{Re}(k)>2\rho_0$. Hence, 
\begin{eqnarray*}
\left|\sum_{\pi\in\widehat{M}}c(\varphi,\gamma,\pi,k)\right|&=& (\cosh L_\gamma-(m_\gamma^{m^{-1}})_{1,1})^{-\rho_0}\left|\int_{\mathbb{R}^{l-1}}(|u|^2+1)^{-k}I_\gamma(\varphi)(\exp X_{T(u)})du  \right|\\ & \leq & 
e^{-\rho_0L_\gamma} \int_{\mathbb{R}^{l-1}}(|u|^2+1)^{-\mathrm{Re}(k)}|I_\gamma(\varphi)(\exp X_{T(u)})|du\\
&\leq & C(\varphi)\cdot L_\gamma e^{-\rho_0L_\gamma} 
\end{eqnarray*}
with $C(\varphi):=C(\varphi;2\rho_0)$, see (\ref{eq: conpk}).
\end{proof}

We continue with the proof of Proposition \ref{prop: zetf}. By Lemma \ref{auxi} we know that for $\mathrm{Re}(k)>2\rho_0$ \[\left|\mathcal Z(k;\varphi)\right|\leq C(\varphi)\sum_{1\neq [\gamma]\in C\Gamma}L_\gamma e^{-\rho_0L_\gamma}. \]

But for any $1\neq \gamma\in C\Gamma$ we can find some natural number $n$ such that $n-1\leq L_\gamma\leq n$. Hence $L_\gamma e^{-\rho_0L_\gamma}\leq n e^{-\rho_0(n-1)}=e^{\rho_0} ne^{-\rho_0n}$. It follows that $\sum_{1\neq [\gamma]\in C\Gamma}L_\gamma e^{-\rho_0L_\gamma}$ is convergent as it is bounded from above by \[e^{\rho_0}\sum_{n=1}^\infty ne^{-\rho_0 n}\] which is convergent as $\rho_0>0$.

\end{proof}

\begin{remark}
We follow Anantharaman and Zelditch in calling $\mathcal R(\sigma)$ and $\mathcal Z(\sigma)$ zeta functions, instead of logarithmic derivatives of zeta functions, which would be more correct in view of the classical Selberg Zeta function $$Z_S(k):=\prod_{[\gamma]}\prod_{s=0}^\infty\left(1-e^{-(s+k)L_\gamma}\right).$$ 

Normally one would pass from $ \frac{d}{dk}\ln Z_S $ to $Z_S$ by showing that $\frac{d}{dk}\ln Z_S$ has simple poles with integer residue. As we will see in Chapter \ref{chap: meromorph}, $\mathcal Z(\sigma)$ still has simple poles in some cases, but we can not guarantee that its residues are integers.  
\end{remark}

\begin{theorem}\label{th: zeta}
Let $X=G/K$ be a real hyperbolic space and $\Gamma\subset G$ a uniform lattice. Let $\varphi$ be a non trivial automorphic eigenfunction and $k\in\mathbb C$. We have \[\mathcal{Z}(k;\varphi)= \sum_{m=0}^\infty \beta(k-\rho_0;m)\mathcal{R}(k+2m;\varphi)\] for $k\in \mathbb{C}$ with $ \mathrm{Re}(k)>2\rho_0 $, where the coefficients $\beta(k-\rho_0;m)$ are determined by Lemma \ref{binom}. 
\end{theorem}
\begin{proof} A slight modification of the proof of \cite[Lemma 9.3.]{AZ} works also here. 
We have $e^{-t}\cdot \cosh t=\frac{1}{2}(1+e^{-2t})$ for $t\in \mathbb{R}$. Hence,
\begin{eqnarray*}
(e^{-t})^{k-\rho_0}&=& \left(\frac{1}{2}\right)^{k-\rho_0}(\cosh t)^{-(k-\rho_0)} \left(1+e^{-2t}\right)^{k-\rho_0}\\ &=&
\left(\frac{1}{2}\right)^{k-\rho_0}(\cosh t)^{-(k-\rho_0)} \left(1+(\cosh t-\sinh t)^2\right)^{k-\rho_0}
\\ &=& \left(\frac{1}{2}\right)^{k-\rho_0}(\cosh t)^{-(k-\rho_0)} \left(1+\left(y-\sqrt{y^2-1}\right)^2\right)^{k-\rho_0},
\end{eqnarray*}
where $y=\cosh t$ and using $\cosh^2t-\sinh^2t=1$. We continue with
\begin{eqnarray*}
(e^{-t})^{k-\rho_0}&=&\left(\frac{1}{2}\right)^{k-\rho_0}(\cosh t)^{-(k-\rho_0)} \left(1+(y^2-2y\sqrt{y^2-1}+y^2-1) \right)^{k-\rho_0}\\&=& 
\left(\frac{1}{2}\right)^{k-\rho_0}(\cosh t)^{-(k-\rho_0)} \left(2y^2-2y^2\sqrt{1-y^{-2}}\right)^{k-\rho_0}\\
&=&\left(\frac{1}{2}\right)^{k-\rho_0}(\cosh t)^{-(k-\rho_0)} (2y^2)^{k-\rho_0}\left(1-\sqrt{1-y^{-2}}\right)^{k-\rho_0}\\
&=&(\cosh t)^{-(k-\rho_0)} (y^2)^{k-\rho_0}\left(1-\sqrt{1-y^{-2}}\right)^{k-\rho_0}.
\end{eqnarray*}

We assume now $t\neq 0$, so that $\cosh t>1$. By the preceding Lemma \ref{binom} there are coefficients $\beta(k-\rho_0;m)$ such that
\[\left(1-\sqrt{1-y^{-2}}\right)^{k-\rho_0}=(y^2)^{-(k-\rho_0)}\sum_{m=0}^\infty \beta(k-\rho_0;m)(y^2)^{-m}\]

Hence,
\begin{eqnarray*}
(\cosh t)^{-(k-\rho_0)}(y^2)^{k-\rho_0}(y^2)^{-(k-\rho_0)}\sum_{m=0}^\infty \beta(k-\rho_0;m)(y^2)^{-m}&=&\sum_{m=0}^\infty \beta(k-\rho_0;m)(\cosh t)^{-(k+2m-\rho_0)}.
\end{eqnarray*}

Thus, we have shown for any $t\neq 0$
\[(e^{-t})^{k-\rho_0}=\sum_{m=0}^\infty \beta(k-\rho_0;m)(\cosh t)^{-(k+2m-\rho_0)}\]
and it follows
\begin{eqnarray*}
\sum_{1\neq \gamma}\sum_{\pi\in\widehat{M}}c(\varphi_n,\gamma,\pi,k)e^{-(k-\rho_0)L_\gamma}&=&\sum_{1\neq \gamma}\sum_{\pi\in\widehat{M}}c(\varphi_n,\gamma,\pi,k)\sum_{m=0}^\infty \beta(k-\rho_0;m)(\cosh L_\gamma)^{-(k+2m-\rho_0)}\\
&=& \sum_{m=0}^\infty \beta(k-\rho_0;m)\sum_{1\neq \gamma}\sum_{\pi\in\widehat{M}}c(\varphi_n,\gamma,\pi,k)(\cosh L_\gamma)^{-(k+2m-\rho_0)},
\end{eqnarray*}
where $L_\gamma=\sqrt{2(l-1)}^{-1}l_\gamma>0$ for $\gamma\neq 1$.
\end{proof}


\section{Simplifications - the zeta function for $SO_o(1,2)$ and $SO_o(1,3)$}\label{sec: trivial}

In this section we want to explain how the zeta function \[\mathcal{R}(k;\varphi)=\sum_{1\neq[\gamma]\in C\Gamma}\sum_{\pi\in\widehat M}c(\varphi,\gamma,\pi,k) (\cosh L_\gamma)^{-k+\rho_0}\] resp. $c(\varphi,\gamma,\pi,k)$ simplify if $\Gamma$ satisfies a special property. Namely, we have seen before that every $\gamma\in\Gamma$ is conjugated in $G$ to some $a_\gamma m_\gamma\in AM$. The assumption we make is that $m_\gamma$ is always central in $M$. 

This assumption is satisfied for any uniform lattice, if $G=SO_o(1,2)$ or $SO_o(1,3)$. In this case $M$ is trivial resp. isomorphic to $SO(2)$, i.e. abelian. Unfortunately, we do not know of any examples in $SO_o(1,l)$, for $l>3$. The case $SO_o(1,2)$ has also been considered in \cite{AZ}.

Let $ \mathbf{1}$ denote the trivial representation of $M$ on $ \mathbb{C}$ and $\varphi\neq 1$ a non trivial automorphic eigenfunction with eigenvalue $\mu=-\frac{1}{4}(\rho_0+\frac{r^2}{\rho_0})$. Further, let us assume $m_\gamma^{m^{-1}}=m_\gamma$ for all $m\in M$  and $\gamma\in\Gamma$. We now look at the definition (\ref{def: coeff}) of $c(\varphi,\gamma,\pi,k)$
\begin{eqnarray*}
c(\varphi,\gamma,\pi,k)&=&\omega_{l-1}d_\pi\int_M \frac{1}{(\cosh L_\gamma-(m_\gamma^{m^{-1}})_{1,1})^{\frac{l-1}{2}}}
\int_{0}^\infty s^{l-2}\left(s^2+1 \right)^{-k}\nonumber\\&&\cdot \left(\chi_\pi*I_\gamma(\varphi)\right)\left(m\cdot \exp \sqrt{\frac{\cosh L_\gamma}{\cosh L_\gamma-(m_\gamma^{m^{-1}})_{1,1}}}sX_{e_1}\right)dsdm,
\\&=:&(*).
\end{eqnarray*}

In Proposition \ref{chap: geom prop minv} we observed that for $\sigma\in C(\Gamma\backslash G)$, $n\in N$ and $z\in M_{m_\gamma}A$
\begin{equation*}
\int_{A/ <a_{\gamma_0}>}\sigma(\alpha_\gamma^{-1}azn)da=\int_{\alpha_\gamma \Gamma_\gamma \alpha_\gamma^{-1}\backslash G_{a_\gamma m_\gamma}}\sigma(\alpha_\gamma^{-1}xzn)dx=\int_{\alpha_\gamma \Gamma_\gamma \alpha_\gamma^{-1}\backslash G_{a_\gamma m_\gamma}}\sigma(\alpha_\gamma^{-1}xn)dx.
\end{equation*}

In particular this is true for $z\in M_{m_\gamma}=M$ by assumption. It follows that the weight function \[I_\gamma(\varphi)(n)=\int_{A/ <a_{\gamma_0}>}\varphi(\alpha_\gamma^{-1}an)da=\int_{\alpha_\gamma \Gamma_\gamma \alpha_\gamma^{-1}\backslash G_{a_\gamma m_\gamma}}\varphi(\alpha_\gamma^{-1}xn)dx\] is also left-$M$-invariant since $M_{m_\gamma}=M$, i.e. $G_{a_\gamma m_\gamma}=MA$ for every $\gamma\in\Gamma$. Thus, $\chi_\pi*I_\gamma(\varphi)=0$ for all $\pi\neq 1$ and $c(\varphi,\gamma,\mathbf{1},k)=0$ for any non trivial $\pi$. Then $ \mathcal{R}(k;\varphi)$ simplifies to
\begin{equation}\label{eq: simple}\mathcal{R}(k;\varphi)=\sum_{1\neq[\gamma]\in C\Gamma}c(\varphi,\gamma,\mathbf{1},k) (\cosh L_\gamma)^{-k+\rho_0}.\end{equation}

But $\pi= \mathbf{1}$ implies $p=d=0$, see Theorem \ref{chap: ode main}, hence we get the following proposition.
\begin{proposition}\label{prop: simple}
For $\varphi\neq 1$, $\gamma\neq 1$ and $k\in \mathbb{C}$ with $\mathrm{Re}(k)>2\rho_0$
the coefficient $c(\varphi,\gamma,\mathbf{1},k)$ can be computed to 
\begin{eqnarray}\label{eq: c1}
c(\varphi,\gamma,\mathbf 1,k)&=& \omega_{l-1}\sqrt{2(l-1)}\left(\int_{c_{\gamma_0}}\varphi\right) \frac{1}{(\cosh L_\gamma-(m_\gamma)_{1,1})^{\rho_0}}\nonumber\\ &&\cdot\int_0^\infty s^{l-2}(s^2+1)^{-k}  {}_2F_1\left(a,b,\rho_0;-s^2\frac{\cosh L_\gamma}{\cosh L_\gamma-(m_\gamma)_{1,1}}\right)ds,
\end{eqnarray}
where \[\frac{1}{(\cosh L_\gamma-(m_\gamma)_{(1,1)})^{\rho_0} }
= \frac{2^{\rho_0}}{e^{\rho_0 L_\gamma }\det\left(1-\mathrm{Ad}(m_\gamma a_\gamma)^{-1}|_{\mathfrak n}\right)} .\]
\end{proposition}

For $l=2$ and $l=3$ we thus get

\begin{equation*}
\mathcal{R}(k;\varphi)=\sum_{1\neq[\gamma]\in C\Gamma}c(\varphi,\gamma,\mathbf{1},k) (\cosh L_\gamma)^{-k+1/2}\mbox{ , } l=2 
\end{equation*}
resp.
\begin{equation*}
\mathcal{R}(k;\varphi)=\sum_{1\neq[\gamma]\in C\Gamma}c(\varphi,\gamma,\mathbf{1},k) (\cosh L_\gamma)^{-k+1}\mbox{ , } l=3 
\end{equation*}
where 

\begin{equation*}
c(\varphi,\gamma,\mathbf 1,k)=  \frac{2\left(\int_{c_{\gamma_0}}\varphi\right)}{\sinh L_\gamma/2}\int_0^\infty \left(s^2+1\right)^{-k} {}_2F_1\left(a,b,\frac{1}{2} ;-s^2 \frac{\cosh L_\gamma}{\cosh L_\gamma-1} \right)ds
\end{equation*}
for  $l=2$ resp. 
\begin{eqnarray*}
c(\varphi,\gamma,\mathbf 1,k)&=&  \frac{2\omega_{2} \left(\int_{c_{\gamma_0}}\varphi\right)}{\cosh L_\gamma-(m_\gamma)_{1,1}}\int_0^\infty s\left(s^2+1 \right)^{-k} {}_2F_1\left(a,b,1;-s^2\frac{\cosh L_\gamma}{\cosh L_\gamma-(m_\gamma)_{1,1}}\right)ds
\end{eqnarray*}
for $l=3$. 

\begin{proof}
First, by the assumptions on $m_\gamma$ and as $I_\gamma(\varphi)$ is left-$M$-invariant we get
\begin{eqnarray*}
c(\varphi,\gamma,\mathbf{1},k)&=&\omega_{l-1}\int_M \frac{1}{(\cosh L_\gamma-(m_\gamma^{m^{-1}})_{1,1})^{\frac{l-1}{2}}}
\int_{0}^\infty s^{l-2}\left(s^2+1 \right)^{-k}\nonumber\\&&\cdot I_\gamma(\varphi)\left(m\cdot \exp \sqrt{\frac{\cosh L_\gamma}{\cosh L_\gamma-(m_\gamma^{m^{-1}})_{1,1}}}sX_{e_1}\right)dsdm,
\\&=&
\omega_{l-1} \frac{1}{(\cosh L_\gamma-(m_\gamma)_{1,1})^{\frac{l-1}{2}}}
\int_{0}^\infty s^{l-2}\left(s^2+1 \right)^{-k}\nonumber\\&&\cdot I_\gamma(\varphi)\left(\exp \sqrt{\frac{\cosh L_\gamma}{\cosh L_\gamma-(m_\gamma)_{1,1}}}sX_{e_1}\right)ds
\\&=:&(+).
\end{eqnarray*}

The equality in $(\ref{eq: c1})$ of Proposition \ref{prop: simple} then follows directly from $(+)$ and the following observations 
\begin{eqnarray*}I_\gamma(\varphi)\left(\exp \sqrt{\frac{\cosh L_\gamma}{\cosh L_\gamma-(m_\gamma)_{1,1}}}sX_{e_1}\right)&=& I_\gamma(\varphi)(e){}_2F_1\left(a,b,\rho_0;-s^2\frac{\cosh L_\gamma}{\cosh L_\gamma-(m_\gamma)_{1,1}}\right)\\
&=& \sqrt{2(l-1)}\left(\int_{c_{\gamma_0}}\varphi\right){}_2F_1\left(a,b,\rho_0;-s^2\frac{\cosh L_\gamma}{\cosh L_\gamma-(m_\gamma)_{1,1}}\right),
\end{eqnarray*}
see equations (\ref{eq: l2hyp}) for $l=2$ resp. (\ref{eq: lghyp}), $l\geq 3$, with $p=d=0$ and also (\ref{eq: inivalsimp}) and Remark \ref{rem: inival}. To get the second equality about $(\cosh L_\gamma-(m_\gamma)_{(1,1)})^{-\rho_0}$ we use two observations.

\begin{lemma}
Let $a\in A^+$ and $m\in M$. The map \index{$h(ma)$, diffeomorphism on $N$} $h(ma):n\mapsto m^{-1}a^{-1}n^{-1}amn$ is a diffeomorphism of $N$ and the Haar measure $dn$ transforms to \[dh(ma)(n)=\mathrm{det}\left(1-\mathrm{Ad}(ma)^{-1}|_{\mathfrak n}\right)dn.\]
\end{lemma}

\begin{proof}
See \cite[Ch. I Lem.5.4.]{GASS} or \cite[Th.11.24.]{Wi}.
\end{proof}
\begin{corollary}\label{cor: bo}
For integrable functions $f$ on $N$, $a\in A^+$ and $m\in M$
\[\int_Nf(n)dn=\mathrm{det}\left(1-\mathrm{Ad}(ma)^{-1}|_{\mathfrak n}\right)\int_Nf(m^{-1}a^{-1}n^{-1}amn)dn\] resp. if $f$ is also bi-$K$-invariant, then
\[\int_N f(an)dn=\mathrm{det}\left(1-\mathrm{Ad}(ma)^{-1}|_{\mathfrak n}\right)\int_N f(n^{-1}amn)dn .\]
\end{corollary}

We already know that for $s\in\mathbb R$ \[f_k\left((\exp -sX_{e_1}) a_\gamma m(\exp sX_{e_1})\right)=\left(-(m_\gamma)_{1,1}s^2+(1+s^2) \cosh L_\gamma \right)^{-k},\] see $(\ref{eq: fk1})$. Also for $t\in\mathbb R$ \[
f_k(\exp tH_0\exp sX_{e_1})=f_k(a_t\exp sX_{e_1})=\left(\cosh t +\frac{s^2}{2}e^t\right)^{-k},
\] see $(\ref{eq: fkatn})$. We finally recall the polar coordinates formula \[\int_N f(n)dn=\omega_{l-1}\int_0^\infty s^{l-2}f(\exp sX_{e_1})ds\] for bi-$M$-invariant, integrable functions $f$ on $N$. Thus for $f=f_k$, $a=a_\gamma=\exp L_\gamma H_0$ and $m=m_\gamma$ we get that

\begin{eqnarray*}
&&\frac{1}{\det\left(1-\mathrm{Ad}(m_\gamma a_\gamma)^{-1}|_{\mathfrak n}\right)}\int_N f_k(a_\gamma n)dn
\\&=&\frac{\omega_{l-1}}{\det\left(1-\mathrm{Ad}(m_\gamma a_\gamma)^{-1}|_{\mathfrak n}\right)} \int_0^\infty \int_M f_k\left(a_\gamma m \exp sX_{e_1} m^{-1}\right)dmds\\
&=&  \frac{\omega_{l-1}}{\det\left(1-\mathrm{Ad}(m_\gamma a_\gamma)^{-1}|_{\mathfrak n}\right)} \int_0^\infty \int_M f_k\left(a_\gamma \exp sX_{e_1} \right)dmds 
\\&=& \frac{\omega_{l-1}}{\det\left(1-\mathrm{Ad}(m_\gamma a_\gamma)^{-1}|_{\mathfrak n}\right)}\int_0^\infty s^{l-2}\left( \frac{s^2}{2}e^{L_\gamma}+\cosh L_\gamma \right)^{-k}ds\\
&=& \frac{\omega_{l-1}}{\det\left(1-\mathrm{Ad}(m_\gamma a_\gamma)^{-1}|_{\mathfrak n}\right)} \cosh^{-k}L_\gamma\int_0^\infty s^{l-2}\left( \frac{s^2e^{L_\gamma}}{2\cosh L_\gamma}+1 \right)^{-k}ds\\
&=& \frac{\omega_{l-1}}{\det\left(1-\mathrm{Ad}(m_\gamma a_\gamma)^{-1}|_{\mathfrak n}\right)} \cosh^{-k}L_\gamma e^{-\frac{l-1}{2}L_\gamma }(2\cosh L_\gamma)^{\frac{l-1}{2} }\int_0^\infty s^{l-2}(s^2+1)^{-k}ds\\
&=& \frac{\omega_{l-1}}{\det\left(1-\mathrm{Ad}(m_\gamma a_\gamma)^{-1}|_{\mathfrak n}\right)} 2^{ \frac{l-1}{2} }e^{-\frac{l-1}{2}L_\gamma } \cosh^{-k+\frac{l-1}{2}} L_\gamma  \int_0^\infty s^{l-2}(s^2+1)^{-k}ds\\
&=:&(*)
\end{eqnarray*}
equals

\begin{eqnarray*}
\int_{N}f_k(n^{-1}m_\gamma a_\gamma n)dn&=& \omega_{l-1}\int_0^\infty s^{l-2}\int_M f_k\left(\exp -sX_{e_1}(m_\gamma^{m^{-1}})a_\gamma \exp sX_{e_1}\right)dmds
\\&=& \omega_{l-1}\int_M \int_0^\infty s^{l-2}\left( -(m_\gamma^{m^{-1}})_{1,1}s^2+(1+s^2)\cosh L_\gamma\right)^{-k}dsdm, \\&&  (m_\gamma^{m^{-1}})_{1,1}=(m^{-1}mm)_{1,1}\\
&=& \omega_{l-1}\cosh^{-k}L_\gamma \int_M\int_0^\infty s^{l-2}\left( s^2\left(1-\frac{(m^{m'})_{1,1}}{\cosh t} \right)+1\right)^{-k}dsdm'\\
&=& \omega_{l-1}\cosh^{-k}L_\gamma\int_M\left(1-\frac{(m_\gamma^{m^{-1}})_{1,1}}{\cosh L_\gamma} \right)^{-\frac{l-1}{2} }dm \int_0^\infty s^{l-2}(1+s^2)^{-k}ds\\
&=& \omega_{l-1}\cosh^{-k+\frac{l-1}{2}}L_\gamma\int_M \left(\frac{1}{\cosh L_\gamma-(m_\gamma^{m^{-1}})_{1,1}} \right)^{\frac{l-1}{2} }dm \int_0^\infty s^{l-2}(1+s^2)^{-k}ds\\
&=:& (**). 
\end{eqnarray*}
Therefore, Corollary \ref{cor: bo} and comparing $(*)$ to $(**)$ gives

\begin{equation*}
\int_M\left(\frac{1}{\cosh L_\gamma -(m_\gamma^{m^{-1}})_{1,1}} \right)^{\frac{l-1}{2} }dm = \frac{1}{\det\left(1-\mathrm{Ad}(m_\gamma a_\gamma)^{-1}|_{\mathfrak n}\right)} 2^{ \frac{l-1}{2} }e^{-\frac{l-1}{2}L_\gamma }.
\end{equation*}

In particular if $m_\gamma$ is central then $(m_\gamma^{m^{-1}})_{(1,1)}=(m_\gamma)_{1,1}$ is actually independent of $m$ and as $\int_M=1$

\begin{eqnarray*}
\int_M\left(\frac{1}{\cosh t-(m_\gamma)_{1,1}} \right)^{\frac{l-1}{2} }dm&=&\frac{1}{(\cosh t-(m_\gamma)_{1,1})^{\frac{l-1}{2}} }\\
&=& \frac{1}{\det\left(1-\mathrm{Ad}(m_\gamma a_\gamma)^{-1}|_{\mathfrak n}\right)} 2^{ \frac{l-1}{2} }e^{-\frac{l-1}{2}L_\gamma }.
\end{eqnarray*}

Hence, we can replace \begin{eqnarray}\label{eq: jactransf}
\frac{1}{(\cosh L_\gamma-(m_\gamma)_{(1,1)})^{\rho_0 }}&=&\frac{1}{(\cosh L_\gamma-(m_\gamma)_{(1,1)})^{\frac{l-1}{2}} }
\nonumber\\&=& \frac{1}{\det\left(1-\mathrm{Ad}(m_\gamma a_t)^{-1}|_{\mathfrak n}\right)} 2^{ \frac{l-1}{2} }e^{-\frac{l-1}{2}L_\gamma }\nonumber\\
&=& \frac{1}{\det\left(1-\mathrm{Ad}(m_\gamma a_t)^{-1}|_{\mathfrak n}\right)} 2^{ \rho_0 }e^{-\rho_0 L_\gamma}
\end{eqnarray}
in the first equality of (\ref{eq: c1}).

\end{proof}

We can even get rid off the integrals in the formula (\ref{eq: c1}) for $c(\varphi,\gamma,\mathbf{1})$. Putting for the moment $z:=\frac{\cosh L_\gamma}{(\cosh L_\gamma-(m_\gamma^{m^{-1}})_{1,1})}\overset{(\ref{eq: jactransf})}=\frac{2^{\rho_0}\cosh L_\gamma}{\det\left(1-\mathrm{Ad}(m_\gamma a_\gamma)^{-1}|_{\mathfrak n}\right)e^{\rho_0 L_\gamma}}$ we have to compute for $\mathrm{Re}(k)>2\rho_0$ the integral 

\begin{eqnarray*}
&&\int_0^\infty s^{l-2}(s^2+1)^{-k} {}_2F_1\left(a,b,\rho_0;-s^2\frac{\cosh L_\gamma}{\cosh L_\gamma-(m_\gamma)_{1,1}}\right)ds\\&=&\int_0^\infty s^{l-2}\left( s^2+1\right)^{-k} {}_2F_1(a,b,\rho_0;-zs^2)ds.
\end{eqnarray*}

By the transformation $s\mapsto \sqrt{s}$ this turns into

\begin{equation}\label{eq: intransform}
\frac{1}{2}\int_{0}^\infty s^{\frac{l-2}{2} }s^{-\frac{1}{2} }(s+1)^{-k}{}_2F_1(a,b,\rho_0;-zs)ds.
\end{equation}

Noting that $\rho_0=\frac{l-1}{2}$ the transformation $s\mapsto \frac{s}{z}$ yields
\begin{eqnarray*}\label{eq: hypint}
&&\frac{1}{2}\int_{0}^\infty s^{\frac{l-2}{2} }s^{-\frac{1}{2} }(s+1)^{-k}{}_2F_1(a,b,\rho_0;-zs)ds\nonumber\\&=&
\frac{1}{2}z^{k-\rho_0}\int_0^\infty s^{\rho_0-1}(s+z)^{-k}{}_2F_1(a,b,\rho_0;-s)ds\nonumber\\&=& \frac{1}{2}z^{k-\rho_0} \cdot \frac{\Gamma(\rho_0)\Gamma(a-\rho_0+k)\Gamma(b-\rho_0+k)}{\Gamma(k)\Gamma(a+b-\rho_0+k)}\\&&\cdot {}_2F_1(a-\rho_0+k,b-\rho_0+k,a+b-\rho_0+k;1-z)\nonumber
\\&=& \frac{1}{2}z^{k-\rho_0} \cdot \frac{\Gamma(\rho_0)\Gamma(a-c+k)\Gamma(b-\rho_0+k)}{\Gamma(k)^2}\\&&\cdot {}_2F_1(a-\rho_0+k,b-\rho_0+k,a+b-\rho_0+k;1-z)\nonumber
\\&=& \frac{1}{2}z^{k-\rho_0} \cdot \frac{\Gamma(\rho_0)\Gamma(k-a)\Gamma(k-b)}{\Gamma(k)^2}\\&&\cdot {}_2F_1(a-\rho_0+k,b-\rho_0+k,k;1-z),
\end{eqnarray*}
where we used an integral transform for hypergeometric functions, see \cite[20.2 (10)]{Ba} which is valid for $\mathrm{Re}(a-\rho_0+k)=\mathrm{Re}(k-b)=\mathrm{Re}(k-\frac{\rho_0}{2}+\frac{ir}{2})>0$, $\mathrm{Re}(b-\rho_0+k)=\mathrm{Re}(k-a)=\mathrm{Re}(k-\frac{\rho_0}{2}-\frac{ir}{2})>0$ for $k\in \mathbb{C}$ with $\mathrm{Re}(k)>2\rho_0$. 

From $(\ref{eq: c1})$ we thus get

\begin{eqnarray}\label{eq: coefsimex}
c(\varphi,\gamma,\mathbf 1,k)&=&\omega_{l-1}\sqrt{2(l-1)}\left(\int_{c_{\gamma_0}}\varphi\right) \frac{1}{(\cosh L_\gamma-(m_\gamma)_{1,1})^{\rho_0}}\nonumber\\ &&\cdot \frac{1}{2}\left(\frac{\cosh L_\gamma}{\cosh L_\gamma-(m_\gamma)_{1,1}}\right)^{k-\rho_0}\frac{\Gamma(\rho_0)\Gamma(-\frac{1}{2}(\rho_0-ir)+k)\Gamma(-\frac{1}{2}(\rho_0+ir)+k)} {\Gamma(k)^2}\nonumber\\&&\cdot {}_2F_1\left(-\frac{1}{2}(\rho_0-ir)+k,-\frac{1}{2}(\rho_0+ir)+k,k;1-\frac{\cosh L_\gamma}{\cosh L_\gamma-(m_\gamma)_{1,1}}\right)\nonumber\\
&= & \frac{\omega_{l-1}\sqrt{2(l-1)}\left(\int_{c_{\gamma_0}}\varphi\right) }{2}\left(\frac{\cosh L_\gamma}{\cosh L_\gamma -(m_\gamma)_{(1,1)}}\right)^{k-2\rho_0}\nonumber \\
&&\cdot\frac{\Gamma(\rho_0)\Gamma(-\frac{1}{2}(\rho_0-ir)+k)\Gamma(-\frac{1}{2}(\rho_0+ir)+k)} {\Gamma(k)^2}\nonumber\\&&\cdot {}_2F_1\left(-\frac{1}{2}(\rho_0-ir)+k,-\frac{1}{2}(\rho_0+ir)+k,k;1-\frac{\cosh L_\gamma}{\cosh L_\gamma-(m_\gamma)_{1,1}}\right)
\end{eqnarray}
where $\frac{2^{\rho_0}\cosh L_\gamma}{\det\left(1-\mathrm{Ad}(m_\gamma a_\gamma)^{-1}|_{\mathfrak n}\right)e^{\rho_0 L_\gamma}}=\frac{\cosh L_\gamma}{\cosh L_\gamma-(m_\gamma)_{1,1}}$.

\begin{remark}
Recall the definition (\ref{def: coeff}) for $c(\varphi,\gamma,\pi,k)$. For $G=SO_o(1,l)$ with $l\geq 4$ the definition of $c(\varphi,\gamma,\pi,k)$ can be simplified at least one more step by using a computer algebra program like MATHEMATICA in order to carry out the $ds$-integration. Let $z:=\frac{\cosh L_\gamma}{\cosh L_\gamma-(m_\gamma^{m^{-1}})_{(1,1)}}$, then 
\begin{eqnarray*}
&&\int_0^\infty s^{l-2}(s^2+1)^{-k}s^{2p}{}_2F_1\left(a+p,b+p,1+2\sqrt{\frac{(\rho_0)^2}{4}-d };-zs^2\right)ds\\
&=& \frac{1}{2}\left(z^{k-\rho_0-p}\frac{\Gamma(1+\frac{1}{2}\sqrt{(l-3)^2-16d} )\Gamma(k-a)\Gamma(k-b)\Gamma(p-k+\rho_0)}{\Gamma(1+k-\rho_0-p+\frac{1}{2}\sqrt{(l-3)^2-16d})\Gamma(a+p)\Gamma(b+p)}\right.
\\&&\cdot {}_3F_2\left( \{k,k-a,k-b\},\{k+1-p-\rho_0,1+\frac{1}{2}\sqrt{(l-3)^2-16d}+k-\rho_0-p\};z\right)\\&&+\left.
\frac{\Gamma(k-\rho_0-p)\Gamma(\rho_0+p)}{\Gamma(k)}{}_3F_2\left(\{a+p,b+p,\rho_0+p\},\{1+\frac{1}{2}\sqrt{(l-3)^2-16d},\rho_0-k+p\};z\right)\right).
\end{eqnarray*}

Here \[{}_3F_2\left(\{\alpha_1,\alpha_2,\alpha_4\},\{\beta_1,\beta_2\};z\right):=\sum_{k=0}^\infty \frac{(\alpha_1)_k(\alpha_2)_k(\alpha_3)_k}{(\beta_1)_k(\beta_2)_k}\frac{z^k}{k!}\] denotes the generalized hypergeometric function, see for example \cite[\S 11]{Ol}.
\end{remark}

\section{$\sigma\equiv 1$ - the classical Selberg zeta function}

In this section we want to compare the (classical) Selberg zeta function as one can find in the book \cite{BO} with our zeta function $\mathcal Z(1)$ derived from $\mathcal R(1)$. So we recall that we have considered the auxiliary zeta function from (\ref{def: auxzet}) \[
\mathcal{R}(k;\varphi)=\sum_{1\neq[\gamma]\in C\Gamma}\sum_{\pi\in\widehat M}c(\varphi,\gamma,\pi,k) (\cosh L_\gamma)^{-k+\rho_0} \]
and the zeta function from (\ref{def: zetfunc})
\[
\mathcal{Z}(k;\varphi)=\sum_{1\neq[\gamma]\in C\Gamma}\sum_{\pi\in\widehat M}c(\varphi,\gamma,\pi,k) e^{ -(k-\rho_0)L_\gamma}.
\]

From now on we fix $\varphi\equiv 1$ and go back to Theorem \ref{theo: generaltr} with $\sigma\equiv \varphi\equiv 1$. Then we note that in this case the weight $I_\gamma(\sigma)$, see (\ref{def: weight}), is a constant which equals \[\int_{A/<a_{\gamma_0}>}da=\int_0^{L_{\gamma_0}}dt=\frac{L_\gamma}{n_\gamma}=L_{\gamma_0}=\frac{l_{\gamma_0}}{\sqrt{2(l-1)}} ,\] see Proposition \ref{chap: geoprop l}, where $l_{\gamma_0}$ is the length of the prime closed geodesic belonging to the conjugacy class of $\gamma$. In particular, the weight is $M$-invariant and hence $\chi_\pi*I_\gamma(\sigma)=0$ for any $\pi\neq \mathbf{1}$. Hence,
\begin{equation}\label{eq: tracmod}
\mathrm{Tr}(\sigma\cdot \pi_R(f))=\mathrm{Tr}(\pi_R(f))=f(e)|\Gamma\backslash G|+\sum_{1\neq [\gamma]\in C\Gamma}L_{\gamma_0}\int_Nf(n^{-1}a_\gamma m_\gamma n)dn
\end{equation}
for any $f\in C^\infty(G//K)$ such that $\pi_R(f)$ is of trace-class. 
We note that 
\begin{eqnarray*}
\frac{1}{2}B(k-\rho_0,k)=\int_0^\infty u^{l-2}(u^2+1)^{-k}du&=& \frac{1}{2}\int_0^\infty y^{\frac{n-2}{2}y^{-\frac{1}{2} }(1+y)^{-k}dy }
\\&=&\frac{\Gamma(k-\rho_0)\Gamma(\rho_0)}{2\Gamma(k)}
,
\end{eqnarray*}
see (\ref{eq: betafunc}). Furthermore, 
$\omega_{l-1}$ is positive by definition. Now we modify the function $f_k$ from (\ref{def: fk}) by dividing it by the non-zero number, see (\ref{eq: inthypgam}),
\begin{equation*}
2^{\rho_0}\omega_{l-1}\int_0^\infty s^{l-2}(s^2+1)^{-k}ds=2^{\rho_0-1}\omega_{l-1}B(k-\rho_0,\rho_0).
\end{equation*}

Then replacing $f_k$ by \[\frac{f_k}{2^{\rho_0-1}\omega_{l-1}B(k-\rho_0.\rho_0)}\] in $\pi_R(f_k)$ gives still a trace class operator and in the proof of Proposition \ref{prop: simple} we have computed 
\begin{eqnarray*}
\int_Nf_k(n^{-1}a_\gamma m_\gamma n)dn&=& \frac{1}{\det\left(1-\mathrm{Ad}(m_\gamma a_\gamma)^{-1}|_{\mathfrak n}\right)}\int_N f_k(a_\gamma n)dn\\
&=& \frac{2^{\rho_0-1}e^{-\rho_0t}\omega_{l-1}B(k-\rho_0,\rho_0)}{\det\left(1-\mathrm{Ad}(m_\gamma a_\gamma)^{-1}|_{\mathfrak n}\right)} \cosh^{-k+\rho_0} L_\gamma,
\end{eqnarray*}

Hence,
\[\frac{1}{2^{\rho_0-1}\omega_{l-1}B(k-\rho_0,\rho_0)}\int_N f_k(n^{-1}a_\gamma m_\gamma n)dn=\frac{e^{-\rho_0L_\gamma}}{\det\left(1-\mathrm{Ad}(m_\gamma a_\gamma)^{-1}|_{\mathfrak n}\right)}\cosh^{-k+\rho_0} L_\gamma. \]


We define  \index{$\mathcal{R}(k)=\mathcal{R}(k;1)$, auxiliary zeta function to $\varphi\equiv 1$} for $\mathrm{Re}(k)>2\rho_0$

\begin{equation*}
\mathcal{R}(k;1):=\mathcal{R}(k):= \sum_{1\neq [\gamma]\in C\Gamma} \frac{L_{\gamma_0}e^{-\rho_0 L_\gamma}}{\det\left(1-\mathrm{Ad}(m_\gamma a_\gamma)^{-1}|_{\mathfrak n}\right)}(\cosh L_\gamma)^{(-k+\rho_0)L_\gamma} 
\end{equation*}
resp. \index{$\mathcal{Z}(k)=\mathcal{Z}(k;1)$, zeta function to $\varphi\equiv 1$}

\begin{equation}\label{eq: clzeta1}
\mathcal{Z}(k;1):=\mathcal{Z}(k):= \sum_{1\neq [\gamma]\in C\Gamma} \frac{L_{\gamma_0} e^{-\rho_0 L_\gamma}}{\det\left(1-\mathrm{Ad}(m_\gamma a_\gamma)^{-1}|_{\mathfrak n}\right)}e^{(-k+\rho_0)L_\gamma}. 
\end{equation}

It follows immediately from $(\ref{eq: tracmod})$ that

\begin{eqnarray*}
\mathcal{R}(k)&=&\mathrm{Tr}\left(\pi_R\left(\frac{f_k}{\omega_{l-1} 2^{\rho_0-1}B(k-\rho_0,\rho_0) }\right)\right)-\frac{f_k(e)}{\omega_{l-1} 2^{\rho_0-1}B(k-\rho_0,\rho_0) }\cdot |\Gamma\backslash G|\\
&=& \mathrm{Tr}\left(\pi_R\left(\frac{f_k}{\omega_{l-1} 2^{\rho_0-1} B(k-\rho_0,\rho_0) }\right)\right)-\frac{1}{\omega_{l-1} 2^{\rho_0-1}B(k-\rho_0,\rho_0) }\cdot |\Gamma\backslash G|
\end{eqnarray*}

Hence, $\mathcal R(k)$ is well-defined at least for $\mathrm{Re}(k)>2\rho_0$. Also, the proof of \ref{th: zeta} works as before to get
\[\mathcal{Z}(k)= \sum_{m=0}^\infty \beta(k-\rho_0;m)\mathcal{R}(k+2m)\] for $k\in \mathbb{C}$ with $ \mathrm{Re}(k)>2\rho_0 $.


We compare now $\mathcal Z$ to the logarithmic derivative $L_{S,\chi}(s,\sigma)$ of the zeta function \begin{eqnarray*}
Z_{S, \mathbf{1}}(k,\mathbf 1)&:=&\prod_{1\neq [\gamma]\in C\Gamma}\prod_{r=0}^\infty \det\left(1-\left(S^r\left(\mathrm{Ad}(m_\gamma a_\gamma)|_{\overline{\mathfrak{ n}}}\right)\right)e^{-(k+\rho_0)l_\gamma}\right)
\end{eqnarray*} from \cite[Def. 3.2]{BO}, when  $\chi=\mathbf{1}$ is the trivial character of $\Gamma$, $\sigma=\mathbf{1}$ the trivial representation of $M$ and $S^r$ denotes the $r^{th}$ symmetric power of an endomorphism. By Lemma 3.3. in \cite{BO} $Z_{S,\mathbf{1}}(s,\sigma)$ converges for $\mathrm{Re}(k)>\rho_0$ and $L_{S,\mathbf{1}}(s,\mathbf{1})$ is given by
\begin{equation}\label{eq: clzeta2}
L_{S,\mathbf{1}}(k,\mathbf{1}):=\frac{d}{dk}\ln Z_{S,\mathbf{1}}(k,\mathbf{1})=2\cdot\sum_{1\neq [\gamma]\in C\Gamma}\frac{l_\gamma}{n_\gamma}\frac{(-1)^{l-1}}{\det\left(1-\mathrm{Ad}(m_\gamma a_\gamma)|_{\mathfrak n}\right)}e^{(\rho_0-k)l_\gamma}.
\end{equation}

Recall that $\gamma^{n_\gamma}_0=\gamma$, that is $l_\gamma/n_\gamma=l_{\gamma_0}$. If we compare (\ref{eq: clzeta1}) with (\ref{eq: clzeta2}) we see using \[\det\left(1-\mathrm{Ad}(a_t)|_{\mathfrak{n}}\right)=e^{2\rho_0 t}\det\left(1-\mathrm{Ad}(a_{-t})|_\mathfrak{n}\right)\] that up to the constant $2(-1)^{l-1}$, $\mathcal{Z}(k)$ equals 
\begin{equation}\label{eq: zetacon}
Z_{S,\mathbf 1}(k-\rho_0,\mathbf 1),
\end{equation} in the special case that for all $\gamma\in \Gamma$, $m_\gamma$ is trivial. In particular for $l=2$, i.e. $G=SO_0(1,2)$. More general let us define \index{$c(a_\gamma m_\gamma)$, coefficient}
\[c(a_\gamma m_\gamma):=\frac{l_{\gamma_0}}{\det\left(1-\mathrm{Ad}(a_\gamma m_\gamma)^{-1}|_{\mathfrak n}\right)}. \]

We recall that for $\gamma\in \Gamma$, $l_\gamma=\sqrt{2(l-1)}l_\gamma$, thus  
\[\mathcal{Z}(k)=(\sqrt{2(l-1)})^{-1}\sum_{1\neq [\gamma]\in C\Gamma}c(a_\gamma m_\gamma)e^{-kL_\gamma}\] 
and
\[Z_{S,\mathbf 1}(k+\rho_0,\mathbf 1)=2(-1)^{l-1}\cdot\sum_{1\neq [\gamma]\in C\Gamma}c(a_\gamma^{-1}m_\gamma^{-1})e^{-kl_\gamma}.\]


\chapter{The spectral trace}\label{chap: strace}
In this chapter we want to compute the spectral trace of the trace class operator $\varphi\cdot \pi_R(f_k)$ from Chapter \ref{chap: gtrace}. We obtain a general formula involving Wigner distributions which is valid for all rank one symmetric spaces, see Theorem \ref{thm: strac} and a refined version thereof which involves Patterson-Sullivan distributions in Theorem \ref{thm: stracspec}. A careful analysis of this refined trace formula will show that we can define a meromorphic continuation of $\mathcal R(\varphi)$ and $\mathcal Z(\varphi)$ by Theorem \ref{thm: stracspec}. This will be done in the next chapter.

\section{First computations on the spectral trace}

Let $X=G/K$ be a symmetric space of noncompact type of rank one, $G=NAK$, $M=Z_K(A)$, $\Gamma\subset G$ a uniform lattice and \index{$B$, boundary of $X$} $B=K/M$ the boundary of $X=G/K$. We denote by $SX_\Gamma$ the unit sphere bundle of $X_\Gamma=\Gamma\backslash X$. It can be identified with $\Gamma\backslash G/M$, see Lemma \ref{eq: sphbund}. We also identify the dual \index{$\mathfrak a_\mathbb{C}^*$, complex dual of $\mathfrak a$} $\mathfrak a_{\mathbb C}^*$ of the complexification $\mathfrak a_{\mathbb{C}}=\mathfrak{a}\otimes \mathbb{C}$ with $\mathbb C$ by sending a linear functional $\lambda$ on $\mathfrak a_\mathbb{C}$ to $\lambda(H_0)$, where $H_0$ is defined in (\ref{def: h0}). By $\Omega$ we denote the Casimir operator as in (\ref{def: casimir}).
 
Further, let $\{\varphi_j\}_{j\in\mathbb{N}}$ be an orthonormal basis of $L^2(X_\Gamma)$ of automorphic eigenfunctions, in particular  \[\Omega\varphi_j=-(\lambda_j^2+\rho_0^2)\varphi_j \mbox{ , } j\in\mathbb N_0\] for $\lambda_j\in \mathbb{C}$ and $\varphi_j(\gamma g)=\varphi_j(g)$ for $\gamma\in\Gamma$. Then one can assume that 
\begin{equation*}
0=\lambda_0^2+\rho_0^2\leq \lambda_1^2+\rho_0^2\leq \cdots,
\end{equation*}
see \cite[(12.12)]{Wi}. Hence, there are the two cases $\lambda_j\in \mathbb{R}$ or $\lambda_j\in i\mathbb R$ and there are only finitely many $\lambda_j$ in the latter case, see \cite[Cor. 12.10]{Wi}. Thus, when we consider asymptotics, we can assume that $\lambda_j\in \mathbb{R}$. If $\varphi_j$ is an eigenfunction with $\lambda_j\in\mathbb R$, we say that $\varphi_j$ or $\lambda_j$ lies in the \index{principal series} principal series, otherwise we say that $\varphi_j$ resp. $\lambda_j$ is in the \index{complementary series} complementary series. Let \index{$j_0$, index seperating principal and complementary series} $j_0\in \mathbb{N}$ be such that $\varphi_j$ is in the complementary series iff $j\leq j_0$. We make the assumption that $\lambda_{j_0}$ possibly equals 0, while $\lambda_j\neq 0$ for all $j\neq j_0$.

We recall that we defined the horocycle bracket by \[\langle gK, kM \rangle=A(k^{-1}g),\] see (\ref{def: horobrac}), where $A:G\to \mathfrak{a}$ is the projection belonging to $G=NAK$.
One can show that there is for any automorphic eigenfunction $\varphi$ with eigenvalue $-(\lambda^2+\rho_0^2)$ some \index{$T_\varphi$, boundary value} $T_\varphi$ in  the dual \index{$\mathcal D'(B)$, distributions on $B$} $\mathcal{D}'(B)$ of $C^\infty(B)$,  such that for $z\in X_\Gamma$ \begin{equation}\label{def: boundv}\varphi(z)=\langle e^{(i\lambda+\rho)\langle z,\cdot\rangle},T_\varphi\rangle_{B}=:\int_Be^{(i\lambda+\rho)\langle z,b\rangle}dT_\varphi(b), \end{equation}
see \cite{HS} or \cite[Ch.II \S 4 C(iii)]{GGA}. $T_\varphi$ is unique up to Weyl group action. That is, if we assume $\mathrm{Re}(\lambda)\geq 0$, then $T_\varphi$ is unique. We will always assume that $\mathrm{Re}(\lambda_j)\geq 0$ for all eigenvalues $\mu_j$ and call \index{boundary value} $T_\varphi$ \textit{boundary value} of $\varphi$.

We state an analogue of Proposition 2.10 in \cite{Z} for the special case of an automorphic eigenfunction. It shows that the operator $\pi_R(f)$, defined in (\ref{def: fourierr}), is diagonalized by the orthonormal basis $\{\varphi_j\}_{j\in \mathbb{N}}$. Finally, we recall the definitions of the spherical transform $\mathcal S(f,\cdot)$ in (\ref{def: sphericaltransform})
. 
\begin{proposition}\label{prop:1}
$\pi_R(f)\varphi_j=\mathcal{S}(f,-\lambda_j)\varphi_j$, where $f\in L^1(G//K)$. 
\end{proposition}
\begin{proof}
\cite[Th.7,Cor.8 \text{and following remark}]{Ga2} \label{thm:1}
\end{proof}

By the invariance of the spherical transform under the Weyl group we also have $\mathcal{S}(f,-\lambda_j)=\mathcal{S}(f,\lambda_j)$ for all $j$.

Let us fix an automorphic eigenfunction $\varphi_k$ from the orthonormal basis $\{\varphi_j\}$ as we did in Section \ref{sec: specialization}. 
We suppose in addition that $f\in C^ \infty(G//K)$ is such that $\pi_R(f)$ is of trace class, for example $f=f_k$ from Section \ref{sec: zeta}. We then conclude that, as for 
Proposition \ref{prop: traceclass}, since $\varphi_k$ is bounded, that $\varphi_k\cdot \pi_R(f)$ is also of trace class and maps $L^2(\Gamma\backslash X)$ into itself. 

In contrast to Chapter \ref{chap: gtrace}, where we computed the trace by integrating a kernel for $\varphi_k\cdot \pi_R(f)$ over $\Gamma\backslash X$ to get the trace, we now sum matrix coefficients which we build from pairing $\varphi_k\cdot \pi_R(f)$ with the orthonormal basis $\{\varphi_j\}$. Explicitly, 
\begin{eqnarray}\label{eqn: strace}
\mbox{Tr} (\varphi_k\cdot\pi_R(f)) & = & \underset{j}\sum \langle\varphi_k\cdot\pi_R(f)\varphi_j,\varphi_j\rangle_{L^2(X_\Gamma)}\nonumber\\
& \overset{\text{Prop.} \ref{prop:1}}=& \underset{j}\sum \langle\varphi_k \mathcal{S}(f,-\lambda_j)\varphi_j,\varphi_j \rangle_{L^2(X_\Gamma)}\nonumber\\
&=& \underset{j}\sum \langle\varphi_k\varphi_j,\varphi_j\rangle_{L^2(X_\Gamma)}\mathcal{S}(f,-\lambda_j \rangle.
\end{eqnarray}

Next we define a map $\mathrm{Op}$ from $C(X_\Gamma)$ into the bounded operators on $L^2(X_\Gamma)$ by \begin{equation}\label{def: opspec}
\mathrm{Op}(a)\varphi_j:=a\cdot \varphi_j.
\end{equation}

Then we can infer from (\ref{eqn: strace})
\begin{eqnarray*}
\mbox{Tr} (\varphi_k\cdot\pi_R(f))&=&\underset{j}\sum \langle\varphi_k\varphi_j,\varphi_j\rangle_{L^2(X_\Gamma)}\mathcal{S}(f,-\lambda_j).\\&=& \underset{j}\sum \langle \mathrm{Op}(\varphi_k)\varphi_j,\varphi_j \rangle_{L^2(X_\Gamma)} \mathcal{S}(f,-\lambda_j).
\end{eqnarray*}

\begin{remark}
This definition of $\mathrm{Op}$ is a special instance of a $\psi DO$-calculus for symbols from \cite[Def. 4.14.]{Sch}, that is, a map $\mathrm{Op}$ from symbols $a\in C^\infty(X\times B)$ into continuous operators from \[C^\infty(X)_c\to C^\infty(X),\] resp. from the dual of $C^\infty(X)$ \index{$\mathcal{D}'(X)$, distributions on $X$} \index{$\mathcal E'(X)$, dual of $C^\infty(X)$} \[\mathcal{E}'(X):=(C^\infty(X))'\to \mathcal{D}'(X):=(C_c^\infty(X))'\] into the dual of $C_c^\infty(X)$,  see \cite[Th. 4.15.]{Sch}. For $a\in C_c^\infty(X\times B)$, $\mathrm{Op}(a)$ is given by
\begin{equation}\label{def: generalop}
\mathrm{Op}(a)u(x)= \frac{1}{2}\int_{\mathfrak{a}_+^*}\int_{K/M}e^{(i\lambda+\rho)\langle z,b\rangle}\mathcal{F}(u,\lambda,b)a(x,b)db |c(\lambda)|^{-2}d\lambda,
\end{equation}
where $\mathcal{F}(u,\lambda,b)$ is Helgason's non-euclidean Fouriertransform, see (\ref{def: fouriertransf}), 
$c(\lambda)$ is Harish-Chandra's 
$c$-function and we identified $G/M=X\times K/M$, i.e. $a(z)=a(x,b)$ for $z\in G/M$, $x\in X$, $b\in K/M$. Note that this integral is finite, as $\mathcal{F}(u,\lambda,b)$ is rapidly decreasing in $\lambda$, see \cite[Ch.III Th.5.1.]{GASS}. Furthermore, it satisfies 
\begin{equation}\label{eq: opae}
\mathrm{Op}(a)e^{(i\lambda+\rho)\langle z,b\rangle}=a(z,b)e^{(i\lambda+\rho)\langle z,b\rangle}
\end{equation} 
for $a\in C(G/M)$, see \cite[(4.2)]{HS} or \cite[(4.18)]{Sch}.

The linear map \index{Wigner distribution} $a\mapsto \langle \mathrm{Op}(a)\varphi_j,\varphi_j\rangle_{L^2(X_\Gamma)}$ is called \index{Wigner distribution} the \textit{Wigner distribution} associated with $\varphi_j$ on $C^\infty(SX_\Gamma)$.

With (\ref{eq: opae}) we also see that the definition of $\mathrm{Op}$ from \cite[Def. 4.14.]{Sch} extends (\ref{def: opspec}), since
\begin{eqnarray*}\label{eq: oponb}
\mathrm{Op}(\varphi_k)\varphi_j(x)&=& \mathrm{Op}(\varphi_k)\int_B e^{(i\lambda+\rho)\langle x,b\rangle}dT_{\varphi_j}(b) \mbox{ , } T_{\varphi_j}\in\mathcal D'(B) \mbox{ boundary value, see (\ref{def: boundv})} \nonumber\\ &=& \int_B \mathrm{Op}(\varphi_k)e^{(i\lambda+\rho)\langle x,b\rangle}dT_{\varphi_j}(b)\mbox{ , see \cite[(4.7)]{HS}} \nonumber\\ &=& \int_B\varphi_k(x)e^{(i\lambda+\rho)\langle x,b\rangle}dT_{\varphi_j}(b)\\&& \mbox{ with } \mathrm{Op}(a)e^{(i\lambda+\rho)\langle x,b\rangle}=a(x,\lambda,b)e^{(i\lambda+\rho)\langle x,b\rangle}, \mbox{ see  \cite[(4.18)]{Sch}}\nonumber\\ &=&\varphi_k(x)\int_Be^{(i\lambda+\rho)\langle x,b\rangle} dT_{\varphi_j}(b) =\varphi_k(x)\varphi_j(x).
\end{eqnarray*}
for $x\in X$.
\end{remark}

To sum up, we have the following.

\begin{proposition}\label{thm: strac}
Let $X=G/K$ be a noncompact symmetric space of rank one and $\Gamma\subset G$ a uniform lattice. Further, we fix a orthonormal basis $\{\varphi_j\}_j$ of $L^2(\Gamma\backslash X)$ of automorphic Laplace-eigenfunctions with eigenvalues \[-\left(\lambda_j^2+\rho_0^2\right)\] and choose $\varphi_k$ in $\{\varphi_j\}_j$. Finally, let $f\in C^\infty(G//K)$ be such that $\pi_R(f)$ is of trace class, where $\pi_R$ is the right-regular representation of $G$ on $L^2(\Gamma\backslash G)$. For example $f\in C^\infty_c(G//K)$ suffices. Then $\varphi_k\cdot \pi_R(f)$ is also of trace class with trace given by
\begin{equation}\label{eq: spectracefi}
\mathrm{Tr}(\varphi_k\cdot \pi_R(f))=\sum_j \langle \mathrm{Op}(\varphi_k)\varphi_j,\varphi_j\rangle_{L^2(SX_\Gamma)}\mathcal{S}(f,\lambda_j).
\end{equation}

See (\ref{def: opspec}) for the definition of $\mathrm{Op}(\varphi_k)$. 
\end{proposition}

We call (\ref{eq: spectracefi}) also the \index{spectral trace} spectral trace of the operator $\varphi_k\cdot \pi_R(f)$.

\section{From Wigner to Patterson-Sullivan distributions}\label{conn}

In this section $G/K$ is at first an arbitrary noncompact symmetric space of rank one. From  Theorem \ref{thm: diffeqex} on we will  specialize to real hyperbolic spaces. The next step then is to relate $\langle \mathrm{Op}(\varphi_k)\varphi_j,\varphi_j\rangle$ to the Patterson-Sullivan distributions. 

Recall that $B=K/M$. We consider the diagonal action of of $G$ on \index{$B^{(2)}$, $B^{(2)}=(B\times B)-\Delta(B)$} \[B^{(2)}:=(B\times B)-\Delta(B),\] where $\Delta(B)$ is the diagonal in $B\times B$. One can show that this action is transitive and that the stabilizer of $(M,wM)$, where $w$ is the non-trivial Weyl group element, is $MA$, \cite[Prop. 2.4.]{HS}. Then we identify $B^{(2)}$ with $G/MA$ and write $g(b,b')$ for an element in $G$ with $g(b,b')\cdot (M,wM)=(g(b,b')\cdot M,g(b,b')\cdot wM)=(b,b')$.

We now recall some definitions from \cite{HHS}, \cite{HS} and \cite{Sch}. For a function $f$ on $G/M$ the \index{Radon transform} Radon transform on $G/M$ is given by 
\begin{equation*} 
\mathcal Rf(b,b'):=\int_Af(g(b,b')aM)da,
\end{equation*}
see \cite[Def. 4.3]{HHS}. This is well-defined by the unimodularity of $A$ and maps $C_c^\infty(G/M)$ into $C_c^\infty(G/MA)$, see \cite[Lem. 4.4]{HHS}. Next we define the so-called \index{intertwining operators} \textit{intertwining operators}.
For $\lambda\in\mathfrak{a}_{\mathbb{C}}^*$ and $a\in C(G)$ we set \index{$L_\lambda$, intertwiner} \begin{equation*} L_\lambda a(g):=\int_N e^{-(i\lambda+\rho)H(n^{-1}w)}a(gn)dn,\end{equation*} see \cite[(1.3)]{HS}, \cite[(6.35)]{Sch}, whenever the integral exists, where $H(g^{-1}k)=-A(k^{-1}g)$ and $w$ is the non-trivial Weyl group element. One can show that $L_\lambda$ maps $C_c^\infty(G/M)$ into $C^\infty(G/M)$, \cite[Lem. 6.30.]{Sch}.

We continue with the notion of intermediate values.
For $\lambda\in \mathfrak{a}^*_{\mathbb C}$ we define \index{intermediate value} \index{$d_\lambda(gMA)$, intermediate value} the \textit{intermediate value} on $G/MA$
\begin{equation*} d_\lambda(gMA):=e^{(i\lambda+\rho)(H(g)+H(gw))},\end{equation*}
see \cite[Def. 4.1]{HHS}. Then we are in the position to define \textit{Patterson-Sullivan distributions}. We start with distributions on $G/M$.
\begin{definition}\cite[Def. 4.8]{HHS}, \cite[Def. 6.12]{Sch}\\
\index{Patterson-Sullivan distributions} \index{$\mathrm{PS}_{\varphi}$, Patterson-Sullivan distributions} For an automorphic Laplace-eigenfunction $\varphi$ with eigenvalue \[-\left(\lambda^2+\rho_0^2\right)\] the Patterson-Sullivan distribution $\mathrm{PS}_{\varphi}$ on $G/M$ is given by \begin{equation*} \int_{G/M} a(gM) d\mathrm{PS}_{\varphi}(gM):=\langle a,\mathrm{PS}_{\varphi} \rangle:=\int_{B^{(2)}}\mathcal{R}a(b,b')d_{\lambda}(b,b')dT_{\varphi}(b)dT_{\varphi}(b'),\end{equation*} 
where $T_\varphi$ is the boundary value of $\varphi$.
\end{definition}

We call $\chi\in C_c^\infty(G/M)$ \index{$chi$@$\chi$, cut-off} \index{smooth fundamental domain cut-off} a \textit{smooth fundamental domain cut-off}, if \begin{equation*} \underset{\gamma\in\Gamma}\sum \chi(\gamma g)=1\end{equation*} for all $g\in G$. Smooth fundamental domain cut-off functions exist. If for example $f\in C_c^\infty(G/M)$ is such that $f$ equals 1 on a fundamental domain for $\Gamma$, then $\chi(g)=\frac{f(g)}{\sum_{\gamma\in\Gamma}f(\gamma g)}$ is a smooth cut-off.

We have the following two Theorems, \ref{thm: inter} and \ref{thm: 2}, which are generalizations of Theorem 1.1 and Lemma 6.4 in \cite{AZ}.
\begin{theorem}\label{thm: inter}
\cite[Theorem 1.1]{AZ}, \cite[Theorem 1.1]{HS}, \cite[Theorem 6.40.]{Sch} \\ Let $a\in C^\infty(SX_\Gamma)$. Then \begin{equation*} \langle \mathrm{Op}(a)\varphi_{\lambda},\varphi_{\lambda}\rangle_{L^2(X_\Gamma)} =\langle L_{\lambda}(\chi a),\mathrm{PS}_{\varphi}\rangle_{G/M} ,\end{equation*} where $\chi$ is an arbitrary smooth fundamental domain cut-off.
\end{theorem}

We will use this theorem to extend Patterson-Sullivan distributions $\mathrm{PS}_\varphi$ to the range $L_\lambda(C_c^\infty(G/M))$ of the intertwiner $L_\lambda$. The next lemma allows us to define Patterson-Sullivan distributions on $\Gamma\backslash G/M$.
 
\begin{lemma}\label{lem: 3}\cite[4.12]{HHS}, \cite[3.5]{AZ}\\
Let $T\in \mathcal D'(\Gamma\backslash G/M)$, $a\in C^\infty(\Gamma\backslash G/M)$ and $a_1,a_2\in C_c^\infty(G/M)$ with $\underset{\gamma\in\Gamma}\sum a_i(\gamma gM)=a(gM)$ for $i=1,2$. Then $\langle a_1,T\rangle_{G/M}=\langle a_2,T\rangle_{G/M}$.
\end{lemma}
\begin{proof}
We have
\begin{eqnarray*}
\langle a_i,T \rangle_{G/M} &=& \int_{G/M}a_i(gM)dT(gM)\\
&=& \int_{G/M}\underset{\gamma\in\Gamma}\sum \chi(\gamma gM) a_i(gM)dT(gM)\\ &\overset{gM\mapsto \gamma^{-1}gM}=& \int_{G/M}\underset{\gamma\in\Gamma}\sum \chi(gM)a_i(\gamma^{-1}gM)dT(\gamma^{-1} gM) \\ &=& \int_{G/M}\chi(gM)\underset{\gamma\in\Gamma}\sum a_i(\gamma^{-1}gM)dT(gM)\\ &=& \int_{G/M}\chi(gM)\underset{\gamma\in\Gamma}\sum a_i(\gamma gM)dT(gM)\\&=&\int_{G/M}\chi(gM)a(gM)dT(gM)
\end{eqnarray*}
\end{proof}

Let now $\chi_1,\chi_2\in C_c^\infty(G/M)$ with $\underset{\gamma\in\Gamma}\sum \chi_i(\gamma gM)=1$ for all $g\in G$ and $a$ in $C^\infty(\Gamma\backslash G/M)$. Set $a_i=a\chi_i$, $(i=1,2)$. Then \[\underset{\gamma\in \Gamma}\sum a_i(\gamma gM)=\underset{\gamma\in\Gamma}\sum a(\gamma gM)\chi_i(\gamma gM)=a(gM)\underset{\gamma\in\Gamma}\sum \chi_i(\gamma gM)=a(gM)\] and hence, see also \cite[Lemma 3.5]{AZ}, \begin{equation}\label{eq: independent} \langle a\chi_1,T\rangle_{G/M}=\langle a\chi_2,T\rangle_{G/M}.\end{equation}

One can show that the Patterson-Sullivan distributions are $\Gamma$-invariant on $G/M$, i.e. \[\langle f\circ\gamma  ,\mathrm{PS}_\varphi\rangle_{G/M}=\langle f,\mathrm{PS}_\varphi\rangle_{G/M}\] for all $\gamma\in\Gamma$ and $f\in C_c^\infty(G/M)$, \cite[Prop.4.10]{HHS} or \cite[Prop.6.13.]{Sch}. Here $f\circ \gamma(g)=f(\gamma g)$ for all $g\in G$. Thus they descend to the quotient $\Gamma\backslash G/M$ and define there elements of $\mathcal D'(\Gamma\backslash G/M)$, see \cite[Prop.4.9]{HHS}.
\begin{definition}\label{def: psgamma}\cite[Def.4.13]{HHS}, \cite[Def.6.19.]{Sch}\\
Let $\varphi$ be an automorphic eigenfunction. On $SX_\Gamma=\Gamma\backslash G/M $ the Patterson-Sullivan distributions \index{$\mathrm{PS}_{\varphi}$} $\mathrm{PS}_{\varphi}$ are given by \[\int_{\Gamma\backslash G/M}a(\Gamma gM)d\mathrm{PS}_{\varphi}(\Gamma gM):=\langle a,\mathrm{PS}_{\varphi}\rangle_{SX_\Gamma}:=\langle \chi a , \mathrm{PS}_\varphi\rangle_{G/M},\]
where $ \chi$ is some smooth fundamental domain cut-off. 
\end{definition}

By (\ref{eq: independent}) this definition is independent of the choice of $\chi$.
\begin{remark}\label{rem: pscontd}
In \cite[Prop. 4.9]{HHS} it is shown that the Patterson-Sullivan distributions are continuous distributions on $C^\infty(G/M)$ resp. $C^\infty(\Gamma\backslash G/M)$. That is, there exist a constant $C>0$ depending only on $\varphi$ and a seminorm \index{$||.||$ seminorm on $C^\infty(G/M)$} $||.||$ independent of $\varphi$  such that 
$$|\langle \chi f,\mathrm{PS}_\varphi\rangle_{SX_\Gamma}|\leq C||\chi f||$$ for all $f\in C^\infty(G/M)$. 
\end{remark}

 We state now a lemma which connects the Patterson-Sullivan distributions on $G/M$ with those on $\Gamma\backslash G/M$.
\begin{lemma}\label{lem: pscon}
Let $f\in C_c^\infty(G/M)$, then \[ \langle f,\mathrm{PS_\varphi}\rangle_{G/M}=\left\langle \sum_{\gamma\in\Gamma}  f\circ\gamma ,\mathrm{PS}_\varphi\right\rangle_{\Gamma\backslash G/M}.\]
\end{lemma}
\begin{proof}
This follows since $\mathrm{PS}_\varphi$ is $\Gamma$-invariant. By Definition \ref{def: psgamma} we have for any smooth fundamental domain cut-off $\chi$ and any $f\in C_c^\infty(G/M)$
\begin{eqnarray*}
\left\langle \sum_{\gamma\in\Gamma} f\circ \gamma, \mathrm{PS}\right\rangle_{\Gamma\backslash G/M}&=& \left\langle \chi\cdot \sum_{\gamma\in \Gamma}f \circ \gamma,\mathrm{PS}_\varphi\right\rangle_{G/M}\\&=& 
\left\langle  \sum_{\gamma\in \Gamma}\chi\cdot (f \circ\gamma ),\mathrm{PS}_\varphi\right\rangle_{G/M}
\\&\overset{\Gamma-\text{invariance}}=&
\left\langle f\cdot \sum_{\gamma\in\Gamma} \chi\circ\gamma^{-1} ,\mathrm{PS}_\varphi\right\rangle_{G/M}\\&=&
\langle f,\mathrm{PS}_\varphi\rangle_{G/M}, 
\end{eqnarray*}
as $\sum_{\gamma\in\Gamma}\chi\circ\gamma^{-1} =1$, since $\chi$ is a fundamental domain cut-off.
\end{proof} 

Let $a\in C^\infty(\Gamma\backslash G/M)$ and let \index{$piM$@$\pi_M$, projection $C(G)\to C(G/M)$} $\pi_M$ denote the projection from $C(G)\to C(G/M)$ given by \[\pi_M(f)(g):=\int_M f(gm)dm.\] 

Further, we write \index{${}^n$, $f^n=f\circ r_n$} $f^n:=f\circ r_n$ for $n\in N$ for functions $f$, i.e. $f^{n}(g)=f(gn)$ for $g\in G$. Putting $a=\varphi_k$ and applying the theory of Chapter \ref{chap: ode} we obtain:

\begin{theorem}\label{thm: diffeqex}
Let $G=SO_o(1,l)=ANK$, $M=Z_K(A)$, $X_{e_1}\in\mathfrak n$ as for (\ref{eq: normx1}) and $\Gamma\subset G$ a uniform lattice. Further, let $\varphi_k$ be an automorphic Laplace eigenfunction on $\Gamma\backslash G/K$, then $n\mapsto \langle \pi_M(\varphi_k^n),\mathrm{PS}_{\varphi_j}\rangle_{SX_\Gamma}$ is bi-$M$-invariant and a  solution to \[\Omega f=-\left(\lambda_k^2+\rho_0^2\right)f=-\frac{1}{4\rho_0}(\rho_0^2+r_k^2)f\] on $N$. On the slice $\exp \mathbb{R}^+X_{e_1}$ it satisfies for any $j$ and $k$ for $l>2$
\begin{equation*}
\langle \pi_M(\varphi_k^{\exp sX_{e_1}}),\mathrm{PS}_{\varphi_j}\rangle_{SX_\Gamma}=\langle \varphi_k,\mathrm{PS}_{\varphi_j}\rangle_{SX_\Gamma}\cdot {}_2F_1\left(a,b,\rho_0;-s^2\right),
\end{equation*}
where $a=\frac{1}{2}(\rho_0+ir_k)$ and $b=\frac{1}{2}(\rho_0-ir_k)$ as in Section \ref{sec: hyppr}. For $l=2$ it is given by
\begin{eqnarray*}
\langle \pi_M(\varphi_k^{\exp sX_{e_1}}),\mathrm{PS}_{\varphi_j}\rangle_{SX_\Gamma}&=& \langle \varphi_k,\mathrm{PS}_{\varphi_j}\rangle_{SX_\Gamma}\cdot {}_2F_1\left(a,b,\frac{1}{2} ;-s^2\right)\\&&+2i\langle X_1\varphi_k,\mathrm{PS}_{\varphi_j}\rangle_{SX_\Gamma}\cdot s\cdot {}_2F_1\left(a+\frac{1}{2},b+\frac{1}{2}, \frac{3}{2};-s^2 \right).
\end{eqnarray*}
\end{theorem}
\begin{proof}
As $\Omega$ and translations commute \[\Omega \varphi_k^n=-\frac{1}{4\rho_0}(\rho_0^2+r_k^2)\varphi_k^n\] for all $n\in N$. 
If we express $\varphi_k^n$ in local charts  of the compact manifold $SX_\Gamma=\Gamma\backslash G/M$ as $\varphi_k^n(x_l)$, then all partial derivatives of $ \frac{\partial^\alpha}{\partial x_l^\alpha} \varphi_k^n(x_l)$ in these local coordinates $x_l$ are simultaneously continuous in $x_l$ and $n$.
Thus, we can interchange the distribution with differentiation using \cite[Ch.IV Th.II]{Schw} to get 
\begin{eqnarray*}&&\Omega\left(n\mapsto \int_{\Gamma\backslash G/M}\int_M \varphi^n_k(\Gamma gm)dm \mathrm{PS}_{\varphi_j}(\Gamma gM)\right)\\&=& \int_{\Gamma\backslash G/M}\int_M \Omega\left(n\mapsto \varphi^n_k(\Gamma gm)\right)dm \mathrm{PS}_{\varphi_j}(\Gamma gM)\\ &=&n\mapsto -\frac{1}{4\rho_0}(\rho_0^2+r_k^2)\int_{\Gamma\backslash G/M}\int_M \varphi^n_k(\Gamma gm)dm \mathrm{PS}_{\varphi_j}(\Gamma gM).
\end{eqnarray*} 

Obviously, for all $m_1,m_2\in M$
\begin{eqnarray*}\int_{\Gamma\backslash G/M}\int_M \varphi^{m_1nm_2}_k(\Gamma gm)dm \mathrm{PS}_{\varphi_j}(\Gamma gM)&=& \int_{\Gamma\backslash G/M}\int_M \varphi_k(\Gamma gmm_1nm_2)dm \mathrm{PS}_{\varphi_j}(\Gamma gM)\\&=&
\int_{\Gamma\backslash G/M}\int_M \varphi_k(\Gamma gmn)dm \mathrm{PS}_{\varphi_j}(\Gamma gM) 
\end{eqnarray*}
by unimodularity of $M$ and right-$K$-invariance of $\varphi_k$. Hence, \[F(n):=\int_{\Gamma\backslash G/M}\int_M \varphi^n_k(\Gamma gm)dm \mathrm{PS}_{\varphi_j}(\Gamma gM)=\int_{\Gamma\backslash G/M}\int_M \varphi_k(\Gamma gmn)dm \mathrm{PS}_{\varphi_j}(\Gamma gM) \] is a bi-$M$-invariant solution on $N$, in particular at the origin $n=e$, to \begin{equation}\label{eq: minv}\Omega F+\mu_k F=0,\end{equation} where $\mu_k=\frac{1}{4\rho_0}(\rho_0^2+r_k^2)$ and we can restrict $F$ to the slice $S=\exp \mathbb{R}^+X_{1}=\exp \mathbb{R}^+X_{e_1}$. 
By Theorem \ref{th: odesum} we know that on the slice $S$ equation (\ref{eq: minv}) reads for bi-$M$-invariant functions $F$, see equation (\ref{e1}),
\begin{eqnarray*}
\left((\alpha(H_1)^2s^2+2)\frac{d^2}{ds^2}+\left(\left(\alpha(H_1)^2+2\alpha(H_\rho)\right)s+2\frac{l-2}{s}\right)\frac{d}{ds}+\mu_k\right)F(s)=0.
\end{eqnarray*}

For $l>3$, by Theorem \ref{chap: ode main}, the restriction of the function $F(n)$ to the slice is then given by
\begin{eqnarray*}
F(s):=F(\exp sX_{1})
&=&
F(0)\cdot {}_2F_1\left(a,b,\rho_0;\frac{-s^2}{4(l-1)}\right)\\
&\overset{\mathrm{def. of } F}=& \langle \varphi_k^{\exp 0X_1},\mathrm{PS}_{\varphi_j}\rangle_{SX_\Gamma}\cdot {}_2F_1\left(a,b,\rho_0;\frac{-s^2}{4(l-1)}\right)
\\&=&\langle \varphi_k,\mathrm{PS}_{\varphi_j}\rangle_{SX_\Gamma}\cdot {}_2F_1\left(a,b,\rho_0;\frac{-s^2}{4(l-1)}\right)
\end{eqnarray*}
resp.
\[F(\exp sX_{e_1})=\langle \varphi_k,\mathrm{PS}_{\varphi_j}\rangle_{SX_\Gamma}\cdot {}_2F_1\left(a,b,\rho_0;-s^2\right).\]

The claim for $l=2$ follows similarly.
\end{proof}

\begin{lemma}\label{lem: conec}
For $a\in C^\infty(\Gamma\backslash G/M)$ and any fundamental domain cut-off $\chi$ we have \[\langle \pi_M(\chi a)^n,\mathrm{PS}_\varphi\rangle_{G/M}=\langle \pi_M(a^n),\mathrm{PS}_\varphi\rangle_{\Gamma\backslash G/M}.\]
\end{lemma} 
\begin{proof}
\begin{eqnarray*}
\langle \pi_M((\chi a)^n),\mathrm{PS}_\varphi \rangle_{G/M} dn
&=& \int_{G/M}\int_M (\chi a)(gmn)dm\mathrm{PS}_\varphi(gM)
\\&=&
\int_{\Gamma\backslash G/M} \sum_{\gamma\in\Gamma}\int_M(\chi a)(\gamma gmn)dm\mathrm{PS}_\varphi(\Gamma gM)\mbox{, as }\int_{G/M}=\int_{\Gamma\backslash G/M}\sum_{\gamma\in\Gamma},\\&& \mbox{see Lemma \ref{lem: pscon}}
\\
&=&\int_{\Gamma\backslash G/M}\int_M a(gmn)\sum_{\gamma\in\Gamma}\chi(\gamma gmn)dm\mathrm{PS}_\varphi(\Gamma gM)\mbox{, as } a \mbox{ is $\Gamma$-invariant}\\
&=&\int_{\Gamma\backslash G/M}\int_M a(gmn)dm \mathrm{PS}_\varphi(\Gamma gM)\mbox{ , as $\chi$ is a fundamental domain cut-off} 
\\&=& \int_{\Gamma\backslash G/M} \pi_M(a^n)(g)\mathrm{PS}_\varphi(\Gamma gM)\\
&=& \langle \pi_M(a^n),\mathrm{PS}_\varphi\rangle_{SX_\Gamma}. 
\end{eqnarray*}
\end{proof}

\begin{proposition}\label{prop: absconv}
Let $\varphi_j$ be an automorphic eigenfunction in the principal series and $\varphi_k$ any automorphic eigenfunction with eigenvalues $-\left(\rho_0^2+\lambda_j^2\right)$ resp. $-\frac{1}{4\rho_0}\left(\rho_0^2+r_k^2\right)$. Then \[\int_N e^{-(i\lambda_j+\rho)H(n^{-1}w)}\langle \pi_M(\varphi_k)^n,\mathrm{PS}_{\varphi_j}\rangle_{SX_\Gamma}dn\] converges absolutely.
\end{proposition}
\begin{proof}
We recall that $G=SO_o(1,l)$ and $M\cong SO(l-1)$. First we note that the integrand $n\mapsto e^{-(i\lambda+\rho)H(n^{-1}w)} \langle \pi_M(\varphi_k^n),PS_\lambda\rangle_{SX_\Gamma}$ of (\ref{eq: umform}) is bi-$M$-invariant. We find then on the slice $S=\{\exp sX_{e_1}:s>0\}$ that   
\[e^{-(i\lambda_j+\rho)H(\exp -sX_{e_1}w)}=\left(1+s^2\right)^{-(i\lambda_j+\rho_0)},\] see (\ref{eq: Hnuw}). Further, by Theorem \ref{thm: diffeqex} for $l\geq 3$
\begin{equation*}
\langle \pi_M(\varphi_k^{\exp sX_{e_1}}),\mathrm{PS}_{\varphi_j}\rangle_{SX_\Gamma}=\langle \varphi_k,\mathrm{PS}_{\varphi_j}\rangle_{SX_\Gamma}\cdot {}_2F_1\left(a,b,\rho_0;-s^2\right)
\end{equation*}  
on the slice $S$. Hence, we can use for $l\geq 3$ polar coordinates and we get the following 
\begin{eqnarray*}
&&\int_N e^{-(i\lambda_j+\rho)H(n^{-1}w)} \langle \pi_M(\varphi_k^n),\mathrm{PS}_{\varphi_j}\rangle_{SX_\Gamma}dn\\&=& \omega_{l-1}\int_0^\infty s^{l-2}\left(1+s^2\right)^{-(i\lambda_j+\rho_0)} \langle \pi_M(\varphi_k^{\exp sX_{e_1}}),\mathrm{PS}_{\varphi_j}\rangle_{SX_\Gamma} ds\\
&\overset{\text{Prop. \ref{thm: diffeqex}}}=& \omega_{l-1}\int_0^\infty s^{l-2}\left(1+s^2\right)^{-(i\lambda_j+\rho_0)} {}_2F_1\left(a,b,\rho_0;-s^2\right)ds\cdot \langle \varphi_k,\mathrm{PS}_{\varphi_j}\rangle_{SX_\Gamma}=:(*).
\end{eqnarray*}

For $l=2$ we know that $N=\{\exp sX_{e_1}:s\in\mathbb{R}\}$ and by Theorem \ref{thm: diffeqex} with $a=\varphi_k$ we get
\begin{eqnarray*}
&&\int_N e^{-(i\lambda_j+\rho)H(n^{-1}w)} \langle \pi_M(\varphi_k^n),\mathrm{PS}_{\varphi_j}\rangle_{SX_\Gamma}dn\\&=& \int_{-\infty}^\infty \left(1+s^2\right)^{-(i\lambda_j+\rho_0)} \langle (\varphi_k^{\exp sX_{e_1}}),\mathrm{PS}_{\varphi_j}\rangle_{SX_\Gamma}ds\\&\overset{\text{Prop. \ref{thm: diffeqex}}}=&  \langle \varphi_k,\mathrm{PS}_{\varphi_j}\rangle_{SX_\Gamma}\cdot\int_{-\infty}^\infty \left(1+s^2\right)^{-(i\lambda_j+\rho_0)} {}_2F_1\left(a,b,\rho_0;-s^2\right)ds\\&&+2i\langle X_{e_1}\varphi_k,\mathrm{PS}_{\varphi_j}\rangle_{SX_\Gamma}\cdot \int_{-\infty}^\infty \left(1+s^2\right)^{-(i\lambda_j+\rho_0)} s\cdot {}_2F_1\left(a+\frac{1}{2},b+\frac{1}{2}, \frac{3}{2};-s^2 \right)ds 
\\&=& \langle \varphi_k,\mathrm{PS}_{\varphi_j}\rangle_{SX_\Gamma}\cdot\int_{-\infty}^\infty (1+s^2)^{-(i\lambda_j+\rho_0)} {}_2F_1\left(a,b,\rho_0;-s^2\right)ds=:(**).
\end{eqnarray*}


Here $a=\frac{1}{2}(\rho_0+ir_k)$ and $b=\frac{1}{2}(\rho_0-ir_k)$. We want to use Lemma \ref{thm: absconvofint} from the appendix to this chapter to investigate the convergence of $(*)$ resp. $(**)$. Thus, we have to check whether 
\begin{equation}\label{eq: convcheck}
2\mathrm{Re}(\rho_0+i\lambda_j)>\rho_0+\mathrm{Re}(ir_k).
\end{equation} 

As $\varphi_j$ is in the principal series $2\mathrm{Re}(\rho_0+i \lambda_j)=2\rho_0$. Because \[0\leq \frac{1}{4\rho_0}(\rho_0^2+r_k^2)\] either $r_k\in\mathbb R$ or $r_k\in i\mathbb R$. In any case $\mathrm{Re}(ir_k)\leq 0$ and $(\ref{eq: convcheck})$ is true. The integrals in $(*)$ resp. $(**)$ now converge absolutely by Theorem \ref{thm: absconvofint}.
\end{proof}

From now on we will always assume that $G=SO_o(1,l)$. In order to explain the connection between Wigner- and Patterson-Sullivan distributions we need the following continuity property of the latter ones.  
\begin{proposition}\label{thm: 2}
Let $a\in C^\infty(\Gamma\backslash G/M)$. Then 
\begin{equation*} \langle L_\lambda(\chi a),\mathrm{PS}_\varphi\rangle_{G/M}=\int_N e^{-(i\lambda+\rho)H(n^{-1}w)}\langle \pi_M\left((\chi a)^n\right),\mathrm{PS}_\varphi \rangle_{G/M} dn, \end{equation*} 
where $\chi\in C_c^\infty(G/M)$ is a smooth fundamental domain cut-off and $\varphi$ an automorphic eigenfunction with eigenvalue $-(\lambda^2+\rho_0^2)$.
\end{proposition}
\begin{proof}


By Theorem \ref{thm: inter} we know that $\langle L_\lambda(\chi a),\mathrm{PS}_\varphi\rangle_{G/M}$ is finite. Thus,
\begin{eqnarray*}
\langle L_\lambda(\chi a),\mathrm{PS}_\varphi\rangle &\overset{\text{def}}=& \int_{B^2} \mathcal{R}(L_\lambda(\chi a))(b,b')d_\lambda(b,b')dT_\varphi(b)dT_\varphi(b')\\ 
&\overset{\text{Lemma 5.15 in [HS]}}=&\int_{B^2} \int_X \chi a(z,b) e^{(i\lambda+\rho)(\langle z,b\rangle+\langle z,b'\rangle)}dz dT_\varphi(b)dT_\varphi(b')\\&=:&(*)
\end{eqnarray*}

We then identify $G/M=X\times B=G/K\times K/M$ and assume that $\chi(z,b)=\chi(z)$ is independent of $b$, i.e. $\chi\in C_c^\infty(G/K)$. Further we assume that $T_\varphi$ equals $D_b^\alpha F_\varphi$ for some differential operator $D_b^\alpha$ on $B$ and some $F_\varphi\in C(B)$ in the distributional sense, see \cite[Th. 1.3.]{GO}. That is, \[\langle f,T_\varphi\rangle_B=\int_B (D_b^\alpha f)(b)F_\varphi(b)db\] for all $f\in C^\infty(B)$. Then we also identify $X=AN$, $dx=dnda$. Thus,

\begin{eqnarray*}
(*)
&=& \int_{B^2}\int_{AN} \chi(an) D_b^\alpha D_{b'}^\alpha \left(a(an,b)e^{(i\lambda+\rho)(\langle an,b\rangle+\langle an,b'\rangle)}\right)F_\varphi(b)F_\varphi(b')dndadbdb'\nonumber\\
&\overset{\text{Fubini}}=& \int_N \int_{B^2} \int_A\chi(an)  D_b^\alpha D_{b'}^\alpha \left(a(an,b)e^{(i\lambda+\rho)(\langle an,b\rangle+\langle an,b'\rangle)}\right)F_\varphi(b)F_\varphi(b') dadbdb'dn\nonumber\\
&=&\int_N \int_{B^{(2)}} \int_A\chi(an)  D_b^\alpha D_{b'}^\alpha \left(a(an,b)e^{(i\lambda+\rho)(\langle an,b\rangle+\langle an,b'\rangle)}\right)F_\varphi(b)F_\varphi(b') dadbdb'dn\nonumber\\
&=& \int_N\int_{B^{(2)}} \int_A \chi(an)a(an,b)e^{(i\lambda+\rho)(\langle an,b\rangle+\langle an,b'\rangle)} dadT_\varphi(b)dT_\varphi(b')dn=:(**)\\
\end{eqnarray*}

We write $\chi a$ as a function on $G/M$, compare \cite[(6.15ff)]{HS}. For any $g\in G$ with $(g\cdot M,g\cdot wM)=(b,b')$ we have \[\chi a(gan\cdot o,b)=\chi a(gan\cdot o,g\cdot M)=\chi a(gan\cdot o,gan\cdot M)=\chi a(ganM).\]

Since $dx=dadn$ is $G$-invariant, $(**)$ equals 
\[\int_N\int_{G/MA} \int_A \chi a(ganM)e^{(i\lambda+\rho)(\langle gan\cdot o,g\cdot M\rangle+\langle gan\cdot o,g\cdot wM\rangle)} da dT_\varphi\otimes dT_\varphi(gMA)dn.\]

Then $\langle gan\cdot o,g\cdot M\rangle=\langle gan\cdot o,gan\cdot M\rangle=H(gan)=H(ga)$, since $g\cdot M=b$ and $P=MAN$, in particular $AN$ fixes $M$. 
Further, 
\begin{eqnarray*}
\langle gan\cdot o,g\cdot wM\rangle&=&-H(n^{-1}a^{-1}w)+H(gw)\\
&=&-H(n^{-1}w)+H(a)+H(gw)\\
&=&-H(n^{-1}w)+H(gaw),
\end{eqnarray*} 
since $\langle g\cdot x,g\cdot b\rangle=\langle x,b\rangle+\langle g\cdot o,g\cdot b\rangle$, $H(ga)=H(g)+H(a)$, see \cite[Ch.II (46)]{GASS}, for all $a\in A$ and $g\in G$.  Also $H(waw^{-1})=H(a^{-1})$ for all $a\in A$. Hence,
\begin{eqnarray*}
&&\int_N\int_{G/MA} \int_A \chi a(ganM)e^{(i\lambda+\rho)(\langle gan\cdot o,g\cdot M\rangle+\langle gan\cdot o,g\cdot wM\rangle)} da dT_\varphi\otimes dT_\varphi(gMA)dn\\&=& \int_N\int_{G/MA} \int_A \chi a(ganM)e^{(i\lambda+\rho)(H(ga)+H(gaw))}e^{-(i\lambda+\rho)H(n^{-1}w)}da dT_\varphi\otimes dT_\varphi(gMA)dn
\\&=& \int_{N} \int_{G/MA} \int_A \chi a(ganM) e^{-(i\lambda+\rho)H(n^{-1}w)}d_\lambda(gMA)dT_\varphi da\otimes dT_\varphi(gMA)dnda
\\&=& \int_N e^{-(i\lambda+\rho)H(n^{-1}w)}\int_{G/MA}d_\lambda(gMA) \int_A (\chi a)^n(gaM) dadT_\varphi\otimes dT_\varphi(gMA)dn\\
&=& \int_N e^{-(i\lambda+\rho)H(n^{-1}w)} \int_{G/M} \chi a(gnM) d\mathrm{PS}_\varphi(gM).
\end{eqnarray*}

Now $H$ is the projection belonging to $G=KAN$, which is bi-$M$-invariant, and $w$ normalizes $M$. Hence, for any $m\in M$ \[H(n^{-1}w)=H(n^{-1}ww^{-1}mw)=H(n^{-1}mw)=H(m^{-1}n^{-1}mw).\] 

Furthermore, $\chi a \in C_c^\infty(G/M)$ and $\int_M=1$, so
\begin{eqnarray*}
&&\int_N\int_{G/M} e^{-(i\lambda+\rho)H(n^{-1}w)}\chi a(gn)d \mathrm{PS}_\varphi(gM) dn\\&=& \int_N \int_{G/M} \int_M e^{-(i\lambda+\rho)H(m^{-1}n^{-1}mw)} \chi a (gnm)dm d \mathrm{PS}_\varphi(gM) dn. 
\end{eqnarray*} 

For every $m\in M$ the mapping $n\mapsto mnm^{-1}$ is an automorphism of $N$ fixing $dn$, thus
\begin{eqnarray*}
&&\int_N \int_{G/M} \int_M e^{-(i\lambda+\rho)H(m^{-1}n^{-1}mw)} \chi a (gnm)dm d \mathrm{PS}_\varphi(gM) dn \\&=& \int_N \int_{G/M} \int_M e^{-(i\lambda+\rho)H(n^{-1}w)} \chi a (gmn) dm d \mathrm{PS}_\varphi(gM) dn\\
&=& \int_N e^{-(i\lambda+\rho)H(n^{-1}w)} \int_{G/M}\int_M \chi a (gmn) dm d \mathrm{PS}_\varphi(gM) dn\\
&=& \int_N e^{-(i\lambda+\rho)H(n^{-1}w)}\langle \pi_M\left((\chi a)^n\right),\mathrm{PS}_\varphi\rangle_{G/M}dn.
\end{eqnarray*}

We note that $\pi_M((\chi a)^n)$ is indeed in $C_c^\infty(G/M)$ as $M$ is compact.
\end{proof}



In view of Lemma \ref{lem: conec} we also get  the following corollary.
\begin{corollary}\label{cor: inter}
Let $a\in C^\infty(\Gamma\backslash G/M)$. Then \begin{equation*} \langle L_\lambda(\chi a),\mathrm{PS}_\varphi\rangle_{G/M}=\int_N e^{-(i\lambda+\rho)H(n^{-1}w)} \langle \pi_M(a^n),\mathrm{PS}_\varphi\rangle_{SX_\Gamma}dn. 
\end{equation*}
\end{corollary}

Hence, combining this theorem, Corollary \ref{cor: inter} and Theorem \ref{thm: inter} we get an exact relation between Wigner- and Patterson-Sullivan distributions from the principal series on the level of non-constant automorphic eigenfunctions. 
For $k\in \mathbb{C}$ with $\mathrm{Re}(k)>\rho_0$ let us set \index{$I(a,b,\rho_0,k)$, normalizing constant}
\begin{equation}\label{eq: normcon}I(a,b,\rho_0,k):=\int_0^\infty s^{l-2}(1+s^2)^{-k}{}_2F_1\left(a,b,\rho_0;-s^2\right)ds,\end{equation}
 which depends on $k$ and the eigenfunction $\varphi_j$ with eigenvalue $\mu_j=-\frac{1}{4\rho_0}(\rho_0^2+r_j^2)$, $a=\frac{1}{2}(\rho_0+ir_j)$, $b=\frac{1}{2}(\rho_0-ir_j)$.

\begin{theorem}\label{thm: wigpat}
Let $G=SO_o(1,l)$, $\varphi_j$ be an automorphic eigenfunction in the principal series and $\varphi_k$ be any non-constant automorphic eigenfunctions, then 

\begin{eqnarray}\label{eq: wipsint}
\langle \mathrm{Op}(\varphi_k)\varphi_j, \varphi_j\rangle_{L^2(X_\Gamma)}&=&\omega_{l-1}\int_0^\infty s^{l-2}\left(1+s^2\right)^{-(i\lambda_j+\rho_0)}{}_2F_1\left(a,b,\rho_0;-s^2\right)ds \cdot \langle \varphi_k,\mathrm{PS}_{\varphi_j}\rangle_{SX_\Gamma}\nonumber\\&=&\omega_{l-1} I(a,b,\rho_0,i\lambda_j+\rho_0)\cdot \langle \varphi_k,\mathrm{PS}_{\varphi_j}\rangle_{SX_\Gamma}
\end{eqnarray}
with $I(a,b,\rho_0,i\lambda_j+\rho_0)$ as in (\ref{eq: normcon}).
\end{theorem}
\begin{proof}
For any $\varphi_k$ we have 
\begin{eqnarray}\label{eq: umform}
\langle \mathrm{Op}(\varphi_k)\varphi_j,\varphi_j\rangle&\overset{\ref{thm: inter}}=&\langle L_{\lambda_j}(\chi \varphi_k),\mathrm{PS}_{\varphi_j}\rangle_{G/M}\nonumber\\&\overset{\ref{cor: inter}}=& \int_N e^{-(i\lambda_j+\rho)H(n^{-1}w)} \langle \pi_M(\varphi_k^n),\mathrm{PS}_{\varphi_j}\rangle_{SX_\Gamma}dn.
\end{eqnarray}

As in the proof of Proposition \ref{prop: absconv} we find that for $l\geq 3$
\begin{eqnarray*}
&&\int_N e^{-(i\lambda_j+\rho)H(n^{-1}w)} \langle \pi_M(\varphi_k^n),\mathrm{PS}_{\varphi_j}\rangle_{SX_\Gamma}dn\\
&=& \omega_{l-1}\int_0^\infty s^{l-2}\left(1+s^2\right)^{-(i\lambda_j+\rho_0)} {}_2F_1\left(a,b,\rho_0;-s^2\right)ds\cdot \langle \varphi_k,\mathrm{PS}_{\varphi_j}\rangle_{SX_\Gamma}.
\end{eqnarray*}

For $l=2$ we have  
\begin{eqnarray*}
&&\int_N e^{-(i\lambda_j+\rho)H(n^{-1}w)} \langle \pi_M(\varphi_k^n),\mathrm{PS}_{\varphi_j}\rangle_{SX_\Gamma}dn 
\\&=& \langle \varphi_k,\mathrm{PS}_{\varphi_j}\rangle_{SX_\Gamma}\cdot\int_{-\infty}^\infty (1+s^2)^{-(i\lambda_j+\rho_0)} {}_2F_1\left(a,b,\rho_0;-s^2\right)ds.
\end{eqnarray*}

The  last equality of (\ref{eq: wipsint}) follows from the definition of $I(a,b,\rho_0,k)$, see (\ref{eq: normcon}), and the convergence of the integral $I(a,b,\rho_0,i\lambda_j+\rho_0)$ follows from Theorem \ref{thm: absconvofint}.
\end{proof}

We excluded the case where $\varphi_j$ is a complementary series eigenfunctions, as we do not know whether the integral
\[\int_0^\infty s^{l-2}\left(1+s^2\right)^{-(i\lambda_j+\rho_0)}{}_2F_1\left(\frac{1}{2}(\rho_0+ir_k) ,\frac{1}{2}(\rho_0-ir_k) ,\rho_0;-s^2\right)ds\] converges in this case. It surely converges absolutely if $\varphi_j$ is in the principal series as Theorem \ref{thm: absconvofint} shows.

But we can still say something about the relation between Wigner- and Patterson-Sullivan distributions in the case of the complementary series. Let us define \index{$C_{k,j}$, constant associated to eigenfunctions $\varphi_k$ and $\varphi_j$}
\[C_{k,j}=\left\{\begin{array}{ccc}\frac{\langle \mathrm{Op}(\varphi_k)\varphi_j,\varphi_j\rangle_{L^2(X_\Gamma)}}{\langle \varphi_k,\mathrm{PS}_{\varphi_j}\rangle_{SX_\Gamma}} &, & \langle \varphi_k,\mathrm{PS}_{\varphi_j}\rangle_{SX_\Gamma}\neq 0\\ 0 & , & \mbox{ else}\end{array}\right.\]

\begin{theorem}\label{th: trivcon}
For any automorphic eigenfunctions $\varphi_k$ and $\varphi_j$ we have \[\langle \mathrm{Op}(\varphi_k)\varphi_j,\varphi_j\rangle_{L^2(X_\Gamma)}=C_{k,j}\cdot\langle \varphi_k,\mathrm{PS}_{\varphi_j}\rangle_{SX_\Gamma}.\]
\end{theorem}
\begin{proof}
It suffices to consider the case of $\langle \varphi_k,\mathrm{PS}_{\varphi_j}\rangle_{SX_\Gamma}=0$ and show that it implies $\langle \mathrm{Op}(\varphi_k)\varphi_j,\varphi_j\rangle_{L^2(X_\Gamma)}=0$. If we assume $\langle \varphi_k,\mathrm{PS}_{\varphi_j}\rangle_{SX_\Gamma}=0$, it follows by Theorem \ref{thm: diffeqex} that for any smooth fundamental domain cut-off $\chi$ \[n\mapsto \langle \pi_M(\chi \varphi_k)^n,\mathrm{PS}_{\varphi_j}\rangle_{G/M}\] vanishes identically. Thus,
\[\int_N e^{-(i\lambda_j+\rho)H(n^{-1}w)}\langle \pi_M(\chi \varphi_k)^n,\mathrm{PS}_{\varphi_j}\rangle_{G/M}dn=0.\]
 
But then it follows by Proposition \ref{thm: 2} and Theorem \ref{thm: inter} that
\begin{eqnarray*}
\int_N e^{-(i\lambda_j+\rho)H(n^{-1}w)}\langle \pi_M(\chi \varphi_k)^n,\mathrm{PS}_{\varphi_j}\rangle_{G/M}dn&\overset{\text{Prop. \ref{thm: 2}}}=& \langle L_{\lambda_j}(\chi \varphi_k),\mathrm{PS}_{\varphi_j}\rangle_{G/M}\\ 
&\overset{\text{Thm. \ref{thm: inter}}}=& \langle \mathrm{Op}(\varphi_k)\varphi_j,\varphi_j\rangle_{L^2(X_\Gamma)}=0.
\end{eqnarray*} 
\end{proof}


Note that by Proposition \ref{prop: decay} from the appendix to this chapter we will know how $C_{k,j}$ decays, if $\varphi_j$ is in the principal series. \begin{definition}
For an automorphic eigenfunction $\varphi$ with eigenvalue $-(\lambda^2+\rho_0^2)$ from the principal series we define the \index{normalized Patterson-Sullivan distribution} \textit{normalized Patterson-Sullivan distribution} by \index{$\widehat{\mathrm{PS}}_\varphi$}
\[\langle a,\widehat{\mathrm{PS}}_{\varphi}\rangle_{SX_\Gamma}:=\frac{\langle a,\mathrm{PS}_{\varphi}\rangle_{SX_\Gamma}}{\langle \mathbf 1,\mathrm{PS}_\varphi\rangle_{SX_\Gamma}}.\]
\end{definition}

It will follow from the next proposition that $\langle \mathbf 1,\mathrm{PS}_\varphi\rangle_{SX_\Gamma}\neq 0$, if $\varphi$ is in the principal series. For the normalization of Patterson-Sullivan distributions we have:
\begin{proposition} Let $\varphi$ be an automorphic eigenfunction with eigenvalue $-(\lambda^2+\rho_0^ 2)$ such that $\mathrm{Re}(\lambda)\in \mathfrak{a}_{+}^*$, i.e. such that $\varphi_j$ is in the principal series. Then
\[\langle L_{\lambda}(\chi \mathbf 1),\mathrm{PS}_\varphi\rangle_{G/M}=\langle \mathbf{1},\mathrm{PS}_\varphi\rangle_{\Gamma\backslash G/M}\cdot \mathcal{C}(\lambda).\] 

Here \index{$\mathcal C(\lambda)$, constant depending on $\lambda\in\mathfrak a^*_{\mathbb C}$}  $\mathcal{C}(\lambda)=\int_Ne^{-(i\lambda+\rho)H(n^{-1}w)}dn=\omega_{l-1}\int_0^\infty s^{l-2}\left(1+s^2\right)^{-(i\lambda+\rho_0)}ds$. 
\end{proposition}

\begin{proof}
It is
\begin{eqnarray*}
\langle L_\lambda(\chi \mathbf 1),\mathrm{PS}_\varphi\rangle_{G/M}&=&\int_Ne^{-(i\lambda+\rho)H(n^{-1}w)}\langle \pi_M(\chi^n),\mathrm{PS}_\varphi\rangle_Gdn \\
&=& \int_Ne^{-(i\lambda+\rho)H(n^{-1}w)}\int_{G/M}\int_M\chi^n(gm)dm\mathrm{PS}_\varphi(gM)dn\\ &=& \int_N e^{-(i\lambda+\rho)H(n^{-1}w)}\int_{G/M}\pi_M(\chi^n)(gM)\mathrm{PS}_\varphi(gM)dn\\ &=& \int_Ne^{-(i\lambda+\rho)H(n^{-1}w)}\langle \mathbf{1},\mathrm{PS}_\varphi\rangle_{\Gamma\backslash G/M}dn=\mathcal{C}(\lambda)\langle \mathbf 1,\mathrm{PS}_\varphi\rangle_{\Gamma\backslash G/M}.
\end{eqnarray*}

\end{proof}

\begin{remark}\label{rem: hccfc}
We recall Harish-Chandra's $c$-function which is given by $$c(\lambda)=\int_{\bar N}e^{-(i\lambda+\rho)H(\bar n)}d\bar{n},$$ where the measure $d\bar n$ on $\bar N$ is normalized such that $\int_{\overline{N}}e^{-2\rho\left(H(\bar n)\right)}d\bar n=1$, see Theorem \ref{th: hccf}. Now with our normalization of $dn$ via $N\cong \bar{N}\cong \mathbb{R}^{l-1}$, $X_u\mapsto \theta X_u\mapsto u$, we compute
\begin{eqnarray*}
\int_{\bar{N}}e^{-2\rho \left(H(\bar n)\right)}d\bar n&=& \int_N e^{-2\rho\left( H(n_{-u}w)\right)}du\\
&=& \omega_{l-1}\int_0^\infty (1+s^2)^{-2\rho_0}ds\\
&=& \frac{\omega_{l-1}}{2}B(\rho_0,\rho_0).
\end{eqnarray*}
\end{remark}

Together with Theorem \ref{thm: inter} we get:
\begin{corollary}\label{cor: normasymp}
If $\varphi_j$ is an automorphic eigenfunction in the principal series with eigenvalue $-(\lambda_j^2+\rho_0^2)$, then \[ \langle 1,\mathrm{PS}_{\varphi_j}\rangle_{SX_\Gamma}\cdot \mathcal C(\lambda_j)=1\] resp. \[\langle a,\widehat{\mathrm{PS}}_{\varphi_j}\rangle_{SX_\Gamma}=\mathcal C(\lambda_j)\langle a,\mathrm{PS}_{\varphi_j}\rangle_{SX_\Gamma} \] for $a\in C^\infty(SX_\Gamma)$. 
\end{corollary}

Finally, let us state the refinement of Proposition \ref{thm: strac}, we obtained in this section by Theorems \ref{thm: wigpat} and \ref{th: trivcon}.

\begin{theorem}\label{thm: stracspec}
Let $G=SO_o(1,l)=NAK$, $M=Z_K(A)$, $\Gamma\subset G$ a uniform lattice and $SX_\Gamma=\Gamma \backslash G/M$ the unit sphere bundle. We fix an orthonormal basis $\{\varphi_j\}_j$ of $L^2(\Gamma\backslash X)$ of automorphic Laplace-eigenfunctions with eigenvalues $-\left(\lambda_j^2+\rho_0^2\right)=-\frac{1}{4\rho_0}\left(r_j^2+\rho_0^2\right)$ and choose a non-constant $\varphi_k$ in $\{\varphi_j\}_j$. Finally, let $f\in C^\infty(G//K)$ such that $\pi_R(f)$ is of trace class, where $\pi_R$ is the right-regular representation of $G$ on $L^2(\Gamma\backslash X)$. For example $f\in C^\infty_c(G//K)$ suffices. Then $\varphi_k\cdot \pi_R(f)$ is also of trace class with trace given by 
\end{theorem}
\begin{eqnarray*}\label{eq: stracspec}\mathrm{Tr}(\varphi_k\cdot \pi_R(f))&=&\sum_j \langle \mathrm{Op}(\varphi_k)\varphi_j,\varphi_j\rangle_{L^2(SX_\Gamma)}\mathcal{S}(f,\lambda_j)\\&=&\sum^{j_0}_{j>0} \langle \mathrm{Op}(\varphi_k)\varphi_j,\varphi_j \rangle_{L^2(X_\Gamma)} \mathcal{S}(f,\lambda_j)\\&&+\omega_{l-1}\sum_{j>j_0}I(a,b,\rho_0,\rho_0+i\lambda_j) \mathcal{S}(f,\lambda_j)\cdot\langle \varphi_n,\mathrm{PS}_{\lambda_j}\rangle_{SX_\Gamma}\\&=&
\sum^{j_0}_{j>0} C_{k,j} \mathcal{S}(f,\lambda_j)\cdot\langle \varphi_k,\mathrm{PS}_{\varphi_j}\rangle_{SX_\Gamma)}\\&&+\omega_{l-1}\sum_{j>j_0}I(a,b,\rho_0,\rho_0+i\lambda_j) \mathcal{S}(f,\lambda_j)\cdot\langle \varphi_n,\mathrm{PS}_{\lambda_j}\rangle_{SX_\Gamma}
\end{eqnarray*}
where $a=\frac{1}{2}(\rho_0+ir_k)$, $b=\frac{1}{2}(\rho_0-ir_k)$, $\rho_0=\frac{l-1}{2}$ and $\mathcal S(f)$ denotes the spherical transform. 
\subsection{Convergence and asymptotics of $I(a,b,\rho_0,z)$}\label{chap: convergence}

In this appendix to Section \ref{conn} we want to discuss the convergence and asymptotic behaviour of the integrals

\begin{equation}\label{eq: asconint}
I(a,b,\rho_0,z)=\int_0^\infty s^{l-2}(1+s^2)^{-z}{}_2F_1\left(a,b,\rho_0;-s^2\right)ds,
\end{equation}
we needed in Proposition \ref{prop: absconv}. The next theorem determines $z\in\mathbb{C}$ for which $(\ref{eq: asconint})$ converges absolutely. 

\begin{lemma}\label{thm: absconvofint}
Let $\varphi_k$ be an automorphic Laplace eigenfunction with eigenvalue $\mu_k=-\frac{1}{4\rho_0}(\rho_0^2+r_k^2)$, so $a=\frac{1}{2}(\rho_0+ir_k)$, $b=\frac{1}{2}(\rho_0-ir_k)$. The integral 
\begin{equation*} \int_0^\infty s^{l-2}(1+s^2)^{-z}{}_2F_1\left(a,b,\rho_0;-s^2\right)ds \end{equation*}
 converges absolutely for $\rho_0+\mathrm{Re}(ir_k)<2\mathrm{Re}(z)$, more precisely 
\begin{equation*}
\int_0^\infty \left|s^{n-2}(1+s^2)^{-z}{}_2F_1(a,b,\rho_0;-s^2)\right|ds\leq C \int_{0}^\infty (s+1)^{-2\mathrm{Re}(z)+\rho_0+\mathrm{Re}(ir_k)}ds
\end{equation*}
for some constant $C$ independent of $z$ and $r_k$.
\end{lemma}
\begin{proof}The proof in \cite[Prop. 5.2.]{AZ} generalizes almost verbatim. First we note that $s^{l-2}(s^2+1)^{-z}$ is asymptotically $s^{2\rho_0-1-2\mathrm{Re}(z)}$, where $\rho_0=\frac{l-1}{2}$. The hypergeometric factor can be controlled by a formula for hypergeometric functions, see for example \cite[(4.7.23)]{GV},
\begin{eqnarray}\label{eq1 asy}
{}_2F_1(\alpha,\beta,\gamma;s)&=&\frac{\Gamma(\gamma)\Gamma(\beta-\alpha)}{\Gamma(\beta)\Gamma(\gamma- \alpha)}|s|^{-\alpha}{}_2F_1(\alpha,1-\gamma+\alpha,1-\beta+\alpha;s^{-1})\nonumber\\&&+\frac{\Gamma(\gamma)\Gamma(\alpha-\beta)}{\Gamma(\alpha)\Gamma(\gamma-\beta)}|s|^{-\beta}{}_2F_1(\beta,1-\gamma+\beta,1-\alpha+\beta;s^{-1}).
\end{eqnarray}

Since ${}_2F_1(\alpha,\beta,\gamma;0)=1$ it follows that asymptotically ${}_2F_1(a,b,\rho_0;-s^2)$ equals
\begin{equation*}
\frac{\Gamma(\rho_0)\Gamma(-ir_k)}{\Gamma(\frac{\rho_0}{2}-\frac{ir_k}{2})^2}\cdot |s|^{-(\rho_0+ir_k)}+\frac{\Gamma(\rho_0)\Gamma(ir_k)}{\Gamma(\frac{\rho_0}{2}+\frac{ir_k}{2})^2}\cdot|s|^{-(\rho_0-ir_k)}.
\end{equation*}

Now $\Gamma(x+iy)\sim \sqrt{2\pi}e^{-\frac{\pi}{2}|y| }|y|^{x-\frac{1}{2} }$ for $|y|\to\infty$, see (\ref{eq: gammaasympt}). Hence the ratios $\frac{\Gamma(\rho_0)\Gamma(-ir_k)}{\Gamma(\frac{\rho_0}{2}-\frac{ir_k}{2})^2}$ and $\frac{\Gamma(\rho_0)\Gamma(ir_k)}{\Gamma(\frac{\rho_0}{2}+\frac{ir_k}{2})^2}$ are uniformly bounded in $r_k$. Thus, for some constant $C>0$ independent of $r_k$ and for all $s\geq 0$

\begin{equation*}
\left|{}_2F_1(a,b,\rho_0;-s^2)\right|\leq C(1+s)^{-\rho_0+\mathrm{Re}(ir_k)}.
\end{equation*}

Multiplying the two terms together we find that the integrand 
\begin{equation*}
s^{l-2}\left(s^2+1 \right)^{-z} {}_2F_1(a,b,\rho_0;-s^2)
\end{equation*}
is asymptotically bounded by
\begin{equation*}
C(1+s)^{\rho_0-1+\mathrm{Re}(ir_k)-2\mathrm{Re}(z)},
\end{equation*}
which is integrable on $[0,\infty)$ iff 
\begin{equation*}
\rho_0-1+\mathrm{Re}(ir_k)-2\mathrm{Re}(z)<-1,
\end{equation*} i.e. iff 
\begin{equation*}
\rho_0+\mathrm{Re}(ir_k)<2\mathrm{Re}(z).
\end{equation*}
\end{proof}

In particular $I(a,b,\rho_0,i\lambda+\rho_0)$ converges if $-(\lambda^2+\rho_0^2)$ is an eigenvalue from the principal series, that is, $\lambda\in\mathbb R$. Next, we want to show that $k\mapsto I(a,b,\rho_0,k)$ for fixed $a$ and $b$ can be meromorphically extended to $\mathbb C$.

\begin{lemma}\label{lem: inthypg}Let $\varphi$ be a Laplacian eigenfunction with eigenvalue $-\frac{1}{4\rho_0}(r^2+\rho_0^2)$ , $a=\frac{1}{2}(\rho_0+ir)$, $b=\frac{1}{2}(\rho_0-ir)$.
\begin{eqnarray}\label{eq: inthypgam}I(a,b,\rho_0,z)&=&\frac{1}{2}\frac{\Gamma(\rho_0)\Gamma(a-\rho_0+z)\Gamma(b-\rho_0+z)}{\Gamma(z)^2}\nonumber\\&=&\frac{1}{2}\frac{\Gamma(\rho_0)\Gamma(z-a)\Gamma(z-b)}{\Gamma(z)^2}  
\end{eqnarray} defines a meromorphic continuation to $\mathbb C$ with poles exactly in $a,a-1,a-2,\ldots$ and $b,b-1,b-2,\ldots$. Furthermore, $I(a,b,\rho_0,k)$ does not vanish for $\mathrm{Re}(k)>0$.  
\end{lemma}
\begin{proof}
Let $z\in \mathbb{C}$ with $\mathrm{Re}(z)>\rho_0$. By Lemma \ref{thm: absconvofint} \[I(a,b,\rho_0,z)=\int_0^\infty s^{l-2}(s^2+1)^{-z}{}_2F_1(a,b,\rho_0,-s^2)ds\] converges absolutely. We use now an integral transform from \cite[20.2 (9)]{Ba} to get \begin{eqnarray}\label{eq: inthypconc}\int_0^\infty s^{l-2}(s^2+1)^{-z}{}_2F_1(a,b,\rho_0,-s^2)ds &\overset{s\mapsto \sqrt s}=&\frac{1}{2}\int_0^\infty s^{\rho_0-1}(s+1)^{-k}{}_2F_1(a,b,\rho_0;-s)ds \nonumber\\&=& \frac{1}{2}\frac{\Gamma(\rho_0)\Gamma(z+a-\rho_0)\Gamma(z+b-\rho_0)}{\Gamma(z)^2}\nonumber\\&=& \frac{1}{2}\frac{\Gamma(\rho_0)\Gamma(z-b)\Gamma(z-a)}{\Gamma(z)^2}.\end{eqnarray}

The remaining claims now follow from properties of the Gamma function, see Remark \ref{rem: gamma}. 
\end{proof}

Now we want to examine how 
\begin{equation*}
\int_0^\infty s^{l-2}(1+s^2)^{-\rho_0+i\lambda_j}{}_2F_1(a,b,\rho_0;-s^2)ds, 
\end{equation*}
grows as $\lambda_j$ tends to $\infty$. 
In (\ref{eq: inthypconc}) we showed that

\begin{equation}\label{eq: hypertransform}
\int_0^\infty s^{l-2}(1+s^2)^{-(i\lambda_j+\rho_0)}{}_2F_1(a,b,\rho_0;-s^2)ds
\end{equation}
equals

\begin{eqnarray*}
\frac{1}{2}\frac{\Gamma(\rho_0)\Gamma(\frac{\rho_0}{2}+i(\lambda_j+ \frac{r_k}{2} ) )\Gamma(\frac{\rho_0}{2}+i(\lambda_j-\frac{r_k}{2} ))}{\Gamma(i\lambda_j+\rho_0)^2}
.\end{eqnarray*}

Then we make use of the asymptotic formula for the $\Gamma$-function
\begin{equation*}
\Gamma(x+iy)\sim \sqrt{2\pi}e^{-\frac{\pi}{2}|y|}|y|^{x- \frac{1}{2} }\mbox{ , } y\to\infty,
\end{equation*}
see (\ref{eq: gammaasympt}). It follows that

\begin{equation*}
\int_0^\infty s^{l-2}(1+s^2)^{-(i\lambda_j+\rho_0)}{}_2F_1(a,b,\rho_0;-s^2)ds
\end{equation*}
asymptotically equals, as $\lambda_j\to\infty$,

\begin{eqnarray}\label{eq: gammaasym}
\frac{e^{-\frac{\pi}{2}(\lambda_j+ \frac{r_k}{2} ) }(\lambda_j+\frac{r_k}{2} )^{\frac{\rho_0-1}{2}}e^{-\frac{\pi}{2}(\lambda_j- \frac{r_k}{2} ) }(\lambda_j-\frac{r_k}{2} )^{\frac{\rho_0-1}{2}}}{ e^{-\pi \lambda_j}\lambda_j^{2\rho_0-1}}\sim \lambda_j^{-\rho_0} .
\end{eqnarray}

We just have proved the following proposition which will be useful for the meromorphic continuation in Chapter \ref{chap: meromorph}.
\begin{proposition}\label{prop: decay}
Let $\lambda_j$ be in the principal series, i.e. $\lambda_j>0$. The integral \[\int_0^\infty s^{l-2}\left(s^2+1\right)^{(-i\lambda_j+\rho_0)}{}_2F_1\left(a,b,\rho_0;-s^2\right)ds\] decays asymptotically as $\lambda_j^{-\rho_0}$ for $\lambda_j\to \infty$.
\end{proposition}

\chapter{The meromorphic continuation of $\mathcal R(\varphi)$ and $\mathcal Z(\varphi)$}\label{chap: meromorph}
In this chapter we will at first give meromorphic continuations of $\mathcal R(\varphi)$ and $\mathcal Z(\varphi)$ which were defined in (\ref{def: auxzet}) resp. (\ref{def: zetfunc}) on the complex half plane $\{k\in\mathbb C:\mathrm{Re}(k)>2\rho_0\}$ to all of $\mathbb C$. This is done in Section \ref{sec: mero1} and Section \ref{sec: mcontrz}. For the loaction of possible poles and residues of $\mathcal R(\varphi)$ and $\mathcal Z(\varphi)$ the focus is on a certain strip $\mathcal S$ in $\mathbb C$ and the (main) results are summarized in Section \ref{sec: msum}.

In Section \ref{sec: mnorm} we will shortly explain how to normalize $\mathcal Z(\varphi)$ in order to obtain a simple formula for its residue in the strip $\mathcal S$. The last Section \ref{sec: out} compares our results on the zeta function $\mathcal Z(\varphi)$ in the surface case with the ones from \cite{AZ}, see also Theorem \ref{th: mainth} and (\ref{eq: zetasurf}).
\section{The meromorphic continuation of $\mathcal{R}(\varphi)$}\label{sec: mero1}
In this section we will discuss the meromorphic continuation which we obtain by using the formula for the spectral trace from Chapter \ref{chap: strace}. The case of $G=SO_o(1,2)$ was dealt with before in \cite[Th. 9.1.]{AZ}. 
Let us recall that $G=SO_o(1,l)=KAN$ and we fixed a uniform lattice $\Gamma$ with an orthonormal basis of automorphic Laplace eigenfunctions $\{\varphi_j\}$ in $L^2(\Gamma\backslash G/K)$ with eigenvalues $-(\lambda_j^2+\rho_0^2)=-\frac{1}{4\rho_0}(r_j^2+\rho_0^2)$. Here $\varphi_0,\ldots, \varphi_{j_0}$ are in the complementary series, i.e. $\lambda_0^2,\ldots,\lambda_{j_0}^2\in[-\rho_0^2,0]$, and $\varphi_j$ is in the principal series, i.e. $\lambda_j^2\in (0,\infty)$ for $j>j_0$. 

By the results of Chapter \ref{chap: gtrace}, see Theorem \ref{th: conauxzeta}, we know that the geometric trace of \[\varphi_n\circ\pi_R(f_k)\] for any eigenfunction $\varphi_n\in\{\varphi_j\}$ which is orthogonal to constants, is given by the auxiliary zeta function, see (\ref{def: auxzet}), 
\begin{equation}\label{gtr}
\mathcal{R}(k;\varphi_n)= \sum_{1\neq[\gamma]\in C\Gamma}\sum_{\pi\in\widehat M}c(\varphi_n,\gamma,\pi,k) (\cosh L_\gamma)^{-k+\rho_0}, 
\end{equation}
for $\mathrm{Re}(k)>2\rho_0$. On the other hand, by Theorem \ref{thm: strac} the spectral trace of $\varphi_n\circ\pi_R(f_k)$ is given for $\mathrm{Re}(k)>2\rho_0$ by \[\sum_{j=0}^\infty \langle \mathrm{Op}(\varphi_n)\varphi_j,\varphi_j \rangle_{L^2(X_\Gamma)} \mathcal{S}(f_k,\lambda_j)\] which equals \[\sum_{j=1}^\infty \langle \mathrm{Op}(\varphi_n)\varphi_j,\varphi_j \rangle_{L^2(X_\Gamma)} \mathcal{S}(f_k,\lambda_j)\] as $\varphi_n$ is orthogonal to the constant function by assumption. From Theorem \ref{thm: stracspec} we then infer that the spectral trace of $\varphi_n\cdot\pi_R(f_k)$ equals, if we replace $\langle \mathrm{Op}(\varphi_n)\varphi_j,\varphi_j\rangle_{L^2(X_\Gamma)} $ by $C_{n,j}\langle \varphi_n,\mathrm{PS}_{\varphi_j}\rangle_{SX_\Gamma}$ and keep in mind that $C_{n,j}=\omega_{l-1}I(a,b,\rho_0,\rho_0+i\lambda_j)$, see Theorems \ref{thm: wigpat} and \ref{th: trivcon},
\begin{equation*}
\sum^{j_0}_{j=1} C_{n,j} \mathcal{S}(f_k,\lambda_j)\cdot\langle \varphi_n,\mathrm{PS}_{\varphi_j}\rangle_{SX_\Gamma}
+\omega_{l-1}\sum_{j=j_0+1}^\infty I(a,b,\rho_0,\rho_0+i\lambda_j) \mathcal{S}(f_k,\lambda_j)\cdot\langle \varphi_n,\mathrm{PS}_{\varphi_j}\rangle_{SX_\Gamma}=:(+)
.
\end{equation*}

Then we use the formula for the spherical transform of $\mathcal S(f_k.\lambda_j)$, see  $(\ref{eq: stfk})$, to see that $(+)$ equals
\begin{eqnarray*}
&& 2^{k-2}B(k-\rho_0,\rho_0)\omega_{l-1}\sum^{j_0}_{j=1} C_{n,j}B\left( \frac{k-i\lambda_j-\rho_0}{2}, \frac{k+i\lambda_j-\rho_0}{2}\right)\langle \varphi_n,\mathrm{PS}_{\varphi_j}\rangle_{SX_\Gamma} + \\&&2^{k-2}\omega_{l-1}^2 B(k-\rho_0,\rho_0)\sum_{j=j_0+1}^\infty
 B\left( \frac{k-i\lambda_j-\rho_0}{2}, \frac{k+i\lambda_j-\rho_0}{2}\right)\langle \varphi_n,\mathrm{PS}_{\varphi_j}\rangle_{SX_\Gamma}\\&&\cdot I(a,b,\rho_0,\rho_0+i\lambda_j)\\&=:&(*).
\end{eqnarray*}

Since the operator $\varphi_n\cdot \pi_R(f_k)$ is of trace class for $ \mathrm{Re}(k)>2\rho_0 $ we know that $(*)$ coincides with (\ref{gtr}) for $ \mathrm{Re}(k)>2\rho_0 $. To obtain the meromorphic continuation we want to show that $(*)$ converges for any $k$ in $\mathbb C$ except from the set \index{$\mathcal P$, set of poles}
\begin{eqnarray}\label{eq: poleset}
\mathcal{P}&:=&\{\rho_0\pm i\lambda_j,\rho_0-2\pm i\lambda_j,\ldots:-(\lambda_j^2+\rho_0^2) \mbox{  eigenvalue of the Laplacian}\}\nonumber\\&&\cup\{\rho_0,\rho_0-1,\rho_0-2,\ldots\}.\end{eqnarray} 

It is clear that \[k\mapsto 2^{k-2}B(k-\rho_0,\rho_0)\omega_{l-1}\sum^{j_0}_{j=1} C_{n,j} B\left( \frac{k-i\lambda_j-\rho_0}{2}, \frac{k+i\lambda_j-\rho_0}{2}\right)\langle \varphi_n,\mathrm{PS}_{\varphi_j}\rangle_{SX_\Gamma}\]
defines a meromorphic function on $\mathbb{C}$. The poles correspond to the poles of the Beta functions $B(k-\rho_0,\rho_0)$ and $B\left( \frac{k-i\lambda_j-\rho_0}{2}, \frac{k+i\lambda_j-\rho_0}{2}\right)$, i.e. $k=\rho_0,\rho_0-1,\ldots$ resp. $k=\rho_0\pm i\lambda_j,\rho_0-2\pm i\lambda_j,\ldots$, where $-\left(\lambda_j^2+\rho_0^2\right)$ is an eigenvalue of the Laplacian. 

Thus, for the convergence of $(*)$ we only have to consider the infinite series 
\begin{equation}\label{eq: infseries}
\sum_{j=j_0+1}^\infty\langle \varphi_n,\mathrm{PS}_{\varphi_j}\rangle_{SX_\Gamma}\cdot I(a,b,\rho_0,\rho_0+i\lambda_j)
\cdot B\left( \frac{k-i\lambda_j-\rho_0}{2}, \frac{k+i\lambda_j-\rho_0}{2}\right)
\end{equation}
outside (\ref{eq: poleset}). Let us check each of the three terms in this series separately. At first, we know that
\begin{equation*}
I(a,b,\rho_0,\rho_0+i\lambda_j)=\int_0^\infty s^{l-2}(1+s^2)^{-(i\lambda_j+\rho_0)}{}_2F_1\left(a,b,\rho_0;-s^2\right)ds
\end{equation*}
behaves for $j\to\infty$ as $\lambda_j^{-\rho_0}$, see Proposition \ref{prop: decay}, because $\varphi_j$ is in the principal series. For the term $\langle \varphi_n,\mathrm{PS}_{\varphi_j}\rangle_{SX_\Gamma}$ we claim the following:
\begin{lemma}\label{lem: asymequiv}
For any automorphic eigenfunction $\varphi_n$, $\langle \varphi_n,\widehat{\mathrm{PS}}_{\lambda_j}\rangle_{SX_\Gamma}$ is uniformly bounded , i.e. there is some constant $C>0$ independent of $j$ such that \[\left|\langle \varphi_n,\widehat{\mathrm{PS}}_{\lambda_j}\rangle_{SX_\Gamma}\right|\leq C\] for all $j$.
\end{lemma}
\begin{proof}
In \cite[Th. 1.3.]{HS} it is shown that
\begin{equation*}
\langle \varphi_n,\widehat{\mathrm{PS}}_{\lambda_j}\rangle_{SX_\Gamma}
\end{equation*}
behaves asymptotically for $j\to\infty$ like 
\begin{equation*}
\langle \mathrm{Op}(\varphi_n)\varphi_j,\varphi_j\rangle_{L^2(X_\Gamma)}=\int_{X_\Gamma}\varphi_n(x) |\varphi_j(x)|^2dx.
\end{equation*}

But $X_\Gamma$ is compact so that 
\[\left|\langle \mathrm{Op}(\varphi_n)\varphi_j,\varphi_j\rangle_{L^2(X_\Gamma)}\right|\leq \int_{X_\Gamma}|\varphi_n(x)||\varphi_j(x)|^2dx\leq C,\] where $C:=\max\{\varphi_n(x):x\in X_\Gamma\}$. Hence, $\langle \mathrm{Op}(\varphi_n)\varphi_j,\varphi_j\rangle_{L^2(X_\Gamma)}$ and also
$\langle \varphi_n,\widehat{\mathrm{PS}}_{\lambda_j}\rangle_{SX_\Gamma}$ are bounded uniformly in $j$. 
\end{proof}

Now by definition $\langle \varphi_n,\widehat{\mathrm{PS}}_{\lambda_j}\rangle_{SX_\Gamma}=\mathcal{C}(\lambda_j)\langle \varphi_n,\mathrm{PS}_{\lambda_j}\rangle_{SX_\Gamma}$ and $$\mathcal{C}(\lambda_j)=\frac{\omega_{l-1}}{2}B(i\lambda_j,\rho_0)=\frac{\omega_{l-1}}{2}\frac{\Gamma(i\lambda_j)\Gamma(\rho_0)}{\Gamma(\rho_0+i\lambda_j)}$$ which by (\ref{eq: gammaasympt}) for $\lambda_j\to\infty$ equals $$ \frac{\sqrt{2\pi}e^{-\pi/2|\lambda_j|}|\lambda_j^{-1/2}|}{\sqrt{2\pi}e^{-\pi/2|\lambda_j|}|\lambda_j^{\rho_0-1/2}|}=|\lambda_j|^{-\rho_0}. $$

Corollary \ref{cor: normasymp} implies then:
\begin{corollary}
For fixed $\varphi_n$, $\langle \varphi_n,\mathrm{PS}_{\varphi_j}\rangle_{SX_\Gamma}\sim \lambda_j^{\rho_0}$ as $j\to\infty$.
\end{corollary}

Finally, we see:
\begin{lemma}\label{lem: betaasym}
Let $k\in \mathbb{C}-\{\rho_0\pm i\lambda_j,\rho_0-2\pm i\lambda_j,\ldots:-\left(\lambda_j^2+\rho_0^2\right)\mbox{ eigenvalue of the Laplacian}\}$. Then
$B\left( \frac{k-i\lambda_j-\rho_0}{2}, \frac{k+i\lambda_j-\rho_0}{2}\right)$
behaves asymptotically for $j\to\infty$ like \[e^{-\frac{\pi}{4}(|\mathrm{Im}(k)+\lambda_j|+|\mathrm{Im}(k)-\lambda_j|)}\left|\mathrm{Im}(k)+\lambda_j\right|^{\frac{\mathrm{Re}(k)-\rho_0-1}{2}}\left|\mathrm{Im}(k)-\lambda_j\right|^{\frac{\mathrm{Re}(k)-\rho_0-1}{2}}.\]
\end{lemma}
\begin{proof}
Let us write out
\begin{eqnarray*}
B\left( \frac{k-i\lambda_j-\rho_0}{2}, \frac{k+i\lambda_j-\rho_0}{2}\right)&=&\frac{\Gamma(\frac{k-i\lambda_j-\rho_0}{2})\Gamma(\frac{k+i\lambda_j-\rho_0}{2})}{\Gamma(k-\rho_0)}\\
&=& \frac{\Gamma\left(\frac{\mathrm{Re}(k)-\rho_0+i\left(\mathrm{Im}(k)-\lambda_j\right)}{2}\right)\Gamma \left(\frac{\mathrm{Re}(k)-\rho_0+i\left(\mathrm{Im}(k)+\lambda_j\right)}{2}\right)}{\Gamma\left(\mathrm{Re}(k)-\rho_0+i\mathrm{Im}(k)\right)} .
\end{eqnarray*}

Using the asymptotic formula for the Gamma function, see (\ref{eq: gammaasympt}) in Remark \ref{rem: gamma},
\begin{equation}\label{eq: gammaasym}
\Gamma(x+iy)\sim \sqrt{2\pi}e^{-\frac{\pi}{2}|y|}|y|^{x- \frac{1}{2} }\mbox{ , } y\to\infty,
\end{equation}
we find that $B\left( \frac{k-i\lambda_j-\rho_0}{2}, \frac{k+i\lambda_j-\rho_0}{2}\right)$ behaves for $j\to\infty$ as
\begin{equation*}
e^{-\frac{\pi}{4}(|\mathrm{Im}(k)+\lambda_j|+|\mathrm{Im}(k)-\lambda_j|)}\left|\mathrm{Im}(k)+\lambda_j\right|^{\frac{\mathrm{Re}(k)-\rho_0-1}{2}}\left|\mathrm{Im}(k)-\lambda_j\right|^{\frac{\mathrm{Re}(k)-\rho_0-1}{2}}.
\end{equation*}
\end{proof}

Since the term $e^{-\frac{\pi}{4}(|\mathrm{Im}(k)+\lambda_j|+|\mathrm{Im}(k)-\lambda_j|)}$ dominates the other terms $$\left|\mathrm{Im}(k)+\lambda_j\right|^{\frac{\mathrm{Re}(k)-\rho_0-1}{2}}\left|\mathrm{Im}(k)-\lambda_j\right|^{\frac{\mathrm{Re}(k)-\rho_0-1}{2}}$$ and since the series $\sum_{n\in \mathbb N}ne^{-\alpha n}$ is absolutely convergent for every $\alpha>0$, we see that the series $$\sum_{j=j_0+1}^\infty e^{-\frac{\pi}{4}(|\mathrm{Im}(k)+\lambda_j|+|\mathrm{Im}(k)-\lambda_j|)}\left|\mathrm{Im}(k)+\lambda_j\right|^{\frac{\mathrm{Re}(k)-\rho_0-1}{2}}\left|\mathrm{Im}(k)-\lambda_j\right|^{\frac{\mathrm{Re}(k)-\rho_0-1}{2}}$$ is absolute convergent for any $k\in\mathbb C$.  Thus, by Proposition \ref{prop: decay}, Lemma \ref{lem: asymequiv} and Lemma \ref{lem: betaasym} we find 
the following proposition,
\begin{proposition}\label{prop: merozeta}
The infinite series $(\ref{eq: infseries})$ converges  absolutely for $k$ in $\mathbb C$ outside the set $\mathcal P$, see (\ref{eq: poleset}), and defines there a holomorphic function.
\end{proposition}
 
Hence, we define the meromorphic continuation of $ \mathcal{R}(k;\varphi_n) $ by $(*)$. Let us now focus on a certain strip in $\mathbb C$ and describe the poles and residues of $\mathcal R(k;\varphi_n)$ more precisely. It follows from Proposition \ref{prop: merozeta} that in the strip \index{$\mathcal S$, strip in $\mathbb C$} \[\mathcal{S}:=\{k\in\mathbb C:\rho_0-\frac{1}{2} <\mathrm{Re}(k)<\rho_0+\frac{1}{2}\} \] the poles are at \index{$\mathcal P_{\mathcal S}$, intersection of $\mathcal P$ and $\mathcal S$} \[\mathcal{P}_{\mathcal S}:=\{\rho_0,\rho_0\pm i\lambda_j: -\left(\lambda_j^2+\rho_0^2\right) \mbox{ eigenvalue of the Laplacian}\}\cap \mathcal S.\] 

For computing the residue we make use of the following facts.
\begin{lemma}\label{lem: rescal}See \cite[Kap.III]{FB}\\
\begin{enumerate}
\item If $f$ is holomorphic in $z_0$, then $\mathrm{Res}_{z=z_0}f=0$
\item The residue is linear, i.e. for all functions $f,g$ and $\alpha,\beta\in\mathbb C$ \[\mathrm{Res}_{z=z_0}\left(\alpha f+\beta g\right)=\alpha\mathrm{Res}_{z=z_0}f+\beta \mathrm{Res}_{z=z_0}g.\]
\item If $g$ is holomorphic and bijective on $\mathbb C$, then $\mathrm{Res}_{z=g(z_0)}=\mathrm{Res}_{z=z_0}g'\cdot (f\circ g)$.
\item If $f$ has a pole of order $m$ at $z_0$, then  \[\mathrm{Res}_{z=z_0}f=\frac{d^{m-1}}{dz^{m-1}}|_{z=z_0}\left(\frac{(z-z_0)^m}{(m-1)!}f(z)\right) .\]
\item If $f$ is holomorphic in $z_0$ and $g$ has a pole of order 1 at $z_0$, then \[\mathrm{Res}_{z=z_0}f\cdot g=f(z_0)\mathrm{Res}_{z=z_0}g.\]
\end{enumerate}

From (4) and (5) we can deduce:
\begin{enumerate}[(6)]
\item If $f$ is holomorphic in $z_0$ and $g$ has a pole of order 2 at $z_0$, $g(z)=\sum_{k=-2}^\infty a_k(z-z_0)^k$ in a neighbourhood of $z_0$, then \[\mathrm{Res}_{z=z_0}f\cdot g=f'(z_0)\cdot \frac{a_{-2}}{2}+f(z_0)\mathrm{Res}_{z=z_0}g.\]
\end{enumerate}
\end{lemma}

Let us now compute the residue at $k\in\mathcal P_{\mathcal S}$, i.e. $k=\rho_0$ or $k=\rho_0+i\lambda_j$ and there is an eigenfunction $\varphi_j$ with eigenvalue $-(\lambda_j^2+\rho_0^2)$. If $\varphi_j$ is now in the principal series, then $\lambda_j\neq 0$ and $\int_0^\infty u^{l-2}(1+u^2)^{-(\rho_0+i\lambda_j)}du=\frac{1}{2}B(i\lambda_j,\rho_0)$ is defined. The residue $\mathrm{Res}_{k=\rho_0+i\lambda_j}\mathcal{R}(k;\varphi_n)$ at $k=\rho_0+i\lambda_j$ then equals
\begin{eqnarray*}
&&2^{\rho_0+i\lambda_j-2}\omega_{l-1}^2B(i\lambda_j,\rho_0)\sum_r \langle \varphi_n,\mathrm{PS}_{\varphi_r}\rangle I(a,b,\rho_0,\rho_0+i\lambda_j)\\&\cdot&\mathrm{Res}_{k=\rho_0+i\lambda_j}B\left(
\frac{k-i\lambda_j-\rho_0}{2},\frac{k+i\lambda_j-\rho_0}{2}\right)
,
\end{eqnarray*}
where the (finite) sum runs over all eigenfunctions $\varphi_r$ with eigenvalue $-\left(\lambda_j^2+\rho_0^2\right)$.

Now the residue $\mathrm{Res}_{k=\rho_0+i\lambda_j}B\left(
\frac{k-i\lambda_j-\rho_0}{2},\frac{k+i\lambda_j-\rho_0}{2}\right)$ is given by

\begin{eqnarray*}
\mathrm{Res}_{k=\rho_0+i\lambda_j}B\left(
\frac{k-i\lambda_j-\rho_0}{2},\frac{k+i\lambda_j-\rho_0}{2}\right)
&=& \mathrm{Res}_{k=\rho_0+i\lambda_j}\frac{\Gamma\left(\frac{k-i\lambda_j-\rho_0}{2} \right)\Gamma\left(\frac{k+i\lambda_j-\rho_0}{2} \right)}{\Gamma(k-\rho_0)}
\\&=& \frac{\Gamma(i\lambda_j)}{\Gamma(i\lambda_j)}\cdot\mathrm{Res}_{k=\rho_0+i\lambda_j}\Gamma\left(\frac{k-i\lambda_j-\rho_0}{2} \right)
\\&=& \mathrm{Res}_{k=0}\Gamma(k)
\\&=& 1
\end{eqnarray*}

Consequently, the residue at $k=\rho_0+i\lambda_j$ in the principal series is
\begin{eqnarray*}
\mathrm{Res}_{k=\rho_0+i\lambda_j}\mathcal{R}(k;\varphi_n) &=&2^{\rho_0+i\lambda_j-2}\sum_r \langle \varphi_n,\mathrm{PS}_{\varphi_r}\rangle 
\omega_{l-1}^2B(i\lambda_j,\rho_0)\cdot I(a,b,\rho_0,\rho_0+i\lambda_j)\\
&=& 2^{\rho_0+i\lambda_j-1} \omega_{l-1} I(a,b,\rho_0,\rho_0+i\lambda_j)\sum_r \langle \varphi_n,\widehat{\mathrm{PS}}_{\varphi_r}\rangle 
\\
&=:&\text{(I)},
\end{eqnarray*}
by Remark \ref{rem: gamma} and Corollary \ref{cor: normasymp}, since $2^{-1}\omega_{l-1}B(i\lambda_j,\rho_0)=\mathcal{C}(\lambda_j)$ for $\lambda_j>0$. 

If $\varphi_j$ is in the complementary series and $j<j_0$, then $\lambda_j\neq 0$ and $B(i\lambda_j,\rho_0)$ is defined. The residue $\mathrm{Res}_{k=\rho_0+i\lambda_j}\mathcal{R}(k;\varphi_n)$ at $k=\rho_0+i\lambda_j$ is

\begin{eqnarray*}
&&2^{\rho_0+i\lambda_j-2}B(i\lambda_j,\rho_0)\omega_{l-1}\sum_r C_{n,r}{} \langle \varphi_n,\mathrm{PS}_{\varphi_r}\rangle_{SX_\Gamma}\\&\cdot& \mathrm{Res}_{k=\rho_0+i\lambda_j}B\left(
\frac{k-i\lambda_j-\rho_0}{2},\frac{k+i\lambda_j-\rho_0}{2}\right).
\end{eqnarray*}

As before $\mathrm{Res}_{k=\rho_0+i\lambda_j}B\left(
\frac{k-i\lambda_j-\rho_0}{2},\frac{k+i\lambda_j-\rho_0}{2}\right)=1$ and hence the residue at $k=\rho_0+i\lambda_j$, $j<j_0$, is
\begin{eqnarray*}
\mathrm{Res}_{k=\rho_0+i\lambda_j}\mathcal{R}(k;\varphi_n)=2^{\rho_0+i\lambda_j-2}B(i\lambda_j,\rho_0)\omega_{l-1}\sum_r C_{n,r} \langle \varphi_n,\mathrm{PS}_{\varphi_r}\rangle_{SX_\Gamma}=:\text{(II)},
\end{eqnarray*}
where we again sum over all eigenfunctions $\varphi_r$ with eigenvalue $-\left(\lambda_j^2+\rho_0^2 \right)$. 

The same formula (II) is valid for $j=j_0$ and $\lambda_{j_0}\neq 0$. If $\lambda_{j_0}=0$, then \[B(k-\rho_0,\rho_0)=\frac{\Gamma(k-\rho_0)\Gamma(\rho_0)}{\Gamma(k)}\] has a pole of order 1 at $k=\rho_0$, while \[B\left(\frac{k-i\lambda_{j_0}-\rho_0}{2},\frac{k+i\lambda_{j_0}-\rho_0}{2}\right)=B\left(\frac{k-\rho_0}{2},\frac{k-\rho_0}{2}\right)=\frac{\Gamma\left(\frac{k-\rho_0}{2}\right)^2}{\Gamma(k-\rho_0)}\] has also a pole of order 1 at $k=\rho_0$. The product of both has thus a pole of order 2 at $k=\rho_0$ and the Laurent series around $k=\rho_0$ of the product starts with $(k-\rho_0)^{-2}$. Using Lemma \ref{lem: rescal}(6) the residue $\mathrm{Res}_{k=\rho_0}\mathcal{R}(k;\varphi_n)$ at $k=\lambda_{j_0}+\rho_0=\rho_0$ computes to
\begin{eqnarray*}
&&\mathrm{Res}_{k=\rho_0}\mathcal{R}(k;\varphi_n)\\
&=& \mathrm{Res}_{k=\rho_0}2^{k-2}B(k-\rho_0,\rho_0)\omega_{\ell-1}\sum_s C_{n,s}{\textstyle B\left(\frac{k-i\lambda_s-\rho_0}{2},\frac{k+i\lambda_s-\rho_0}{2}\right)}\langle\varphi_n,\mathrm{PS}_{\varphi_s}\rangle_{SX_\Gamma}\\
&&+\mathrm{Res}_{k=\rho_0}2^{k-2}B(k-\rho_0,\rho_0)\omega_{\ell-1}\sum_r C_{n,r}{\textstyle B\left(\frac{k-i\lambda_r-\rho_0}{2},\frac{k+i\lambda_r-\rho_0}{2}\right)}\langle\varphi_n,\mathrm{PS}_{\varphi_r}\rangle_{SX_\Gamma}\\
&&+2^{\rho_0-2}\omega_{\ell-1}^2\mathrm{Res}_{k=\rho_0}B(k-\rho_0,\rho_0)\sum_{j=j_0+1}^\infty {\textstyle B\left(\frac{k-i\lambda_j-\rho_0}{2},\frac{k+i\lambda_j-\rho_0}{2}\right)}\langle \varphi_n,\mathrm{PS}_{\varphi_j}\rangle_{SX_\Gamma}\\
&&\cdot I(a,b,\rho_0,\rho_0+i\lambda_j)\\
&=&
\omega_{\ell-1}\sum_{s}C_{n,s}{\textstyle B\left(\frac{i\lambda_s}{2},\frac{-i\lambda_s}{2}\right)}\langle\varphi_n,\mathrm{PS}_{\varphi_s}\rangle_{SX_\Gamma}2^{\rho_0-2}\cdot \mathrm{Res}_{k=\rho_0}B(k-\rho_0,\rho_0)\\
&&+\omega_{\ell-1}\sum_{r}C_{n,r}\langle\varphi_n,\mathrm{PS}_{\varphi_r}\rangle_{SX_\Gamma}\cdot \mathrm{Res}_{k=\rho_0}2^{k-2}B(k-\rho_0,\rho_0){\textstyle B\left(\frac{k-i\lambda_j-\rho_0}{2},\frac{k+i\lambda_j-\rho_0}{2}\right)}\\
&&+ \Big(2^{\rho_0-2}\omega_{\ell-1}^2\sum_{j=j_0+1}^\infty {\textstyle B\left(\frac{i\lambda_j}{2},\frac{-i\lambda_j}{2}\right)}\langle \varphi_n,\mathrm{PS}_{\varphi_j}\rangle_{SX_\Gamma}\cdot I(a,b,\rho_0,\rho_0+i\lambda_j)\Big)\mathrm{Res}_{k=\rho_0}B(k-\rho_0,\rho_0)
,
\end{eqnarray*}
where $\sum_s$ runs over all eigenfunctions $\varphi_s$ with eigenvalue $\mu_s=-(\lambda_s^2+\rho_0^2)>-\rho_0^2$ and $\sum_s$ runs over all eigenfunctions $\varphi_r$ with eigenvalue $\mu_r=-(\lambda_r^2+\rho_0^2)=-\rho_0^2$. Then $B(k-\rho_0,\rho_0)=\frac{\Gamma(k-\rho_0)\Gamma(\rho_0)}{\Gamma(k)}$ has a pole of order 1 at $k=\rho_0$, while \[B{\textstyle \left(\frac{k-i\lambda_{j_0}-\rho_0}{2},\frac{k+i\lambda_{j_0}-\rho_0}{2}\right)
=B\left(\frac{k-\rho_0}{2},\frac{k-\rho_0}{2}\right)}
=\frac{\Gamma\left(\frac{k-\rho_0}{2}\right)^2}{\Gamma(k-\rho_0)}\]
has also a pole of order 1 at $k=\rho_0$. The product of both has thus a pole of order 2 at $k=\rho_0$ and the Laurent series around $k=\rho_0$ of of the product starts with $(k-\rho_0)^{-2}$. Using Lemma \ref{lem: rescal}(6) it follows that

\begin{eqnarray*}
&&\mathrm{Res}_{k=\rho_0}2^{k-2}B(k-\rho_0,\rho_0){\textstyle B\left(\frac{k-i\lambda_r-\rho_0}{2},\frac{k+i\lambda_r-\rho_0}{2}\right)}
\\&=& \ln(2)2^{\rho_0-3}+2^{\rho_0-2}\cdot \mathrm{Res}_{k=\rho_0}B(k-\rho_0,\rho_0){\textstyle B\left(\frac{k-i\lambda_r-\rho_0}{2},\frac{k+i\lambda_r-\rho_0}{2}\right)}
\\&=&
\ln(2)2^{\rho_0-3}+2^{\rho_0-2}\cdot \mathrm{Res}_{k=\rho_0}B(k-\rho_0,\rho_0){\textstyle B\left(\frac{k-i\lambda_{j_0}-\rho_0}{2},\frac{k+i\lambda_{j_0}-\rho_0}{2}\right)}
\\&=&
\ln(2)2^{\rho_0-3}+2^{\rho_0-2}\cdot \mathrm{Res}_{k=\rho_0}B(k-\rho_0,\rho_0){\textstyle B\left(\frac{k-\rho_0}{2},\frac{k-\rho_0}{2}\right)}
\end{eqnarray*}

and

\begin{eqnarray*}
&&\mathrm{Res}_{k=\rho_0}B(k-\rho_0,\rho_0){\textstyle B \left(\frac{k-\rho_0}{2},\frac{k-\rho_0}{2}\right)}\\
&=& \mathrm{Res}_{k=\rho_0} {\textstyle \frac{\Gamma\left(\frac{k-\rho_0}{2}\right)^2 }{\Gamma(k-\rho_0)}}\cdot \frac{\Gamma(k-\rho_0)\Gamma(\rho_0)}{\Gamma(k)}
\ =\  \frac{\Gamma(\rho_0)}{\Gamma(\rho_0)}\cdot \mathrm{Res}_{k=\rho_0}{\textstyle\frac{\Gamma\left(\frac{k-\rho_0}{2}\right)^2 } {\Gamma(k-\rho_0)}}\cdot \Gamma(k-\rho_0)\\
&=& \mathrm{Res}_{k=\rho_0}{\textstyle \Gamma\left(\frac{k-\rho_0}{2}\right)^2}
\ =\  \mathrm{Res}_{k=0}{\textstyle \Gamma\left(\frac{k}{2}\right)^2}
\ =\ 2\mathrm{Res}_{k=0}\Gamma(k)^2
\\&=& 2 \frac{d}{dk}|_{k=0} \frac{k^2}{2}\frac{\Gamma(k+1)^2}{k^2}
\ =\  2\Gamma'(1)
\ =\ -2C_\gamma,
\end{eqnarray*}
where $C_\gamma$ is the Euler-Mascheroni constant, see \cite[Kap.IV]{FB}. Also
$$ \mathrm{Res}_{k=\rho_0}B(k-\rho_0,\rho_0)=\mathrm{Res}_{k=0}=1 .$$
Thus,

\begin{eqnarray*}
&&\mathrm{Res}_{k=\rho_0}\mathcal{R}(k;\varphi_n)\\
&=& \omega_{\ell-1}\sum_{s}C_{n,s}{\textstyle B\left(\frac{i\lambda_s}{2},\frac{-i\lambda_s}{2}\right)}\langle\varphi_n,\mathrm{PS}_{\varphi_s}\rangle_{SX_\Gamma}2^{\rho_0-2}
\\&&+ \omega_{\ell-1}\sum_r C_{n,r}\langle \varphi_n,\mathrm{PS}_{\varphi_r}\rangle_{SX_\Gamma}\cdot \left(\ln(2)2^{\rho_0-3}-2^{\rho_0-1}C_\gamma\right)
\\&&+2^{\rho_0-2}\omega_{\ell-1}^2\sum_{j=j_0+1}^\infty {\textstyle B\left(\frac{i\lambda_j}{2},\frac{-i\lambda_j}{2}\right)}\langle \varphi_n,\mathrm{PS}_{\varphi_j}\rangle_{SX_\Gamma}\cdot I(a,b,\rho_0,\rho_0+i\lambda_j)
\\
&=&\omega_{\ell-1}\sum_{s}C_{n,s}{\textstyle B\left(\frac{i\lambda_s}{2},\frac{-i\lambda_s}{2}\right)}\langle\varphi_n,\mathrm{PS}_{\varphi_s}\rangle_{SX_\Gamma}2^{\rho_0-2}
\\&&+\omega_{\ell-1}\left(\ln(2)2^{\rho_0-3}-2^{\rho_0-1}C_\gamma\right)\sum_r C_{n,r}\langle\varphi_n,\mathrm{PS}_{\varphi_r}\rangle_{SX_\Gamma)}\\
\\&&+ 2^{\rho_0-2}\omega_{\ell-1}^2\sum_{j=j_0+1}^\infty {\textstyle B\left(\frac{i\lambda_j}{2},\frac{-i\lambda_j}{2}\right)}\langle \varphi_n,\mathrm{PS}_{\varphi_j}\rangle_{SX_\Gamma}\cdot I(a,b,\rho_0,\rho_0+i\lambda_j)\\
&=:&(III).
\end{eqnarray*}

Finally, if $\lambda_{j_0}\neq 0$, then $\mathcal R(k;\varphi_n)$ has still a pole at $k=\rho_0$ coming from the term $B(k-\rho_0,\rho_0)$. The residue $\mathrm{Res}_{k=\rho_0}\mathcal{R}(k;\varphi_n)$ is
\begin{eqnarray*}
&&\mathrm{Res}_{k=\rho_0} \mathcal{R}(k;\varphi_n)\\
&=& 2^{\rho_0-2}\omega_{\ell-1}\Big(\sum_{j=1}^{j_0}C_{n,j}
{\textstyle B\left(\frac{i\lambda_j}{2},\frac{-i\lambda_j}{2}\right)}
\langle \varphi_n,\mathrm{PS}_{\varphi_j}\rangle_{SX_\Gamma}\\
&&+\omega_{\ell-1}\sum_{j=j_0+1}^\infty {\textstyle B\left(\frac{i\lambda_j}{2},\frac{-i\lambda_j}{2}\right)}
\langle \varphi_n,\mathrm{PS}_{\varphi_j}\rangle_{SX_\Gamma} I(a,b,\rho_0,\rho_0+i\lambda_j)\Big) \mathrm{Res}_{k=\rho_0}B(k-\rho_0,\rho_0)\\
&=& 2^{\rho_0-2}\omega_{\ell-1}
\Big(\sum_{j=1}^{j_0}C_{n,j}
{\textstyle B\left(\frac{i\lambda_j}{2},\frac{-i\lambda_j}{2}\right)}
\langle \varphi_n,\mathrm{PS}_{\varphi_j}\rangle_{SX_\Gamma}\\
&&+\omega_{\ell-1}\sum_{j=j_0+1}^\infty
{\textstyle B\left(\frac{i\lambda_j}{2},\frac{-i\lambda_j}{2}\right)}
\langle \varphi_n,\mathrm{PS}_{\varphi_j}\rangle_{SX_\Gamma}
I(a,b,\rho_0,\rho_0+i\lambda_j)\Big)
\\
&=:&(IV),
\end{eqnarray*}
since $\mathrm{Res}_{k=\rho_0}B(k-\rho_0,\rho_0)=1$.

\section{The meromorphic continuation of $\mathcal Z(\varphi)$}\label{sec: mcontrz}
In this short section we show how the meromorphic continuation of $\mathcal{Z}(k;\varphi)$ is deduced from the one of $\mathcal R(k;\varphi_n)$.

For the zeta function $ \mathcal{Z}(k;\varphi_n) $ we note that the poles of $\mathcal{R}(k+2m;\varphi_n)$ are the $-2m$-shifted poles of $\mathcal{R}(k;\varphi_n)$. Now
\begin{equation}\label{eq: mcont}
\mathcal{Z}(k;\varphi_n)=\sum_{m=0}^\infty \beta(k-\rho_0;m) \mathcal{R}(k+2m;\varphi_n)=\beta(k-\rho_0;0)\mathcal{R}(k;\varphi_n)+\sum_{m=1}^\infty \beta(k-\rho_0;m) \mathcal{R}(k+2m;\varphi_n),
\end{equation}
see Theorem \ref{th: zeta}. 

We also know from Lemma \ref{binom} that $\beta(k-\rho_0;m)\to 0$ for $m\to \infty$ and that $k\mapsto \beta(k-\rho_0;m)$ is holomorphic for any $m\in \mathbb{N}_0$. Furthermore, \[\mathcal{R}(k;\varphi_n)= \sum_{1\neq[\gamma]\in C\Gamma}\sum_{\pi\in\widehat M}c(\varphi_n,\gamma,\pi,k) (\cosh L_\gamma)^{-k+\rho_0}\] and by Lemma \ref{auxi} we know that \[\sum_{\pi\in\widehat M}c(\varphi_n,\gamma,\pi,k)\] is absolutely bounded by \[C(\varphi_n)\cdot L_\gamma e^{-\rho_0 L_\gamma},\] where $L_\gamma$ is the length of the closed geodesic $[\gamma]$ and \[C(\varphi_n):=\frac{\omega_{l-1}\max_{X_\Gamma}\{|\varphi_n|\}}{2}B(\rho_0,\rho_0)\] 
if $\mathrm{Re}(k)>2\rho_0$.
It follows that for $\mathrm{Re}(k)>2\rho_0$
\begin{eqnarray*}
\left|\mathcal{R}(k;\varphi_n)\right|&\leq& \sum_{1\neq[\gamma]\in C\Gamma}\left|\sum_{\pi\in\widehat M}c(\varphi_n,\gamma,\pi,k)\right| (\cosh L_\gamma)^{-\mathrm{Re}(k)+\rho_0}
\\ &\leq & C(\varphi_n)\sum_{1\neq [\gamma]\in C\Gamma} e^{-\rho_0L_\gamma }L_\gamma (\cosh L_\gamma)^{-\mathrm{Re}(k)+\rho_0}
\\ &\leq& C(\varphi_n) (\cosh L_{\mathrm{inf}})^{-\mathrm{Re}(k)+\rho_0} \sum_{1\neq [\gamma]\in C\Gamma} e^{-\rho_0L_\gamma }L_\gamma
\end{eqnarray*}
for the constant $C(\varphi_n)$ depending only on $\varphi_n$. Here \index{$L_{\mathrm{inf}}$} $L_{\mathrm{inf}}=\sqrt{2(l-1)}^{-1}l_{\mathrm{inf}}$, where $l_{\mathrm{inf}}>0$ is the infimum of the set $\{l_\gamma:1\neq [\gamma]\in C\Gamma\}$, see Proposition \ref{prop: inflspec}. Thus \[\left|\mathcal{R}(k;\varphi_n)\right|\leq C\cdot (\cosh L_{\mathrm{inf}})^{-\mathrm{Re}(k)}\] for some constant $C>0$, since \[C(\varphi_n)(\cosh L_{\mathrm{inf}})^{\rho_0} \sum_{1\neq [\gamma]\in C\Gamma} e^{-\rho_0L_\gamma }L_\gamma\] is bounded. 

As a consequence, $\mathcal{R}(k;\varphi_n)=O(e^{-\mathrm{Re}(k)})$ for $\mathrm{Re}(k)>2\rho_0$. In particular, $\mathcal{R}(k+2m;\varphi_n)$ decays like $e^{-\mathrm{Re}(k)+2m}$ for $m\in\mathbb N$ and $m\to\infty$. Hence, (\ref{eq: mcont}) converges for $k$ away from the poles and defines there a meromorphic continuation of the zeta function $ \mathcal{Z}(k;\varphi_n)$. Since in the strip \[\mathcal{S}=\rho_0-\frac{1}{2} <\mathrm{Re}(k)<\rho_0+\frac{1}{2} \] the only poles of $\mathcal R(k;\varphi_n)$ are at \[\{\rho_0,\rho_0\pm i\lambda_j:-(\lambda_j+\rho_0^2) \mbox{ eigenvalue of the Laplacian}\}\cap \mathcal S,\] only the $m=0$-term contributes to poles of $\mathcal{Z}(k;\varphi_n)$ in $\mathcal S$. The residues/poles in the strip $\mathcal S$ of $\mathcal{Z}(k;\varphi_n)$ are hence the same as the ones of $\mathcal{R}(k;\varphi_n)$ modulo $\beta(k-\rho_0;0)$ which equals $2^{\rho_0-k}$ by Lemma \ref{binom}. More precisely, we have to distinguish again the 4 cases from (I) to (IV) and to note that \[\mathrm{Res}_{z=k}\mathcal{Z}(z;\varphi_n)=\beta(k-\rho_0;0)\mathrm{Res}_{z=k}\mathcal{R}(z;\varphi_n)=2^{\rho_0-k}\mathrm{Res}_{z=k}\mathcal{R}(z;\varphi_n)\] for $k\in\mathbb C$ with $\rho_0-\frac{1}{2} <\mathrm{Re}(k)<\rho_0+\frac{1}{2}$.

\section{Summary}\label{sec: msum}
We collect the results of the previous sections:

\begin{theorem}\label{th: merocont}
Let $\varphi_n$ be a non-constant automorphic eigenfunction. The zeta functions $\mathcal{R}(k;\varphi_n)$ and $\mathcal{Z}(k;\varphi_n)$ can be extended meromorphically to $\mathbb C$ by $(*)$ resp. (\ref{eq: mcont}). In the strip $\rho_0-\frac{1}{2} <\mathrm{Re}(k)<\rho_0+\frac{1}{2}$ the only possible poles are at $k\in\{\rho_0,\rho_0\pm i\lambda_j:-\left(\lambda_j^2+\rho_0^2\right) \mbox{ eigenvalue of the Laplacian} \}$. If the eigenvalue $-\left(\lambda_j^2+\rho_0^2\right)$ lies in the principal series and $(I)\neq 0$, there is a pole of order 1 at $k=\rho_0\pm i\lambda_j$, the residue of $\mathcal{R}(k;\varphi_n)$ is then given by (I). If the eigenvalue comes from the complementary series, $\lambda_j\neq 0$ and $(II)\neq 0$, then there is also a pole of order 1 at $k=\rho_0\pm i\lambda_j$ and residue $(II)$. If $\lambda_j=0$, there is a pole of order $2$ and residue determined by $(III)$. If $(I)$ or $(II)$ vanishes, then $\mathcal{R}(\varphi_n)$ can be holomorphically extended to $k=\rho_0\pm i\lambda_j$. The residues of $\mathcal{Z}(\varphi_n)$ are the same modulo $\beta(k-\rho_0;0)=2^{\rho_0-k}$. 
\end{theorem}

\section{Normalization of $\mathcal Z(\varphi)$}\label{sec: mnorm}
In this chapter we have shown so far that $\mathcal{Z}(\varphi_n)$ is a meromorphic function on $\mathbb C$. Let $k_0:=\rho_0+i\lambda$, where $-(\lambda^2+\rho_0^2)$ is an eigenvalue of the Laplacian from the principal series. Then $\mathcal{Z}(\varphi_n)$  has a simple pole at $k_0$ with residue given - up to a non-zero constant - by normalized Patterson-Sullivan distributions, see Equation $(I)$ in Section \ref{sec: mero1}. This constant equals \[\omega_{l-1}I(a,b,\rho_0,\rho_0+i\lambda).\]

To be consistent with \cite{AZ} we can divide $\mathcal Z(\varphi_n)$ by this constant, i.e. we consider the function \begin{equation}\label{eq: normzetaf}k\mapsto \frac{\mathcal{Z}(\varphi_n,k)}{\omega_{l-1}I(a,b,2^{-1},k)}.\end{equation}

Now \[k\mapsto\frac{1}{\omega_{l-1}I(a,b,\rho_0,k)}\overset{\text{Lem. \ref{lem: inthypg}}}=\frac{2\Gamma(k)^2}{\omega_{l-1}\Gamma(k-a)\Gamma(k-b)\Gamma(\rho_0)}\] is a meromorphic function on $\mathbb C$ with poles exactly in $\{0,-1,-2,\ldots\}$. Hence (\ref{eq: normzetaf}) defines also a meromorphic function on $\mathbb C$. In the strip $\rho_0-\frac{1}{2}<\mathrm{Re}(k)<\rho_0+\frac{1}{2}$ the poles of (\ref{eq: normzetaf}) equal the poles of $\mathcal Z(\varphi_n)$. The residue at $k_0$, $\rho_0+\frac{1}{2}<\mathrm{Re}(k_0)<\rho_0-\frac{1}{2}$, is 
\[\mathrm{Res}_{k=k_0}\frac{\mathcal Z(\varphi_n,k)}{\omega_{l-1}I(a,b,\rho_0,k)}=\frac{1}{\omega_{l-1} I(a,b,\rho_0,k_0)}\mathrm{Res}_{k=k_0}\mathcal{Z}(\varphi_n,k).\]

In particular, if $k_0=\rho_0+i\lambda$ comes from the principal series, we deduce from equation $(I)$ that 
\[\mathrm{Res}_{k=k_0}\frac{\mathcal Z(\varphi_n,k)}{\omega_{l-1}I(a,b,\rho_0,k)}=2^{\rho_0-1}\sum_r \langle \varphi_n,\widehat{\mathrm{PS}}_{\varphi_r}\rangle,\] 
where as before the finite sum runs over all eigenfunctions $\varphi_r$ with eigenvalue $-(\lambda^2+\rho_0^2)$. We can of course also  normalize $\mathcal R(\varphi_n)$ by the same constant $\left(\omega_{l-1}I(a,b,\rho_0,k)\right)^{-1}$.

\section{Comparison with the zeta function from \cite{AZ}}\label{sec: out}

Let $G=SO_o(1,2)$. In this section we want to discuss a difference between the zeta function \cite[(1.9ii)]{AZ} and our zeta function, see (\ref{def: zetfunc}), in the case of a compact hyperbolic surface $X_\Gamma =\Gamma\backslash SO_o(1,2)/SO(2)$. In \cite[(1.9ii)]{AZ} we find as a definition
\begin{equation*}
\mathcal{Z}_1(k;\varphi_n):= \sum_{1\neq [\gamma]\in C\Gamma} \frac{e^{-kL_\gamma}}{1-e^{-L_\gamma}}\left(\int_{c_{\gamma_0}}\varphi_n\right)= \sum_{1\neq [\gamma]\in C\Gamma} \frac{e^{-(k-\frac{1}{2} )L_\gamma}}{2\sinh \frac{L_\gamma}{2}}\left(\int_{c_{\gamma_0}}\varphi_n\right)
\end{equation*}

After normalizing our zeta function $\mathcal Z(k;\varphi_j)=\sum_{1\neq [\gamma]\in C\Gamma}c(\varphi_j,\gamma,\mathbf{1},k)e^{-(k-1/2)L_\gamma}$ from Proposition \ref{prop: simple} by $\left(4I(a,b,1/2,k)\right)^{-1}$ the term   $\left(2I(a,b,1/2,k)\right)^{-1} c(\varphi_j,\gamma,\mathbf{1},k)$ takes the form, see equation (\ref{eq: coefsimex}),
\begin{eqnarray*}
&&\left(4I(a,b,1/2,k)\right)^{-1} \cdot c(\varphi_j,\gamma,\mathbf{1},k)\\&\overset{\text{Lem. } \ref{lem: inthypg} }=& \frac{2\Gamma(k)^2}{4\Gamma(\frac{1}{2} )\Gamma(k-a)\Gamma(k-b)}\cdot \sum_{1\neq [\gamma]\in C\Gamma} \frac{1}{\sqrt{\cosh L_\gamma-1}} \sqrt{2}\left(\int_{c_{\gamma_0}}\varphi_n\right)\\&&\cdot \left(\frac{\cosh L_\gamma}{\cosh L_\gamma-1}\right)^{k-1 }{}_2F_1\left(k-a,k-b,k;1-\frac{\cosh L_\gamma}{\cosh L_\gamma- 1} \right)\\
&&\cdot \frac{\Gamma(\frac{1}{2})\Gamma(k-a)\Gamma(k-b)}{\Gamma(k)^2}\\
&=& \sum_{1\neq [\gamma]\in C\Gamma} \frac{\sqrt{2}}{2\sqrt{2}\sinh \frac{L_\gamma}{2} }\left(\int_{c_{\gamma_0}}\varphi_n\right)\\&&\cdot \left(\frac{\cosh L_\gamma}{\cosh L_\gamma-1}\right)^{k-1 }{}_2F_1\left(k-a,k-b,k;1-\frac{\cosh L_\gamma}{\cosh L_\gamma- 1} \right).
\end{eqnarray*}

Hence,
\begin{eqnarray*}
\frac{1}{4 I(a,b,1/2,k)}\mathcal{Z}(k;\varphi_n)&=& \sum_{1\neq [\gamma]\in C\Gamma} \frac{e^{-(k-\frac{1}{2} )L_\gamma}}{2\sinh \frac{L_\gamma}{2} } \left(\int_{c_{\gamma_0}}\varphi_n\right)\\&&\cdot \left(\frac{\cosh L_\gamma}{\cosh L_\gamma-1}\right)^{k-1 } {}_2F_1\left(k-a,k-b,k;1-\frac{\cosh L_\gamma}{\cosh L_\gamma- 1} \right).
\end{eqnarray*}

Here $a=\frac{1}{2}(\rho_0+ir_n)$ and $b=\frac{1}{2}(\rho_0-ir_n)$, if the eigenvalue of $\varphi_n$ is $-\frac{1}{4\rho_0} (\rho_0^2+r_n^2)$. The difference $\mathcal{Z}_1(k;\varphi_n)-\frac{1}{4 I(a,b,1/2,k)}\mathcal{Z}(k;\varphi_n)$ thus equals on $\{k\in \mathbb{C}:\mathrm{Re}(k)>1\}$
\index{$\mathcal D(k;\varphi_n)$, holomorphic function on $\{\mathrm{Re}(k)>0\}$}
\begin{eqnarray*}
\mathcal{D}(k;\varphi_n)&:=&\sum_{1\neq [\gamma]\in C\Gamma} \frac{e^{-(k-\frac{1}{2} )L_\gamma}}{ 2\sinh \frac{L_\gamma}{2} } \left(\int_{c_{\gamma_0}}\varphi_n\right)
\\ &&\cdot \left(1- \left(\frac{\cosh L_\gamma}{\cosh L_\gamma-1}\right)^{k-1 }{}_2F_1\left(k-a,k-b,k;1-\frac{\cosh L_\gamma}{\cosh L_\gamma- 1}\right)\right).
\end{eqnarray*}

On $\{k\in \mathbb{C}:\mathrm{Re}(k)>1\}$, the function $\mathcal D(\varphi_n)$ is holomorphic as it is the difference of two holomorphic functions. Furthermore, it is also holomorphic on the half plane $\{k\in\mathbb C: \mathrm{Re}(k)>0\}$, because $\left(\frac{\cosh L_\gamma}{\cosh L_\gamma-1}\right)^{k-1 }{}_2F_1\left(k-a,k-b,k;1-\frac{\cosh L_\gamma}{\cosh L_\gamma- 1}\right)$ tends to 1 as $L_\gamma\to \infty$, the argument $1-\frac{\cosh L_\gamma}{\cosh L_\gamma- 1}$ lies always in the interval $(-\infty,1)$ and $k\mapsto {}_2F_1(k-a,k-b,k;x)$ is holomorphic on $\{k\in\mathbb C:k\neq 0,-1,-2,\ldots\}$ for fixed $x\in (-\infty,1)$, see \cite[Th. 9.1]{Ol}. It follows that the term \[\left(1- \left(\frac{\cosh L_\gamma}{\cosh L_\gamma-1}\right)^{k-1 }{}_2F_1\left(k-a,k-b,k;1-\frac{\cosh L_\gamma}{\cosh L_\gamma- 1}\right)\right)\] is bounded for all $L_\gamma$. Also $C\cdot \sum_{n\in\mathbb N} \frac{n}{1-e^{-n}}e^{-kn}$ is a summable upper bound to \[\sum_{1\neq [\gamma]\in C\Gamma} \frac{e^{-(k-\frac{1}{2} )L_\gamma}}{ 2\sinh \frac{L_\gamma}{2} } \left(\int_{c_{\gamma_0}}\varphi_n\right)=\sum_{1\neq [\gamma]\in C\Gamma} \frac{e^{-kL_\gamma}}{1-e^{-L_\gamma}}\left(\int_{c_{\gamma_0}}\varphi_n\right)\] for a suitable constant $C>0$.  

Thus, \[\frac{1}{4I(a,b,1/2,k)}\mathcal{Z}(k;\varphi_n)+\mathcal{D}(k;\varphi_n)\] defines a meromorphic continuation of $\mathcal Z_1(k;\varphi_n)$ to the half plane $\{k\in\mathbb C:\mathrm{Re}(k)>0\}$ with the same poles and residues as $\frac{1}{4 I(a,b,1/2,k)}\mathcal{Z}(k;\varphi_n)$.


\chapter{The zeta function on the spherical spectrum}
Let $G=SO_o(1,l)=ANK$ as before. Also let $\Gamma\subset G$ be a uniform lattice. We will use a decomposition of the right-regular representation $\pi_R$ on $L^2(\Gamma\backslash G)$ in order to extend the definition to zeta function $\mathcal Z(\sigma)$, where $\sigma$ is more general than an automorphic eigenfunction of the Laplacian. More precisely, $\sigma$ will be an element of phase space $C^\infty(\Gamma\backslash G/M)$, $M=Z_K(A)$, which satisfies a certain finiteness conditions explained in Section \ref{chap: class1}. In Section \ref{sec: out} we indicate an approach for general $\sigma\in C^\infty(SX_\Gamma)$.

The case of $SO_o(1,2)$ is dealt with in \cite{AZ} and we will basically omit this case as 
In this chapter we want to show how one can pass in Theorem \ref{th: zeta} and \ref{th: conauxzeta}, resp. Theorem \ref{thm: stracspec} and \ref{th: merocont} from automorphic eigenfunctions $\varphi_n$ to $\sigma\in C^\infty(\Gamma\backslash G/M)$ lying in the class 1 spectrum with only finitely many nontrivial components. 
\section{Extension to the spherical spectrum}\label{chap: class1}

To give a dynamical interpretation of Patterson-Sullivan distributions, one needs to pass from automorphic eigenfunctions $\varphi_n$ to arbitrary $\sigma \in C^\infty(\Gamma \backslash G/M)$. We make use of the decomposition of the (right-)regular representation $\pi_R$ of $G$ on $L^2(\Gamma \backslash G)$. So let $G=SO_o(1,l)$, $l>2$, and $\Gamma$ a uniform lattice. Since $\Gamma$ is assumed to be co-compact, the representation decomposes discretely into a direct sum of at most 2 different types, see \cite[Th. 2.7.]{Wi},
\begin{equation}\label{eq: l2dec}
L^2(\Gamma \backslash G)= \bigoplus_{\text{spherical}} V \oplus \bigoplus_{\text{non-spherical}} W. 
\end{equation}

The \index{spherical representation} spherical part is called the \index{spherical spectrum} \textit{spherical spectrum} or also \index{class 1 spectrum} \textit{class 1 spectrum}. Here we call a representation \textit{spherical}, if it possesses a non-trivial $K$-invariant vector. 

Therefore, \index{$V^M$, $M$-invariant vectors in $V$}

\begin{equation}\label{eq: l2decm}
L^2(\Gamma\backslash G)^M=L^2(\Gamma \backslash G/M)= \bigoplus_{\text{spherical}} V^M \oplus \bigoplus_{\text{non-spherical}} W^M . 
\end{equation}

In this section we want to show how one can pass in Theorem \ref{th: zeta} and \ref{th: conauxzeta}, resp. Theorem \ref{thm: stracspec} and \ref{th: merocont} from automorphic eigenfunctions $\varphi_n$ to $\sigma\in C^\infty(\Gamma\backslash G/M)$ lying in the class 1 spectrum with only finitely many nontrivial components. That is, \[\sigma=\sum_{\text{finite}} \sigma_n,\] where each $\sigma_n$ lies in some irreducible spherical (principal series) component $V$ and is right-$M$-invariant. 

We show now that the subspace of $M$-invariant functions in each spherical principal series representation is generated by the spherical vector if we restrict the representation from $G$ to $A$.

\begin{lemma}\label{lem: sphcy}
Let $G=SO_o(1,l)$, $l>2$, and $V$ be an irreducible spherical representation with spherical vector $\varphi$. Then $A$ acts on $V^M$ with cyclic vector $\varphi$, i.e. $U(\mathfrak a)\varphi$ is dense in $V^M$. If $G=SO_o(1,2)$, then \[U(\mathfrak a)\varphi\oplus U(\mathfrak a)X\varphi\] is dense in $V^M=V$, where $ \mathfrak{n}=\mathbb{R}X $.  
\end{lemma}

\begin{proof}

Since $V$ is assumed to be irreducible we have (up to completion)
\begin{equation*}
V^M=(U(\mathfrak g)\varphi)^M=U(\mathfrak{g})^M\varphi ,
\end{equation*}
 since we can average over the compact group $M$, i.e. for $X\in \mathfrak{g}$ we can define the $M$-average $\tilde X$ by 
\begin{equation*}
\tilde{X}f(g):=\int_M \frac{d}{dt}|_{t=0}f(g\exp \mathrm{Ad}(m)tX )dm.
\end{equation*}
Then $\tilde Xf$ is $M$-invariant, if $f$ is. Now
\begin{eqnarray*}
U(\mathfrak{g} )^M\varphi&=& U(\mathfrak{a}\oplus \mathfrak{n} )^M\varphi\\
&=& U(\mathfrak{a} )U(\mathfrak{n} )^M\varphi.
\end{eqnarray*}

Since $l>2$, $M$ is not trivial and it follows that $U(\mathfrak{n})^M$ is generated by the Laplacian $\sum_i X_i^2$ on $\mathfrak n$, see \cite[Prop. 4.11.]{GGA}. Further, we know by Chapter \ref{chap: casimir} 
\begin{equation*}
\sum_i X_i^2=\frac{1}{2}\left(-\Omega-H^2+2H_\rho)\mbox{ mod }  \mathfrak{k}U(\mathfrak{g}\right).
\end{equation*}

Because the space of $K$-invariant elements in $V$ is one dimensional and because $\Omega\varphi$ is also $K$-invariant, we deduce that $\varphi$ is a Casimir eigenfunction, let us assume $\Omega\varphi=\mu\varphi$. Hence, the effect of applying $\sum_i X_i^2$ to $\varphi$ can be expressed by elements in $U(\mathfrak{a})$, i.e.
\begin{equation*}
\left(\sum_i X_i^2\right)\varphi=\frac{1}{2}\left(-\mu -H^2+2H_\rho \right)\varphi,
\end{equation*}
where $\mu$ is the eigenvalue of $\varphi$.

If $l=2$, i.e. $M$ is trivial, we just note that $U(\mathfrak{n} )^M=U(\mathfrak{n} )$ is generated by $X$, since $ \mathfrak{n}=\mathbb{R}X $.
\end{proof}

If we assume $l>2$ and $V$ is some irreducible representation in the class 1 spectrum in (\ref{eq: l2dec}), the lemma shows that restricting the irreducible $G$-representation $\pi_R$ on $V$ from $G$ to $A$ yields an $A$-representation on $V^M$ with cyclic vector $\varphi$, if $\varphi$ is the (normalized) $K$-fixed vector in $V$. It follows that the induced representation of $L^1(A)$ on $V^M$ is also cyclic with cyclic vector $\varphi$, \cite[13.3.5]{Dix}.  

Thus, if $\sigma_n$ is an $M$-invariant element of the irreducible spherical component $V$ with (normalized) $K$-fixed vector $\varphi_n$, then there is some \index{$alpha_n$@$\alpha_n$, coefficient function in $L^1(A)$} $\alpha_{n}\in L^1(A)$ such that \begin{equation}\label{eq: sdec}\sigma_n=\int_A\alpha_{n}(a)\pi_R(a) \varphi_n da.
\end{equation}

By invariance under the geodesic flow for any $a\in A$ and any $\mathrm{PS}_{\varphi_j}$
$$\langle \pi_R(a) \varphi_n,\mathrm{PS}_{\varphi_j}\rangle_{SX_\Gamma}=\langle  \varphi_n,\mathrm{PS}_{\varphi_j}\rangle_{SX_\Gamma}.$$  

Hence, as $\mathrm{PS}_{\varphi_j}$ is a continuous functional on $C^\infty(SX_\Gamma)$, see Remark \ref{rem: pscontd}, and as $SX_\Gamma$ is compact
\begin{eqnarray}\label{eq: dsdec}\langle \sigma_n,\mathrm{PS}_{\varphi_j}\rangle_{SX_\Gamma}&=& \left\langle\int_A\alpha_{n}\pi_R(a) \varphi_n da,\mathrm{PS}_{\varphi_j}\right\rangle_{SX_\Gamma}\nonumber\\
&=& \int_A\alpha_n(a)\langle \pi_R(a)\varphi_n,\mathrm{PS}_{\varphi_j}\rangle_{SX_\Gamma} da 
  \\&=&\int_A\alpha_{n}(a)\langle  \varphi_n,\mathrm{PS}_{\varphi_j}\rangle_{SX_\Gamma}da\nonumber\\
&=& \int_A\alpha_n(a)da \cdot \langle  \varphi_n,\mathrm{PS}_{\varphi_j}\rangle_{SX_\Gamma}  .\nonumber\end{eqnarray}

That is for any $\mathrm{PS}_{\varphi_j}$, the value $\langle \varphi_n,\mathrm{PS}_{\varphi_j}\rangle_{SX_\Gamma}$ essentially determines $\mathrm{PS}_{\varphi_j}$ on the whole $M$-invariant part $V^M$ of the irreducible spherical component $V$ associated with $\varphi_n$. We also remark that (\ref{eq: dsdec}) remains true if we replace $\mathrm{PS}_{\varphi_j}$ by $\widehat{\mathrm{PS}}_{\varphi_j}$.
 
For $G=SO_o(1,2)$ see \cite[Th. 9.6.]{AZ}. Here, for spherical irreducible components $V$ of $L^2(\Gamma\backslash SO_o(1,2))$ one basically needs knowledge of $\mathrm{PS}_{\varphi}$ on $\varphi_n$ and $X\varphi_n$, if $\varphi_n$ is the (normalized) spherical vector of $V$. If $V$ is not spherical, in other words, if $V$ is a discrete series representation, then $\mathrm{PS}_{\varphi}$ on $V$ is determined by its value on the lowest weight vector.

Going back to $G=SO_o(1,l)$, $l>2$, we define for $\sigma=\sum^{\text{finite}}_n \sigma_n=\sum_n^{\text{finite}}\int_A\alpha_{n}(a)\pi_R(a)\varphi_n da$, see (\ref{eq: sdec}), \index{$\mathcal Z(\sigma)$, zeta function ass. with $\sigma$}
\begin{equation}\label{def: zetagen}
\mathcal{Z}(\sigma):=\sum^{\text{finite}}_n\int_A\alpha_n(a)da\cdot \mathcal{Z}(\varphi_n)=\sum_{1\neq[\gamma]\in C\Gamma}\sum_{\pi\in\widehat{M}}c(\gamma,\sigma,\pi,k) e^{-(k-\rho_0)\log a_\gamma},
\end{equation} if $\varphi_n$ is the normalized $K$-fixed vector of the component of $\sigma_n$ in $L^2(\Gamma\backslash G)$, see $(\ref{eq: l2dec})$. Here \[c(\gamma,\sigma,\pi,k):=\sum_n^{\text{finite}}\int_A\alpha_n(a)da\cdot c(\gamma,\varphi_n,\pi,k).\]

Now the analogue of Theorem  \ref{th: merocont} holds.

\begin{proposition}\label{th: geodirrnonsphere}
Let $\sigma$ be a function in $C^\infty(\Gamma\backslash G/M)$ with only finitely many nontrivial components in the spherical spectrum (\ref{eq: l2dec}) and no component in the non spherical spectrum. The zeta function $\mathcal{Z}(\sigma)$ defined by (\ref{def: zetagen}) is a meromorphic function on $\mathbb{C}$. In the strip $\mathcal S=\{k\in \mathbb{C}:\rho_0-\frac{1}{2} <\mathrm{Re}(k)<\rho_0+\frac{1}{2}\}$, the poles are at \[\{\rho_0,\rho_0\pm i\lambda: -(\rho_0+\lambda^2) \mbox{ eigenvalue of the Laplacian}\}\cap \mathcal S.\] 

If $-(\lambda^2+\rho_0^2)$ is an eigenvalue of the Laplacian from the principal series, the residue at $k=\rho_0+i\lambda$ is up to the non-zero constant $\omega_{l-1}I(a,b,\rho_0,\rho_0+i\lambda)$ \begin{equation}\label{eq: sumar}
\sum_{\text{finite}} \langle \sigma,\widehat{PS}_\varphi\rangle,
\end{equation}
where this finite sum ranges over all normalized Patterson-Sullivan distributions $\widehat{PS}_\varphi$ associated to Laplace eigenfunctions with eigenvalue $-(\lambda^2+\rho_0^2)$. 
\end{proposition}

\begin{proof}
The meromorphic continuation of $\mathcal Z(\sigma)$ follows directly from Theorem \ref{th: merocont} and the fact that $\sigma$ has only finitely many nontrivial components. In order to compute the residue at $k=\rho+i\lambda\in \mathcal{S}$ we recall that we can assume that there are $\alpha_{n}\in L^1(A)$ such that $$ \sigma=\sum_n^{\text{finite}}\sigma_n=\sum_n^{\text{finite}}\int_A\alpha_{n}(a)\pi_R(a)\varphi_nda,$$ see (\ref{eq: sdec}). This implies, see (\ref{eq: dsdec}), 
\begin{equation}\label{eq: dsdec2} \langle \sigma,\widehat{\mathrm{PS}}_{\varphi},\rangle_{SX_\Gamma}=\sum_n^{\text{finite}}\int_A\alpha_n(a)da\langle \varphi_n,\widehat{\mathrm{PS}}_\varphi\rangle_{SX_\Gamma}.\end{equation}   

For the residue at $k_0\in \mathcal{S}$ we have
$$\mathrm{Res}_{k=k_0}\mathcal{Z}(k;\sigma)=\sum_n^{\text{finite}}\int_A\alpha_n(a)da\mathrm{Res}_{k=k_0}\mathcal{Z}(k;\varphi_n) .$$

In particular, if $k=\rho_0+i\lambda\in\mathcal S$, $-(\lambda^2+\rho_0^2)$ an eigenvalue from the principal series, it follows from Theorem \ref{th: merocont} that 
\begin{eqnarray*}
\mathrm{Res}_{k=\rho_0+i\lambda} \mathcal{Z}(k;\sigma)&=& \sum_n^{\text{finite}} \int_A\alpha_n(a)da\omega_{l-1}I(a,b,\rho_0,\rho_0+i\lambda) \sum_{r:\lambda_r^2=\lambda^2 }^{\text{finite}} \langle \varphi_n,\widehat{\mathrm{PS}}_{\varphi_r}\rangle_{SX_\Gamma}\\
&=& \omega_{l-1}I(a,b,\rho_0,\rho_0+i\lambda)\sum_{r:\lambda_r^2=\lambda^2 }^{\text{finite}} \sum_n^{\text{finite}} \int_A\alpha_n(a)da \langle \varphi_n,\widehat{\mathrm{PS}}_{\varphi_r}\rangle_{SX_\Gamma}\\
&\overset{(\ref{eq: dsdec2})}=& \omega_{l-1}I(a,b,\rho_0,\rho_0+i\lambda)\sum_{r:\lambda_r^2=\lambda^2 }^{\text{finite}} \langle \sigma,\widehat{\mathrm{PS}}_{\varphi_r}\rangle_{SX_\Gamma}.
\end{eqnarray*}
\end{proof}

For the general case of a $K$-finite Casimir eigenfunction of type $\delta$, $\delta\in \widehat{K}$ we just state the following.

\begin{lemma}\label{th: geodirrsphere}
Let $l>2$ and $V$ be any irreducible representation of $G$ with $K$-finite $\Omega$-eigenfunction $v$ of type $\delta$, $(\delta,V_\delta)\in \widehat{K} $, i.e. $ \mathrm{span}\{K\cdot v\}=\mathrm{span}\{v_1,\ldots,v_{d_\delta}\}\cong V_\delta $ and $\Omega v=\mu v$ for some $\mu\in \mathbb{C}$. Then $\bigoplus_{i=1}^{d_\delta} U(\mathfrak{a} )v_i$ contains $V^M$.
\end{lemma} 
\begin{proof}
We use the Cartan decomposition $ \mathfrak{g}=\mathfrak{p}\oplus \mathfrak{k} $. Then by irreducibility $V=U(\mathfrak{g} )v=S(\mathfrak{p} )U(\mathfrak{k} )v$, where $ S(\mathfrak{p} ) $ is the symmetric algebra of $ \mathfrak{p} $. Since $v$ is $K$-finite of type $\delta$, we have that $U(\mathfrak{k} )v$ is contained in $ \mathrm{span}\{v_1,\ldots,v_{d_\delta}\} $, where the $v_i$ are $K$-translates which are $K$-finite of type $\delta$ and $\Omega$-eigenfunctions, $\Omega v_i=\mu_i v_i$, since $\Omega$ lies in the center of $ U(\mathfrak{g} ) $. Thus, 
\begin{equation*}
V^M\subset \bigoplus_{i} S(\mathfrak{p} )^M v_i. 
\end{equation*}

But $ S(\mathfrak{p} )^M=U(\mathfrak a)S( \mathfrak{p} )^K $, as $l>2$, where $ U(\mathfrak{p} )^K $ is generated by $\sum_j T_j^2$, $T_j$ forming an orthonormal basis of $ \mathfrak{p} $. Now $ \Omega $ can be expressed as $ \Omega=\sum_j T_j^2 -\sum_l W_t^2$, where the $W_t$ form an orthonormal base of $ \mathfrak{k} $. Hence for any $i$,
\begin{equation*}
\sum_j T_j^2v_i=\Omega v_i+\sum_l W_t^2v_i =\mu_i v_i+\sum_l W_t^2v_i\in \mathrm{span}\{v_1,\ldots, v_{d_\delta}\}.
\end{equation*}

It follows that 
\begin{equation*}
V^M\subset \bigoplus_i U(\mathfrak{a} )v_i.
\end{equation*}

\end{proof}

\section{Outlook}\label{sec: out}
It remains an open problem what are good choices for $\sigma$ coming from the $M$-invariant part $V^M$ of a non-spherical representation $V$ in the decomposition (\ref{eq: l2decm}), if $G=SO_o(1,l)$ with $l>2$. Lemma \ref{th: geodirrsphere} shows that if   $V$ has a $K$-finite $\Omega$-eigenfunction $v$ then $\mathrm{PS}_{\varphi_j}$ is  basically determined by its value on finitely many $K$-translates $v_i$ of $v$. Thus, one would like to associate a zeta function $\mathcal Z(v_i)$ to $v_i$ as one did in the case of an automorphic eigenfunction $\varphi$. Here for example the problem occurs how to associate an operator with $v_i$ which maps $L^2(X_\Gamma)$ into itself and how to compute its trace. For $G=SO_o(1,2)$, \cite{AZ} gives a solution to these problems, at least if $\sigma$ has only finitely many nontrivial components in the decomposition (\ref{eq: l2dec}).  

An approach would be to consider instead of $\sigma\in C^\infty(X_\Gamma)$ and $f_k\in C^\infty(G//K)$ sections of vector bundles $\Sigma\in C^\infty(X_\Gamma\times_K V_\delta)$ and $F_k\in C^\infty(X\times_K V_{\delta^*})$ for $K$-types $\delta$, $k\in\mathbb C$, so that $\Sigma\cdot \pi_R(F_k)$ maps $L^2(X_\Gamma)$  into itself and is an operator of trace class on $L^2(X_\Gamma)$.

\bibliography{Phd2}

\bibliographystyle{amsalpha}
\printindex
\end{document}